\documentclass[10pt,a4paper]{article}

\usepackage[english]{babel}

\usepackage{amsfonts,amssymb,amsmath,amsthm,upgreek}
\usepackage{mathtools,units,xspace}
\usepackage{mathrsfs,booktabs,fullpage,afterpage}
\usepackage{hyperref}

\usepackage[caption=false]{subfig}
\usepackage[pdftex]{graphicx}

\usepackage{algorithm, algorithmic}
\graphicspath{{fig/}{}}

\usepackage{tikz}
\usepackage{pgfplots}
\usepackage{tikz-3dplot,etoolbox}
\usetikzlibrary{arrows,shapes.geometric,calc}
\pgfplotsset{compat=newest}

\usepackage{tikz}
\usetikzlibrary{external}
\usepackage{environ,ifdraft}
\makeatletter
\providecommand{\env@tikzpicture@save@env}{}
\providecommand{\env@tikzpicture@process}{}
\makeatother

\ifdraft{%
\newcounter{tikzfigcntr}
\RenewEnviron{tikzpicture}[1][]{\null\hfill\stepcounter{tikzfigcntr} \fbox{This is \texttt{tikzpicture}~\thetikzfigcntr}\hfill\null}
}{}

\renewcommand{\appendix}{%
  \par
  \setcounter{section}{0}%
  \renewcommand{\thesection}{\Alph{section}}%
}

\PassOptionsToPackage{numbers}{natbib}

\newenvironment{talign}
 {\align}
 {\endalign}
\newenvironment{talign*}
 {\csname align*\endcsname}
 {\endalign}

\def\bs{\boldsymbol}
\def\bb{\mathbb}
\def\cl{\mathcal}
\def\sf{\mathsf}
\def\ts{\textstyle}

\DeclareMathOperator{\diag}{{diag}}

\DeclareMathOperator{\Ex}{{\bb {E}}}
\DeclareMathOperator{\argmin}{{arg min}}
\DeclareMathOperator{\tr}{tr}

\def\eg{\emph{e.g.},~}
\def\etal{{\emph{et al.}}}
\def\ie{\emph{i.e.},~}

\usetikzlibrary{arrows,shapes.geometric,calc,decorations.markings}
\definecolor{mathblue}{HTML}{3F3D9A}
\definecolor{mathpurp}{HTML}{9A3D71}
\definecolor{mathsand}{HTML}{9A8C3D}
\definecolor{mathgrn}{HTML}{0F9F4F}
\definecolor{mathgrnl}{HTML}{B7E2CA}
\definecolor{azu}{RGB}{128,179,255}
\definecolor{yel}{RGB}{255,235,154}
\definecolor{azuwhi}{RGB}{205,225,255}

\def\A{{\bs{A}}}
\def\a{\bs{a}}
\def\ail{\a_{i,l}}
\def\B{{\bb B}}
\def\bLambda{\bs\Lambda}
\def\bnu{\bs{\nu}}
\def\bxi{{\bs{\xi}}}

\def\eps{\bs{\varepsilon}}

\def\F{\bs{F}}
\def\H{\mathcal{H}}
\def\I{\bs{I}}
\def\onep{{{\bs 1}^\perp_m}}
\def\P{{\bs P}}
\def\Ponep{\P_{\onep}}
\def\Pro{\mathbb{P}}
\def\R{\mathbb{R}}
\def\u{{\bs{u}}}
\def\v{{\bs{v}}}
\def\vone{\bs{1}}
\def\vzer{\bs{0}}
\def\x{{\bs{x}}}
\def\y{{\bs{y}}}
\def\iid{\textrm{i.i.d.}\xspace}
\def\rv{{r.v.}\xspace}
\def\rvs{{r.v.}'s\xspace}
\def\rhs{{r.h.s.}\xspace}
\def\lhs{{l.h.s.}\xspace}
\newcommand{\scp}[2]{\langle #1,\,#2\rangle}

\newcommand{\new}[1]{#1}

\newtheorem*{theorem*}{Theorem}
\newtheorem{theorem}{Theorem}[section]
\newtheorem*{definition*}{Definition}
\newtheorem{definition}{Definition}[section]
\newtheorem*{hypothesis*}{Hypothesis}

\newtheorem*{proposition*}{Proposition}
\newtheorem{proposition}{Proposition}[section]
\newtheorem*{corollary*}{Corollary}
\newtheorem{corollary}{Corollary}[section]
\newtheorem*{lemma*}{Lemma}

\newtheorem*{remark*}{Remark}
\newtheorem{remark}{Remark}[section]

\newcommand{\tinv}[1]{\ts\tfrac{1}{#1}}

\numberwithin{equation}{section}

\pgfplotsset{
	compat=newest,
	legend style = {font=\small},
}

\title{{\bf Through the Haze: a Non-Convex Approach to \\ Blind \new{Gain} Calibration for Linear Random Sensing Models}}

\author{Valerio Cambareri$^{*}$ and Laurent Jacques\thanks{VC and LJ are with Image and Signal Processing Group
  (ISPGroup), ICTEAM/ELEN, Universit\'e catholique de Louvain (UCL).  E-mail:
\url{laurent.jacques@uclouvain.be},
\url{valerio.cambareri@uclouvain.be}. The authors are funded by the
Belgian F.R.S.-FNRS. Part of this study is funded by the project
\textsc{AlterSense} (MIS-FNRS).}}

\date{\today}

\begin{document}

\maketitle

\begin{abstract}
	{
		Computational sensing strategies often suffer from calibration errors in the physical implementation of their ideal sensing models. Such uncertainties are typically addressed by using multiple, accurately chosen training signals to recover the missing information on the sensing model, an approach that can be resource-consuming and cumbersome. Conversely, blind calibration does not employ any training signal, but corresponds to a bilinear inverse problem whose algorithmic solution is an open issue. 
		We here address blind calibration as a non-convex problem for linear random sensing models, in which we aim to recover an unknown signal from its projections on sub-Gaussian random vectors, each subject to an unknown \new{positive} multiplicative factor (or gain). To solve this optimisation problem we resort to projected gradient descent starting from a suitable, carefully chosen initialisation point. An analysis of this algorithm allows us to show that it converges to the exact solution provided a sample complexity requirement is met, \ie relating convergence to the amount of information collected during the sensing process. Interestingly, we show that this requirement grows linearly (up to $\log$ factors) in the number of unknowns of the problem. This sample complexity is found both in absence of prior information, as well as when subspace priors are available for both the signal and gains, allowing a further reduction of the number of observations required for our recovery guarantees to hold. Moreover, in the presence of noise we show how our descent algorithm yields a solution whose accuracy degrades gracefully with the amount of noise affecting the measurements. 
		Finally, we present some numerical experiments in an imaging context, where our algorithm allows for a simple solution to blind calibration of the gains in a sensor array. 
	} \\
	
	{\noindent {\em Keywords: Blind calibration, non-convex optimisation, sample complexity, bilinear inverse problems.}
		%%%% If classification number provided then
		\\
		\em 2000 Math Subject Classification: 94A15, 94A20, 90C26,15A29.
	}
\end{abstract}

	\section{Introduction}
		\label{sec:introduction}

		The problem of recovering an unknown signal measured or transmitted by means of an inaccurate sensing model is of crucial importance for modern sensing strategies relying on the solution of inverse problems. In such problems, exact prior information on the sensing model is paramount to accurately reconstruct the original signal. 
		Compressed Sensing (CS)~\cite{Donoho2006} has emerged as a powerful framework to design new sensing strategies employing sub-Gaussian random matrix ensembles given their remarkable properties (see,  \eg~\cite{BaraniukDavenportDeVoreEtAl2008}). However, model errors inevitably affect its physical implementation and can significantly degrade signal recovery, as first studied by Herman and Strohmer~\cite{HermanStrohmer2010}. In particular, such model errors may arise from physical causes such as unknown convolution kernels~\cite{AhmedRechtRomberg2014,AhmedCosseDemanet2015,BahmaniRomberg2015a} affecting the measurements; unknown attenuations or gains on the latter coefficients, \eg pixel response non-uniformity~\cite{HayatTorresArmstrongEtAl1999} or fixed-pattern noise in imaging systems; complex-valued (\ie gain and phase) errors in sensor arrays~\cite{FriedlanderStrohmer2014,LingStrohmer2015a,BilenPuyGribonvalEtAl2014}.	
		
		Assuming such errors remain stationary throughout the sensing process, the use of linear random operators in CS does suggest that repeating the acquisition, \ie taking several {\em snapshots} under new independent draws of a random sensing operator could suffice to {\em diversify} the measurements and extract the information required to learn both the unknown signal and the model error. In this paper we adopt this general principle to achieve the {\em blind calibration} of sensor gains, that is the joint recovery of an unknown signal and some unknown multiplicative factors (\ie the {\em gains}) not accounted for in the assumed sensing model. Our method is inspired by recent results on fast, provably convergent algorithms for phase retrieval~\cite{CandesLiSoltanolkotabi2015,WhiteSanghaviWard2015,SunQuWright2016} and entails solving a non-convex problem by means of a descent algorithm that is presented below. Most importantly, this paper is concerned with finding the conditions under which the convergence of our algorithm to the exact solution is guaranteed, \ie a bound on the number of measurements $m$ and snapshots $p$ collected during the sensing process, along with some mild requirements on the entity of the gains.  Hence, our main concern will be to establish a {\em sample complexity}, \ie a lower bound on the total amount of observations collected during the sensing process: we will see that this bound must fulfil $mp = {\cl O}\big((m + n) \log^2(m(p+n))\big)$ with $n$ the dimensionality of the signal and $m$ that of the gains, up to a condition on $n = {\cl O}(\log mp)$. This methodology allows for provably exact blind calibration under some mild hypotheses on the gains for the sensing models we describe in the next section. 
		
	\subsection{Sensing Models}

        Our paper focuses on blind calibration for systems based on linear random sensing modalities. As described hereafter, we will assume that each observation is obtained by projection on a random vector $\bs a \in \bb R^n$ that is composed of $n$ independent and identically distributed {(\iid)}~components drawn as $a_j \sim_{\rm \iid} X$, $j \in [n]$, where $X$ is a centred sub-Gaussian random variable (r.v.)~having unit variance and sub-Gaussian norm $\alpha > 0$ (for a thorough introduction to such concepts, we refer the reader to~\cite[Section 5.2.5]{Vershynin2012a}). We recall that Gaussian, Bernoulli or bounded r.v.'s all pertain to the class of sub-Gaussian r.v.'s for some finite~$\alpha$. In this context, the sensing models tackled by this paper are defined as follows\footnote{The notation used in this section anticipates the one that is fully explained in Sec.~\ref{sec:notation}.}.
        
		\begin{definition}[Uncalibrated Multi-Snapshot Sensing Model]
			\label{def:model}	
			We define {\em uncalibrated multi-snapshot sensing model} any instance of
			\begin{equation}
				\label{eq:measurement-model-matrix} 
				\y_{l}\ =\ \diag({\bs g})\A_{l} \x\ =\ \big[ g_1 (\a_{1,l}^\top\x),\ \cdots, g_m (\a_{m,l}^\top \x) \big]^\top,\quad l \in [p],
			\end{equation}  
			where the $m$-dimensional measurement vector $\y_l \in \R^m$ is the $l$-th {\em snapshot} associated to the $l$-th random sensing matrix $\A_l =(\bs a_{1,l},\,\cdots,\bs a_{m,l})^\top\in \R^{m\times n}$ with $\ail \sim_{\iid} \bs a$,  \ie all the matrices $\{\A_l: l \in [p]\}$ are \iid. The vector $\x\in\R^n$ is an unknown, unstructured {\em signal} and ${\bs g} \in \R^m_+$ are unknown, positive and bounded {\em gains}, both quantities remaining fixed throughout the $p$ snapshots entailed by the sensing process. 
			
			The {\em noisy} case of this model is given hereafter by considering additive and bounded disturbances, \ie 
			\begin{equation}
				\label{eq:measurement-model-matrix-noise}
				\y_{l} = \diag({\bs g}) \,\A_{l} \x + \bnu_{l}, \ l \in [p],
			\end{equation}  
			where the noise vectors $\bnu_{l} \in \R^m$, $l \in [p]$ are collected in a matrix $\bs N \in \R^{m\times p}$ with $\sigma \coloneqq \ts\tfrac{1}{\sqrt{m p}}\|\bs N\|_F < \infty$.
		\end{definition}

		These {\em bilinear} sensing models are related to computational sensing applications in which unknown ${\bs g}$ are associated to positive gains in a sensor array, while $p$ random matrix instances can be applied on a source $\x$ by means of a suitable, typically programmable medium. In particular, the setup in \eqref{eq:measurement-model-matrix} matches compressive imaging configurations~\cite{Romberg2009,BjorklundMagli2013,DegrauxCambareriGeelenEtAl2014,BahmaniRomberg2015a,DumasLodhiBajwaEtAl2016}~with an important difference in that the absence of {\em a priori} structure on $(\x,{\bs g})$ in Def.~\ref{def:model} implies an over-Nyquist sampling regime with respect to $n$,~ \ie exceeding the number of unknowns as $m p \geq n + m$. 
		When the effect of ${\bs g}$ is critical, \ie assuming $\diag({\bs g}) \approx \I_m$ would lead to an inaccurate recovery of $\x$, finding solutions to \eqref{eq:measurement-model-matrix} in $(\x,{\bs g})$ justifies a possibly over-Nyquist sampling regime (that is $mp > n$) as long as both quantities can be recovered accurately (\eg as an on-line calibration modality). 
		To show that more efficient sampling regimes in $mp$ are possible we now introduce known subspace priors, paving the way to actual blind calibration for CS.  
		\begin{definition}[Uncalibrated Multi-Snapshot Sensing Model with Subspace Priors.]
			\label{def:modelsub}	
			Given two subspaces $\cl B \subset \bb R^m$ and $\cl Z \subset \bb R^n$, of dimension $h \coloneqq \dim \cl B \leq m$ and $k \coloneqq \dim \cl Z \leq n$, with orthonormal bases $\bs B \in \R^{m \times h}$ and ${\bs Z} \in \R^{n \times k}$, respectively, we define {\em uncalibrated sensing model with subspace priors} any instance of~\eqref{eq:measurement-model-matrix} (or of~\eqref{eq:measurement-model-matrix-noise}, in the presence of noise) where $\x \coloneqq {\bs Z} {\bs z} \in \cl Z$ for ${\bs z} \in \R^k$, and ${\bs g} \coloneqq {\bs B} {\bs b}$ for $\bs b \in \R^h$. 
			%				
			%			The {\em noisy} case of this model is given hereafter by considering additive and bounded disturbances, \ie 
			%			\begin{equation}
			%			\label{eq:measurement-model-matrix-noise}
			%			\y_{l} = \diag({\bs B \bs b}) \,\A_{l} \bs Z \bs z + \bnu_{l}, \ l \in [p],
			%			\end{equation}  
			%			where the noise vectors $\bnu_{l} \in \R^m, l \in [p]$ are collected in a matrix $\bs N \in \R^{m\times p}$ with $\sigma \coloneqq \ts\tfrac{1}{\sqrt{m p}}\|\bs N\|_F < \infty$.
			%			
		\end{definition}
		This known subspace prior is specially relevant for the signal domain, since such models are not necessarily present in the gain domain (\eg the gains can be fully random due to the nature of the device that captures the measurements). 		 
		When compared to {\em sparse} models for the signal and gains (\ie when either $\x$ or ${\bs g}$ lie in a union of low-dimensional canonical subspaces) Def.~\ref{def:modelsub} amounts to knowing the support of their respective sparse representations. Thus, while enforcing actual sparsity priors in the signal domain seems numerically feasible~\cite{CambareriJacques2017} we leave its theoretical analysis for a future communication given the depth of the additional considerations required to prove it. % rigorously. % enforce it with the non-convex approach presented in this paper.  
\medskip

Note that our analysis of both previous bilinear sensing models will exploit the fact that the deviation between the gains $\bs g$ and $\vone_m$ (up to a scaling factor) is significant, but not too large and, in fact, bounded in $\ell_\infty$-norm. Depending on which model is considered, this assumption is first detailed in Sec.~\ref{sec:nonconv}, then specified in the known subspace case of Sec.~\ref{sec:blind-calibration-with-subspace-priors}. 
		
		It is finally worth noting that the above sensing models are strongly related and could be generalised to {\em blind deconvolution} by modifying \eqref{eq:measurement-model-matrix}, \ie by letting the measurements 
		\begin{equation}
			\label{eq:bctobd}
			\y_{l} = ({\F}_m^{-1} {\bs g}) \circledast \A_l \x = {\F}_m^{-1} \diag({\bs g}) {\F_m} \,\A_{l} \x, \ l \in [p], 
		\end{equation}
		with $\F_m$ being the $m$-dimensional discrete Fourier transform and $\circledast$ the circular convolution operator. However, assessing the performances of our algorithm for this case with generally complex-valued $(\x,{\bs g})$ is beyond the scope of this contribution. 
\medskip

	{Before proceeding, let us state a minor working hypothesis on the interaction of the input vector $\bs x$ with the distribution $X$ controlling the random vectors $\{\ail\}$ in the very special case where $\bb E X^4 = 1$, \eg if $X$ is a Bernoulli \rv. 
  \begin{hypothesis*}[Bernoulli Restriction]
    \label{hyp:hyp-non-1-sparse-if-bern}
   If $\bb E X^4 = 1$, we will additionally assume $n>1$ and that there exists some constant $c >0$ such that
   $$
   \|\hat{\bs x}\|_4^4 \leq 1 - \tfrac{c}{n}, 
   $$
   with $\hat{\bs x} = \tfrac{\bs x}{\|\bs x\|^2}$.  
  \end{hypothesis*}
This hypothesis is purely technical and is likely an artifact of our proofs\footnote{{Moreover, this restriction could be relaxed to assuming $\|\hat{\bs x}\|_4^4 \leq 1 - \tfrac{c}{n^s}$ for any power $s>1$, as this would only change the universal constants appearing in all our results.}}. Its origin is actually found in the use of a matrix version of Bernstein's concentration inequality \cite{AhmedRechtRomberg2014,LiLingStrohmerEtAl2016,Tropp2015} (see Prop.~\ref{prop:mbi} in App.~\ref{sec:appA}) that imposes a non-vanishing matrix variance for the concentrating sum of centred matrices. Our hypothesis is also minor as $(i)$~by Jensen's inequality we have $\bb E X^4 \geq (\bb E X^2)^2 = 1$, and $(ii)$~when $\bb E X^4=1$, this just prevents us to take (in this very specific case) vectors lying ``too close'' to any of the coordinate axes $\bs c$ of $\bb R^n$, which are the only unit vectors whose $\|\bs c\|_4^4 = \|\bs c\|^4 = 1$. A possible strategy to avoid this restriction could consist in forming a mixture $X' \sim (1-\lambda) X + \lambda Y$ of $X$ with another arbitrary centred distribution $Y$ with $\bb EY^4 > 1$, $\bb E Y^2=1$ and $\lambda \in [0,1]$. For instance we could take $Y \sim \cl N(0,1)$, whose $\bb E Y^4 = 3$. Then, by linearity of the expectation operator with respect to the involved probability density measure, $\bb E X'^4 = (1 - \lambda) + \lambda \bb E Y^4 > 1$, so that for small $\lambda$ the distribution $X'$ is arbitrarily close to $X$, yet so that the fourth-moment hypothesis above is satisfied. 
} 		
	\subsection{Relation to Prior Works}			

		{\em General Literature on Blind Calibration}: \ Prior approaches to blind calibration of sensor gains include convex or alternating optimisation algorithms~\cite{BalzanoNowak2008,BilenPuyGribonvalEtAl2014,LiporBalzano2014} as well as message-passing approaches~\cite{SchuelkeCaltagironeZdeborova2015}. In more detail, Balzano and Nowak~\cite{BalzanoNowak2008} use a signal-domain subspace prior to infer sensor gains in absence of random sensing operators, \ie in a conventional sampling scheme. This approach is extended by Lipor and Balzano~\cite{LiporBalzano2014} in the presence of errors in the prior. Bilen \etal~\cite{BilenPuyGribonvalEtAl2014} use a sparse signal model for multiple inputs and solve a convex version of blind calibration for complex-valued gains. This is numerically shown to be successful in achieving blind calibration for CS. The approach of Sch{\"u}lke \etal~\cite{SchuelkeCaltagironeZdeborova2015} is based on a generalised approximate message passing framework, therefore taking into account a probabilistic model for the signal and gains% (thus assuming a very different perspective than ours)
		. All the former approaches are aided by multiple input signals (\eg $\x_l$,~$l \in [p]$) instead of taking new draws of the sensing operator itself, while there are clear assumptions on the independence of such signals. Moreover, no formal recovery guarantee is given in these works, \ie no requirement is obtained on the number of measurements required to perform provably exact blind calibration. 
		
		We now proceed to a comparison of our setup with three prior contributions that are closer to our aims, \ie they develop algorithms with exact recovery guarantees (and their required sample complexity) for either blind calibration or blind deconvolution. To do so, we recast our setup as follows. Let us define ${\bs w} \coloneqq \tfrac{\vone_p}{\sqrt p} \otimes {\bs g} \in \cl W \subset \bb R^{mp}$, with $\otimes$ the Kronecker product and $\cl W$ a subspace, by repeating $p$ times the same gains $\bs g$. Then we collect the snapshots of \eqref{eq:measurement-model-matrix} in 
		\begin{equation}
		\label{eq:strohmertous}
			\y \coloneqq \begin{bmatrix}
			\y_1 \\ \vdots \\ \y_p 
			\end{bmatrix} = \diag({\bs w}) \bs A \bs x, \ \bs A \coloneqq \begin{bmatrix}
			\A_1 \\ \vdots \\ \A_p 
			\end{bmatrix}.
		\end{equation}
		Let us assume $\bs x = \bs Z \bs z$ for some known subspace defined by a basis $\bs Z \in \bb R^{n \times k}$. 
		Depending on whether we are considering Def.~\ref{def:model} or Def.~\ref{def:modelsub} we will then have either $\cl W \coloneqq \big\{\tfrac{\vone_p}{\sqrt p} \otimes \bs v : \v \in \R^{m}\big\}$ with $m = \dim \cl W$, or $\cl W \coloneqq \big\{(\tfrac{\vone_p}{\sqrt p} \otimes \bs B) \bs v : \v \in \R^{h}\big\}$ with $h = \dim \cl W$ since $\bs g = \bs B \bs b$, $\bs B \in \bb R^{m \times h}$. 
		Moreover, let us briefly introduce the definition of {\em coherence}\footnote{Note that $ \upmu_{\max} \big(\tfrac{\vone_p}{\sqrt p} \otimes \bs B\big) = \sqrt{\tfrac{mp}{h}} \max_{j \in [mp]} \tinv{\sqrt p } \|  ({\vone_p}\otimes \bs B)^\top \bs c_j\| = \upmu_{\max} (\bs B)$.} of $\bs B$ as 
		\[
		\upmu_{\max}(\bs B) \coloneqq \sqrt{\tfrac{m}{h}} \max_{i\in[m]} \|\bs B^\top \bs c_i\| \in \big[1, \sqrt{\tfrac{m}{h}}\big],
		\] 
		\new{with $\bs c_i$ denoting the canonical basis vectors.}  
		This quantity frequently appears in related works \cite{AhmedRechtRomberg2014,LingStrohmer2015a} and shall be explained carefully in Sec.~\ref{sec:blind-calibration-with-subspace-priors}. Another recurring quantity that will not return in our analysis and main results is
		\[
		\upmu_{\rm p} \coloneqq \sqrt{m} \tfrac{\| \bs g \|_\infty}{\|\bs g\|},
		\] 
		which measures the peak-to-energy ratio of a specific instance of $\bs g$. In our developments this will be implicitly bounded as $\upmu_{\rm p}  < 1+\rho$ for a value $\rho < 1$ and will be therefore considered as a constant smaller than $2$.
		
		We omit from the following comparison the recent contribution of Ahmed \etal~\cite{AhmedKrahmerRomberg2016}, which tackles provably exact recovery for blind deconvolution in a different setup, \ie when both the signal and gains are sparse (and, due to this assumption, both are required to have such a sparse representation on random bases verifying several conditions). 
		\medskip
		
		\noindent {\em Ahmed, Recht and Romberg~\cite{AhmedRechtRomberg2014}:} \ The use of a {\em lifting} approach for the solution of bilinear inverse problems was first proposed in this fundamental contribution, which addressed the problem of {blind deconvolution}~(see also~\cite{AhmedCosseDemanet2015,BahmaniRomberg2015a,LingStrohmer2015}). As pointed out in \eqref{eq:bctobd} this sensing model encompasses blind calibration up to taking, in all generality, complex $(\bs x, \bs g)$ due to the application of $\bs F_m$. Loosely speaking, Ahmed \etal~assume a deterministic, known subspace prior on $\bs g = \bs B \bs b$ as we do. However, the random sensing matrix $\bs A$ in \eqref{eq:strohmertous} is assumed \iid Gaussian, whereas we consider \iid sub-Gaussian $\bs A$. Moreover, no prior is considered on the signal\footnote{In Ahmed \etal~our signal $\bs x$ is the message $\bs m$ that undergoes encoding by a random matrix $\bs C$, the latter being equivalent to $\bs A$ in \eqref{eq:strohmertous}.} $\bs x$.
		
		To show provably exact recovery, the authors leverage guarantees based on constructing a dual certificate for their lifted, semidefinite (\ie convex) problem via the so-called ``golfing scheme''~\cite{Gross2011}. The sample complexity required to recover exactly $(\bs x, \bs g)$ in~\cite[Thm. 1]{AhmedRechtRomberg2014} is then shown to be of the order of
		$m p = \cl O(\max\{\upmu^2_{\max} h, \upmu_{\rm p}^2 \, n\} \log^3 (mp))$. This is equivalent to what we will find in our Thm.~\ref{theorem:convergence-subs}, $mp = \cl O\big((n + \upmu^2_{\max}(\bs B) h) \log^2(m(p+n))\big)$, if no subspace prior holds on $\bs x$ (\ie $\bs Z = \I_n$). This equivalence of sample complexities (up to $\log$ factors), even if obtained using different algorithmic frameworks and in a setup more general than ours, suggests that this rate is somehow intrinsic to this class of bilinear inverse problems. 
		\medskip
		
		\noindent {\em Ling and Strohmer~\cite{LingStrohmer2015a}:} \ This recent contribution also proposed a lifting approach to jointly recover $(\x,{\bs g})$ in \eqref{eq:measurement-model-matrix}, \ie specifically for blind calibration. 
		The setup of~\cite{LingStrohmer2015a} is indeed the closest to the ones analysed in this paper. A first and foremost difference is in that, by letting $\A_l, l \in [p]$ be \iid sub-Gaussian random matrices we have partitioned the sensing operator in several independent snapshots, as opposed to a single snapshot in \cite{LingStrohmer2015a}. Moreover, Ling and Strohmer assume that the signal $\x$ is complex-valued and sparse, while we tackled only known subspace priors. Their vector ${\bs w}$ in \eqref{eq:strohmertous} is also complex-valued and described by a known subspace. Hence, their setup is still more general than the models we address in this paper. 
		
		The techniques used in Ling and Strohmer are similar in principle to those of Ahmed, Recht and Romberg: both contributions use lifting and obtain recovery guarantees via the aforementioned golfing scheme. In detail,~\cite[Thm.  3.1]{LingStrohmer2015a} shows that the solution to \eqref{eq:strohmertous} for \iid Gaussian $\A_l$ can be found with a sample complexity given by~\cite[(3.15)]{LingStrohmer2015a}, that is ${mp} = {\cl O}(m n \log (m n) \, \log^2 (mp))$ in absence of subspace priors. Otherwise, when such priors are given, this requirement would be $mp = {\cl O}(\upmu^2_{\max}(\bs B) h k \log h n \, \log^2 (mp))$. Our main results, while less general, will otherwise show that $mp$ $=$ $\cl O\big((k + \upmu^2_{\max}(\bs B) h)$ $\log^2(m(p+n))\big)$ when $n = {\cl O}(\log mp)$.
		
		From a numerical standpoint the main limitation of the lifting approach of~\cite{LingStrohmer2015a} is computational, \ie a semidefinite problem must be solved to recover a large-scale rank-one matrix $\x {\bs g}^\top$. This approach becomes computationally inefficient and unaffordable quite rapidly as $m$ and $n$ in the uncalibrated sensing model exceed a few hundreds. This is the main reason why we have undertaken the challenge of devising a non-convex optimisation framework for blind calibration, drawing mainly from the principles and concepts presented by Cand\`es \etal~\cite{CandesEldarStrohmerEtAl2015} in the context of phase retrieval. Just like a non-convex approach to phase retrieval based on a simple gradient descent algorithm with a careful initialisation strategy improved upon the former work of Cand\`es \etal~\cite{CandesStrohmerVoroninski2013}, we found that the same type of approach can lead to highly efficient blind calibration in the presence of unknown gains, that is computationally affordable for very high-dimensional signals and gains. %, contrarily to~\cite{LingStrohmer2015a}. 		
		\medskip
		
		\noindent {\em Li, Ling, Strohmer and Wei~\cite{LiLingStrohmerEtAl2016}:} \ 	
		Between our first short communication \cite{CambareriJacques2016} and the finalisation of this paper a remarkable contribution from Li \etal~\cite{LiLingStrohmerEtAl2016} showed that a non-convex approach to blind deconvolution is indeed capable of provably exact recovery, in a framework that would be applicable to blind calibration. Both our work and Li \etal~were inspired by the same non-convex approach of Cand\`es \etal~\cite{CandesLiSoltanolkotabi2015}. Thus, there are indeed some similarities between our paper and Li \etal~since,  loosely speaking, the general methodology both papers rely on is $(i)$ the definition of an initialisation for a descent algorithm, and $(ii)$ the verification of some mild regularity conditions on the gradient of a non-convex objective in a neighbourhood defined by the initialiser, so that a suitable gradient descent  algorithm converges to the exact solution. Moreover, both Li \etal~and our work consider known subspace priors on the signal and the gains, so our sensing model (when suitably recast as in \eqref{eq:strohmertous}) is essentially equivalent to theirs. 
		
		There are however some important differences worth pointing out. Our approach uses $\bs A$ comprised of $p$ independent snapshots with \iid sub-Gaussian sensing matrices. Li \etal~tackle the case of a complex Gaussian sensing matrix $\bs A$. In our case, we will bound {\em a priori} the perturbation between the uncalibrated and true sensing model, \ie we will have a condition on $\|\bs g - \vone_m\|_\infty$ for $\bs g^\top \vone_m = m$. Li \etal~instead assume a prior on the aforementioned peak-to-energy ratio $\upmu_{\rm p}$, as well as a small $\upmu_{\max}(\bs B)$. Our initialisation is a simple back-projection in the signal domain, which also uses the boundedness of $\|\bs g - \vone_m\|_\infty$ and the use of $p$ snapshots. Li \etal~use a spectral initialisation closer to what was done in \cite{CandesLiSoltanolkotabi2015}, followed by an optimisation problem that enforces a constraint on $\upmu_{\rm p}$. Given such priors, these must then be enforced by our respective algorithms: in our case, we carry out a projected gradient descent that minimises a non-convex objective with a convex constraint, with the projector affecting only the gain domain (so the gains verify the bound on $\|\bs g - \vone_m\|_\infty$). The method of Li \etal~uses instead a regularised objective function without constraints, which however must enforce their condition on $\upmu_{\rm p}$. 
		
		In terms of sample complexity, our results range from a worst-case setting in which no subspace model is assumed to the case of known subspace priors, where we require $mp = \cl O\big((k + \upmu^2_{\max}(\bs B) h) \log^2(m(p+n))\big)$. In Li \etal~the obtained sample complexity, by a comparison through \eqref{eq:strohmertous}, would be $mp =  {\cl O}(\max\{\upmu^2_{\max}(\bs B) h, \upmu_{\rm p}^2 k\} \log^2(mp))$ under a careful account of $\upmu_{\max}(\bs B)$. Thus, while the problem setup of Li \etal~is quite general, their sample complexity obtained in~\cite[Thm. 3.2]{LiLingStrohmerEtAl2016} is substantially the same and indeed close to that of Ahmed \etal. In fact, we stress that the sample complexities we obtain and the former ones are substantially similar and given by the intrinsic properties of this bilinear inverse problem. 
		Hence, we conclude that our work is essentially alternative to Li \etal, as it applies to sub-Gaussian sensing matrices and uses some specific conditions related to blind calibration of sensor gains, while admittedly not addressing the more general case of \eqref{eq:bctobd}. Moreover, the theory we leverage to prove our main results is slightly simpler.
		
		% Finally, we note that the sample complexity bounds of both~\cite{LingStrohmer2015a} and~\cite{LiLingStrohmerEtAl2016} depend explicitly on the value of $\upmu_{\max} \coloneqq \sqrt{m}\max_{\bs g \in \cl B} \tfrac{\|\bs g\|_\infty}{\|\bs g\|}$, \ie the {\em coherence} of the gains $\bs g$. This quantity affects their bound as $\upmu^2_{\max} \in [1, {m}]$. Interestingly, in our setup that only tackles a special instance of blind calibration having practical interest, this value is fixed in function of the amount of perturbation affecting the gains. However, the numerical results 
		
	\subsection{Main Contributions and Outline}	
	%	\mtodo{Modify outline to account for moved stability theorem etc.}
	
	The main contribution of this paper is in showing that for all the uncalibrated sensing models specified in Def.~\ref{def:model} and Def.~\ref{def:modelsub} it is possible to prove that a very simple and efficient projected gradient descent algorithm promoting the Euclidean data fidelity to the measurements actually converges (under very mild hypotheses) to the exact solution, which verifies \eqref{eq:measurement-model-matrix} up to an unrecoverable scaling factor. This descent algorithm strongly relies on an initialisation strategy that is, in fact, a simple back-projection of the measurements. This provides an unbiased estimate of the signal-domain solution under the hypotheses of Def.~\ref{def:model}, and puts the first iteration of our descent algorithm in a neighbourhood of the global minimiser. Once this neighbourhood can be shown to be sufficiently small, and provided that the perturbations with respect to the available information on the sensing model are also small (in particular, far from the loss of information corresponding to any zero gain $g_i = 0$) the behaviour of the gradients used in the iterates of our algorithms is close to its expectation as the {sample complexity} $mp$ grows, \ie as a consequence of the concentration of measure phenomenon.  
	This allows us to find the conditions on $mp$ that ensure convergence to the exact solution, depending on which sensing model is chosen. In particular, our substantial contribution is in showing that this sample complexity grows as $mp = \cl O\big((n + m) \log^2(m(p+n))\big)$, \ie only proportionally to the number of unknowns $n + m$ of this problem (up to $\log$ factors). Moreover, when this number of unknowns is reduced by the use of known subspace priors, \ie as in Def.~\ref{def:modelsub}, we show that this complexity only needs to grow as $mp = \cl O\big((k + \upmu^2_{\max} h) \log^2(m(p+n))\big)$, \ie again linearly (up to $\log$ factors and the effect of a coherence parameter $\upmu_{\max} \in \big[1,\sqrt{\tfrac{m}{h}}\big]$) in the number of unknowns $k$ and $h$. 
		
	Note that a short communication that partially overlaps with this work was published by the authors~\cite{CambareriJacques2016}. It is here improved and expanded with revised proofs, results on the stability of the proposed approach in the presence of noise, and the possibility of introducing known subspaces to model the signal and gain domains with the aim of reducing the sample complexity of this problem, \ie minimising the amount of required snapshots in Def.~\ref{def:modelsub} when more information on the unknowns is available. 
	
	The rest of this paper is structured as follows. In Sec.~\ref{sec:nonconv} we formalise the blind calibration problem in its non-convex form and in absence of priors. We explore some of its geometric properties, both in expectation as well as for a finite number of snapshots. There, we define the main notions of distance and neighbourhood used in proving the properties of this problem, and we see that some local convexity properties do hold in expectation. Our main algorithm is then introduced in Sec.~\ref{sec:solbypgd} and followed by its convergence guarantees in absence of priors, which actually enables a simple understanding of our main results.  
	In Sec.~\ref{sec:blind-calibration-with-subspace-priors} we discuss a modification of our algorithm in the case of known subspace priors for $\x$ and ${\bs g}$. We show how the anticipated sample complexity improvement is achieved using such prior models and discuss the convergence of a descent algorithm that enforces them properly. Most importantly, the proofs for this case are the cornerstones of this paper, and serve to prove the somewhat simpler results we obtained in absence of priors and advocated in \cite{CambareriJacques2016}. In Sec.~\ref{sec:nstability-sub} a noise stability analysis is then proposed to show how the accuracy of the solution is affected by the presence of bounded noise in the measurements. 
	In Sec.~\ref{sec:numerical-experiments} we provide numerical evidence on our algorithm's empirical phase transition, highlighting the regime that grants exact recovery of the signal and gains. This is followed by an empirical discussion of the descent algorithm's step size update, and by an assessment of the stability of our algorithm in the presence of noise. All three cases are carried out in absence of priors, \ie in a worst-case setup.~ A practical application of our method to a realistic computational sensing context, both in absence and in presence of known subspace priors, concludes the experiments carried out in this paper. The proofs of all mathematical statements and tools used in this paper are reported in the appendices.  
	
	\subsection{Notation and Conventions}
		\label{sec:notation}
		The notation throughout the paper is as follows: vectors and matrices are denoted by boldface lower-case and upper-case letters respectively, \eg $\bs q$ and $\bs Q$, while scalars and constants are denoted as $q$, $Q$. % Moreover, we will distinguish optimisation variables by the use of Greek letters and the corresponding data (either known or unknown) as Roman letters. 
		The vectors $\vzer_q$ and $\vone_q$ indicate a vector of dimension $q$, respectively of all zeros or ones and with size specified in the subscript. The identity matrix of dimension $q$ is $\I_q$. The rows, columns and entries of a matrix $\bs Q$ will be denoted as  $\bs Q_{j, \cdot}$, $\bs Q_{\cdot, j}$ and $Q_{ij}$ respectively. An exception to this rule are the rows of the sensing matrices $\bs A_{l}$, denoted as $\ail^\top$ (\ie as the column vectors $\ail$). 
		Sets and operators are generally denoted with calligraphic capital letters, \eg $\cl S$, $\cl A$. The $n$-variate Gaussian distribution is denoted as $\cl N(\vzer_n,\I_n)$, the uniform distribution is $\cl U_{\cl C}$ over a set $\cl C$ specified at the argument. The usual big-O notation is indicated by $\cl O$. The Kronecker product between vectors of matrices is denoted by $\otimes$. Collections of vectors or matrices can be indexed by single or double subscripts, \eg $\a_i$ or $\ail$.  For the sake of brevity, we may omit the boundaries of sums that are identical throughout the development: unless otherwise stated, $\ts\sum_i$ denotes $\ts\sum^{m}_{i=1}$, $\ts\sum_l$ denotes $\ts\sum^{p}_{l=1}$, $\ts\sum_{i,l}$ denotes $\ts\sum^{m}_{i=1} \ts\sum^{p}_{l=1}$.  For some integer $q>0$, we denote the set $[q] \coloneqq \{1, \ldots, q\}$. The norms in $\ell_p(\R^q)$ are denoted as usual with $\|\cdot\|_p$ with $\|\cdot\| = \|\cdot\|_2$. The spectral norm of a matrix reads $\|\cdot\|$, while the Frobenius norm and scalar product are $\|\cdot\|_F$ and $\langle\cdot,\cdot\rangle_F$, respectively. Spheres and balls in $\ell_p(\R^q)$ will be denoted by $\bb S^{q-1}_p$ and $\B^q_p$ respectively. For matrices, we also introduce the Frobenius sphere and ball, defined respectively as $\bb S_F^{n\times m}=\{\bs U \in \bb R^{n \times m}: \|\bs U\|_F=1\}$ and $\bb B_F^{n\times m}=\{\bs U \in \bb R^{n \times m}: \|\bs U\|_F\leq 1\}$. The projection operator on a closed convex set $\cl C$ is ${\cl P}_{\cl C}$, while orthogonal projection on a linear subspace $\cl B$ is denoted by the projection matrix $\P_{\cl B}$. The symbol $\sim_{\rm \iid}$ indicates that a collection of random variables (abbreviated as \rv) or vectors on the left hand side (\lhs) of the operator are independent and follow the same distribution given on the right hand side (\rhs). The Orlicz norm of a random variable $A$ is denoted as $\|A\|_{\psi_q} \coloneqq \sup_{r \geq 1} \ts r^{-1/q} {\Ex[|A|^r]}^{1/r}$ with the sub-exponential and sub-Gaussian norms being the cases for $q = 1$ and $q = 2$ respectively. We will often resort to some constants $C,c > 0$, as traditional in the derivation of non-asymptotic results for sub-Gaussian random vectors and matrices. The value of these constants is not relevant and may change from line to line, as long as it does not depend on the problem dimensions. The quantity $\bs\nabla_\v f(\ldots, \v, \ldots)$ denotes the gradient operator with respect to the vector $\v$ specified in the subscript, as applied on a function $f$. In absence of subscripts, it is the gradient of $f$ with respect to all of its components. $\H f$ denotes the Hessian matrix of $f$. The set $\Pi_+^m = \{\v \in \R^m_+, \, \vone^\top_m \v = m\}$ denotes the scaled {\em probability simplex}.  We will also refer to the orthogonal complement $\onep \coloneqq \{\v \in \R^m : \vone_m^\top \v= 0\} \subset \bb R^m$. Moreover, when vectors and matrices are projected or lie on the latter subspace they will be denoted with the superscript $\cdot^{\perp}$.  The canonical basis vectors of $\R^m$ are denoted by ${\bs c}_i$, $i \in [m]$. Their projection on $\onep$ is ${\bs c}^{\perp}_i \coloneqq \Ponep {\bs c}_i$. The operator $\succeq$ denotes the L\"owner ordering on the convex cone of symmetric positive-semidefinite matrices (if strict, this is denoted as $\succ$). The absence of this ordering is denoted as $\nsucceq$. The restriction of this ordering to test vectors belonging to a set $\cl A$ is denoted in the subscript, \eg $\succ_{\cl A}$. The accent $\tilde{\cdot}$ will denote the noisy version of a quantity at the argument that was previously defined in absence of noise, while the accent $\bar{\cdot}$ denotes the estimates attained by an optimisation algorithm. The accent $\hat{\cdot}$ denotes unit vectors obtained from the argument ($\hat{\bs q} = \tfrac{\bs q}{\|\bs q\|}$) or unit matrices with respect to the Frobenius norm ($\widehat{\bs Q} = \tfrac{\bs Q}{\|\bs Q\|_F}$). The superscript $\cdot^{\rm s}$ denotes a quantity defined to accommodate subspace priors, while the superscript $.^{\rm c}$ denotes the complementary of an event.

%%% Local Variables:
%%% mode: latex
%%% TeX-master: "iai_blindcalibration_2016"
%%% End:

\section{A Non-Convex Approach to Blind Calibration}
	\label{sec:nonconv}
	We now proceed to introduce our main non-convex optimisation problem and its geometric properties, depending on the dimensions of the setup in Def.~\ref{def:model}.
	\subsection{The Blind Calibration Problem}
	The formulation of an inverse problem for~\eqref{eq:measurement-model-matrix} is quite natural by means of a Euclidean data fidelity objective function $f(\bxi,{\bs\gamma}) \coloneqq  \tfrac{1}{2mp} {\ts \sum_{l=1}^{p} \left\Vert \diag({\bs\gamma}) \A_l \bxi - \y_l\right\Vert^2}$ (this is further expanded in Table~\ref{tab:quants} for both finite and asymptotic $p$). Since no {\em a priori} structure is assumed on the solution $(\x,{\bs g}) \in \R^n \times \R^m_+$ for now, let us operate in the overdetermined case $m p \geq n + m$ and solve the optimisation problem% ,%.
	\begin{equation}%
		\label{eq:bcpearly}
		(\bar{\x}, \bar{{\bs g}}) \coloneqq {\mathop{\argmin}_{(\bxi,  {\bs\gamma}) \in \R^n \times \R^m}}  \tfrac{1}{2mp} {\ts \sum_{l=1}^{p} \left\Vert \diag({\bs\gamma}) \A_l \bxi - \y_l\right\Vert^2},
	\end{equation}%
	with $\y_l \in \R^m$, $\A_l\in \R^{m\times n}$, $l \in [p]$ as in Def.~\ref{def:model}. 
	To begin with, replacing $\y_l$ by its model~\eqref{eq:measurement-model-matrix} in~\eqref{eq:bcpearly} shows that, in absence of prior information on $\bs x$ or ${\bs g}$, all points in 
	\[
	\cl X \coloneqq \{(\bxi,{\bs\gamma}) \in\R^n \times \R^m : \bxi= {\alpha}^{-1} \x, \, {\bs\gamma} = \alpha {\bs g}, \, \alpha \in \R \setminus \{0\}\}
	\]
	%\footnote{This minimiser's uniqueness can only be shown in a neighbourhood and will be proved afterwards. % , as a consequence of Prop.~\ref{prop:regu}.
	%}
	are global minimisers of $f(\bxi,{\bs\gamma})$ up to an unrecoverable scaling factor $\alpha$. This ambiguity is an inevitable aspect of many bilinear inverse problems that is well recognised in the referenced literature (for a more general theory on the identifiability of this class of problems, we refer the reader to some recent contributions~\cite{BahmaniRomberg2015a,KechKrahmer2016,ChoudharyMitra2013}). 
	In most applications this scaling ambiguity is equivalent to ignoring the norm of the original signal and is an acceptable loss of information. Thus, since we also know that $\bs g$ is positive, we could apply a constraint such as ${\bs\gamma} \in \Pi_+^m \coloneqq \{\v \in \R^m_+, \, \vone^\top_m \v = m\}$ in~\eqref{eq:bcpearly}, which would fix the $\ell_1$-norm of the gain-domain solution. This yields the minimiser $(\x^\star,{\bs g}^\star) \coloneqq \big( \tfrac{\|{\bs g}\|_1}{m} \x, \tfrac{m}{\|{\bs g}\|_1}{\bs g}\big) \in \cl X \cap (\R^n \times \Pi_+^m)$, \ie scaled by $\alpha = \tfrac{m}{\|{\bs g}\|_1}$. In other words,~\eqref{eq:bcpearly} has only one non-trivial global minimiser over $\cl X  \cap (\R^n \times \Pi_+^m)$ and at least one in $\R^n \times \Pi_+^m$ (this minimiser's uniqueness can only be shown in a neighbourhood and shall be proved afterwards).  
	As a result, since  ${\bs g} \in \R^m_+$ in~Def.~\ref{def:model} is positive, bounded and close to unity, the gain-domain solution of~\eqref{eq:bcpearly} with the additional constraint $\bs\gamma \in \Pi_+^m$ will actually be
	\[
		{\bs g}^\star \in {\cl G}_{\rho} \subset \Pi^m_+, \, {\cl G}_{\rho} \coloneqq \vone_m + \onep \cap \rho\,\B^m_\infty, \, \rho < 1,
	\]
	for a maximum deviation $\rho \geq \|{\bs g}^\star-\vone_m\|_\infty$ which we assume known at least for what concerns the analysis of the proposed algorithms. % and so that $\rho < 1$. 
	Thus, under the constraint $\bs \gamma \in {\cl G}_{\rho}$ we can specify that ${\bs g}^\star \coloneqq \vone_m + {\bs e}$ for ${\bs e} \in \onep \cap \rho\,\B^m_\infty$ as well as ${\bs\gamma} \coloneqq \vone_m + \eps$ for $\eps \in \onep \cap \rho\,\B^m_\infty$. With these simple considerations obtained by construction from Def.~\ref{def:model} we arrive to the following problem, that is simply \eqref{eq:bcpearly} with ${\bs\gamma} \in {\cl G}_{\rho} \subset \Pi_+^m$. % , which is essentially equivalent to~\eqref{eq:bcpearly} when $\rho < 1$.
	%
	%%%%%%%%%%%%%%%%%%%%%%%%%%%%%%%%%%%%%%%%%%%%%%%%%%%%%%%%%%%%%%%%%%
	\begin{definition}[Non-convex Blind Calibration Problem]\label{def:ncbc}
		We define {\em non-convex blind calibration} the optimisation problem
		\begin{equation}%
		\label{eq:bcp}
		(\bar{\x}, \bar{{\bs g}}) \coloneqq {\mathop{\argmin}_{(\bxi,{\bs\gamma}) \in \R^n \times \cl G_\rho}}  \tfrac{1}{2mp} {\ts \sum_{l=1}^{p} \left\Vert \diag({\bs\gamma}) \A_l \bxi - \y_l\right\Vert^2},
		\end{equation}%
		where
		\begin{equation}
			{\cl G}_{\rho} \coloneqq \big\{\bs v \in \bb R^m : \bs v \in \vone_m + \onep \cap \rho\,\B^m_\infty\big\},
		\end{equation} 
		given $\rho < 1$, $\y_l \in \R^m$, $\A_l\in \R^{m\times n}$, $l \in [p]$ as in Def.~\ref{def:model}.
	\end{definition}
	%	\begin{definition}[Non-convex Blind Calibration Problem]\label{def:ncbc}
	%		We define {\em non-convex blind calibration} the optimisation problem
	%		\begin{equation}%
	%			\label{eq:bcp}
	%			(\bar{\x}, \bar{{\bs g}}) \coloneqq {\mathop{\argmin}_{\bxi \in \R^n, {\bs\gamma} \in \Pi_+^m}}  \tfrac{1}{2mp} {\ts \sum_{l=1}^{p} \left\Vert \diag({\bs\gamma}) \A_l \bxi - \y_l\right\Vert^2},
	%		\end{equation}%
	%		given $\y_l \in \R^m, \A_l\in \R^{m\times n}, l \in [p]$ as in Def.~\ref{def:model}.
	%	\end{definition}
	%%%%%%%%%%%%%%%%%%%%%%%%%%%%%%%%%%%%%%%%%%%%%%%%%%%%%%%%%%%%%%%%%%%%%%%%%%%%%%%%%%%%%%%%%%%%%%%%%%%
	For finite $p$, we remark that the minimiser $(\bs x^\star, \bs g^\star)$ of~\eqref{eq:bcp} is still not necessarily the only one: to begin with, we would have to ensure that $\cl K \coloneqq \bigcap \limits_{l=1}^p {\rm Ker} \A_l = \{\vzer_n\}$ which, however, holds with probability $1$ if $mp > n$ and $\A_l$ is an \iid sub-Gaussian random matrix.	Hence, taking the number of measurements $mp \geq n+m$ to be sufficiently large reduces, but does not generally exclude the possibility of stationary points in the domain of~\eqref{eq:bcp}. This possibility will only be cleared out later in Cor.~\ref{coro:unique}, where we will be able to show that the magnitude of the gradient of~\eqref{eq:bcp} is non-null in a neighbourhood of the global minimiser $(\x^\star,{\bs g}^\star)$. 
	However, since $f(\bxi,{\bs\gamma})$ is non-convex (as further detailed below), devising an algorithm to find such a minimiser is non-trivial and requires a better understanding of the geometry of this problem, starting from its discussion in the next section.

	Hereafter, we shall simplify the notation $(\bs x^\star, \bs g^\star)$ to $(\bs x, \bs g)$, \ie in all generality, we can assume $\bs g$ is directly normalized so that $\|\bs g\|_1= m$, \ie so that it lies in $\cl G_\rho \subset \Pi^m_+$.
	
	\subsection{The Geometry of Blind Calibration}
	\label{sec:geom-blind}
	\subsubsection{Preliminaries}
	\label{sec:prelims}
	We now expand the objective function $f(\bxi,{\bs\gamma})$ of~\eqref{eq:bcp}, its gradient~$\bs\nabla f(\bxi,{\bs\gamma}) = \begin{bsmallmatrix}(\bs\nabla_\bxi f(\bxi,{\bs\gamma}))^\top & (\bs\nabla_{\bs\gamma} f(\bxi,{\bs\gamma}))^\top\end{bsmallmatrix}^\top$ and Hessian matrix $\H f(\bxi,{\bs\gamma})$ as reported in Table~\ref{tab:quants} in two useful forms for this paper (the matrix form is more convenient for the implementation, the explicit form as a function of the sensing vectors $\ail$ is analogue to that used in the proofs of our main results). There, we confirm that $f(\bxi,{\bs\gamma})$ is generally non-convex. In fact, as noted in~\cite{LingStrohmer2015} it is bilinear and {\em biconvex}, \ie convex once either $\bxi$ or ${\bs \gamma}$ are fixed in~\eqref{eq:bcp} (this suggests that an alternating minimisation with sufficiently many samples $mp$ could also converge, although proving this is an open problem). As a confirmation of this, there exist plenty of counterexamples for which the Hessian matrix $\H f(\bxi,{\bs\gamma}) \nsucceq 0$, the simplest being $(\bxi, {\bs\gamma}) = (\vzer_n,\vzer_m)$. 

	Moreover, note that the constraint in \eqref{eq:bcp} is so that each $\bs \gamma = \vone_m + \bs \varepsilon$ with $\bs \varepsilon \in \onep$. Thus, the steps that our descent algorithm will take must lie on $\onep$, so in our analysis we will also consider the projected gradient and Hessian matrix components by suitably projecting them on $\onep$. % that are related to the unknown gains on $\onep$. 
	Hence, we define the projected gradient as 
	\[
		\ts\bs\nabla^\perp f(\bxi,{\bs\gamma})\coloneqq\begin{bsmallmatrix}\I_n& \vzer_{n \times m} \\ \vzer_{m\times n} & \Ponep \end{bsmallmatrix} \bs\nabla f(\bxi,{\bs\gamma}) = \begin{bmatrix}
		\bs\nabla_\bxi f(\bxi,{\bs\gamma})\\
		\bs\nabla^\perp_{\bs\gamma} f(\bxi,{\bs\gamma})
		\end{bmatrix},
	\]
	whose component
	 $\bs\nabla^\perp_{\bs\gamma} f(\bxi,{\bs\gamma}) \coloneqq\Ponep \bs\nabla_{\bs\gamma} f(\bxi,{\bs\gamma})$, and the projected Hessian matrix 
	\[
	\H^\perp f(\bxi,{\bs\gamma}) \coloneqq \begin{bsmallmatrix}\I_n& \vzer_{n \times m} \\ \vzer_{m\times n} & \Ponep \end{bsmallmatrix} \H f(\bxi,{\bs\gamma}) \begin{bsmallmatrix}\I_n& \vzer_{n \times m} \\ \vzer_{m\times n} & \Ponep \end{bsmallmatrix},
	\]
	both fully developed in Table~\ref{tab:quants}, where the projection matrix $\Ponep \coloneqq \I_m - \tfrac{1}{m} \vone_m \vone^\top_m$. 
	These quantities allow us to discuss our problem~\eqref{eq:bcp} in terms of the deviations ${\bs e},\eps \in \onep$~around $\vone_m$, where by $\bs e \in \onep  \cap \rho \B^m_\infty$ we are allowed to carry out a perturbation analysis for sufficiently small values of $\rho<1$. 
	
	\subsubsection{Distances and Neighbourhoods}
	Since $f(\bxi,{\bs\gamma})$ is non-convex, applying a constraint as in~\eqref{eq:bcp} will not grant the convexity of problem~\eqref{eq:bcp} on its domain. However, a {\em local} notion of convexity may hold when testing those points that are close, in some sense, to the global minimiser of interest $(\x,{\bs g})$. Hence, we define a {\em distance} and a {\em neighbourhood} of the global minimiser as follows. 
	We first note that the pre-metric 
	\begin{equation}
		{\Delta}_F(\bxi,{\bs\gamma}) \coloneqq \tfrac{1}{m} \big\|\bxi {\bs\gamma}^\top - \x {\bs g}^\top \big\|^2_F	= 2 \Ex f(\bxi,{\bs\gamma})
		\label{eq:distF}
	\end{equation}
	is exactly the expected objective, the last equivalence being immediate from the form reported in Table~\ref{tab:quants}. Thus, the corresponding distance ${\Delta}^{\scriptstyle\frac{1}{2}}_F(\bxi,{\bs\gamma})$ is a naturally balanced definition that does not depend on the scaling, and truly measures the distance between $(\bxi,{\bs\gamma})$ and $(\x,{\bs g})$ in the product space $\R^n \times \R^m$. However, this choice would complicate the convergence proof for the algorithms described below, so we resort to the simpler
	\begin{equation}
	{\Delta}(\bxi,{\bs\gamma}) \coloneqq \|\bxi-\x\|^2 + \tfrac{\|\x\|^2}{m}\|{\bs \gamma}-{\bs g}\|^2.
	\label{eq:dist}
	\end{equation}
	To relate~\eqref{eq:dist} and~\eqref{eq:distF} note that, for $(\bxi,{\bs\gamma}) \in \R^n \times {\cl G}_{\rho}, \rho \in (0,1)$,  we have the bounds
	\begin{equation}
	\label{eq:boundsondelta}
	(1- \rho){\Delta}(\bxi,{\bs\gamma})\leq{{\Delta}_F(\bxi,{\bs\gamma})} \leq(1 + 2\rho){\Delta}(\bxi,{\bs\gamma}).
	\end{equation}
	Little is then lost in using ${\Delta}(\bxi,{\bs\gamma})$ in the following (a proof of~\eqref{eq:boundsondelta} is reported in App.~\ref{sec:appB}). 
	Thus, with this simpler definition of distance (noting that~\eqref{eq:dist} is still a pre-metric) we may define a neighbourhood of $(\x,{\bs g})$ as follows.
		\begin{definition}[($\kappa,\rho$)-neighbourhood]
			\label{def:neighbour}
			We define ($\kappa,\rho$)-neighbourhood of the global minimiser $(\x,{\bs g})$ the set 
			\begin{equation}
			\label{eq:neigh}
			{{\cl D}}_{\kappa,\rho} \coloneqq \{(\bxi,{\bs\gamma})\in \R^n\times {\cl G}_\rho : \Delta(\bxi,{\bs\gamma}) \leq \kappa^2 \|\x\|^2\},
			\end{equation}
			for $\kappa,\rho \in [0,1)$.
		\end{definition}
	Geometrically, we remark that~\eqref{eq:neigh} is simply the intersection of an ellipsoid in $\R^n \times \R^m$, as defined by $\Delta(\bxi,{\bs\gamma}) \leq \kappa^2 \|\x\|^2$, with $\R^n \times {\cl G}_{\rho}$. Hence, for some fixed and sufficiently small $\kappa,\rho$, such a neighbourhood defines a set of points in the product space on which we will be able to bound functions of the projected gradient. Before these considerations, we try and develop some intuition on the role of the total number of observations $mp$ in making the problem solvable by means of a descent algorithm as $p\rightarrow \infty$.
	
	%%%%%%%%%%%%%%%%%%%%%%%%%%
	\begin{table*}[tb]
			\centering
			\resizebox{\textwidth}{!}{
					\def\arraystretch{2}
					\begin{tabular}{ccc}
						\toprule
						\bf Quantity & \bf  Finite-sample value $(p < \infty)$ & \bf  Expectation $({\bb E}_{\ail}, p\rightarrow \infty)$ \\
						\toprule
						$f(\bxi,{\bs\gamma})$ &
						$\tfrac{1}{2mp}\ts \sum_{l=1}^{p} \left\Vert\diag({{\bs \gamma}}) \A_l \bxi - \diag({\bs g})\A_l \x\right\Vert^2 = \tfrac{1}{2mp}\ts \sum_{i,l} (\gamma_i \ail^\top\bxi - g_i \ail^\top \x)^2$
						& $ \begin{gathered} \tfrac{1}{2m}\left(\|\bxi\|^2 \|{\bs \gamma}\|^2\right)+\tfrac{1}{2m}\|\x\|^2 \|{\bs g}\|^2 - 2 ({\bs\gamma}^\top {\bs g})(\bxi^\top \x) \end{gathered}$ \\
						\midrule
						$\bs\nabla_{\bxi} f(\bxi,{\bs\gamma})$ & {$\tfrac{1}{mp}\ts \sum^p_{l=1} \A^\top_l \diag({\bs\gamma}) \left(\diag({\bs\gamma})\A_l\bxi-\diag({\bs g})\A_l\x\right) = \tfrac{1}{mp}\ts \sum_{i,l} \gamma_i \ail \big(\gamma_i \ail^\top\bxi - g_i \ail^\top \x\big)$} & $\tfrac{1}{m}\left[\|{\bs \gamma}\|^2 \bxi - ({\bs\gamma}^\top {\bs g}) \x \right]$\\
						$\bs\nabla_{\bs \gamma} f(\bxi,{\bs\gamma})$& $\tfrac{1}{mp}\ts \sum^p_{l=1}  \diag(\A_l \bxi) \left(\diag({\bs\gamma}) \A_l \bxi - \diag({\bs g}) \A_l \x\right)=  \tfrac{1}{mp}\ts \sum_{i,l} (\ail^\top \bxi) \big(\gamma_i \ail^\top\bxi - g_i \ail^\top \x\big) {\bs c}_i $ & $\tfrac{1}{m} \left[\|\bxi\|^2 {\bs\gamma} - (\bxi^\top\x) {\bs g} \right]$\\
						$\bs\nabla^\perp_{{\bs \gamma}} f(\bxi,{\bs\gamma})$ & $\tfrac{1}{mp}\ts \sum^p_{l=1}  \Ponep \diag(\A_l \bxi) \left(\diag(\vone_m + {\bs \varepsilon}) \A_l \bxi - \diag(\vone_m + {\bs e}) \A_l \x\right)=\tfrac{1}{mp}\ts \sum_{i,l} (\ail^\top \bxi) \big((1+\varepsilon_i) \ail^\top\bxi - (1+e_i) \ail^\top \x\big) {\bs c}^{\perp}_i$ & $\tfrac{1}{m} \left[\|\bxi\|^2 \eps - (\bxi^\top\x) {\bs e} \right]$\\
						\midrule
						$\H f(\bxi,{\bs\gamma})$ & 
						$ \begin{gathered}
						\tfrac{1}{mp}\ts \sum^p_{l=1} 
						\begin{bsmallmatrix} \A_l^\top \diag({\bs \gamma})^2 \A_l & \A_l^\top\,\diag(2 \diag({\bs\gamma})\A_l \bxi - \diag({\bs g}) \A_l \x )\\ 
						\diag(2 \diag({\bs\gamma}) \A_l\bxi - \diag({\bs g}) \A_l \x)\,\A_l & \diag( \A_l \bxi)^2
						\end{bsmallmatrix} \\
						=\tfrac{1}{mp}\ts \sum_{i,l} 
						\begin{bsmallmatrix} \gamma^2_i \ail \ail^\top & \ail\ail^\top\,(2\gamma_i \bxi - g_i \x){\bs c}^\top_i \\ 
						{\bs c}_i (2\gamma_i \bxi - g_i \x)^\top \ail\ail^\top & (\ail^\top \bxi)^2 {\bs c}_i {\bs c}_i^\top
						\end{bsmallmatrix}
						\end{gathered} $ 
						& 
						$\tfrac{1}{m} {\begin{bsmallmatrix} \|{\bs \gamma}\|^2 \I_n & 2 \bxi {\bs\gamma}^\top - \x {\bs g}^\top \\ 2 {\bs\gamma} \bxi^\top - {\bs g} \x^\top & \|\bxi\|^2 \I_m\end{bsmallmatrix}}$
						\\[4ex]					
						$\H^\perp f(\bxi,{\bs\gamma})$ & 
						$\begin{gathered}
						{\tfrac{1}{mp}\ts \sum^p_{l=1} \begin{bsmallmatrix} \A_l^\top \diag({\bs\gamma})^2 \A_l & \A_l^\top \diag(2 \diag({\bs \gamma})\A_l\bxi - \diag({\bs g})\A_l\x)\Ponep \\ 
							\Ponep \diag(2 \diag({\bs\gamma})\A_l\bxi - \diag({\bs g})\A_l\x)^\top\A_l & \Ponep \diag(\A_l\bxi)^2 \Ponep 
							\end{bsmallmatrix}} \\
						=\tfrac{1}{mp}\ts \sum_{i,l} 
						\begin{bsmallmatrix} \gamma^2_i \ail \ail^\top & \ail\ail^\top\,(2\gamma_i \bxi - g_i \x)(\bs c^\perp_i)^\top \\ 
						{\bs c}^\perp_i  (2\gamma_i \bxi - g_i \x)^\top \ail\ail^\top & (\ail^\top \bxi)^2 {\bs c}^\perp_i ({\bs c}^\perp_i)^\top
						\end{bsmallmatrix}
						\end{gathered}$ 
						& {$\tfrac{1}{m} {\begin{bsmallmatrix} \|{\bs \gamma}\|^2 \I_n & 2 \bxi \eps^\top - \x {\bs e}^\top \\ 2 \eps \bxi^\top - {\bs e} \x^\top & \|\bxi\|^2 \Ponep \end{bsmallmatrix}}$}\\
						\midrule
						$(\bxi_0,\,{\bs \gamma}_0)$ & $\left(\tfrac{1}{mp}\ts \sum^{p}_{l = 1} \left(\A_l\right)^\top \diag({\bs g})\A_l \x,\, \vone_m\right)= \left(\tfrac{1}{mp}\ts \sum_{i,l} g_i \ail \ail^\top \x,\, \vone_m\right)$ & $\left(\tfrac{\|{\bs g}\|_1}{m}\x,\,\vone_m\right)$\\ 
						\bottomrule
					\end{tabular}}
		\caption{\label{tab:quants}Finite-sample and expected values of the objective function, its gradient and Hessian matrix, and the initialisation point for the problem in Def.~\ref{def:ncbc}~and in absence of noise.}
	\end{table*}
	%%%%%%%%%%%%%%%%%%%%%%%%%%
	
	To do this, we look at the geometric behaviour of the objective minimised in~\eqref{eq:bcp} to understand whether a region exists where the problem shows local convexity. Intuitively, we generate a random instance of~\eqref{eq:measurement-model-matrix} for $n=2$, $m=2$ and $\ail \sim_{\rm \iid} \cl N (\vzer_n, \I_n)$, with $(\x,\bs g)\coloneqq\big(\tfrac{1}{\sqrt 2}[1 -1]^\top,\vone_2+ \tfrac{\sqrt 2}{25}[1 -1]^\top\big) = (\x,\bs g)$. The amount of snapshots is varied as $p =\{2^0,2^1,\,\cdots, \infty\}$. We measure the $\log f(\bxi,{\bs\gamma})$ for $\bxi \in \bb S^1_2$ and {${\bs\gamma}  = \vone_2 + r  \tfrac{1}{\sqrt 2}[\begin{smallmatrix}1 & -1\end{smallmatrix}]^\top \in {\cl G}_{\rho}$} depending only on a parameter $r < \rho$. As for the case $p \rightarrow \infty$ we simply report the logarithm of $\Ex f(\bxi,{\bs\gamma})$ in Fig.~\ref{fig:excyl}. As $mp > n+m$ there is one global minimiser at {$(\x,{\bs g})$}~for $r = \rho = 8 \cdot 10^{-2} $, whose neighbourhood exhibits a smooth behaviour suggesting local convexity as $p$ increases, and thus the existence of a {\em basin of attraction} around the minimiser. In other words, there will be a critical value of $m p$ for which the objective function behaves arbitrarily close to its expectation which, as we will see below and suggested  by our example, is indeed locally convex for sufficiently small values of $\kappa,\rho$. This property is investigated in the following section.
	
		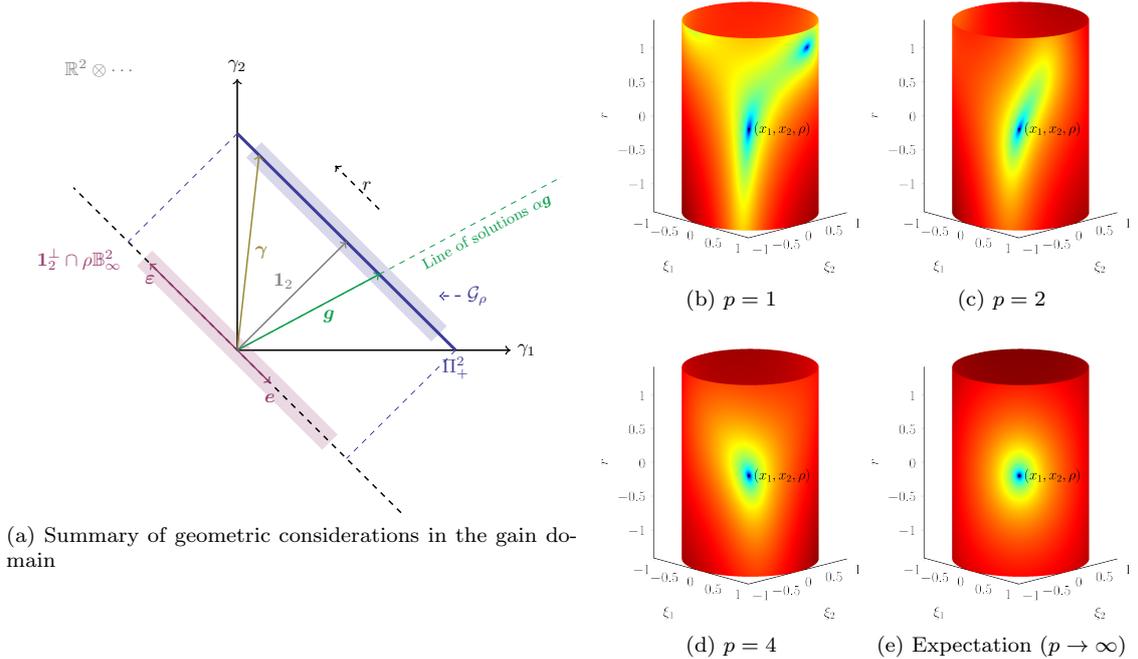
\begin{figure}[t]
			\null\hfill
			\begin{minipage}{0.45\textwidth}
				\subfloat[Summary of geometric considerations in the gain domain]{
					\tikzset{myptr/.style={thick,decoration={markings,mark=at position 1 with %
								{\arrow[scale=1]{>}}},postaction={decorate}}}
					\resizebox{\linewidth}{!}{
						\begin{tikzpicture}[scale=2]
						% Draw axes
						\draw [myptr] (0,0) -- (0,2.5) node (yaxis) [above] {$\gamma_2$};
						\draw [myptr] (0,0) -- (2.5,0) node (xaxis) [right] {$\gamma_1$};
						\fill[mathblue,opacity=0.2,rotate around={-45:(1,1)}] (-0.2,0.9) rectangle (2.2,1.1);
						\fill[mathpurp,opacity=0.2,rotate around={-45:(0,0)}] (-1.2,-0.1) rectangle (1.2,0.1);
						\node [align=right] at (-1.25,2.6) [gray] {$\R^2 \otimes \cdots$};
						\draw[mathblue,ultra thick] (0,2)--(2,0);
						\draw[mathsand,myptr] (0,0) -- (0.2,1.8) node [midway,anchor=west] {${\bs \gamma}$};
						\node at (2,0) [anchor=north,mathblue] {$\Pi^2_+$};
						\draw[black,thick, dashed] (1.5,-1.5) -- (-1.5, 1.5);
						\node at (-1,1) [anchor=north east,mathpurp] {$\vone^\perp_2 \cap \rho \B^2_\infty$};
						\draw[black,myptr,gray] (0,0) -- (1,1) node [midway,above, xshift=-1ex] {$\vone_2$};
						\draw[mathgrn,myptr] (0,0) -- (1.3,0.7) node [midway,right,xshift=1ex,yshift=-0.5ex] {${\bs g}$};
						\draw[mathgrn,dashed] (0,0) -- (2.9724,1.6) node [anchor=north,xshift=0ex,sloped,near end] {\footnotesize Line of solutions $\alpha {\bs g}$};
						\draw[mathpurp,myptr] (0,0) -- (0.3,-0.3) node [below, yshift=-0.5ex,mathpurp] {${\bs e}$};
						\draw[mathpurp,myptr] (0,0) -- (-0.8,0.8) node [below, yshift=-0.5ex,mathpurp] {$\eps$};
						\draw[mathblue,dashed] (2,0) -- (1,-1);
						\draw[mathblue,dashed] (0,2) -- (-1,1);
						\draw[black,myptr,dashed] (1.3,1.3) -- (0.9,1.7) node [midway,right,xshift=-0.25ex, yshift=0.25ex] {$r$}; 
						\node (lc) at (2.2,0.5) [mathblue] {${\cl G}_{\rho}$};
						\draw[mathblue,myptr,dashed] (lc) -- ++(-0.35,0);
						
						%\node[align=left] at (-1,-0.75) {\small $n=m=2$ \\ \small  $\ts{\bs g} = \vone_m + \rho \tfrac{1}{\sqrt{2}}\begin{bmatrix}-1 & 1\end{bmatrix}^\top$ \\ \small  $\ts{\bs\gamma} = \vone_m + r \tfrac{1}{\sqrt{2}}\begin{bmatrix}-1 & 1\end{bmatrix}^\top$\\\small  $\ts\ail \mathop{\sim}_{\rm \iid} \mathcal{N}(\vzer_n,\I_n)$};
						
						%Let $m = 2, n = 2, {\bs g} = \vone_m + \rho \tfrac{1}{\sqrt{2}}\begin{bmatrix}-1 & 1
						%\end{bmatrix}^\top, \x =\tfrac{1}{\sqrt{2}}\begin{bmatrix}1 & -1
						%\end{bmatrix}^\top,, \ail \displaystyle \mathop{\sim}_{\rm \iid} \mathcal{N}(\vzer_n,\I_n)$.\\[0.5ex]
						\end{tikzpicture}
					}
				}
			\end{minipage}
			\hfill
			\begin{minipage}{0.45\textwidth}
				\null\hfill % ,natwidth=1338,natheight=1432
				\subfloat[$p=1$]{\includegraphics[width=0.49\linewidth]{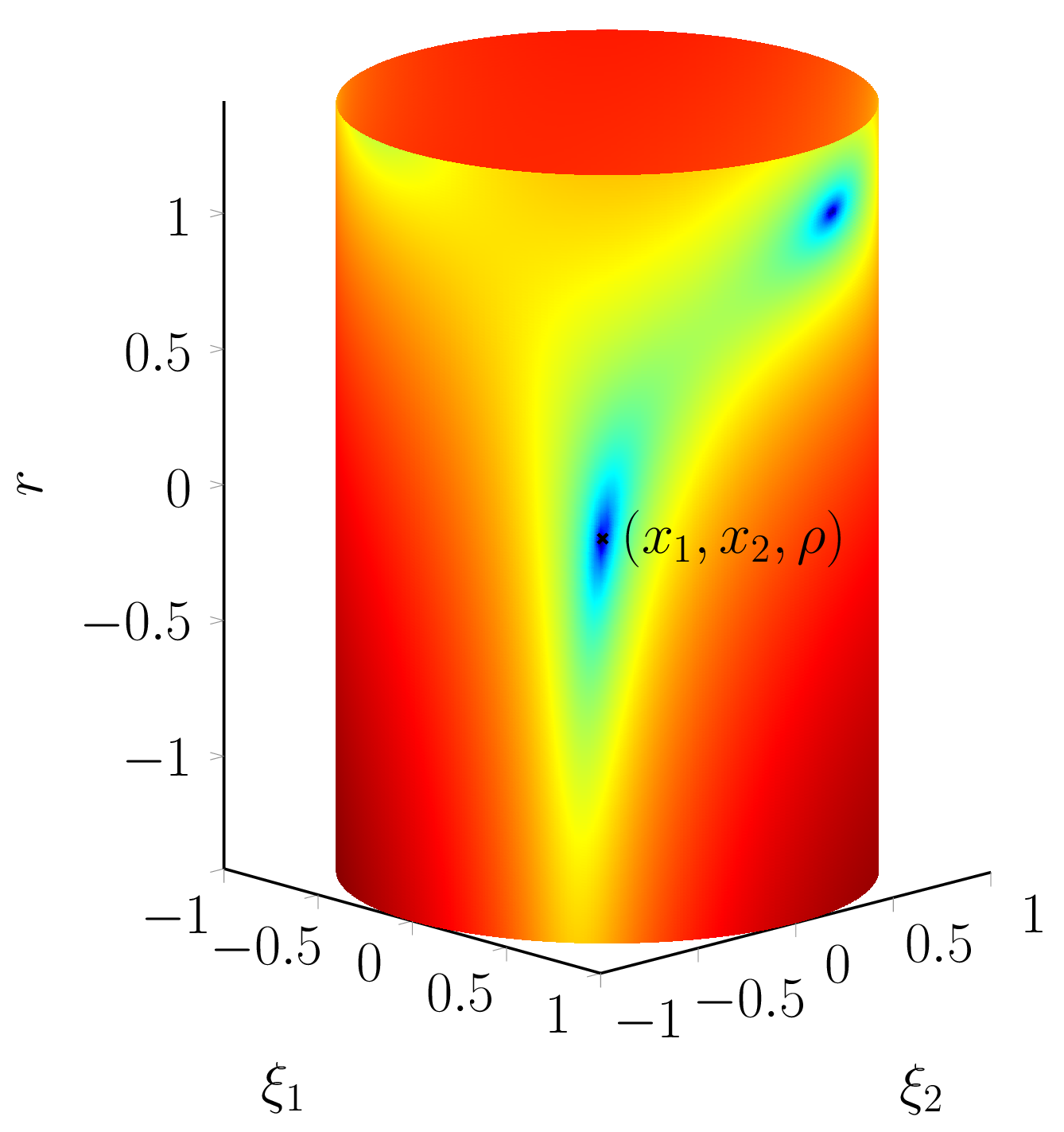}}\hfill
				\subfloat[$p=2$]{\includegraphics[width=0.49\linewidth]{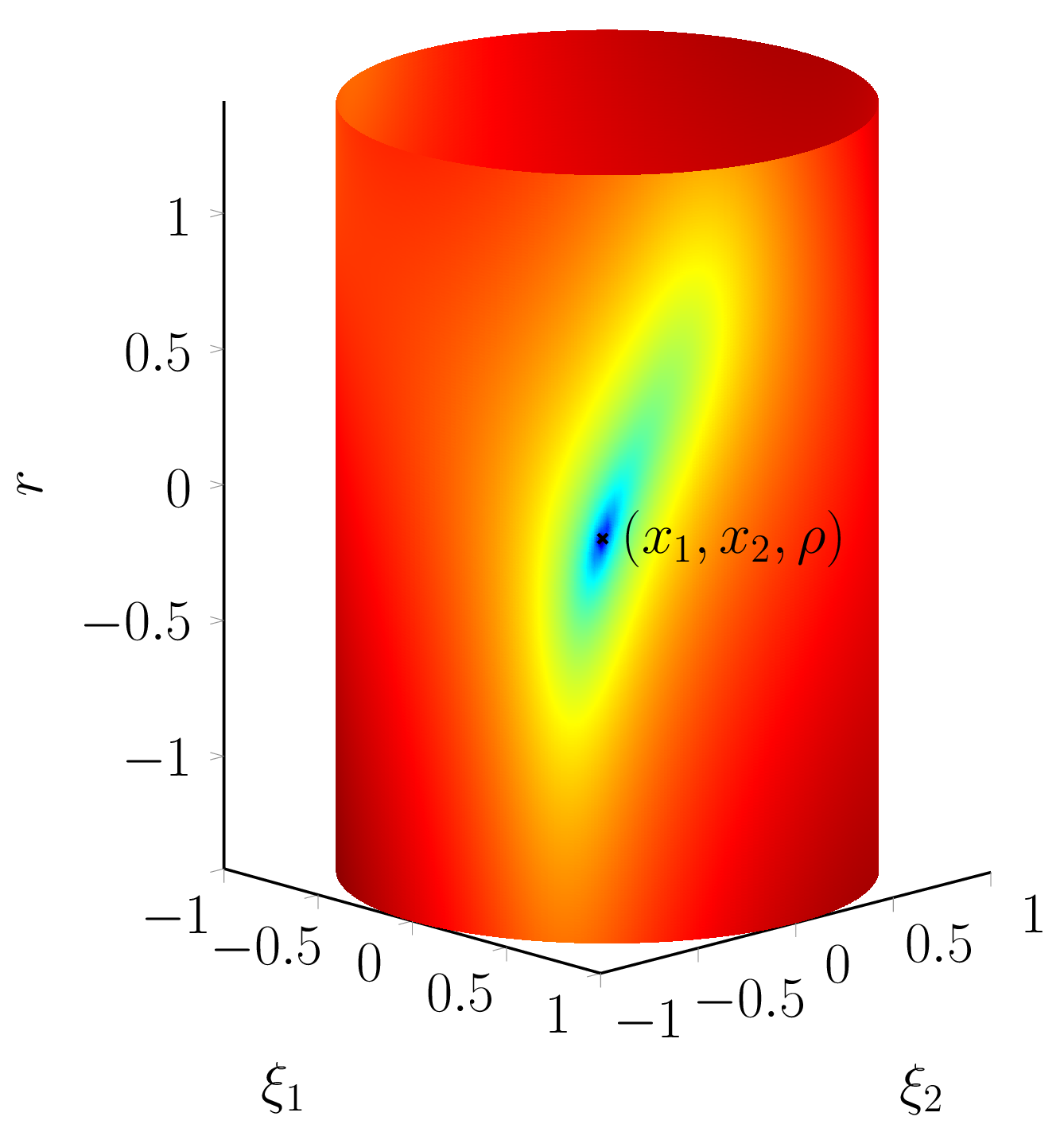}}\hfill\null\\
				\null\hfill
				\subfloat[$p=4$]{\includegraphics[width=0.49\linewidth]{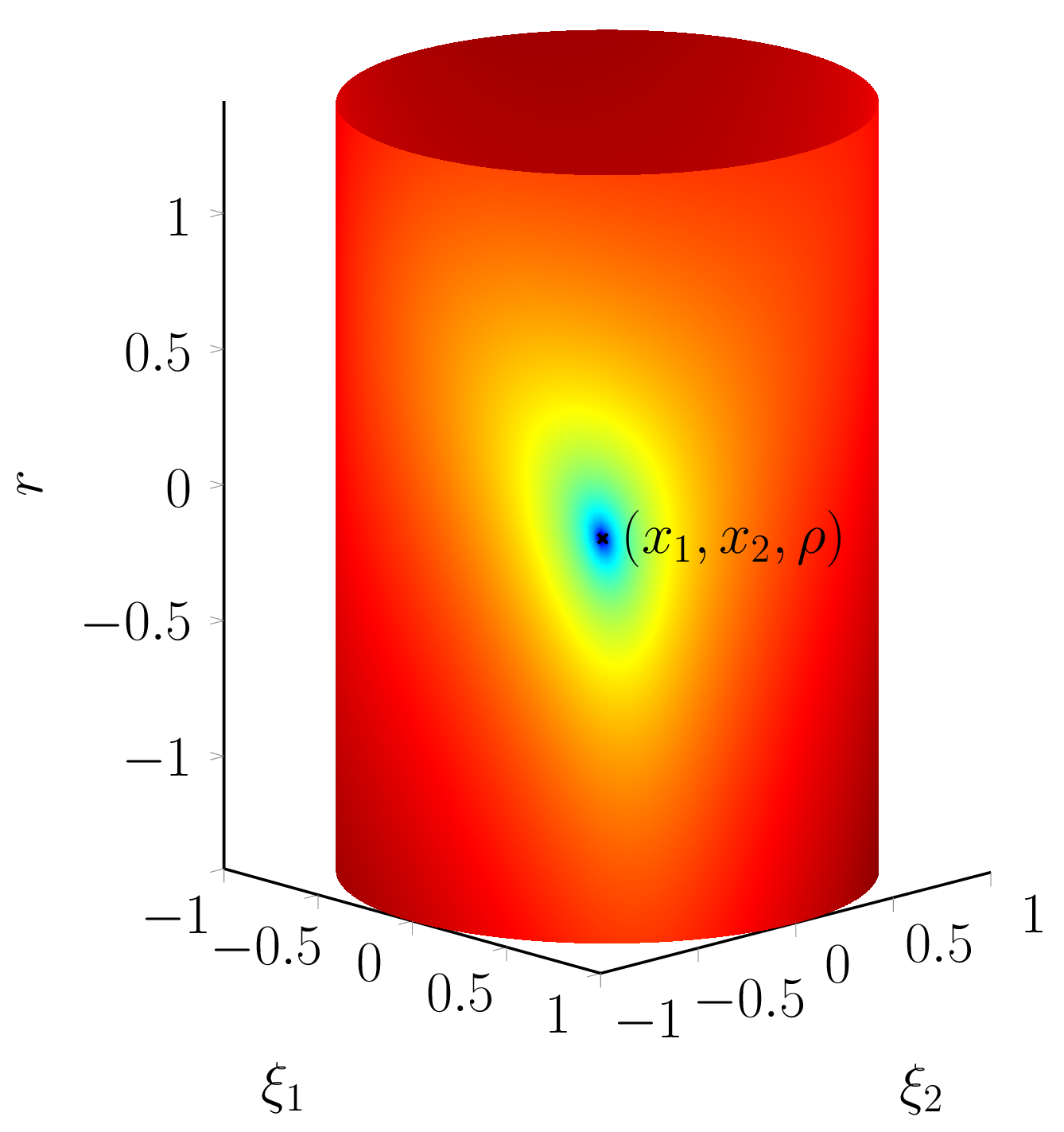}}\hfill
				\subfloat[Expectation $(p\rightarrow\infty)$]{\includegraphics[width=0.49\linewidth]{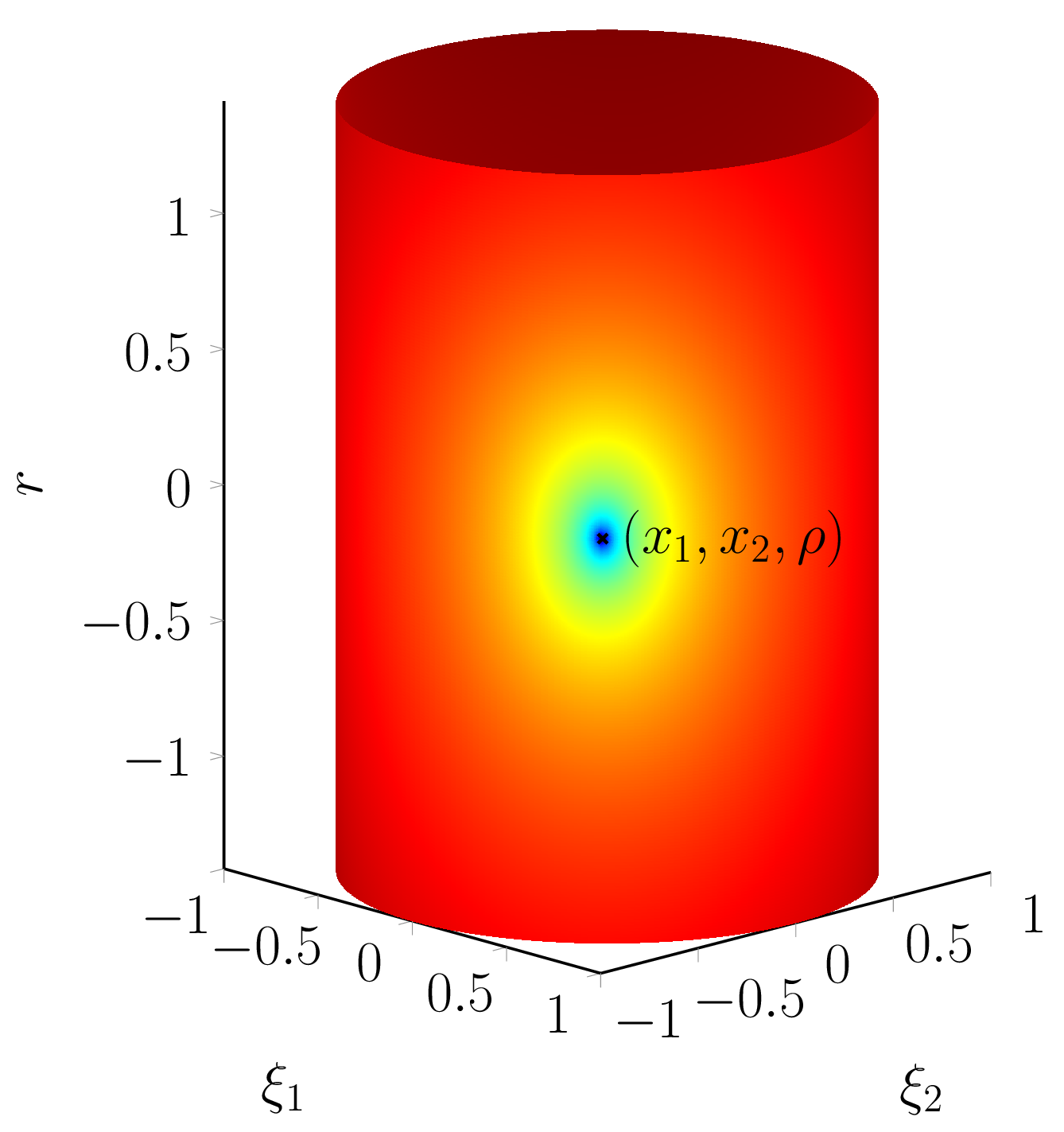}}\hfill\null\\
			\end{minipage}
			\hfill\null
			\caption{\label{fig:excyl}Geometric intuition on an instance of~\eqref{eq:bcp} for $n=2,m=2$: graphical representation of the parametrisation of the gain domain (a); heat map of $\log f(\bxi,{\bs\gamma})$ for the numerical example given in the text at $\|\bxi\|=1$ and increasing $p \rightarrow \infty$ (b--e). The vertical axis reports the parameter $r\in [-\sqrt 2, \sqrt 2]$.}
		\end{figure}
		
	\subsubsection{\new{Local Convexity in Expectation}}
	We proceed by highlighting two basic facts regarding~\eqref{eq:bcp} for $p \rightarrow \infty$, a case that is reported in Table~\ref{tab:quants}, where all finite-sample expressions of the quantities therein are unbiased estimates of their expectation with respect to  the \iid sensing vectors $\ail$. 
	To begin with, we define a set of test vectors $\cl V \coloneqq \R^n \times \onep$ used in the two following results. This is simply the set of all test vectors in the product space that are orthogonal to the direction $\vone_m$ in the gain domain. The proof of all following statements in this section is reported in App.~\ref{sec:appB}.
	\begin{proposition}[Global minimiser in expectation of~\eqref{eq:bcp}]
		\label{prop:expstatpoint}
		In expectation, the only stationary point of~\eqref{eq:bcp} is $\ts(\x,{\bs g}) \coloneqq \big( \tfrac{\|{\bs g}\|_1}{m} \x, \tfrac{m}{\|{\bs g}\|_1}{\bs g}\big)$. There, we have that $\Ex \H^\perp f(\x,{\bs g})\succ_{\cl V} 0$.
	\end{proposition}
	\noindent Once this stationary point is established, simple analysis of the projected Hessian matrix in expectation for all $(\bxi,{\bs\gamma}) \in {{\cl D}}_{\kappa,\rho}$ shows the following Proposition. 
	\begin{proposition}[{Convexity in expectation of~\eqref{eq:bcp} in a $(\kappa,\rho)$-neighbourhood}]
		\label{prop:expconvexity}
		For any $(\bxi,{\bs\gamma}) \in {\cl D}_{\kappa,\rho}$ for some $\ts{\kappa \in [0, 1), \rho \in \left[0,1-\tfrac{\sqrt{3} \kappa}{\sqrt{m}(1-\kappa)}\right)}$, we have that $\Ex \H^\perp f(\bxi,{\bs\gamma})\succ_{\cl V} 0$.
	\end{proposition}
	We remark two interesting aspects of the upper bound $\rho < 1-\tfrac{\sqrt{3} \kappa}{\sqrt{m}(1-\kappa)}$. Note how this can be made arbitrarily close to $1$ both when $m \rightarrow \infty$, \ie asymptotically in the number of measurements, and when $\kappa \rightarrow 0$, \ie when the basin of attraction is made arbitrarily small around the global minimiser. 
	
	Thus, in expectation~\eqref{eq:bcp} is locally convex on ${{\cl D}}_{\kappa,\rho}$ when the variations with respect to the gain domain are taken on $\onep$. However, this last information is not particularly useful in practice, since in expectation we would already have an unbiased estimate of $\x$ in the form of $\bxi_0$ in Table \ref{tab:quants}, from which ${\bs g}$ would be easily obtained. 
	This would suggest the need for a non-asymptotic analysis to check the requirements on $mp$ needed to benefit of the consequences of local convexity for finite $p$.  
	Rather than testing \new{local convexity by the positive-semidefiniteness of the Hessian matrix} (as done, \eg in Sanghavi \etal~\cite{WhiteSanghaviWard2015}) the theory we develop in Sec.~\ref{sec:recguars} follows the methodology of Cand\`es \etal~\cite{CandesLiSoltanolkotabi2015}, \ie a simple first-order analysis of the local properties of the gradient $\bs\nabla^\perp f(\bxi,{\bs\gamma})$ and initialisation point $(\bxi_0,{\bs \gamma}_0)$ which suffice to prove our recovery guarantees.

%%% Local Variables:
%%% mode: latex
%%% TeX-master: "iai_blindcalibration_2016"
%%% End:

\section{Blind Calibration by Projected Gradient Descent}
	\label{sec:solbypgd}
	In this section we discuss the actual algorithm used to solve \eqref{eq:bcp} in its non-convex form with the observations made in Sec.~\ref{sec:nonconv}. The algorithm is a simple projected gradient descent, as detailed below. We proceed by stating our method and providing right after the recovery guarantees ensured by a sample complexity requirement given on $m p$.
	\subsection{Descent Algorithm}
	\label{sec:algo}
	The solution of \eqref{eq:bcp} is here obtained as summarised in Alg.~\ref{alg1} and consists of an {\em initialisation} $(\bxi_0,{\bs \gamma}_0)$ followed by projected gradient descent. Similarly to~\cite{CandesLiSoltanolkotabi2015} we have chosen an initialisation $\bxi_0$ that is an unbiased estimator of the exact signal-domain solution as $p\rightarrow\infty$, \ie $\Ex\bxi_0 = \x$. This is indicated in Table~\ref{tab:quants} and straightforward since $\Ex \ail \ail^\top = \I_n$. For $p<\infty$ we will show in Prop.~\ref{prop:init} that, for $mp \gtrsim (n+m) \log n$ and sufficiently large $n$, the initialisation lands in $(\bxi_0,{\bs \gamma}_0) \in {\cl D}_{{\kappa},\rho}$ for $\rho \in [0, 1)$ and $\kappa>0$~with high probability. As for the gains, we initialise ${\bs \gamma}_0 \coloneqq \vone_m \ (\eps_0 \coloneqq \vzer_m)$. This initialisation allows us to establish the radius of the basin of attraction around $(\x, {\bs g})$. As a minor advantage, let us also remark that this initialisation is specific to our sensing model, but computationally cheaper than spectral initialisations such as those proposed in \cite{CandesLiSoltanolkotabi2015,LiLingStrohmerEtAl2016}. 
	
	Since $\rho < 1$ is small, we perform a few simplifications to devise our solver to \eqref{eq:bcp}. While generally we would need to project\footnote{Computationally this projection is entirely feasible, but would critically complicate the proofs.} any step in ${\bs \gamma}$ on $\cl G_\rho$, we first update the gains with the projected gradient $\bs\nabla^\perp_{\bs \gamma} f(\bxi_j, {\bs \gamma}_j)$ in step 5 (see Table \ref{tab:quants} for its expression). Then we apply ${\cl P}_{{\cl G}_{\rho}}$, \ie the projector on the closed convex set ${{\cl G}_{\rho}}$ (step 6). This is algorithmically easy, as it can be  obtained by, \eg alternate projection on convex sets\footnote{\new{Indeed, ${\cl P}_{{\cl G}_{\rho}}$ is obtained algorithmically by simple alternate projections between $\Ponep$ and ${\cl P}_{\rho\,\B^m_\infty}$ and adding $\vone_m$ to the output.}}. However, 
	 this step is merely a formal requirement to ensure that each iterate ${\bs \gamma}_{j+1} \in {\cl G}_{\rho} \subset \Pi^m_+$ for some fixed $\rho$ when proving the convergence of Alg.~\ref{alg1} to $(\x,{\bs g})$, but practically not needed, as we have observed that step 6 can be omitted since $\check{{\bs \gamma}}_{j+1} \in {\cl G}_{\rho}$ is always verified in our experiments. 
	
	Thus, Alg.~\ref{alg1} is as simple and efficient as a first-order descent algorithm with the projected gradient $\bs\nabla^\perp f(\bxi,{\bs \gamma})$. The proposed version performs two line searches in step 3 that can be solved in closed-form at each iteration, as made clearer in Sec.~\ref{sec:stepsize}. These searches are simply introduced to improve the convergence rate, but are not necessary and could be replaced by a careful choice of some fixed steps $\mu_\bxi,\mu_{\bs \gamma} > 0$ (for which, in fact, our main theoretical results are developed).
	
	\begin{algorithm}[!t]
		\caption{Non-Convex Blind Calibration by Projected Gradient Descent.}
		\label{alg1}
		{
			\begin{algorithmic}[1]
				\STATE Initialise $\bxi_0 \coloneqq \tfrac{1}{m p}
				\sum^{p}_{l = 1} \left(\A_l\right)^\top \y_l,
				\, {\bs \gamma}_0 \coloneqq \vone_m, \, j \coloneqq0$. 
				\WHILE{stop criteria not met}
				\STATE $\begin{cases} 
				\mu_\bxi \coloneqq \argmin_{\upsilon \in \R} f(\bxi_{j}-\upsilon\bs\nabla_\bxi f({\bxi_{j}, {\bs \gamma}_{j}}), {\bs \gamma}_{j}) \\
				\mu_{\bs \gamma} \coloneqq \argmin_{\upsilon \in \R} f(\bxi_{j},{\bs \gamma}_{j}-\upsilon\bs\nabla^\perp_{\bs \gamma} f({\bxi_{j}, {\bs \gamma}_{j}})) 
				\end{cases}$
				\STATE $\bxi_{j+1} \coloneqq \bxi_{j} - \mu_\bxi \bs\nabla_\bxi f({\bxi_{j}, {\bs \gamma}_{j}})$
				\STATE $\check{{\bs \gamma}}_{j+1} \coloneqq {\bs \gamma}_{j} - \mu_{\bs \gamma} \, \bs\nabla^\perp_{\bs \gamma} f({\bxi_{j}, {\bs \gamma}_{j}})$ 
				\STATE ${{\bs \gamma}}_{j+1} \coloneqq {\cl P}_{{\cl G}_{\rho}} \check{{\bs \gamma}}_{j+1}$
				\STATE $j \coloneqq j + 1$
				\ENDWHILE
			\end{algorithmic}
		}
	\end{algorithm}
	
	\subsection{Convergence Guarantees}
	\label{sec:recguars}
	We now obtain the conditions that ensure convergence of this descent algorithm to the exact solution $(\x,{\bs g})$ for some fixed step sizes $\mu_\bxi,\mu_{\bs \gamma}$. This is done in three steps: $(i)$ the initialisation $(\bxi_0,{\bs \gamma}_0)$ is shown to lie in ${\cl D}_{{\kappa},\rho}$ for $\rho \in [0,1)$ and some small $\kappa$ with high probability, \ie the probability that this is not verified decays exponentially in the problem dimensions; $(ii)$ {\em uniformly} on ${\cl D}_{{\kappa},\rho}$ we are able to show that, with high probability, a gradient descent update decreases the distance to $(\x,{\bs g})$ by bounding the magnitude of $\bs\nabla^\perp f(\bxi,{\bs \gamma})$ as well as its angle with the direction of the global minimiser; $(iii)$ still by uniformity, applying this property to any iterate $(\bxi_j,{\bs \gamma}_j)$ of Alg.~\ref{alg1} leads to finding some fixed step values $\mu_\bxi,\mu_{\bs \gamma}$ that grant convergence to $(\x,{\bs g})$ as $j\rightarrow \infty$, \ie what is commonly referred to as {\em linear convergence} (with respect to $\log j$). 
	The proofs of all following statements in this section are reported in App.~\ref{sec:appC}. The general proof strategy relies on concentration inequalities for functions of our~\iid sub-Gaussian sensing vectors ${\ail}$ as well as a particular application of matrix Bernstein's inequality. These fundamental mathematical tools are defined in App.~\ref{sec:appA}. 
	
	We first focus on the initialisation and assess its non-asymptotic properties in terms of the distance attained with respect to the global minimiser.
	\begin{proposition}[Initialisation Proximity]
		\label{prop:init}
		Let $(\bxi_0,{\bs \gamma}_0)$ be as in Table~\ref{tab:quants}. Given ${\delta} \in (0,1)$, $ t \geq 1$, provided $\ts mp \gtrsim {\delta}^{-2} (n + m) \log \tfrac{n}{\delta} $ and $n \gtrsim t \log mp$, with probability exceeding 
		\begin{equation}
		\label{eq:bound1}
		1 - C e^{- c {\delta}^2 m p} - (mp)^{-t}
		\end{equation}
		for some $C, c > 0$,  we have that $\|\bxi_0 - \x\| \leq {\delta}
		\|\x\|$. 
		Since ${\bs \gamma}_0 \coloneqq\bs 1_m$ we also have $\|{\bs \gamma}_0 - {\bs g}\|_\infty\leq\rho<1$. Thus $(\bxi_0,{\bs \gamma}_0) \in
		{\cl D}_{{\kappa},\rho}$ with the same probability and $\kappa_0 \coloneqq \sqrt{\delta^2+\rho^2}$.
	\end{proposition}
	
	Secondly, we develop the requirements for convergence, \ie for ${\cl D}_{{\kappa_0},\rho}$ to be a {basin of attraction} around the global minimiser. Provided that the initialisation lands in ${\cl D}_{{\kappa_0},\rho}$, any update from $(\bxi,{\bs \gamma}) \in {\cl D}_{{\kappa_0},\rho}$ to some $\bxi_+ \coloneqq \bxi - \mu_\bxi \bs\nabla_\bxi f({\bxi, {\bs \gamma}}), \check{{\bs \gamma}}_+ \coloneqq {\bs \gamma} - \mu_{\bs \gamma} \bs\nabla^\perp_{\bs \gamma} f({\bxi,{\bs \gamma}})$ has distance from the solution
	\begin{align}
	{\Delta} (\bxi_+,\check{{\bs \gamma}}_+) & = {\Delta} (\bxi,{\bs \gamma}) -2 \big(\mu_\bxi \langle \bs\nabla_\bxi f(\bxi,{\bs \gamma}), \bxi-\x\rangle + \mu_{\bs \gamma} \tfrac{\|\x\|^2}{m}\langle \bs\nabla^\perp_{\bs \gamma} f(\bxi,{\bs \gamma}), {\bs \gamma}-{\bs g}\rangle\big) \nonumber  \\ 
	& \quad + {\mu_\bxi^2} \|\bs\nabla_\bxi f(\bxi,{\bs \gamma})\|^2 + {\mu_{\bs \gamma}^2}\tfrac{\|\x\|^2}{m}\|\bs\nabla^\perp_{\bs \gamma} f(\bxi,{\bs \gamma})\|^2. \label{eq:updatedist-main}
	\end{align}	
	
	As detailed in App.~\ref{sec:appC}, it is therefore clear from \eqref{eq:updatedist-main} that \new{a lower bound on $\big\langle \bs\nabla^\perp f(\bxi, {\bs \gamma}),\, [ (\bxi-\x)^\top ({\bs \gamma}-{\bs g})^\top]^\top \big\rangle$ and upper bounds on both $\|\bs\nabla_{\bs \xi} f(\bxi,{\bs \gamma})\|$ and $\|\bs\nabla^\perp_{\bs \gamma} f(\bxi,{\bs \gamma})\|$}, holding for all points in ${\cl D}_{{\kappa_0},\rho}$ will yield the requirements to attain the desired convergent behaviour (for sufficiently small step sizes), also provided these bounds can be expressed as a function of $\Delta(\bxi,{\bs \gamma})$. This leads to the following proposition formulated for a general neighbourhood of radius $\kappa>0$.  
	
	\begin{proposition}[Regularity condition in ${\cl D}_{{\kappa},\rho}$]
		\label{prop:regu}
		Given $\delta \in (0,1)$, $t\geq1$ and $\kappa > 0$, provided $\rho < \tfrac{1-4\delta}{31}$ and 
                \begin{align*}
                  % \label{eq:glob-requirements-bis}
                  \ts n&\ts \gtrsim t \log(mp),\\
                  \ts mp&\ts \gtrsim t\delta^{-2} (n + m)\, 
                          \log^2(m(p+n)) \log(\frac{1}{\delta}).
                \end{align*}
                with probability exceeding $1 - C (mp)^{-t}$ for some $C>0$, we have for all $(\bxi,{\bs \gamma}) \in {\cl D}_{{\kappa},\rho}$
		\begin{align}
                  \label{eq:bounded-curvature}
		\left\langle \bs\nabla^\perp f(\bxi, {\bs \gamma}), \begin{bsmallmatrix} \bxi-\x \\ {\bs \gamma}-{\bs g} \end{bsmallmatrix} \right\rangle  & \geq \eta\, {\Delta}(\bxi,{\bs \gamma}), 
		\end{align}\new{\vspace{-7mm}
		\begin{align}
		\label{eq:bounded-lipschitz-xi}        
		\|\bs\nabla_{\bs \xi} f(\bxi, {\bs \gamma}) \|^2 & \leq L_{\bs \xi}^2\ {\Delta}(\bxi,{\bs \gamma}),\\
		\label{eq:bounded-lipschitz-gamma-perp}        
		\|\bs\nabla_{\bs \gamma}^\perp f(\bxi, {\bs \gamma}) \|^2 & \leq L_{\bs \gamma}^2\ {\Delta}(\bxi,{\bs \gamma}),
		\end{align}
                 where $\eta \coloneqq 1-31\rho-4\delta \in (0,1)$, $L_{\bs \xi} = 8 \sqrt 2$ and $L_{\bs \gamma} \coloneqq 4\sqrt 2(1+\kappa) \|\bs x\|$.}
	\end{proposition}
	
	The interpretation of \eqref{eq:bounded-curvature}, \eqref{eq:bounded-lipschitz-xi} and \eqref{eq:bounded-lipschitz-gamma-perp} is clear, and analogue to the method pursued by Cand\`es \etal~\cite{CandesLiSoltanolkotabi2015}~in the context of phase retrieval. The first condition or {\em bounded curvature}~ensures that the angle between the gradient and the direction of the global minimiser is not too large. The second condition or {\em Lipschitz gradient} gives a bound on the maximum magnitude that the gradient is allowed to assume. Moreover, the bounded curvature in \eqref{eq:bounded-curvature} implies that 
	\[
	\|\bs\nabla^\perp f(\bxi, {\bs \gamma})\| \left\|\begin{bsmallmatrix} \bxi-\x \\ {\bs \gamma}-{\bs g} \end{bsmallmatrix} \right\| \geq \left\langle \bs\nabla^\perp f(\bxi, {\bs \gamma}), \begin{bsmallmatrix} \bxi-\x \\ {\bs \gamma}-{\bs g} \end{bsmallmatrix} \right\rangle \geq \eta\, {\Delta}(\bxi,{\bs \gamma}) > 0,
	\]
	which holds for all points $(\bxi, \bs \gamma) \in \cl D_{\kappa,\rho}$ except the solution $(\x, \bs g)$ in which the distance is $0$. Since for all such points $\left\|\begin{bsmallmatrix} \bxi-\x \\[-.5mm] {\bs \gamma}-{\bs g} \end{bsmallmatrix} \right\| > 0$, this straightforwardly proved the following claim. \new{Moreover, the orthogonal projection $\Ponep$ also implies $\|\bs\nabla f(\bxi,{\bs \gamma})\| \geq  \|\bs\nabla^\perp f(\bxi, {\bs \gamma})\| > 0$.}
	\begin{corollary}[Uniqueness of the minimiser in ${\cl D}_{{\kappa},\rho}$]
		\label{coro:unique}
		Under the conditions of Prop.~\ref{prop:regu}, the only stationary point of \eqref{eq:bcp} in ${\cl D}_{\kappa,\rho}$ is $\ts(\x,{\bs g})$. 
	\end{corollary}
	\medskip
	
	With the conditions obtained in Prop.~\ref{prop:regu} we are in the position of finding the values of the step sizes $\mu_\bxi, \mu_{\bs \gamma}$ that make convergence to the global minimiser possible, clearly after jointly verifying the previous properties of the initialisation and the basin of attraction established in Prop.~\ref{prop:init}.
	
	\begin{theorem}[Provable Convergence to the Exact Solution]
		\label{theorem:convergence}
		Given $\delta \in (0,1)$ and $t\geq1$, let us take $\ts n \gtrsim t \log(mp)$, 
                  $mp \gtrsim t\delta^{-2} (n + m) \log^2 (m(p+n))\log(\frac{1}{\delta})$, $\rho < \tfrac{1-4\delta}{31}$. There exists $\mu_0 >0$ with% \footnote{The actual constant is loosely estimated as $\mu_0 \leq \tfrac{1}{2048} \cdot \frac{\min(1, \|\bs x\|^2)}{m(1+\|\bs x\|^2)} \eta$, as reported in the proof of this result.}
					\new{$$
					\ts \mu_0 \lesssim \frac{\eta}{m},
					$$}%
for $\eta = 1-4\delta-31\rho \in (0,1)$ such that for any $0<\mu<\mu_0$ and with probability exceeding $1 - C (mp)^{-t}$ for some $C>0$, Alg.~\ref{alg1} with step sizes set to $\mu_\bxi \coloneqq \mu$ and $\mu_{\bs \gamma} \coloneqq \mu \tfrac{m}{\|\x\|^2}$ has distance decay
		\begin{equation}
		\ts {\Delta}(\bxi_j,{\bs \gamma}_j)\leq
		(1- \eta \mu)^j \big(\delta^2+\rho^2)\|\x\|^2,
		\label{eq:convres}
		\end{equation}
		at any iteration $j > 0$. Hence, ${\Delta}(\bxi_j,{\bs \gamma}_j)\underset{j\rightarrow\infty}{\longrightarrow}0$. 
	\end{theorem}

	Thus, we have verified that performing blind calibration by means of our non-convex algorithm in the noiseless case of Def.~\ref{def:model} provably recovers the exact solution with linear convergence, under some requirements on $\rho$ due to Prop.~\ref{prop:regu} and essentially when ${m}p \gtrsim (n+m) \log^2(m(p+n))$. 
	
	It is worth noting that the condition on $m p$ appearing in Prop.~\ref{prop:regu} improves upon the rate advocated in~\cite{CambareriJacques2016}, $\sqrt{m} p \gtrsim (n+m)\log n$. The additional $\sqrt{m}$ factor was indeed due to some technical conditions when proving the regularity condition, that were here refined using a more sophisticated application of matrix Bernstein's inequality~\cite{Tropp2015}. 
	
	Moreover, the experiments suggest that the requirements on $\rho$ are actually much weaker than the conditions given in this theorem (and in Prop.~\ref{prop:regu}). 	
	Interestingly, we also note that the ratio between $\tfrac{\mu_{\bs \xi}}{\mu_{\bs \gamma}} \simeq \tfrac{\|\bs x\|^2}{m}$ is also observed in our experimental results when these step sizes are not fixed, but rather varied according to the line searches reported in Alg.~\ref{alg1}.  
	Finally, we remark that the obtained sample complexity ${m}p \gtrsim (n+m) \log^2 (m(p+n))$ is essentially equivalent to that found in prior literature when devising recovery guarantees for different algorithms~\cite{AhmedRechtRomberg2014,LiLingStrohmerEtAl2016}.
	
%%% Local Variables:
%%% mode: latex
%%% TeX-master: "iai_blindcalibration_2016"
%%% End:
		
\section{Blind Calibration with Subspace Priors}
		\label{sec:blind-calibration-with-subspace-priors}
		
			%			\begin{figure}[t]
			%				\centering
			%				\includegraphics[width=3in]{bcsubcoh.pdf}
			%				\caption{\label{fig:subscoh}Geometric intuition on the role of coherence in \eqref{eq:bcpk}: for $m = 3$, $h = 2$ we have highlighted $\cl G_\rho = \vone_m + \bs 1^\perp_m \cap \rho \bb B^m_\infty$ for $\rho = 0.99$, and $\cl B_\rho$ as the intersection of $\cl G_\rho$ with the subspace $\cl B$ (dashed: identity, dotted: DCT).}
			%			\end{figure}
			
			We now apply the same principles developed in Sec.~\ref{sec:nonconv} to define a simple modification of Alg.~\ref{alg1} specialised to the sensing model in Def.~\ref{def:modelsub},~\ie when 
                        \begin{equation}
                          \label{eq:sub-sens-model}
                          \ts \y_l = \diag(\bs B \bs b) \A_l {\bs Z} {\bs z}\ \in \R^m,\quad l \in [p].
                        \end{equation}
			Thus, we consider known subspace priors on the signal ${\bs x} = {\bs Z} {\bs z} \in \cl Z \subset \bb R^n$ and the gains ${\bs g} = \bs B \bs b \in \cl B \subset \bb R^m_+$. To better specify the latter subspace, we resume our previous considerations in Sec.~\ref{sec:nonconv} and also assume that ${\bs g} = \vone_m + \bs e \in \vone_m + \onep \cap \rho \B^m_\infty$ for $\rho < 1$,~\ie the gains are small perturbations around $\vone_m$. It is then clear that $\bs g \in {\cl B} \cap (\vone_m +\onep \cap \rho \B^m_\infty)$. For simplicity we can then consider that our basis for $\cl B$ is  $\bs B \coloneqq \big(\tfrac{\vone_m}{\sqrt{m}},\ \bs B^\perp\big)$ where $\bs B^\perp \in \R^{m\times h-1}$ is comprised of $h-1$ orthonormal vectors in $\R^m$ that span $\onep$. To fix ideas, such $\bs B$ could be constructed as $h$ basis elements of the discrete cosine transform (DCT), including the so-called ``DC component'' $\tfrac{\vone_m}{\sqrt{m}}$. To conclude our specification of $\bs g = \bs B \bs b$ the coefficients $\bs b \coloneqq \big(\sqrt m,\ (\bs b^\perp)^\top\big){}^\top$, where $\bs b^\perp \in \R^{h-1}$ corresponds to ${\bs e} = \bs B^\perp \bs b^\perp  \in \onep \cap \rho \B^m_\infty$. 
			This holds in all generality, since we can apply a scaling factor to the gains ${\bs g} \in \R^m_+$ so that they follow this model. In this setup, the blind calibration problem is summarised as the following definition.
			
			\begin{definition}[Non-convex Blind Calibration Problem with Subspace Priors]
				\label{def:blind-calibr-with-sub}
				Given two subspaces $\cl B \subset \bb R^m$ and $\cl Z \subset \bb R^n$, with dimension $h \coloneqq \dim \cl B \leq m$ and $k \coloneqq \dim \cl Z \leq n$, and with orthonormal bases $\bs B \in \R^{m\times h}$ and ${\bs Z} \in \R^{n \times k}$ (\ie tight frames), respectively, we define {\em non-convex blind calibration with subspace priors} the optimisation problem 
					\begin{equation}%
					\label{eq:bcpk}
					(\bar{{\bs z}}, \bar{\bs b}) \coloneqq {\mathop{\argmin}_{({\bs \zeta},{\bs \beta}) \in {\cl Z} \times \cl B_\rho}} \tfrac{1}{2mp} {\ts \sum_{l=1}^{p} \left\Vert \diag({\bs B }{\bs \beta}) \A_l {\bs Z} {\bs \zeta} - \y_l\right\Vert^2},
					\end{equation}%
				where 
				\begin{equation}
				\label{eq:subconstr}
				{\cl B}_\rho \coloneqq \left\{\bs v \in \R^h : \bs B \bs v \in \vone_m + \bs 1^\perp_m \cap \rho \bb B^m_\infty\right\},
				% {\cl B}_\rho \coloneqq \left\{\bs v \in \R^h : v_1 = \sqrt{m}, \bs B \bs v \in \vone_m + \bs 1^\perp_m \cap \rho \bb B^m_\infty\right\},
				\end{equation}
				% \mtodo{Fix this definition. The one of $\cl B$ is wrong.} -- I think this is correct? With the given B it must be like this. The first coefficient associated to \vone_m/\sqrt{m} must be \sqrt{m}.
				given $\rho < 1$, $\y_l = \diag({\bs B}{\bs b}) \A_l {\bs Z} {\bs z} \in \R^m$, $l \in [p]$ as in Def.~\ref{def:modelsub}.
			\end{definition}
			Moreover, we also define\footnote{Hereafter $\upmu_{\max}(\bs B)$ always refers to $\bs B$, as shall be clear from the context. Hence the simplified notation $\upmu_{\max}$.} the \emph{coherence} of $\bs B$ as 
			\new{\begin{equation}
				\label{eq:mumax}
				\ts\upmu_{\max} \coloneqq \sqrt\frac{m}{h} \max_{i \in [m]} \|\bs B^\top \bs c_i\| \, \in \, [1, \sqrt\frac{m}{h}].    
			\end{equation}
			}
	        This quantity also appears in \cite{AhmedRechtRomberg2014} and has a fundamental role in characterising $\bs B$ and its effect on the possibility of recovering $\bs b$. It measures how evenly the energy of $\bs B$ is distributed among its rows, given that the basis vectors (\ie its columns) are orthonormal. In particular, two classical examples are in order. As mentioned above, we could form $\bs B$ as $h$ basis elements drawn at random from the $m$ elements that define the DCT in $\R^m$. \new{Since the DCT is a \emph{universal basis}, with entries smaller than $c/\sqrt m$ in amplitude for some constant $c>0$ \cite{puy2012universal}, t}his actually leads to a low-coherence basis $\bs B$, \ie $\upmu_{\max} \simeq 1$ \new{as imposed by $\|\bs B^\top \bs c_i\|^2 = \sum_{j=1}^h |B_{ij}|^2 \leq c^2 h/m$}. On the other hand, if\footnote{This case is not possible in our model since $\tfrac{\vone_m}{\sqrt m}$ must be a basis vector. It is however easy to find highly coherent examples of $\bs B$ containing $\frac{\vone_m}{\sqrt{m}}$ and working on the remaining $h-1$ basis vectors of $\bs B^\perp$.} 
	        $\bs B := [\bs I_h, \vzer_{h \times (m-h)}]^\top$ it is straightforward that $\upmu_{\max} = \sqrt\frac{m}{h}$, \ie the worst case setting in which only $h$ out of $m$ rows are contributing to the energy in $\bs B$. We shall use this quantity in deriving the sample complexity for the known subspace case. 

			Let us now proceed by taking the gradient of the objective function
			\begin{equation}
\label{eq:sub-f-cost}
				f^{\rm s}({\bs \zeta},{\bs \beta}) \coloneqq  \tfrac{1}{2mp} {\ts \sum_{l=1}^{p} \left\Vert \diag({\bs B}{\bs \beta}) \A_l {\bs Z} {\bs \zeta} - \y_l\right\Vert^2},
			\end{equation}
			that is the vector $\bs\nabla f^{\rm s}({\bs \zeta},{\bs \beta}) \coloneqq 
			\begin{bsmallmatrix} \bs\nabla_{{\bs \zeta}} f^{\rm s}({\bs \zeta},{\bs \beta})^\top &  \bs\nabla_{{\bs \beta}} f^{\rm s}({\bs \zeta},{\bs \beta})^\top  \end{bsmallmatrix}^\top$. 
			Firstly, in the signal-domain subspace we have 
			\begin{align}
				\bs\nabla_{{\bs \zeta}} f^{\rm s}({\bs \zeta},{\bs \beta}) & =  \tfrac{1}{mp} \ts\sum^p_{l=1} {\bs Z}^\top \A^\top_l \diag({\bs B}{\bs \beta}) \left(\diag({\bs B}{\bs \beta}) \A_l {\bs Z} {\bs \zeta} - \diag({\bs B \bs b}) \A_l {\bs Z} {\bs z}\right) \nonumber \\ 
				\label{eq:gradw-bis}
				& = \tfrac{1}{mp} \ts\sum_{i,l}  (\bs c^\top_i {\bs B}{\bs\beta}){\bs Z}^\top  \ail \left[(\bs c^\top_i {\bs B}{\bs\beta})  \ail^\top \bs Z \bs\zeta  - (\bs c^\top_i {\bs B}{\bs b}) \ail^\top \bs Z \bs z\right]\\
				& = {\bs Z}^\top \bs\nabla_{\bxi} f(\bxi,{\bs \gamma})\big\vert_{\bxi = {\bs Z} {\bs \zeta}, {\bs \gamma} = {\bs B}{\bs \beta}}
				\label{eq:gradw}
			\end{align}
			Then we can take the partial derivatives with respect to the gain-domain subspace components $\beta_j, j\in[h]$, yielding
			\begin{align*}
				\frac{\partial}{\partial \beta_j}  f^{\rm s}({\bs \zeta},{\bs \beta}) & \coloneqq \tfrac{1}{mp} \ts\sum^p_{l=1} \left(\frac{\partial}{\partial \beta_j} (\sum^{h}_{t=1} {\bs B}_{\cdot, t} \beta_t)\right)^\top \diag(\A_l {\bs Z} {\bs \zeta}) \left(\diag({\bs B}{\bs \beta}) \A_l {\bs Z} {\bs \zeta} - \diag({\bs B \bs b}) \A_l {\bs Z} {\bs z}\right) \\
				& = \tfrac{1}{mp} \ts\sum_{i,l} ({\bs B}^\top_{\cdot,j} {\bs c}_i)  ( \ail^\top {\bs Z} {\bs\zeta}) \left[({\bs c}^\top_i {\bs B}{\bs\beta}) \ail^\top \bs Z \bs\zeta  - ({\bs c}^\top_i {\bs B} {\bs b}) \ail^\top \bs Z \bs z\right] \\
				& = {\bs B}^\top_{\cdot,j} \bs\nabla_{{\bs \gamma}} f(\bxi,{\bs \gamma})\big\vert_{\bxi = {\bs Z} {\bs \zeta}, {\bs \gamma} = {\bs B}{\bs \beta}}.
			\end{align*}
			where $\bs B_{\cdot,j} \in \R^m, j\in[h]$ denotes the columns of $\bs B$. Thus, we can collect 
			\begin{align}
				\label{eq:gradz}
				\bs\nabla_{{\bs \beta}} f^{\rm s}({\bs \zeta},{\bs \beta}) & = \tfrac{1}{mp} {\bs B}^\top \left(\ts\sum_{i,l} ( \ail^\top {\bs Z} {\bs\zeta}) \left[({\bs c}^\top_i {\bs B}{\bs\beta}) \ail^\top \bs Z \bs\zeta  - ({\bs c}^\top_i {\bs B} {\bs b}) \ail^\top \bs Z \bs z\right] {\bs c_i}\right) \nonumber \\
				& =  {\bs B}^\top \bs\nabla_{{\bs \gamma}} f(\bxi,{\bs \gamma})\big\vert_{\bxi = {\bs Z} {\bs \zeta}, {\bs \gamma} = {\bs B}{\bs \beta}}.
			\end{align}
			% \mtodo{This part is {\em wrong} as it is. It should be rewritten with a definition of $\bs\nabla^\perp f^{\rm s}(\bs \zeta, \bs \beta)$, of $\bs\nabla^\perp_{\bs \beta} f(\bs\zeta,\bs\beta)$, etc.}
			% \mtodo{Replace all instances of the gradient with the new explanation and definition.}
			
			We now elaborate on the constraint~\eqref{eq:subconstr} in a fashion similar to what we did for $\bs \gamma \in \cl G_\rho$ in Sec.~\ref{sec:nonconv}. The constraint now imposes $\bs \beta \in \cl B_\rho \subset \vone_m + \onep$, so the steps of our descent algorithm will still lie on $\onep$ in the gain domain. However, since the optimisation in \eqref{eq:bcpk} is with respect to the subspace coefficients $\bs\beta$, the orthogonal projector that maps $\bs v \in \bb R^h$ on the projection of $\onep$ in the subspace $\cl B$ 
			is instead $\bs P^{\rm s}_{\onep} \coloneqq \bs B^\top \Ponep \bs B = \diag \big(\begin{bsmallmatrix} 0 & \vone^\top_{h-1}\end{bsmallmatrix}\big)$, that is the operator that sets the first component to $v_1 = 0$. Hence, we can define the projected gradient
			\[
			\ts\bs\nabla^\perp f^{\rm s}(\bs\zeta,{\bs\beta})\coloneqq\begin{bsmallmatrix}\I_k & \vzer_{k \times h} \\ \vzer_{h\times k} & \bs P^{\rm s}_{\onep} \end{bsmallmatrix} \bs\nabla f^{\rm s}(\bs \zeta,{\bs\beta}) = \begin{bmatrix}
			\bs\nabla_{\bs\zeta} f^{\rm s}(\bs\zeta,{\bs\beta})\\
			\bs\nabla^\perp_{\bs\beta} f^{\rm s}(\bs\zeta,{\bs\beta})
			\end{bmatrix},
			\]
			where 
			\[
			\bs\nabla^\perp_{\bs\beta}f^{\rm s}(\bs\zeta,{\bs\beta}) \coloneqq \bs P^{\rm s}_{\onep}	\bs\nabla_{\bs\beta} f^{\rm s}(\bs\zeta,{\bs\beta}) = \bs B^\top \bs\nabla^\perp_{{\bs \gamma}} f(\bxi,{\bs \gamma})\big\vert_{\bxi = {\bs Z} {\bs \zeta}, {\bs \gamma} = {\bs B}{\bs \beta}} = \left.\begin{bsmallmatrix} 0 \\ ({\bs B^\perp})^\top \bs\nabla_{{\bs \gamma}} f(\bxi,{\bs \gamma}) \end{bsmallmatrix}\right\vert_{\bxi = {\bs Z} {\bs \zeta}, {\bs \gamma} = {\bs B}{\bs \beta}} ,
			\]
			the latter equalities being due to the fact that $\diag \big(\begin{bsmallmatrix} 0 & \vone^\top_{h-1}\end{bsmallmatrix}\big) \bs B^\top \bs\nabla_{{\bs \gamma}} f(\bxi,{\bs \gamma}) =  \bs B^\top \bs\nabla^\perp_{{\bs \gamma}} f(\bxi,{\bs \gamma})$. 
			
			%			\[
			%				\bs\nabla_{{\bs \beta}} f^{\rm s}({\bs \zeta},{\bs \beta}) = {\bs B}^\top \bs\nabla^\perp_{{\bs \gamma}} f(\bxi,{\bs \gamma})\big\vert_{\bxi = {\bs Z} {\bs \zeta}, {\bs \gamma} = {\bs B}{\bs \beta}} = \left.\begin{bsmallmatrix} 0 \\ ({\bs B^\perp})^\top \bs\nabla^\perp_{{\bs \gamma}} f(\bxi,{\bs \gamma}) \end{bsmallmatrix}\right\vert_{\bxi = {\bs Z} {\bs \zeta}, {\bs \gamma} = {\bs B}{\bs \beta}}.
			%			\]
			
			The definitions introduced in Sec.~\ref{sec:geom-blind} also require a few changes to be adapted to the sensing model with subspace priors. Indeed, we have just developed the gradient components of $\bs\nabla^\perp f^{\rm s} ({\bs \zeta},{\bs \beta})$ that can be obtained by those of $\bs\nabla^\perp f(\bxi,{\bs \gamma})$. Moreover, it is straightforward to obtain the initialisation ${\bs \zeta}_0$ in the signal-domain subspace from $\bxi_0$ as ${\bs \zeta}_0 \coloneqq {\bs Z}^\top \bxi_0$, while we can let ${\bs \beta}_ 0 \coloneqq \begin{bsmallmatrix} \sqrt m \\ \vzer_{h-1}\end{bsmallmatrix} = \bs B^\top \vone_m$, equivalently to what we adopted in Sec.~\ref{sec:nonconv}. By our choice of ${\bs Z}$ and recalling $\x = {\bs Z} {\bs z}$ we see that the initialisation is still so that
			\[
				\Ex {\bs \zeta}_0 = \tfrac{1}{m p}  
				\ts\sum_{i,l} (\bs c^\top_i \bs B \bs b) \underbrace{\Ex {\bs Z}^\top \ail  \ail^\top {\bs Z}}_{\I_k} \bs z  = {\bs z} ,
			\]
			thus yielding an unbiased estimate of the exact subspace coefficients $\bs z$. The neighbourhood of the global minimiser $({\bs z},{\bs b})$ (for which we require at least $mp \geq k + h$) is then defined using the distance  
			\begin{equation}
				{\Delta}({\bs \zeta},{\bs \beta}) \coloneqq \|{\bs \zeta}-{\bs z}\|^2 + \tfrac{\|{\bs z}\|^2}{m}\|{\bs \beta}- {\bs b}\|^2
				\label{eq:dist_subs}
			\end{equation}
			and noting that $({\bs \zeta}-{\bs z})^\top ({\bs \zeta}-{\bs z}) \equiv (\bxi-\x)^\top (\bxi-\x),  ({\bs \beta}- {\bs b})^\top ({\bs \beta}- {\bs b})\equiv ({\bs \gamma}-{\bs g})^\top ({\bs \gamma}-{\bs g})$ with our choices of ${\bs B},{\bs Z}$ as tight frames. Indeed, this guarantees that~\eqref{eq:dist_subs} has the same value as~\eqref{eq:dist}, so we can modify Def.~\ref{def:neighbour} to express it in terms of the coefficients in their respective subspaces, \ie 
			\begin{equation}
				{{\cl D}^{\rm s}_{\kappa,\rho}} \coloneqq \{({\bs \zeta},{\bs \beta})\in \R^k \times {\cl B}_\rho : \Delta({\bs \zeta},{\bs \beta}) \leq \kappa^2 \|{\bs z}\|^2\}, \ \rho \in [0, 1).
			\end{equation}
			Thus, we are allowed to use the very same theory to obtain analogue results with respect to Prop.~\ref{prop:init} and~\ref{prop:regu}, anticipating the advantage of a lower sample complexity given by the use of subspace priors. In fact, it is simply shown that the former propositions are special cases of Prop.~\ref{prop:initsub} and~\ref{prop:regusub}, as obtained when $\cl B \coloneqq \bb R^m$ and $\cl Z \coloneqq \bb R^n$, \ie $h=m$, $k=n$. Hence, to show convergence we will have to prove Prop.~\ref{prop:initsub} and~\ref{prop:regusub}.
			
			% ALGORITHM
			\begin{algorithm}[!t]
				{   \small
					\begin{algorithmic}[1]
						\STATE Initialise ${\bs \zeta}_0 \coloneqq \tfrac{1}{m p}
						\sum^{p}_{l = 1} \left(\A_l {\bs Z} \right)^\top \y_l,
						\, {\bs \beta}_0 \coloneqq \begin{bsmallmatrix}
						\sqrt m \\ \vzer_{h-1}
						\end{bsmallmatrix}, \, j \coloneqq0$. 
						\WHILE{stop criteria not met}
						% \item[] COMMENT{Line search w.r.t. $\bxi,{\bs \beta}$}
						\STATE $\begin{cases} 
						\mu_{\bs \zeta} \coloneqq \argmin_{\upsilon \in \R} f^{\rm s}({\bs \zeta}_{j}-\upsilon\bs\nabla_{\bs \zeta} f^{\rm s}({{\bs \zeta}_{j}, {\bs \beta}_ {j}}), {\bs \beta}_ {j}) \\
						\mu_{{\bs \beta}} \coloneqq \argmin_{\upsilon \in \R} f^{\rm s}({\bs \zeta}_{j},{\bs \beta}_ {j}-\upsilon\bs\nabla^\perp_{\bs \beta} f^{\rm s}({{\bs \zeta}_{j}, {\bs \beta}_ {j}})) 
						\end{cases}$
						\STATE ${\bs \zeta}_{j+1} \coloneqq {\bs \zeta}_{j} - \mu_{\bs \zeta} \bs\nabla_{\bs \zeta} f^{\rm s}({{\bs \zeta}_{j}, {\bs \beta}_ {j}})$
						\STATE $\check{{\bs \beta}}_{j+1} \coloneqq {\bs \beta}_ {j} - \mu_{\bs \beta} \, \bs\nabla^\perp_{\bs \beta} f^{\rm s}({{\bs \zeta}_{j}, {\bs \beta}_ {j}})$ % COMMENT{Update estimates}
						\STATE ${{\bs \beta}}_{j+1} \coloneqq{\cl P}_{\cl B_\rho} \check{{\bs \beta}}_{j+1}$
						\STATE $j \coloneqq j + 1$
						\ENDWHILE
					\end{algorithmic}
				}
				\caption{\label{alg2} Non-Convex Blind Calibration by Projected Gradient Descent with Known Signal and Gain Subspaces.}
			\end{algorithm}
			
			For what concerns the descent algorithm, a modification is straightforwardly obtained as Alg.~\ref{alg2}. The main difference is that the optimisation is carried out on the subspace coefficients $\bs \zeta$ and ${\bs\beta}$ rather than the signal and gains, with the gradient expressions given in~\eqref{eq:gradw} and~\eqref{eq:gradz}. The descent will run from $({\bs \zeta}_0, {\bs \beta}_ 0)$ up to $({\bs \zeta}_j, {\bs \beta}_ j)$ at the $j$-th iteration with the updates specified in steps 4:-6:. Again, the gradient update in the gain-domain subspace is followed by the projection on $\cl B_\rho$, which is a closed convex set. This is still algorithmically simple and not needed  in the numerical results (\ie it is only required by our proofs), as for $\cl G_\rho$ in absence of subspace priors. % . This projection is simply ${\cl P}_{\cl B_\rho}({{\bs \beta}}) = {\bs B}^\top {\cl P}_{\rho \B^m_\infty} ({\bs B^\perp}{{\bs \beta}})$ for ${\bs \beta} \in \R^h$. % , which still satisfies the constraint that since by construction ${\bs B}^\top \vone_m=\vzer_h$. 
			
			We now establish in detail how the main theoretical statements are modified, as an extension of our results in Sec.~\ref{sec:recguars}. The proofs of the next statements are reported in App.~\ref{sec:appC} and rely on the fact that the technical arguments used in proving the results of Sec.~\ref{sec:recguars}, as already mentioned, are actually a special case of those in the subspace case, with a sample complexity that is reduced thanks to the low-dimensional description of the signal and gains.  			
			
			Firstly, we have that the sample complexity requirements for the initialisation in Prop.~\ref{prop:init} are essentially reduced to $mp \gtrsim (k+h)\log n$.
			
			\begin{proposition}[Initialisation Proximity with Subspace Priors]
				\label{prop:initsub}
				Let $({\bs \zeta}_0,{\bs \beta}_ 0)$ be as in Alg.~\ref{alg2}. Given ${\delta} \in (0,1)$, $t \geq 1$, provided $\ts mp \gtrsim {\delta}^{-2} (k+h) \log \tfrac{n}{\delta} $ and  $n \gtrsim t \log mp $, with probability exceeding
				\begin{equation}
				\label{eq:bounds1}
				1 - C e^{- c {\delta}^2 m p} - (mp)^{-t}
				\end{equation}
				for some $C,c>0$, we have that $\|{\bs \zeta}_0 - {\bs z} \| \leq {\delta} \|{\bs z}\|$. Thus, $({\bs \zeta}_0,{\bs \beta}_ 0) \in
				{\cl D}^{\rm s}_{{\kappa_0},\rho}$ with the same probability and $\kappa_0 \coloneqq \sqrt{\delta^2+\rho^2}$.
			\end{proposition}
			
                          This following proposition actually reduces the regularity condition formulated in Prop.~\ref{prop:regu} according to the knowledge of the subspaces where both the signal and the gain lie and given a general neighbourhood of radius $\kappa > 0$ around the solution.
                          \begin{proposition}[Regularity Condition with Subspace Priors]
				\label{prop:regusub}
				Given $\delta \in (0,1)$, $t \geq 1$ and $\kappa>0$ , 
                                provided $\rho < \tfrac{1-4\delta}{31}$ and 
                \begin{align*}
                  \label{eq:glob-requirements-reg-prior}
                   \ts n&\ts \gtrsim t \log(mp),\\
                   \ts mp&\ts \gtrsim \delta^{-2} (k + \upmu^2_{\max} h)\, \log^2(m(p+n))\log(\frac{1}{\delta}),            
                \end{align*}
                with probability exceeding $1 - C (mp)^{-t}$ for some $C>0$ and for all $(\bs \zeta,\bs \beta) \in \cl D^{\rm s}_{\kappa,\rho}$, we have
\new{
                \begin{subequations}
                  \begin{align}
\ts \left\langle \bs\nabla^\perp f^{\rm s}({\bs \zeta}, {\bs \beta}), \left[\begin{smallmatrix} {\bs \zeta}-{\bs z} \\ {\bs \beta}- {\bs b} \end{smallmatrix}\right] \right\rangle&\ts \geq \eta\, {\Delta}({\bs \zeta},{\bs \beta}), \\
\|\bs\nabla_{\bs \zeta} f^{\rm s}({\bs \zeta}, {\bs \beta}) \|^2&\ts \leq L_{\bs \zeta}^2\ {\Delta}({\bs \zeta},{\bs \beta}),\\
\ts \|\bs\nabla^\perp_{{\bs\beta}} f^{\rm s}(\bs\zeta, \bs\beta)\|^2 \leq \|\bs\nabla_{{\bs\beta}} f^{\rm s}(\bs\zeta, \bs\beta)\|^2&\ts \leq L_{\bs \beta}^2\ \Delta(\bs \xi, \bs \gamma),                    
                \end{align}
                \label{eq:reg-event-prop-reg-subspace}
                \end{subequations}
                 where $\eta \coloneqq 1-31\rho-4\delta \in (0,1)$, $L_{\bs \zeta} = 8 \sqrt 2$ and $L_{\bs \beta} \coloneqq 4\sqrt 2(1+\kappa) \|\bs z\|$.}
			\end{proposition}

			Given Prop.~\ref{prop:initsub} and~Prop.~\ref{prop:regusub}, we finally obtain a more general version of Thm.~\ref{theorem:convergence} as follows.
			\begin{theorem}[Provable Convergence to the Exact Solution with Subspace Priors]
                          \label{theorem:convergence-subs}
                          Given $\delta \in (0,1)$ and $t\geq1$, let us take $n \gtrsim t \log(mp)$, $mp \gtrsim \delta^{-2} (k + \upmu^2_{\max} h)\, \log^2(m(p+n))\log(\frac{1}{\delta})$ and $\rho < \tfrac{1-4\delta}{31}$. There exists $\mu_0 >0$ with 
\new{$$
\ts \mu_0 \lesssim \frac{\eta}{m},
$$}%
for $\eta = 1-31\rho-4\delta \in (0,1)$ such that for any $0<\mu<\mu_0$ and with probability exceeding $1 - C (mp)^{-t}$ for some $C>0$, Alg.~\ref{alg2} with step sizes set to $\mu_{\bs \zeta} \coloneqq \mu$ and $\mu_{\bs \beta} \coloneqq \mu \tfrac{m}{\|\bs z\|^2}$ has distance decay
                          \begin{equation}
                            \ts {\Delta}({\bs \zeta}_j,{\bs \beta}_j) \leq
                            (1- \eta \mu)^j \big(\delta^2+\rho^2)\|\bs z\|^2, 
                            \label{eq:convres}
                          \end{equation}
                          at any iteration $j > 0$. Hence, ${\Delta}({\bs \zeta}_j,{\bs \beta}_j)\underset{j\rightarrow\infty}{\longrightarrow}0$. 
                        \end{theorem}

			Thus, subspace priors allow for a provably convergent application of Alg.~\ref{alg2}. Let us note that $mp \gtrsim \delta^{-2} (k + \upmu^2_{\max} h)\, \log^2(m(p+n))\log(\frac{1}{\delta})$ is a clearly more stringent requirement than $mp \gtrsim (k+h) \log n$, emerging from Prop.~\ref{prop:initsub}. 
			% Further assumptions on the tight frame $\bs B$, such as fixing its coherence, could in fact allow us to have a dependency on $k+h$ rather than $k+m$. We have however chosen to adopt a worst-case perspective that does not depend explicitly on the said coherence or on the structure of $\bs B$. % with a multiplicative factor depending 
			%
			However, we see that if $\upmu_{\max} \simeq 1$, and if $h$ and $k$ are small before $m$, having $p=1$ is allowed by the requirement. This will be also confirmed in our experiments, provided that both $k \ll n$, $h \ll m$ and $\upmu_{\max}$ is sufficiently small. This allows for a single-snapshot application of our setup and algorithm to the blind calibration of sensors with side information regarding the subspaces to which $(\bs x,{\bs g})$ belong. 
	
%%% Local Variables:
%%% mode: latex
%%% TeX-master: "iai_blindcalibration_2016"
%%% End:

\section{Blind Calibration in the Presence of Noise}
\label{sec:nstability-sub}

%%%%%%%%%%%%%%%%%%%%% 
\begin{table}[t]
	\centering
     \resizebox{\textwidth}{!}{	
        \def\arraystretch{2}

          % \scriptsize
          \begin{tabular}{c c c}
            \toprule
            \bf  Noisy-case quantity & \bf Expression (or relation w.r.t. noiseless-case counterpart)\\
            \toprule
            $\tilde{f}^{\rm s}({\bs \zeta},{\bs \beta})$ &
                                                           $f^{\rm s}({\bs \zeta},{\bs \beta}) + \tfrac{1}{2}\sigma^2 - \tfrac{1}{mp}\ts \sum_{i,l} \nu_{i,l} (\gamma_i \ail^\top\bs Z\bs \zeta - g_i \ail^\top \bs Z \bs z)$\\
            \midrule
            $\bs\nabla_{\bs \zeta} \tilde{f}^{\rm s}(\bs \zeta,\bs \beta) =  {\bs Z}^\top \bs\nabla_{\bxi} \tilde{f}(\bxi,{\bs \gamma})\big\vert_{\bxi = {\bs Z} {\bs \zeta}, {\bs \gamma} = {\bs B}{\bs \beta}}$ & ${\bs Z}^\top \bs\nabla_{\bxi} f(\bxi,{\bs \gamma})\big\vert_{\bxi = {\bs Z} {\bs \zeta}, {\bs \gamma} = {\bs B}{\bs \beta}} - \tfrac{1}{mp}\ts \sum_{i,l} \nu_{i,l} \gamma_i {\bs Z}^\top\ail$\\
            $\bs\nabla^\perp_{\bs \beta} \tilde{f}^{\rm s}(\bs \zeta,\bs \beta) = {\bs B}^\top \bs\nabla^\perp_{{\bs \gamma}} \tilde{f}(\bxi,{\bs \gamma})\big\vert_{\bxi = {\bs Z} {\bs \zeta}, {\bs \gamma} = {\bs B}{\bs \beta}}$ & ${\bs B}^\top \bs\nabla^\perp_{{\bs \gamma}} f(\bxi,{\bs \gamma})\big\vert_{\bxi = {\bs Z} {\bs \zeta}, {\bs \gamma} = {\bs B}{\bs \beta}}-\tfrac{1}{mp}\ts \sum_{i,l} \nu_{i,l} (\ail^\top {\bs Z} {\bs \zeta})  {\bs B}^\top {\bs c}^\perp_i$\\
            \midrule
            $(\tilde{\bs \zeta}_0,\,\tilde{\bs \beta}_0 \coloneqq \bs \beta_0)$& $\left(
                                                                          \tfrac{1}{mp}\ts \sum_{i,l} g_i \bs Z^\top \ail \ail^\top \bs Z \bs z +  \nu_{i,l} \bs Z^\top \ail,\, \begin{bsmallmatrix} \sqrt m \\ \vzer_h \end{bsmallmatrix}\right)$\\
            \bottomrule
          \end{tabular}
    }
	\caption{\label{tab:quantsnoise-sub}Objective function, its gradient, and the initialisation point for the problem in Def.~\ref{def:blind-calibr-with-sub}~and in the presence of noise. The expressions are expanded to highlight the noise-dependent terms against the corresponding noiseless-case values. Note that above $\gamma_i = {\bs c}^\top_i {\bs B }{\bs \beta}$ and $g_i = {\bs c}^\top_i {\bs B }{\bs b}$.}
\end{table}

This purpose of this section is to study the stability of the solution produced by Def.~\ref{def:blind-calibr-with-sub} when an additive and bounded noise $\{\bs \nu_l\}$ affects the measurements according to 
\begin{equation}
  \label{eq:noisy-mod-with-prior}
  \y_l = \diag({\bs B \bs b}) \A_l \bs Z \bs z + \bnu_l,\quad l \in [p],
\end{equation}
\ie following the setup\footnote{Let us recall the energy of the signal is preserved in its representation in $\cl Z$ so that $\|\bs x\|=\|\bs z\|$.} of the uncalibrated multi-snapshot sensing model with subspace priors described in Def.~\ref{def:modelsub}. We recall the assumption that all noise vectors are collected as the columns of a matrix $\bs N \coloneqq (\bs \nu_1,\,\cdots,\bs \nu_p)$, defining $\sigma \coloneqq \tfrac{1}{\sqrt{mp}} \|\bs N\|_F$ as a bound on the amount of noise injected into the model over $p$ snapshots. 		

While Alg.~\ref{alg2} still estimates $(\bar{\bs z},\bar{{\bs b}})$ in \eqref{eq:bcpk} as before, the presence of $\bnu_l$ in \eqref{eq:measurement-model-matrix-noise} modifies the values of the objective function $\tilde{f}(\bs \zeta, \bs \beta)$, as well those of its projected gradient. In Table \ref{tab:quantsnoise-sub}, the impact of noise is highlighted as a deviation from the noiseless quantities discussed in Sec.~\ref{sec:blind-calibration-with-subspace-priors}, as obtained by plugging in the measurements $\y_l$ defined in \eqref{eq:noisy-mod-with-prior}. This deviation vanishes as $\sigma \simeq 0$. Our theoretical results will be modified accordingly, \ie the initialisation $\tilde{\bs \zeta}_0$ will be negatively affected and require more observations with respect to Prop.~\ref{prop:initsub} to attain the same distance with respect to $\bs z$. Similarly, some noise-dependent terms will worsen the upper and lower bounds of Prop.~\ref{prop:regusub}. However, it is still possible to show in the following Thm.~\ref{theorem:stability} that the algorithm is robust and converges to a minimiser whose distance from the exact solution is controlled by $\sigma$. 

We begin by discussing the effect of noise on the initialisation point. The proofs of this and all following statements in this section are reported in App.~\ref{sec:appD-sub}.

\begin{proposition}[Initialisation Proximity in the Presence of Noise]
  \label{prop:initnoisy}
  Under the noisy sensing model \eqref{eq:noisy-mod-with-prior} with
  $\sigma \coloneqq \tinv{\sqrt{mp}} \|\bs N\|_F$, let $(\tilde{\bs \zeta}_0, \tilde{\bs \beta}_0)$ be as in
  Table \ref{tab:quantsnoise-sub}. Given $\delta \in (0,1)$, $t \geq 1$, provided $\ts mp \gtrsim {\delta}^{-2} (k + h) \log \tfrac{n}{\delta}$ and $n \gtrsim t \log mp $, with probability exceeding 
  $$
  1 - C [e^{-c\delta^2 mp} + (mp)^{-t}]
  $$
  for some $C,c>0$, we have $(\tilde{\bs \zeta}_0, \bs \beta_0) \in {\cl D}^{\rm s}_{\tilde{\kappa}_0, \rho}$ with
  \begin{equation}
    \label{eq:tilde-kappa-0-def}
      \ts \tilde{\kappa}_0^2 \coloneqq \delta^2 + \rho^2 + \frac{4\sigma^2}{\|\bs z\|^2}.    
  \end{equation}
\end{proposition}

Note how the quality of the initialisation increases with a ``signal-to-noise ratio''-like quantity $\tfrac{\|\bs x\|}{\sigma}$ in~\eqref{eq:measurement-model-matrix-noise}. 

\new{
  \begin{remark}
Our stability analysis in Prop.~\ref{prop:initnoisy}, Prop.~\ref{prop:regunoisy} and in Thm~\ref{theorem:stability} does not use any statistical property of the noise $\bs N$. For instance, assuming $\bs N$ to be an additive white Gaussian noise could lead, from a standard statistical argument related to the consistency of the maximum likelihood estimator, to a reduction of the noise impact as $mp$ increases. This interesting improvement is, however, postponed to a future study.
  \end{remark}
}

% From the statement of Proposition 5.1, it seems that you don't use any statistical property of noise. As we know, if $\bs \nu_l$ are Gaussian, the effect of noise gets weaker once $p$ gets larger. This observation comes from the consistency of maximal likelihood estimator. However, I don't see it appears in (5.2) since the $\tilde{\kappa}_0$ actually depends on $\sigma^2$ and does not become smaller if $p$ gets larger. You may want to mention that as a remark under Prop.5.1.		

We now apply again the regularity condition and assess how the presence of noise modifies the requirements on the projected gradient used in Alg.~\ref{alg2}, \ie on
\[
  \bs\nabla^\perp \tilde{f}^{\rm s}(\bs \zeta, \bs \beta) \coloneqq 
  \ts \begin{pmatrix}
    \ts \bs\nabla_{\bs \zeta} \tilde{f}^{\rm s}(\bs \zeta, \bs \beta)\\
    \ts \bs\nabla_{\bs \beta}^\perp \tilde{f}^{\rm s}(\bs \zeta, \bs \beta)
  \end{pmatrix},
\]
with $\bs\nabla_{\bs \beta}^\perp \tilde f^{\rm s}(\bs \zeta, \bs \beta) = \bs P^{\rm s}_{\onep} \tilde f^{\rm s}(\bs \zeta, \bs \beta) = \big(\bs B^\top (\bs I_m - \frac{\bs 1 \bs 1^\top}{m})\bs B\big) \bs\nabla_{\bs \beta} \tilde f^{\rm s}(\bs \zeta, \bs \beta)$, as defined in Sec.~\ref{sec:blind-calibration-with-subspace-priors}.

\begin{proposition}[Regularity condition in the Presence of Noise]
  \label{prop:regunoisy}
  Given $\delta \in (0,1)$, $t\geq1$ and $\kappa > 0$, provided $\rho < \tfrac{1-4\delta}{31}$,  $n \gtrsim t \log(mp)$, and 
  $mp \gtrsim \delta^{-2} (k + \upmu^2_{\max} h)\, \log^2(m(p+n))\log(\frac{1}{\delta})$,  with probability exceeding $1 - C (mp)^{-t}$ for some $C>0$, we have for all $(\bs \zeta, \bs \beta) \in {\cl D}^{\rm s}_{\kappa,\rho}$,
  \begin{subequations}
      \label{eq:reg-noise}
    \begin{align}
      \label{eq:bounded-curvature-noise}
      \big\langle \bs\nabla^\perp \tilde{f}^{\rm s}({\bs \zeta}, {\bs \beta}), \big[\begin{smallmatrix} {\bs \zeta}-{\bs z} \\ {\bs \beta}- {\bs b} \end{smallmatrix}\big] \big\rangle&\geq \eta\, {\Delta}({\bs \zeta},{\bs \beta}) - o_C\,\sigma,
    \end{align}\new{\vspace{-7mm}
\begin{align} 
   \label{eq:bounded-lipschitz-noise-zeta}        
      \|\bs\nabla_{\bs \zeta} \tilde{f}^{\rm s}({\bs \zeta}, {\bs \beta}) \|^2&\leq
                                                                          2L_{\bs \zeta}^2\ {\Delta}({\bs \zeta},{\bs \beta}) + o_{L,\bs \zeta}\,\sigma^2,\\
   \label{eq:bounded-lipschitz-noise-beta}        
      \|\bs\nabla_{\bs \beta}^\perp \tilde{f}^{\rm s}({\bs \zeta}, {\bs \beta}) \|^2 \leq \|\bs\nabla_{\bs \beta} \tilde{f}^{\rm s}({\bs \zeta}, {\bs \beta}) \|^2&\leq
                                                                          2 L_{\bs \beta}^2\ {\Delta}({\bs \zeta},{\bs \beta}) + o_{L,\bs \gamma}\,\sigma^2,
    \end{align}
}  
  \end{subequations}
\new{with $\eta$, $L_{\bs \zeta}$ and $L_{\bs \beta}$ determined in Prop.~\ref{prop:regusub}, $o_C = \sqrt{2}(1+ 3\kappa) \|{\bs z}\|$, $o_{L,\bs \zeta} := 32$, and $o_{L,\bs \beta} \coloneqq 4\,(1+\kappa)^2\,\|\bs z\|^2$.}
\end{proposition}

As clear from these propositions, the sample complexity is not modified by the presence of noise (\ie the nature of the dependency on the problem dimensions remains unaltered), but compared to~\eqref{eq:reg-event-prop-reg-subspace} the bounds in~\eqref{eq:reg-noise} are directly impacted by the noise level $\sigma$ as determined by two factors $o_C$ and $o_L$ depending only on the radius $\kappa$ of the considered neighbourhood $\cl D^{\rm s}_{\kappa,\rho}$ and on $\|\bs z\|$. As made clear in the next theorem, these disturbances have direct impact on the quality attained asymptotically in the number of iterations of Alg.~\ref{alg2}. 

\begin{theorem}[Stable Recovery of the Exact Solution]
  \label{theorem:stability}
  Given $\delta \in (0,1)$ and $t\geq1$, let us take $n \gtrsim t \log(mp)$, $mp \gtrsim \delta^{-2} (k + \upmu^2_{\max} h)\, \log^2(m(p+n))\log(\frac{1}{\delta})$ and $\rho < \tfrac{1-4\delta}{31}$. If $\sigma \lesssim \|\bs z\|$, there exists $\mu_0 >0$ with 
  \new{$$
  \ts \mu_0 \lesssim \frac{\eta}{m}\, \min(1,\frac{\|\bs z\|^2}{\sigma^2}),
  $$}%
for $\eta = 1-4\delta-31\rho \in (0,1)$ such that for any $0<\mu<\mu_0$ and with probability exceeding $1 - C (mp)^{-t}$ for some $C>0$, Alg.~\ref{alg2} with step sizes set to $\mu_{\bs \zeta} \coloneqq \mu$ and $\mu_{\bs \beta} \coloneqq \mu \tfrac{m}{\|\bs z\|^2}$ has distance decay
\begin{equation}
\label{eq:updnoisedec}
\new{\ts \Delta(\bs\zeta_{j}, \bs\beta_{j}) \lesssim (1 - \eta\mu)^{j} \|\bs z\|^2 + \frac{1}{{\eta}} \sigma \|\bs z\|.}
\end{equation}
Hence, \new{$\lim_{j\to + \infty}{\Delta}({\bs \zeta}_j,{\bs \beta}_j) \lesssim \frac{1}{\eta} \sigma \|\bs z\|$.}
\end{theorem}
Note that this result holds identically for Alg.~\ref{alg1}, since the latter is only a particular case of Alg.~\ref{alg2} in absence of subspace priors. Let us also emphasise that our result was shown for $\sigma \lesssim \|\bs x\|$, \ie when $\sigma$ is a small fraction of the signal energy. If not, the dependency reported in \eqref{eq:updnoisedec} will not be linear but in general polynomial with respect to $\tfrac{\sigma}{\|\bs z\|}$, as shall be seen in the proof of Thm.~\ref{theorem:stability}.

Finally, note that the dependency on $\|\bs z\|$, which is generally unknown, is not concerning since the initialisation $\|\bs \zeta_0\| \in [(1-\delta) \|\bs z\|, (1+\delta)\|\bs z\|]$ for some $\delta \in (0,1)$ can still be used as a rough estimate of the former.

%%% Local Variables:
%%% mode: latex
%%% TeX-master: "iai_blindcalibration_2016"
%%% End:
	
% !TeX spellcheck = en_GB
% !TeX encoding = UTF-8
\section{Numerical Experiments}
		\label{sec:numerical-experiments}
		We now introduce some experiments and applications of our blind calibration framework to assess its practical performances for finite values of $m,n,p$ and in settings agreeing with our sensing models. In the following we will adopt as a figure of merit 
		\[
		\text{RMSE}_{\max} \coloneqq 20 \log_{10} \max \left\{\tfrac{\|\bar{\x} - \x\|}{\|\x\|},\tfrac{\|\bar{{\bs g}}-{\bs g}\|}{\|{\bs g}\|} \right\},
		\]
		\ie the {\em maximum relative mean square error}  taking the worst-case performances achieved by the estimates $(\bar{\x}, \bar{\bs g})$ between the signal and gain domain. This is used to assess when the recovery of the signal and gains on any one instance of \eqref{eq:bcp} or \eqref{eq:bcpk} is achieved successfully by either Alg.~\ref{alg1} (in absence of priors) or \ref{alg2} (with subspace priors). For our experiments, we have chosen to generate $\A_l$ in our sensing models with \iid random sensing vectors distributed as $\ail \sim_{\rm \iid} \cl N(\vzer_n,\I_n)$. However, the same performances can be achieved in high dimensions with other \iid sub-Gaussian random vectors, such as those with symmetric Bernoulli-distributed entries, in a fashion fully compatible with the Gaussian case. Moreover, while the theory in this paper addresses only sub-Gaussian random matrix ensembles, the experiments can be empirically run when $\A_l$ is implemented (\eg optically) as a random convolution~\cite{Romberg2009} albeit requiring a higher number $p$ of snapshots to achieve the exact solution. This suggests that our framework could be extended to random matrix ensembles not covered by the theory in Sec.~\ref{sec:appC}, increasing the applicability of the proposed framework to implementation-friendly and fast sensing matrix configurations. 
		
		As a point of comparison we will adopt the least-squares solution in absence of prior information on the gains, \ie 
		\begin{equation}
		\label{eq:ls}
		\bar{\x}_{\rm ls} \coloneqq {\mathop{\argmin}_{\bxi \in \R^n}}  \tfrac{1}{2mp} {\ts \sum_{l=1}^{p} \left\Vert \A_l \bxi - \y_l\right\Vert^2},
		\end{equation}%
		given $\y_l \in \R^m, \A_l\in \R^{m\times n}, l \in [p]$ as in Def.~\ref{def:model} (its extension to the known subspace case is trivial and merely involves solving \eqref{eq:ls} with respect to $\bs\xi = \bs Z\bs\zeta$). \new{Indeed, this convex problem can be solved by taking the gradient and solving the resulting linear system via, \eg the LSQR algorithm \cite{paige1982lsqr}}. 
		This comparison merely aims to show the advantages of performing blind calibration with respect to ignoring the effect of $\bs g$. Resorting to~\eqref{eq:ls} as a valid reference problem is in fact due to the observation that most of the algorithms addressing blind calibration employ additional assumptions or a different setup, such as signal-domain sparsity or multiple and possibly independent inputs, that are incompatible with our sensing model.
					
			\begin{figure}
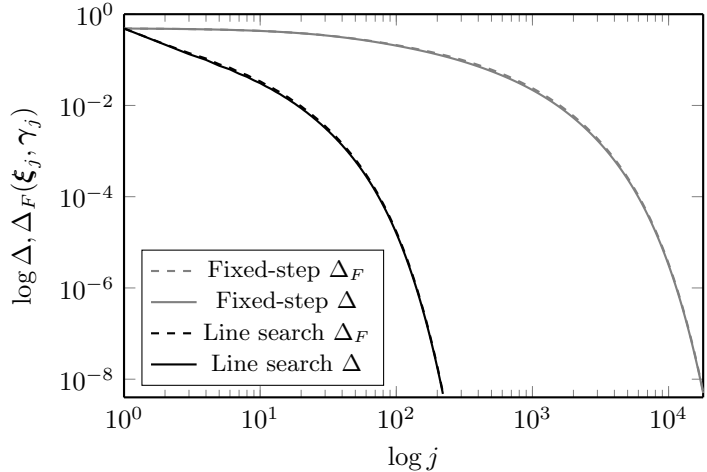

				\centering
%
				\caption{\label{fig:decay} Residual evolution of two exemplary runs of Alg.~\ref{alg1}: with fixed steps $\mu_\bxi,\mu_{\bs \gamma}$ (gray); with the line search updates given in \eqref{eq:step1} and \eqref{eq:step2} (black).}
			\end{figure}					
		
		{
			% \input{pht_fig.tex}
			%%% EMPIRICAL PHASE TRANSITION %%%
			\begin{figure}[!p]
				\definecolor{mycolor1}{rgb}{0.254,0.254,0.254}%
				\definecolor{mycolor2}{rgb}{0.508,0.508,0.508}%
				\definecolor{mycolor3}{rgb}{0.762,0.762,0.762}%
				\centering
				\subfloat[{$n=2^6,\rho=10^{-3}$}]{%
						\includegraphics{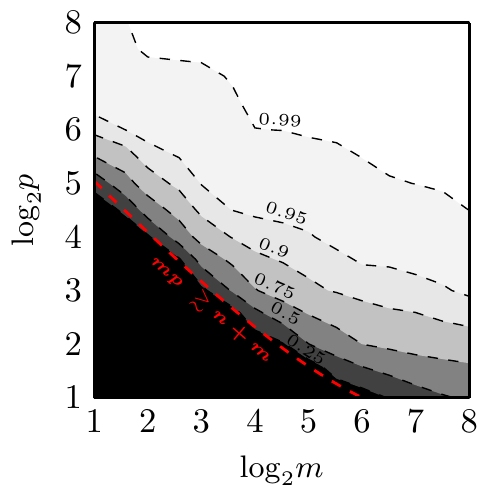}
				}
				\hfill
				\subfloat[{$n=2^6,\rho=10^{-2}$}]{%
						\includegraphics{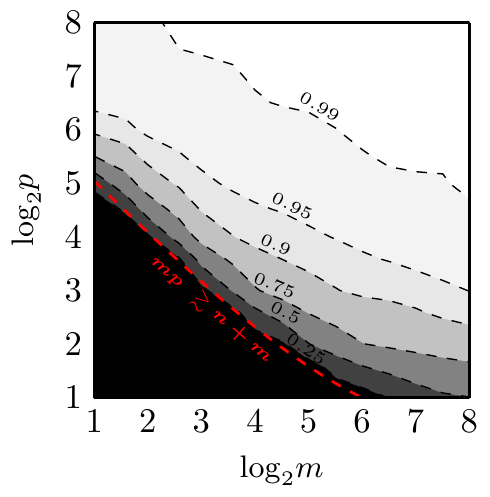}
				}
				\hfill
				\subfloat[{$n=2^6,\rho=10^{-1}$}]{%
						\includegraphics{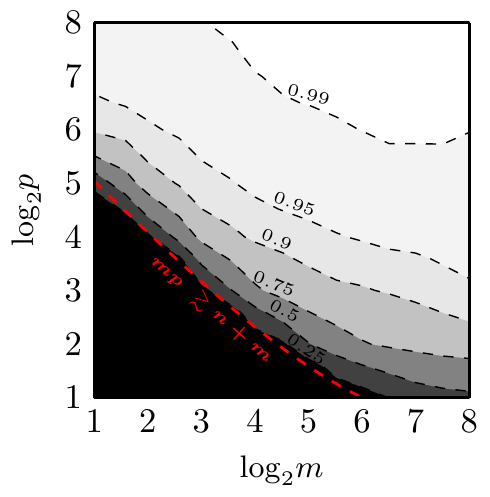}
				}\\
				\subfloat[{$n=2^7,\rho=10^{-3}$}]{%
						\includegraphics{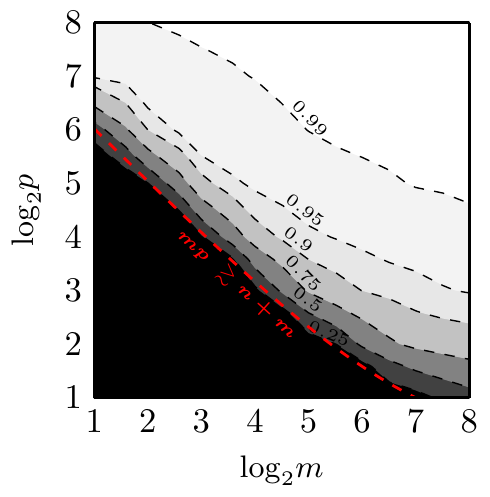}
				}
				\hfill
				\subfloat[{$n=2^7,\rho=10^{-2}$}]{%
						\includegraphics{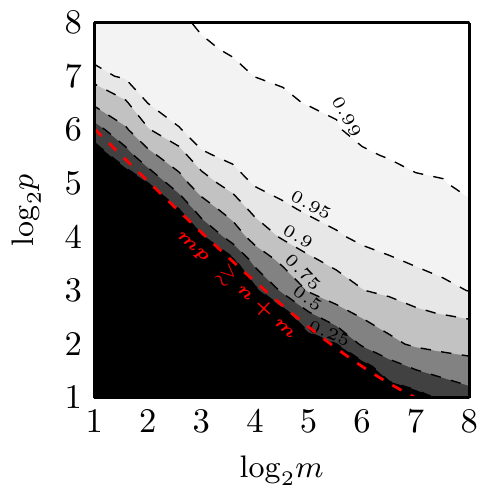}
				}
				\hfill
				\subfloat[{$n=2^7,\rho=10^{-1}$}]{%
						\includegraphics{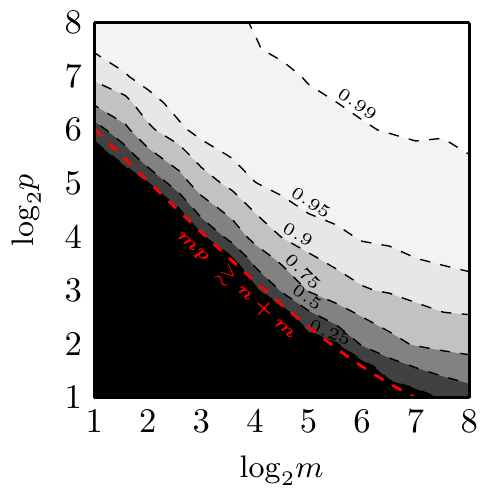}
				}\\
				\subfloat[{$n=2^8,\rho=10^{-3}$}]{%
						\includegraphics{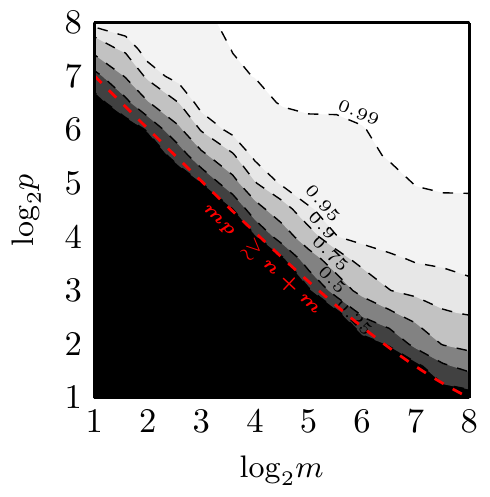}
				}
				\hfill
				\subfloat[{$n=2^8,\rho=10^{-2}$}]{%
						\includegraphics{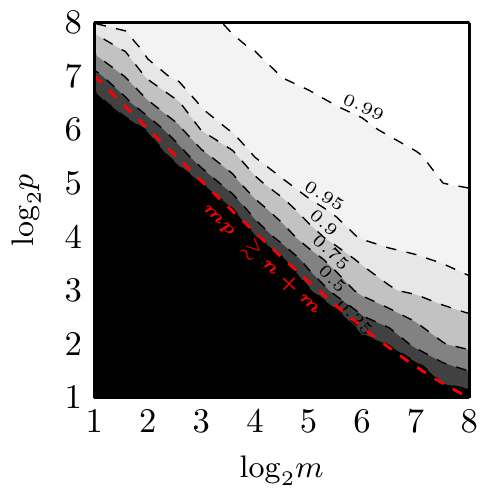}
				}
				\hfill
				\subfloat[{$n=2^8,\rho=10^{-1}$}]{%
						\includegraphics{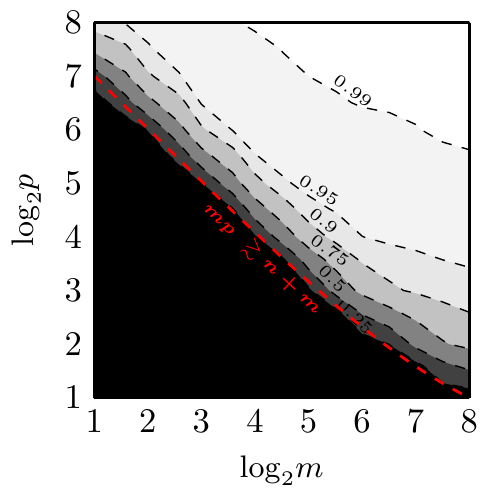}
				}
				\caption{\label{fig:ept}Empirical phase	transition of \eqref{eq:bcp} for increasing values of $n$ (top to bottom) and {$\rho$} (left to right). We report the contours $\{0.25,0.5,0.75,{0.9},0.95,0.99\}$ of the probability of exact recovery $\rm P_T$. The superimposed curve in red corresponds to the sample complexity bound obtained in Theorem \ref{theorem:convergence}.}
			\end{figure}%		
		}
		\subsection{Step Size Updates}
		\label{sec:stepsize}
			We begin our experiments by addressing a computational issue arising in Alg.~\ref{alg1}, that is the choice of a step size for the gradient descent updates. Thm.~\ref{theorem:convergence} confirms that there is a scaling factor $\tfrac{\mu_\bxi}{\mu_{\bs \gamma}} \simeq \tfrac{m}{\|\x\|^2}$ between the two fixed step sizes that ensures convergence to the global minimiser once we fix, \eg a sufficiently small value for $\mu_\bxi = \mu$ (and set $\mu_{\bs \gamma}$ accordingly). 
			
			However, in the practical application of our method we have found that updating the step sizes with the line searches reported in Alg.~\ref{alg1} is advantageous in terms of rate of convergence. Firstly, for some fixed values of $(\bxi_j, {\bs \gamma}_j)$ the line searches can be solved in closed form at iteration $j$ as 
			\begin{align}
			\mu_\bxi & \coloneqq \frac{\sum^p_{l=1} \big\langle \diag({\bs \gamma}_j) \A_l \bxi_j,  \diag({{\bs \gamma}})_j \A_l \bs\nabla_\bxi f(\bxi_j,{\bs \gamma}_j) \big\rangle}{\sum^p_{l=1} \|\diag({{\bs \gamma}})_j \A_l \bs\nabla_\bxi f(\bxi_j,{\bs \gamma}_j)\|^2}\label{eq:step1}\\
			\mu_{\bs \gamma} & \coloneqq \frac{\sum^p_{l=1} \big\langle \diag({\A_l \bxi_j}) {\bs \gamma}_j,  \diag({\A_l \bxi_j}) \bs\nabla^\perp_{\bs \gamma} f(\bxi_j,{\bs \gamma}_j) \big\rangle}{\sum^p_{l=1} \|\diag({\A_l \bxi_j}) \bs\nabla^\perp_{\bs \gamma} f(\bxi_j,{\bs \gamma}_j)\|^2}\label{eq:step2}
			\end{align}
			Equivalent formulas can be straightforwardly derived for Alg.~\ref{alg2}. Note that \eqref{eq:step1} and \eqref{eq:step2} do not constitute a single exact line search, which would require solving another bilinear problem jointly with respect to the two step sizes. In fact, \eqref{eq:step1} and \eqref{eq:step2} are exact line searches when the signal- or the gain-domain iteration is fixed, in which case the two corresponding optimisation problems are convex and solved in closed-form; this strategy is, in all the simulations we carried out, a stable way of updating the step size that largely outperforms the choice of a fixed value for $\mu_\bxi$ and $\mu_{\bs \gamma}$. 
			
			As an exemplary case, we randomly generated a typical problem instance with $\x \in \bb S^{n-1}$ and ${\bs g} \in \vone_m + \onep \cap \rho {\bb S}^{m-1}_\infty$, setting $n = 256$, $m = 64$, $\rho = 0.99$, and taking $p = 10$ snapshots so that $\tfrac{mp}{n+m} = 2$. Then, we ran Alg.~\ref{alg1} with either: $(i)$ $\mu_\bxi = \mu = 10^{-4}$, $\mu_{\bs \gamma} = m \mu$, or $(ii)$ the step-size updates in \eqref{eq:step1},\eqref{eq:step2}. The results are reported in Fig. \ref{fig:decay} in terms of the decay of both $\Delta(\bxi_j,\bs\gamma_j)$ (solid lines) and $\Delta_F(\bxi_j,\bs\gamma_j)$ (dashed lines), as a means to allow us to validate numerically the use of \eqref{eq:dist} instead of \eqref{eq:distF} to assess the convergence of our descent algorithm. Indeed, as expected $\Delta(\bxi_j,\bs\gamma_j)$ has the same decay as $\Delta_F(\bxi_j,\bs\gamma_j)$ up to some constant factor. For what concerns the step-size choice, there is a clear advantage in choosing the rules in \eqref{eq:step1}, \eqref{eq:step2} (black lines). For the given instance, convergence up to an objective value $f(\bxi_j, {\bs \gamma}_j) = 10^{-8}$ is attained after $j = 220$ iterations. On the other hand, the same problem instance using a fixed step size (with $\mu$ chosen fairly by inspecting the values given by the update rules) does converge to the global minimiser with the same criterion, but requires $j = 17951$ iterations, \ie convergence is slower by two orders of magnitude in the number of iterates. With both step-size choices, the algorithm terminates with ${\rm RMSE}_{\max} > \unit[86.49]{dB}$. Thus, since convergence and accuracy do not ultimately depend on our step-size choice, we adopt the faster update rule for all further experiments. 
			In addition, the step-size updates in \eqref{eq:step1}, \eqref{eq:step2} did verify $\tfrac{\mu_{\bs \gamma}}{\mu_\bxi} \simeq \tinv{m}$ empirically in our experiments as the number of iterations $j$ increases. 
			
			Let us finally mention that analogue updates to \eqref{eq:step1},\eqref{eq:step2} are easily derived for Alg.~\ref{alg2} and consistently yield faster convergence than choosing fixed step sizes as those established in Thm.~\ref{theorem:convergence-subs}. 
		\subsection{Empirical Phase Transition}
		\label{sec:ept}
			To characterise the {\em phase transition} of \eqref{eq:bcp}, that is the transition between a region in which Alg.~\ref{alg1} successfully recovers $(\x,{\bs g})$ with probability $1$ and that in which it does with probability $0$, we ran some extensive simulations by generating $256$ random instances of \eqref{eq:measurement-model-matrix} for each $n = \{2^1, \ldots, 2^8\}$, probing the same range for $m$ and $p$; we also varied for each configuration $\rho = \{10^{-3}, 10^{-2}, \ldots, 1\}$, generating ${\bs g} = \vone_m + {\bs e} $ with ${\bs e}$ drawn uniformly at random on $\onep \cap \rho\,\bb S^{m-1}_\infty$. Then we evaluated ${\rm P}_T \coloneqq \Pro\left[\ts\max\left\{\tfrac{\|\bar{{\bs g}}-{\bs g}\|}{\|{\bs g}\|},\tfrac{\|\bar{\x}-\x\|}{\|\x\|}\right\} < T\right]$ on the trials with threshold $T =  10^{-3}$ chosen according to the stop criterion $f(\bxi,{\bs \gamma}) < 10^{-8}$. Of this large dataset we report the cases specified in Fig.~\ref{fig:ept}, highlighting the contour lines of ${\rm P}_T$~ for $\rho = \{10^{-3}, 10^{-2}, 10^{-1}\}$ estimated from the outcome of our experiments as a function of $\log_2 m$ and $\log_2 p$. There, we also plot a curve corresponding to the sample complexity bound that grants the verification of Thm.~\ref{theorem:convergence} (up to $\log$ factors), with the corresponding phase transition occurring at $\log_2 p \simeq \log_2\left(\tfrac{n}{m} + 1\right)$. This curve matches the empirical phase transition up to a shift that is due to the required accuracy (\ie the value of $\delta$) of the concentration inequalities. Moreover, we appreciate how an increase in $\rho$ does affect, albeit mildly, the requirements on $(m,p)$ for a successful recovery of $(\x, {\bs g})$, and just as expected given its effect on the initialisation and the distance decay in \eqref{eq:convres}. % \mtodo{Valerio, the following sentence seems wrong now we play with $\upmu_{\max}$. Please clarify and update all text elsewhere connected to this} $[\rightarrow]$ 
			\new{
			\subsection{Subspace Priors and the Effect of Coherence}
				Since Alg.~\ref{alg2} generalises Alg.~\ref{alg1} to known subspaces, we now focus on how the phase transition in Fig.~\ref{fig:ept} changes under subspace priors, as Thm.~\ref{theorem:convergence-subs} suggests that the transition will occur when $ m p \gtrsim k+\upmu_{\max}^2 h$ (up to $\log$ factors).  				
				Indeed, the second term scales with $\upmu_{\max} \in [1, \sqrt{\tfrac{m}{h}}]$ defined in \eqref{eq:mumax}, \ie the larger $\upmu_{\max}$, the farther the phase transition curve will be in the $(\log_2 m,\log_2 p)$ diagram.
				
				As $\upmu_{\max}$ depends on the nature of the known subspace prior $\bs B \coloneqq \big(\tfrac{\vone_m}{\sqrt{m}},\ \bs B^\perp\big)$ in Def.~\ref{def:blind-calibr-with-sub}, let us first fix $\bs Z$ by drawing its $k$ column vectors $\bs Z_i \sim \cl N(\vzer_n,\I_n)$ and running Gram-Schmidt orthonormalisation to comply with the hypothesis $\bs Z^\top \bs Z = \bs I_k$. Then, we compare three cases of $\bs B$ ranging between low and high coherence as follows; firstly, let us take $\bs B^\perp \coloneqq {\bs C}_m  {\bs S}_\Omega$ with ${\bs C}_m$ an orthonormal basis of $\bb R^m$ specified below, and $\bs S_{\Omega}$  the selection operator at a randomly drawn index set $\Omega \subset [m] : |\Omega| = h-1$. We now proceed to detail the different choices of $\bs C_m$ and their coherence bounds.

				\begin{itemize}
				
				\item {\emph{Case 1} (DCT):}  we let $\bs C_m$ be the $m$-dimensional type-II DCT matrix, \ie
				\[
				(\bs C_m)_{i,j} = 
				\begin{cases} 
				\tinv{\sqrt{m}}, & i \in [m], j = 1\\
				\sqrt{\tfrac{2}{m}} \cos(\tfrac{\pi}{2m} (j-1) (2 i-1)), & i \in [m], j \in [m] \setminus \{1\}
				\end{cases}.
				\]
				In this case, we pose that $1\notin \Omega$ to avoid selecting $\tfrac{\vone_m}{\sqrt{m}}$ as by construction it is already the first column of $\bs B$.
				It is then simply estimated that \[
				\upmu_{\max} = \sqrt{\tfrac{m}{h}} \max_{i\in [m]} \sqrt{\tinv{m}+\sum_{j \in \Omega} {\tfrac{2}{m}} \cos^2(\tfrac{\pi}{2m} (j-1) (2 i-1))} <  
				\sqrt{\tfrac{m}{h}}  \sqrt{\tfrac{2 h-1}{m}} < \sqrt{2},
				\]
				\ie that this ``DCT''  case always attains relatively low coherence, similarly to the  discrete Fourier transform (see \cite[Sec.~I-D]{AhmedRechtRomberg2014});
				
				\item {\emph{Case 2} (Id.):} we let $\bs C_m \coloneqq \begin{bmatrix}
				{\bs U} & \bs U^\perp
				\end{bmatrix}$, where  $\bs U \coloneqq \begin{bsmallmatrix} \tinv{\sqrt{m}} \vone_m & \sqrt{\tfrac{m}{m-1}}\begin{psmallmatrix}1-\tinv{m} \\ -\tinv{m} \vone_{m-1} \end{psmallmatrix} \end{bsmallmatrix}$ and $\bs U^\perp$ is an orthonormal basis for the null space of $\bs U^\top$ (again, obtained by the Gram-Schmidt process). We pose again $1\notin \Omega$, and note that the resulting $\bs C_m$ resembles an identity matrix with a negative offset (hence the shorthand ``Id.''). In fact, since the first row obtained this way always has only two non-zero elements, it is easily shown that $\|\bs B^\top \bs c_1 \|^2_2 = \tinv{m} +  (1-\tinv{m})^2{\tfrac{m}{m-1}} = 1$ which sets $\upmu_{\max} = \sqrt{\tfrac{m}{h}}$. This choice attains the upper-bound for $\upmu_{\max}$, and is thus expected to exhibit the worst-case phase transition;
				
				\item {\emph{Case 3} (Rand.):} we let $\bs C_m \coloneqq {\bs U} \bs V$, where ${\bs U} \in \bb R^{m\times m-1}$ is now an orthonormal basis for the null space of $\vone_m$ (again, obtained by the Gram-Schmidt process). We then apply a random rotation by $\bs V \in \bb R^{m-1 \times h-1}$ generated by drawing $h-1$ column vectors $\bs V_j \sim \cl N(\vzer_{m-1},\bs I_{m-1})$, orthogonalising them afterwards. 
				This results in $\bs V^\top$ and $\bs C^\top_m$ verifying the restricted isometry property. In particular, by \cite[Lemma 3]{DavenportBoufounosWakinEtAl2010} it follows that the rows of $\bs B$ have norm $\|\bs B^\top \bs c_i\| \leq \tinv{\sqrt{m}} + (1+\delta) \sqrt\frac{h-1}{m} \lesssim \sqrt{\tfrac{2h}{m}} $, hence $\upmu_{\max} \gtrsim \sqrt{2}$ as in the ``DCT'' case.
				
				\end{itemize}
				
				With these three cases at hand, we may now run some simulations to assess how Alg.~\ref{alg2} actually performs on randomly generated instances of \eqref{eq:sub-sens-model}. To do so, we repeat the generation of Sec.~\ref{sec:ept} for $n = 2^8$, $k = 2^6$, varying $h = \{2^4,2^5\}$ and $m = \{2^{\log_2 h}, \ldots, 2^8\}$ since only $m \geq h$ is meaningful. The resulting phase transitions at probability $0.9$ are reported in Fig.~\ref{fig:phtalg2}. There, we can see that a large $\upmu_{\max}$ does have a negative effect on the phase transition for successful recovery, that is consistent with the result in Thm.~4.1.
				
				% However, we anticipate that such experiments will yield results exhibiting a transition at $ \log_2 p \simeq \log_2 \big(m^{-1}(k+\upmu_{\max}^2)\big)$, \ie at worst at $\log_2 p \simeq \log_2 \big({\tfrac{k}{m} + 1}\big)$ when $\upmu_{\max} = \sqrt{\tfrac{m}{h}}$~(both rates up to $\log$ factors at the argument). 
			}
		\begin{figure}[t]
		\new{
			\centering
			\null\hfill
			\subfloat[$h = \dim \cl B = 2^4$]{
				\begin{tikzpicture}
				\begin{axis}[%
				width=4.747in,
				height=4.747in,
				scale only axis,
				scale=0.4,
				point meta min=0,
				point meta max=1,
				separate axis lines,
				every outer x axis line/.append style={thick,black},
				every x tick label/.append style={font=\color{black}},
				xmin=2.99995,
				xmax=8.00005,
				xlabel={$\log_2{m}$},
				xtick = {1,...,8},
				every outer y axis line/.append style={thick,black},
				every y tick label/.append style={font=\color{black}},
				ymin=-0.00005,
				ymax=8.00005,
				ylabel={$\log_2{p}$},
				ytick = {0,...,8},
				set layers,
				every axis plot/.append style={on layer=pre main},
				axis background/.style={fill=white},
				legend style={at={(0.92,0.92)},anchor=north east,legend cell align=right,align=right,draw=none}
				]
				
				\addplot[draw=black,loosely dashed,thick]
				table[row sep=crcr] {%
					x	y\\
					6.93940023116219	0\\
					6.75488750216347	0.470046082949308\\
					6.60765600410156	1\\
					6.5077946401987	1.18435181840909\\
					6.2667865406949	1.41166486997403\\
					6	1.49937103966411\\
					5.82321532235684	1.58496250072116\\
					5.75488750216347	1.68552927939257\\
					5.52356195605701	1.95108985698734\\
					5.40097973838324	2\\
					5.28540221886225	2.13937734674223\\
					5	2.50613141333454\\
					4.91195959261382	2.58496250072116\\
					4.75488750216347	2.81359719330634\\
					4.57323308558907	3\\
					4.52356195605701	3.11632777002978\\
					4.32192809488736	3.48962330861467\\
					4.19859200658636	3.58496250072116\\
					4	3.75491320085689\\
				};
				
				\addplot[densely dotted,draw=black,thick]
				table[row sep=crcr] {%
					x	y\\
					4	3.880992831748\\
					4.32192809488736	3.72056881236672\\
					4.52356195605701	3.64245687291563\\
					4.75488750216347	3.60204223320177\\
					4.82916401665938	3.58496250072116\\
					5	3.54850142344043\\
					5.10761067268577	3.58496250072116\\
					5.28540221886225	3.66381962558414\\
					5.52356195605701	3.78834741072633\\
					5.75488750216347	3.8593560633673\\
					6	3.87085269304355\\
					6.2667865406949	3.86608629885969\\
					6.5077946401987	3.93341644396596\\
					6.67626704608377	4\\
					6.75488750216347	4.0557404363863\\
					7	4.04775734058628\\
					7.25738784269265	4.02303058457729\\
					7.5077946401987	4.04390115780602\\
					7.75488750216347	4.01777525159453\\
					7.83659166810898	4\\
					8	3.98588744206471\\
				};
				
				\addplot[solid,draw=gray,thick]
				table[row sep=crcr] {%
					x	y\\
					6.9179496299615	0\\
					6.75488750216347	0.480122324159021\\
					6.64463609393772	1\\
					6.5077946401987	1.25821566454428\\
					6.2667865406949	1.4608997582649\\
					6	1.53033164162297\\
					5.85237542743936	1.58496250072116\\
					5.75488750216347	1.7007100104802\\
					5.52356195605701	1.9431166879161\\
					5.42203069967398	2\\
					5.28540221886225	2.30218441980382\\
					5.02778251688039	2.58496250072116\\
					5	2.60924660972151\\
					4.75488750216347	2.87596580481322\\
					4.56395213077401	3\\
					4.52356195605701	3.05823156115776\\
					4.32192809488736	3.52196653910503\\
					4.2317882283189	3.58496250072116\\
					4	3.75221641834099\\
				};

				\addlegendentry{DCT}
				\addlegendentry{Id.}
				\addlegendentry{Rand.}
				
				\fill[gray,opacity = 0.1] (axis cs:3,0) rectangle (axis cs:4,8);
				\node at (axis cs:3.5,0.5) [color=black,opacity=1] {\tiny$m < h$};
				\end{axis}
				\end{tikzpicture}%
			}
			\hfill
			\subfloat[$h = \dim \cl B = 2^5$]{
				\begin{tikzpicture}
				\begin{axis}[%
				width=4.747in,
				height=4.747in,
				scale only axis,
				scale=0.4,
				point meta min=0,
				point meta max=1,
				separate axis lines,
				every outer x axis line/.append style={thick,black},
				every x tick label/.append style={font=\color{black}},
				xmin=2.99995,
				xmax=8.00005,
				xlabel={$\log_2 {m}$},
				xtick = {1,...,8},
				every outer y axis line/.append style={thick,black},
				every y tick label/.append style={font=\color{black}},
				ymin=-0.00005,
				ymax=8.00005,
				ylabel={$\log_2 {p}$},
				ytick = {0,...,8},
				set layers,
				every axis plot/.append style={on layer=pre main},
				axis background/.style={fill=white},
				legend style={at={(0.92,0.92)},anchor=north east,legend cell align=right,align=right,draw=none}
				]
				
				\addplot[draw=black,loosely dashed,thick]
				table[row sep=crcr] {%
					x	y\\
					5	3.99208375341975\\
					5.28540221886225	3.64848864856996\\
					5.3287039892613	3.58496250072116\\
					5.52356195605701	3.36560156295072\\
					5.69880858189524	3\\
					5.75488750216347	2.90719658401218\\
					6	2.60044897457484\\
					6.01699277329267	2.58496250072116\\
					6.2667865406949	2.35985913002479\\
					6.5077946401987	2.05131250006326\\
					6.53613097757998	2\\
					6.75488750216347	1.76017294203348\\
					6.89481002117137	1.58496250072116\\
					7	1.49908157956265\\
					7.25738784269265	1.295843103813\\
					7.5077946401987	1.045221539956\\
					7.55232881880863	1\\
					7.75488750216347	0.735955056179776\\
					8	0.269938650306749\\
					nan	nan\\
				};
				
				\addplot[densely dotted,draw=black,thick]
				table[row sep=crcr] {%
					x	y\\
					5	3.88359867836317\\
					5.28540221886225	3.93582961357783\\
					5.52356195605701	3.94084745921493\\
					5.6675503061845	4\\
					5.75488750216347	4.06072662186241\\
					6	4.14615333755655\\
					6.2667865406949	4.25466306623921\\
					6.5077946401987	4.33546303820587\\
					6.75488750216347	4.39353258868101\\
					7	4.42747082518549\\
					7.25738784269265	4.44562069399027\\
					7.5077946401987	4.44768341170092\\
					7.75488750216347	4.42233444632808\\
					8	4.41287603917248\\
					nan	nan\\
				};
				
				\addplot[solid,draw=gray,thick]
				table[row sep=crcr] {%
					x	y\\
					5	3.89691879102878\\
					5.27125004272032	3.58496250072116\\
					5.28540221886225	3.56863796581731\\
					5.52356195605701	3.39519205734339\\
					5.75488750216347	3.00254331522053\\
					5.75694727105285	3\\
					6	2.81657518758463\\
					6.224583020133	2.58496250072116\\
					6.2667865406949	2.54021127115779\\
					6.5077946401987	2.18110294139974\\
					6.64818831176959	2\\
					6.75488750216347	1.90987757158517\\
					7	1.58928580800531\\
					7.00417385690853	1.58496250072116\\
					7.25738784269265	1.37743194490975\\
					7.5077946401987	1.09275620216039\\
					7.58363502357402	1\\
					7.75488750216347	0.765494137353433\\
					8	0.415274463007158\\
					nan	nan\\
				};

				\addlegendentry{DCT}
				\addlegendentry{Id.}
				\addlegendentry{Rand.}
				%	\addlegendentry{Former Bound}
				
				\fill[gray,opacity = 0.1] (axis cs:3,0) rectangle (axis cs:5,8); 
				\node at (axis cs:4,0.5) [color=black,opacity=1] {\tiny$m < h$};
				
				\end{axis}
				\end{tikzpicture}%
			}
			\hfill\null
			\caption{\label{fig:phtalg2}Phase transition of Alg.~\ref{alg2} at probability $0.9$ for $n=2^8$, $k=2^6$, and three different cases of $\bs B$ (\ie above each curve, the corresponding probability of successful recovery exceeds $0.9$).}}
		\end{figure}
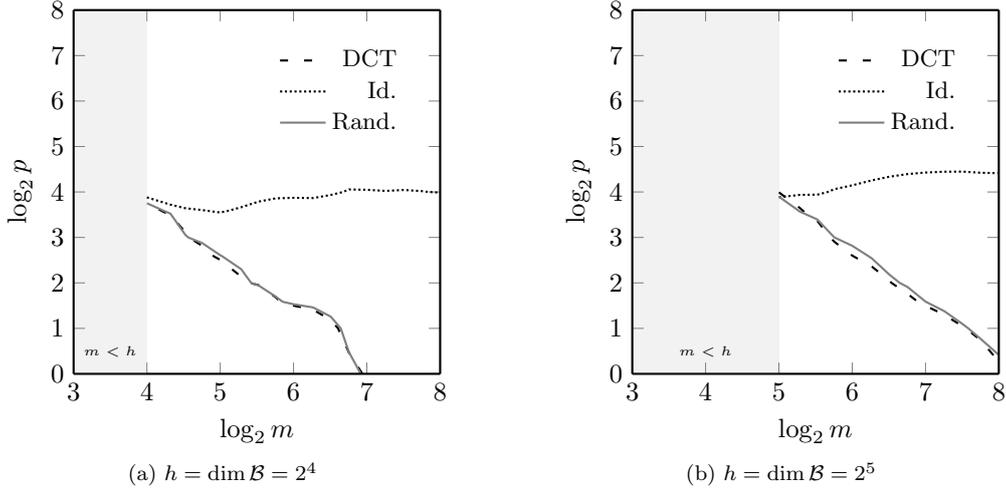
		
		\subsection{Noise Stability}
			\begin{figure}[!t]
				\centering
				\begin{tikzpicture}
				\begin{axis}[%
				width=3in,
				height=2in,
				scale only axis,
				unbounded coords=jump,
				separate axis lines,
				every outer x axis line/.append style={black},
				every x tick label/.append style={font=\color{black}},
				xmin=5,
				xmax=85,
				xlabel={${-20\log_{10}{\sigma}} \ (\unit{dB})$},
				every outer y axis line/.append style={black},
				every y tick label/.append style={font=\color{black}},
				ymin=-120,
				ymax=0,
				ylabel={$20 \log_{10}\Ex {\rm RMSE}_{\max} \ (\unit{dB})$},
				axis background/.style={fill=white},
				legend style={at={(0.05,0.05)},anchor=south west,legend cell align=left,align=left,draw=black}
				]
				\addplot [color=mathblue,line width=0.75pt,only marks,mark=+,mark options={solid}]
				table[row sep=crcr]{%
					10	-5.29268236869554\\
					20	-20.9007843716715\\
					30	-31.5147348714436\\
					40	-41.5316504944748\\
					50	-51.4163019398795\\
					60	-60.1371859280373\\
					70	-69.0315417863631\\
					80	-73.7155871693971\\
				};
				\addlegendentry{$p = 4$};
				
				\addplot [color=mathgrn,line width=0.75pt,only marks,mark=triangle*,mark options={solid}]
				table[row sep=crcr]{%
					10	-21.1870013280045\\
					20	-31.1879386808852\\
					30	-41.2138231776866\\
					40	-51.1319522279787\\
					50	-61.1222979505465\\
					60	-70.9175507334258\\
					70	-80.3731213950819\\
					80	-85.8151992710055\\
				};
				\addlegendentry{$p = 16$};
				
				\addplot [color=mathpurp,line width=0.75pt,only marks,mark=square*,mark options={solid}]
				table[row sep=crcr]{%
					10	-27.8570849822191\\
					20	-37.8035785102795\\
					30	-47.8245628228927\\
					40	-57.7750766075998\\
					50	-67.2463903420312\\
					60	-77.4026243519139\\
					70	-86.9457786884811\\
					80	-94.7564428261315\\
				};
				\addlegendentry{$p = 64$};
				
				\addplot [color=mathsand,line width=0.75pt,only marks,mark=*,mark options={solid}]
				table[row sep=crcr]{%
					10	-34.0640112029873\\
					20	-44.0016723550515\\
					30	-54.0765191092057\\
					40	-64.0204621310396\\
					50	-73.7623860132102\\
					60	-83.7072616363363\\
					70	-93.5922052497583\\
					80	-100.850265890787\\
				};
				\addlegendentry{$p = 256$};
				
				\addplot [color=mathblue,dashed,line width=0.5pt,forget plot]
				table[row sep=crcr]{%
					0	8.12217234997463\\
					5	0.822388422670697\\
					10	-6.18958207711677\\
					15	-12.9137391493878\\
					20	-19.3500827941423\\
					25	-25.4986130113804\\
					30	-31.359329801102\\
					35	-36.9322331633072\\
					40	-42.2173230979959\\
					45	-47.2145996051682\\
					50	-51.924062684824\\
					55	-56.3457123369633\\
					60	-60.4795485615861\\
					65	-64.3255713586925\\
					70	-67.8837807282825\\
					75	-71.154176670356\\
					80	-74.136759184913\\
					85	-76.8315282719535\\
					90	-79.2384839314776\\
				};
				\addplot [color=mathgrn,dashed,line width=0.5pt,forget plot]
				table[row sep=crcr]{%
					0	-9.33791373004191\\
					5	-15.0227511590819\\
					10	-20.5972049322976\\
					15	-26.0612750496892\\
					20	-31.4149615112566\\
					25	-36.6582643169998\\
					30	-41.7911834669188\\
					35	-46.8137189610137\\
					40	-51.7258707992843\\
					45	-56.5276389817308\\
					50	-61.2190235083531\\
					55	-65.8000243791512\\
					60	-70.2706415941251\\
					65	-74.6308751532748\\
					70	-78.8807250566004\\
					75	-83.0201913041017\\
					80	-87.0492738957789\\
					85	-90.9679728316319\\
					90	-94.7762881116607\\
				};
				\addplot [color=mathpurp,dashed,line width=0.5pt,forget plot]
				table[row sep=crcr]{%
					0	-17.1057239193638\\
					5	-22.3892214960612\\
					10	-27.6199070358082\\
					15	-32.7977805386047\\
					20	-37.9228420044507\\
					25	-42.9950914333462\\
					30	-48.0145288252913\\
					35	-52.9811541802859\\
					40	-57.89496749833\\
					45	-62.7559687794236\\
					50	-67.5641580235668\\
					55	-72.3195352307594\\
					60	-77.0221004010017\\
					65	-81.6718535342934\\
					70	-86.2687946306346\\
					75	-90.8129236900254\\
					80	-95.3042407124657\\
					85	-99.7427456979555\\
					90	-104.128438646495\\
				};
				\addplot [color=mathsand,dashed,line width=0.5pt,forget plot]
				table[row sep=crcr]{%
					0	-23.0873725067917\\
					5	-28.4459772863299\\
					10	-33.7442494298663\\
					15	-38.9821889374007\\
					20	-44.1597958089332\\
					25	-49.2770700444638\\
					30	-54.3340116439925\\
					35	-59.3306206075193\\
					40	-64.2668969350442\\
					45	-69.1428406265671\\
					50	-73.9584516820882\\
					55	-78.7137301016073\\
					60	-83.4086758851245\\
					65	-88.0432890326398\\
					70	-92.6175695441532\\
					75	-97.1315174196647\\
					80	-101.585132659174\\
					85	-105.978415262682\\
					90	-110.311365230188\\
				};
				\end{axis}
				\end{tikzpicture}%
				\caption{\label{fig:noise}Performances of Alg.~\ref{alg1} in the presence of noise. This experiment details the exemplary case $m = n = 2^8, \rho = 10^{-1}$ for different values of $p$; we report the average relative mean-square error as a function of the noise level $\sigma$.}
			\end{figure}
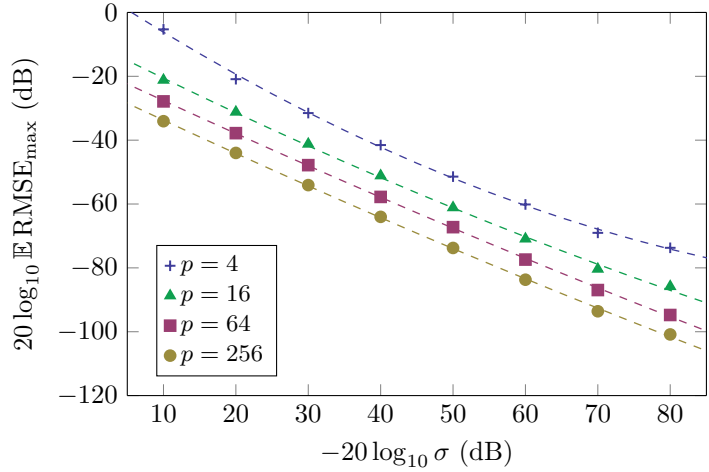
			
			As a numerical confirmation of the theory devised in Sec.~\ref{sec:nstability-sub} we run some numerical experiments to verify the effect of noise on the estimate obtained by Alg.~\ref{alg1}, \ie without subspace priors for simplicity. The simulations are carried out as follows: we generate $64$ instances\footnote{This lower value is due to the fact that, when our algorithm does not converge, meeting the stop criterion requires a larger number of iterations due to the presence of noise.} of \eqref{eq:measurement-model-matrix-noise}, fixing $n = m = 2^8, \rho = 10^{-1}$ and varying $p =\{2^1, \ldots, 2^8\}$. Moreover, we generate additive and bounded noise instances collected in a matrix $\bs N$ whose entries are drawn as  $\nu_{i,l} \sim_{\rm \iid} \cl N(0,1)$ and normalised afterwards to meet a given $\sigma = \tfrac{1}{\sqrt{mp}}\|{\bs N}\|_F$. This noise level value is then varied as $\sigma = \unit[\{10, 20, \ldots, 80\}]{dB}$. Since the objective function will reach a noise-dependent value, we change the stop criterion to the condition 
			$\max\left\{\tfrac{\|\bxi_{j+1} -\bxi_{j}\|}{\|\bxi_j\|}, \tfrac{\|{\bs \gamma}_{j+1} -{\bs \gamma}_{j}\|}{\|{\bs \gamma}_j\|} \right\} < 10^{-6}$, 
			terminating the algorithm and returning its best estimate up to the given tolerance. To extract a single figure of merit from the outcomes of this experiment, we compute the $\Ex {\rm RMSE}_{\max}$ on our set of $64$ trials, with $\Ex$ here denoting the {\em sample average} over the trials' outcomes. 
			The results after running Alg.~\ref{alg1} are reported in Fig.~\ref{fig:noise} and confirm the predicted graceful decay in the $\Ex  {\rm RMSE}_{\max}$ as a function of $\sigma$, \ie the achieved relative mean-square error decreases linearly (in a $\log-\log$ scale) with the amount of noise injected into the model, as suggested from Thm.~\ref{theorem:stability}.
				
	\subsection{A Computational Imaging Example}
	% Vignetting, PRNU, such examples
	
	We envision that our framework could be applied broadly to sensing systems where obtaining calibrated gains as formulated in our models is a critical issue. Generally speaking, whenever there are means to capture measurements of the type $\A_l \x$, and whenever the sensors (antennas, pixels, nodes) assigned to capturing the output of this operation are subject to gain uncertainties, it is worth putting into account the presence of ${\bs g}$ and to calibrate the sensing system against it. Clearly, the main requirement is indeed the introduction of several random draws of the sensing matrices $\A_l, l\in[p]$ which, depending on $m,n$ and the presence of subspace priors, will allow for lower values of $p$ as shown by our main sample complexity results. As an example, one could consider an image formation model in which a source $\x$ illuminates a programmable medium, set to apply a sensing matrix $\A_l$ and to capture the output of this operation by means of a focal plane array. The disturbance could then be regarded as uncalibrated fixed pattern noise on this sensor array, as stylised in Fig.~\ref{fig:modello}, or in fact as any attenuation such as some ``haze'' affecting the sensor array in a multiplicative fashion. In fact, this type of issue could physically arise in imaging modalities where $\A_l$ are random convolutions rather than sub-Gaussian random matrices, as anticipated before. With a due gap between the theory covered in this paper and the actual nature of the sensing operator, our algorithm is still practically applicable and will eventually converge once a sufficient amount of observations with random sensing matrices is collected. 
	
	To apply our result in a realistic computational imaging context, we assume that $\x$ is a $n=\unit[128\times128]{pixel}$ monochromatic image acquired by an uncalibrated sensing device that implements \eqref{eq:measurement-model-matrix} in which its $m=\unit[64\times64]{pixel}$ sensor array has an unknown set of gains ${\bs g}\in \Pi^m_+$. This set of gains is specialised in two cases described hereafter. For the sake of this application example, we maintain $\A_l$ comprised of \iid rows $\ail\sim \cl N(\vzer_n,\I_n)$. 
	
	\paragraph{Blind Calibration in Absence of Priors}
	{
	\begin{figure}[t]
		\centering
            \includegraphics[width=.6\textwidth]{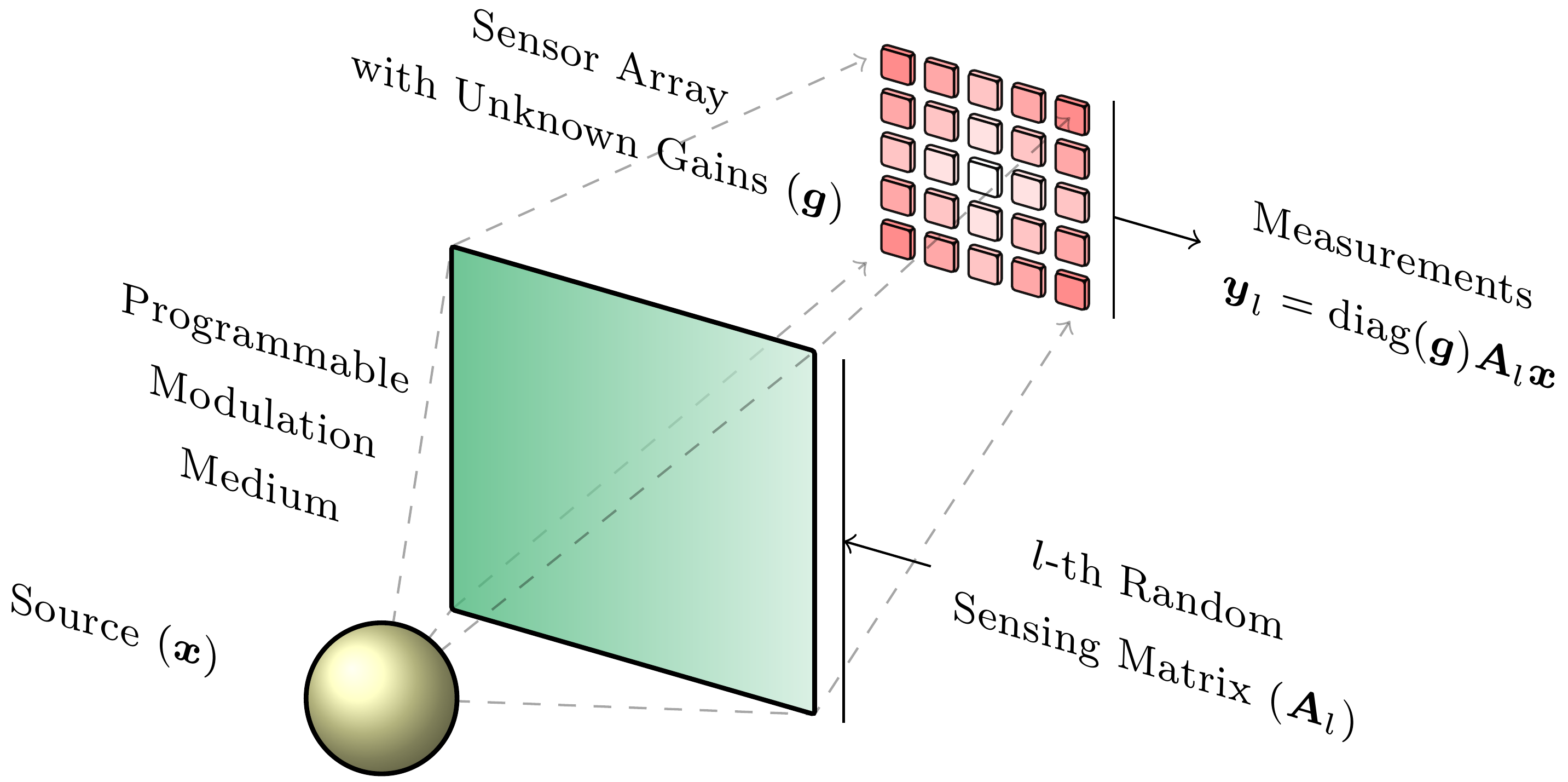}
	 	\caption{\label{fig:modello} A computational imaging model: our blind calibration framework entails the joint recovery of the source $\x$ and sensor gains ${\bs g}$ by exploiting multiple random sensing matrices $\A_l$ (\eg $p$ programmable random masks in a random convolution setup~\cite{Romberg2009}). The intensity of ${\bs g}$ is represented in shades of red as a possible {vignetting} of the sensor array.}
	 \end{figure}
	 }
	 	\begin{figure}[p]
	 		\null\hfill
	 		\subfloat[{$\x$ (true signal).}]{
	 			\includegraphics{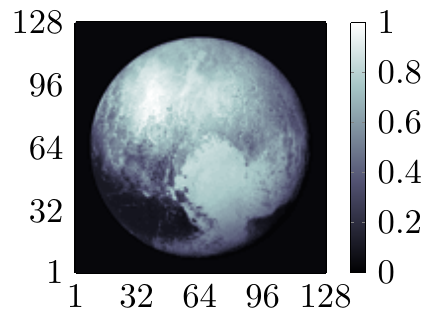}
	 		}
	 		\hfill
	 		\subfloat[{$\bar{\x}_{\rm ls}$ recovered by least squares with model error; ${\rm RMSE} = \unit[-9.22]{dB}$.}]{
	 			\includegraphics{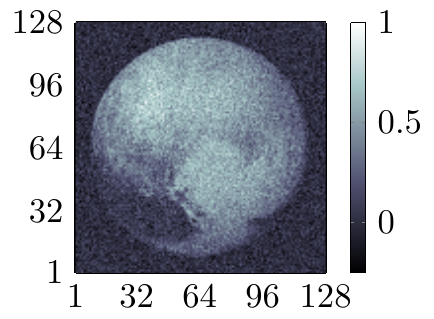}
	 		}
	 		\hfill
	 		\subfloat[{$\bar{\x}$ recovered by Alg.~\ref{alg1}; ${\rm RMSE} = \unit[-149.96]{dB}$.}]{
	 			\includegraphics{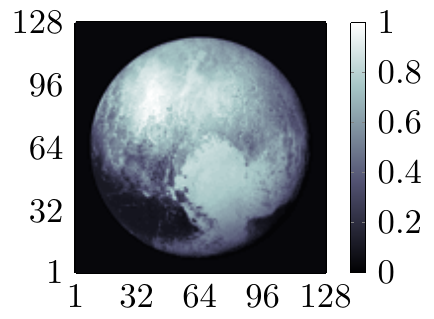}
	 		}
	 		\hfill\null			
	 		\\
	 		\null\hfill
	 		\subfloat[{${\bs g}$ (true sensor gains), $\rho = 0.99$.}]{				
	 			\includegraphics{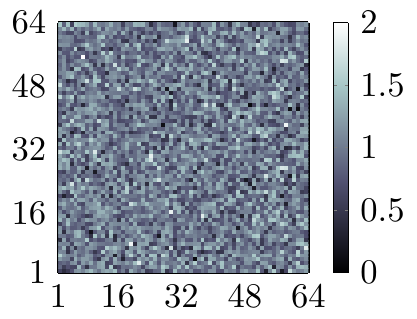}
	 		}
	 		\hfill
	 		\subfloat[{$ \bar{{\bs g}}$ recovered by Alg.~\ref{alg1}; ${\rm RMSE} = \unit[-145.39]{dB}$.}]{
	 			\includegraphics{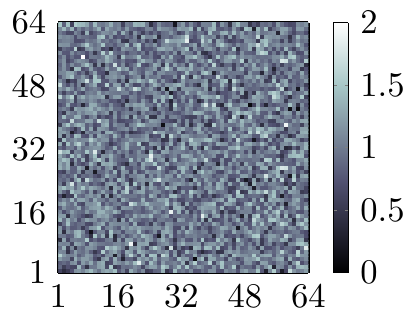}
	 		}
	 		\hfill\null
	 		\caption{\label{fig:randexp} A high-dimensional example of blind calibration for computational imaging. The unknown gains ${\bs g}$ $(m=\unit[64\times64]{pixel})$ and signal $\x$ $(n = \unit[128\times128]{pixel})$ are perfectly recovered with $p = 10$ snapshots.}
	 	\end{figure}
	 	\begin{figure}[p]
	 		\null\hfill
	 		\subfloat[{$\x$ (true signal)}]{
	 			\includegraphics{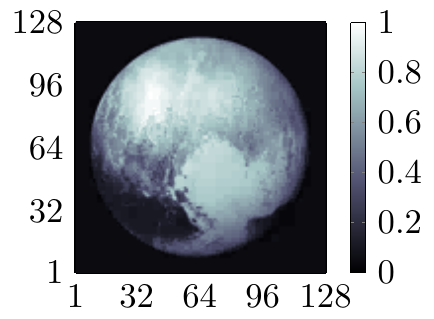}
	 		}
	 		\hfill
	 		\subfloat[{$\bar{\x}_{\rm ls}$ recovered by least squares with model error; ${\rm RMSE} = \unit[-3.42]{dB}$.}]{
	 			\includegraphics{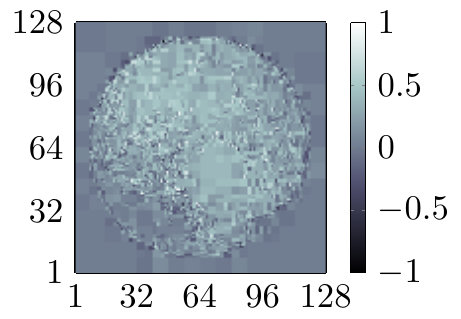}
	 		}
	 		\hfill
	 		\subfloat[{$\bar{\x}$ recovered by Alg.~\ref{alg2}; ${\rm RMSE} = \unit[-138.84]{dB}$.}]{
	 			\includegraphics{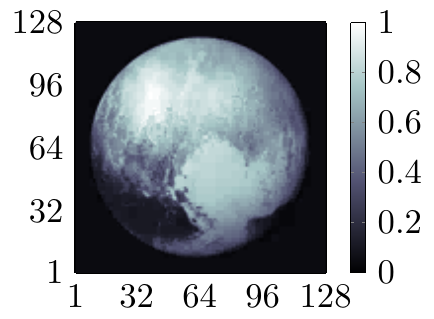}
	 		}
	 		\hfill\null			
	 		\\
	 		\null\hfill
	 		\subfloat[{${\bs g}$ (true sensor gains), $\rho = 0.99$}]{				
	 			\includegraphics{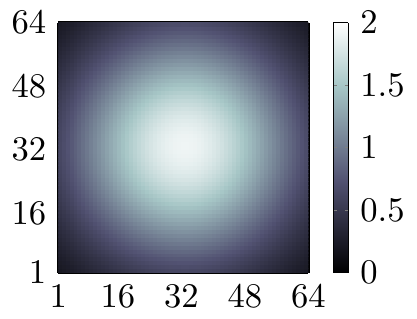}
	 		}
	 		\hfill
	 		\subfloat[{$ \bar{{\bs g}}$ recovered by Alg.~\ref{alg2}; ${\rm RMSE} = \unit[-144.99]{dB}$.}]{
	 			\includegraphics{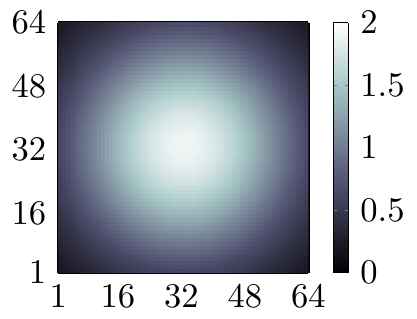}
	 		}
	 		\hfill\null
				\caption{\label{fig:subexp} A high-dimensional example of blind calibration for computational imaging with subspace priors. The unknown gains ${\bs g}$ $(m=\unit[64\times64]{pixel})$ and signal $\x$ $(n = \unit[128\times128]{pixel})$ are perfectly recovered with $p = 1$ snapshot, since they are described by known subspaces of dimension $h=256$ and $k=2730$, respectively.}
		\end{figure}
		
		In this first example we randomly draw the gains ${\bs g} \in \Pi^m_+$ at random as ${\bs g} \in \vone_m + \onep \cap \rho \bb S^{m-1}_\infty$, fixing $\rho=0.99$. We capture $p=10$ snapshots, again so that $\tfrac{mp}{n+m} = 2$. By running Alg.~\ref{alg1} we obtain the results depicted in Fig.~\ref{fig:randexp}. The recovered $(\bar{\x},\bar{{\bs g}})\approx({\x},{{\bs g}})$ by solving \eqref{eq:bcp} attains ${\rm RMSE}_{\max} \approx \unit[-138.84]{dB}$ in accordance with the stop criterion at $f(\bxi_k,{\bs \gamma}_k) < 10^{-7}$. 
		Instead, by fixing ${\bs \gamma} \coloneqq \vone_m$ and solving \eqref{eq:ls} only with respect to $\bxi$, the least-squares solution $\bar{\x}_{\rm ls}$ reaches a ${\rm RMSE} \coloneqq \tfrac{\|\bar{\x}-\x\|}{\|\x\|} \approx \unit[-9.22]{dB}$. 	
	
	\paragraph{Blind Calibration with Subspace Priors}
	
		We now proceed to evaluate the effect of subspace priors on the same exemplary case, with $\x$ constructed so that it matches the known subspace prior ${\bs x} = {\bs Z} {\bs z}$, where ${\bs Z}$ is a set of $k=2730$ basis elements of a two-dimensional Haar orthonormal wavelet basis in $\R^{n}$. Moreover, we generate the gains with a subspace prior that entails ${\bs g} = {\bs B}{\bs b}$, where ${\bs B}$ is a set of $h = 256$ basis elements of a two-dimensional discrete cosine transform basis in $\R^m$, including the DC component as its first column. The vector ${\bs e} = {\bs B}^\perp {\bs b}^\perp$ is drawn with a low-pass profile as shown in Fig.~\ref{fig:subexp} as a means to simulate a structured model for the gains. We also fix $\|{\bs e}\|_\infty = \rho$ on the generated profile, with $\rho = 0.99$. Substantially, what changes with respect to the previous experiment is the introduction of a known subspace prior for the signal and gains, both benefiting from such a low-dimensional model. 
% \mtodo{Valerio, here also the next sentences should be clarified now that we have $mp \gtrsim \delta^{-2} (k + \upmu^2_{\max} h)\, \log^2(m(p+n))\log(\frac{1}{\delta})$ theoretically.}
Due to this additional prior we can take a single snapshot (\ie $p=1$) provided that $\bs B$ has sufficiently low $\upmu_{\max}$ and that $m > c ({k + \upmu^2_{\max} h}) \, \log^2(m(n+1))$ for some $c > 0$, \ie $m$ must exceed, up to some constant and $\log$ factors, ${k + \upmu^2_{\max} h}$.
%, even if our theory hints that generally we should fulfil the condition $p \gtrsim \log mp$ even with subspace priors. 

		Then, by running Alg.~\ref{alg2} we obtain the results depicted in Fig.~\ref{fig:subexp}. The recovered $(\bar{\x},\bar{{\bs g}})\coloneqq (\bs Z \bar{\bs z},\bs B \bar{{\bs b}})\approx({\x},{{\bs g}})$ achieves a ${\rm RMSE}_{\max} = \unit[-138.84]{dB}$ in accordance with the stop criterion at $f^{\rm s}(\bs \zeta_k,{\bs\beta}_k) < 10^{-7}$. 
		Instead, by fixing ${\bs\beta} \coloneqq (\sqrt{m}, \vzer_m^\top)^\top$ and solving \eqref{eq:ls} when setting $\bs\xi =\bs Z\bs \zeta$ and with respect to $\bs\zeta$ yields a least-squares solution $\bar{\bs z}_{\rm ls}$ that reaches a ${\rm RMSE} \coloneqq \tfrac{\|\bar{\bs z}_{\rm ls}-\bs z \|}{\|\bs z\|} \approx \unit[-3.42]{dB}$. 	
	
	\medskip
	
	We have therefore seen how the algorithms devised in this paper work on practical instances of the uncalibrated sensing models studied in this paper. The application of this bilinear inverse problem in the solution of physical instances of \eqref{eq:measurement-model-matrix} or \eqref{eq:measurement-model-matrix-noise} is an open subject for future developments.
		
\section{Conclusion}
	\label{sec:conclusion}
		We presented and solved a non-convex formulation of blind calibration in the specific, yet important case of linear random sensing models affected by unknown gains. In absence of {\em a priori} structure on the signal and gains, we have devised and analysed a descent algorithm based on a simple projected gradient descent. In this case, our main results have shown that provable convergence of the algorithm can be achieved at a rate  $mp$ $=$ ${\cl O}\big((m+n)$ $\log^2(m(p+n))\big)$ that is linear (up to $\log$ factors) with respect to the number of unknowns in this bilinear inverse problem. Similarly, using a subspace prior on both the signal and gains, a straightforward extension of the previous descent algorithm into Alg.~\ref{alg2} proved that the sample complexity ensuring convergence in this setting is~$mp$ $=$ ${\cl O}\big((k + \upmu^2_{\max} h)$ $\log^2(m(p+n))\big)$, leading to a worst-case value of $mp$ $=$ ${\cl O}\big((k + m)$ $\log^2(m(p+n))\big)$ when $\bs B$ achieves maximum coherence $\upmu_{\max} = \sqrt{\tfrac{m}{h}}$. 
		% We presented and solved a non-convex formulation of blind calibration in the specific, yet important case of linear random sensing models affected by unknown gains. In absence of {\em a priori} structure on the signal and gains, we have devised and analysed a descent algorithm based on a simple projected gradient descent. In this case, our main results have shown that provable convergence of the algorithm can be achieved at a linear rate (up to $\log$ factors) $mp \gtrsim \delta^{-2} (m+n)\, \log^2(m(p+n))\log(\frac{1}{\delta})$ with respect to the number of unknowns in this bilinear inverse problem. Similarly, using a subspace prior on both the signal and gains, a straightforward extension of the previous descent algorithm into Alg.~\ref{alg2} proved that the sample complexity ensuring convergence in this setting is $mp \gtrsim \delta^{-2} (k + \upmu^2_{\max} h)\, \log^2(m(p+n))\log(\frac{1}{\delta})$, leading to a worst-case value of $mp \gtrsim \delta^{-2} (k + m)\, \log^2(m(p+n))\log(\frac{1}{\delta})$ when $\bs B$ achieves maximum coherence. 
		
		We envision that our results could be extended to the case of complex gains (\ie ${\bs g} \in \bb C^m$) and complex sensing operators $\bs A_l$ by means of Wirtinger calculus (up to redefining our bounded $\ell_\infty$-norm assumption made throughout this paper). Another important extension is the adaptation of Alg.~\ref{alg2}~to enforce the sparsity of $\x$ (or ${\bs g}$, although this may not be verified in practice). We explored numerically this possibility in \cite{CambareriJacques2017}. By using sparsity we expect a reduction of the sample complexity that is similar to the subspace case, but without using such strong assumptions as prior information on a fixed, known tight frame~$\bs Z$ in whose span~${\bs x}$ must lie. 
		
		Finally, our general technique is in line with the surge of new results on non-convex problems with linear, bilinear or quadratic random models. In fact,
                the core principle underlying this work is that, as
                several random instances of the same non-convex
                problem are taken, we approach its behaviour in
                expectation which, as previously discussed, is
                significantly more benign than the non-asymptotic
                case. This is an extremely general approach that could
                be applied to very different models than the ones we
                discussed.
%%% Local Variables:
%%% mode: latex
%%% TeX-master: "iai_blindcalibration_2016"
%%% End:
	
% The next command is useful for hyperref to help it to reference anything in the appendix
\renewcommand{\theHsection}{A\arabic{section}}
\appendix
	\section*{Appendices}
	We now provide arguments to support the main results in this paper. Unless otherwise noted, the proofs are first given for the known subspace case discussed in Sec.~\ref{sec:blind-calibration-with-subspace-priors}. In fact, the proofs for our statements in absence of priors (\ie as discussed throughout Sec.~\ref{sec:nonconv}) are particular cases of the corresponding statements in the known subspace case (Sec.~\ref{sec:blind-calibration-with-subspace-priors}). 
	
	\section{Technical Tools}
	\label{sec:appA}
	
	We begin by introducing some technical results used in the following proofs. Hereafter, we recall that $\bs a \in \bb R^n$ denotes a sub-Gaussian and isotropic random vector comprised of \rvs $a_i \sim_{\rm iid} X$, $i \in [n]$, with $X$ a centred and unit-variance sub-Gaussian \rv with sub-Gaussian norm $\|X\|_{\psi_2} = \alpha>0$. 
	
	\begin{proposition}[Concentration Bound on $\ell_2$-Norm of sub-Gaussian \rv's]
		\label{prop:truncation1}
		Let $\a_i \sim_{\iid} \bs a$ be a set of $q\geq 2$
                \iid sub-Gaussian random vectors. There exists a
                value $\vartheta > 1$ only depending on $\alpha$ such
                that, provided $n \gtrsim t \log q$ for some $t \geq 1$,
                the event% \mtodo{Not sure $\cl E''$ is needed anymore}
		\begin{gather}
		\cl E_{\max} \coloneqq \big\{\max_{i \in [q]} \|\a_i\|^2
                  \leq \vartheta n \big\}, \label{eq:truncation1} %\\
		% \cl E'' \coloneqq \big\{\max_{i \in [q]}
                %   \big| \hat{\bs x}^\top \a_i\big|^2 \leq \vartheta\,
                %   t\log q \big\}, \label{eq:truncation2}
		\end{gather}
                occurs with probability exceeding $1-2 q^{-t}$.
	\end{proposition}
	
	\begin{proof}[Proof of Prop.~\ref{prop:truncation1}]
		Since all $\a_i \sim \bs a$ are comprised of \iid sub-Gaussian
                \rvs with sub-Gaussian norm $\alpha$, then
                $\|\a_i\|^2$ 
                % and $|\hat{\bs x}^\top \a_i|$ 
                can be bounded 
                by standard concentration inequalities~\cite[Cor.~5.17]{Vershynin2012a}. In detail, since
                $\|\a_i\|^2$, $i \in [q]$ are sub-exponential
                variables with $\bb E \|\a_i\|^2 = n$, there exists
                some $c > 0$ such that the probability of the
                complementary event is bounded as 
				\[
                  \Pro[\cl E_{\max}^{\rm c}] \leq q\,\Pro\big[\|\a\|^2 >
                  (1+\vartheta') n\big] \leq q\,\Pro\big[\big|\|\a\|^2 -
                  n\big| > \vartheta' n\big] < 2 \exp\big(\log q - c
                  \min\big\{\tfrac{\vartheta'^2}{\alpha^4}, \tfrac{\vartheta'}{\alpha^2}\big\} n\big),
                \]
                with $\vartheta' \coloneqq \vartheta - 1 > 0$. 
		% \[
                %   \Pro[\cl E'^{\rm c}] \leq q\,\Pro\big[\|\a\|^2 >
                %   \vartheta n\big] \leq q\,\Pro\big[\big|\|\a\|^2 -
                %   n\big| > (\vartheta -1)n\big] < 2 \exp(\log q - c
                %   (\vartheta -1) n).
		% \]
                Therefore, provided $
                  \vartheta'\min(\vartheta',
                  \alpha^2) > \tfrac{2}{c}\alpha^4$ (which leads to a
lower bound on $\vartheta$ only depending on $\alpha$) we find $\Pro[\cl
E_{\max}^{\rm c}] < 2q^{-t}$ when $n
                \geq t \log q$ for $t \geq 1$.
% Similarly, since $\hat{\bs x}^\top \bs a$ is sub-Gaussian with
% $\|\hat{\bs x}^\top \bs a\|_{\psi_2} \lesssim \alpha$ and
%                 $\bb E\, \hat{\bs x}^\top \bs a = 0$, 
%                 \begin{align*}
%                 \ts\Pro[\cl E''^{\rm c}]&\leq q\,\Pro\big[|\hat{\bs x}^\top
%                   \bs a|^2 > \vartheta \,t \log q\big] = q\,\Pro\big[|\hat{\bs x}^\top
%                   \bs a| > \sqrt{\vartheta \,t \log q}\,\big]\\
% &\ts < e \exp(\log q -
%                   c \vartheta \, t\log q).
%                 \end{align*}
% 	Therefore, if $\vartheta > \tfrac{2}{c}$ and $t\geq 1$, we find
% 	$\Pro[\cl E''^{\rm c}] < q^{-t}$. Therefore, for some $\vartheta$ bigger than
% 	the two lower bounds (implicitly) defined above, we have $\Pro[\cl E' \cap \cl
% 	E'']\geq 1- (2+e) q^{-t} \geq 1 - 5 q^{-t}$ by union bound. 
	\end{proof}

	We now introduce a proposition that provides a concentration inequality for a weighted sum of matrices (these being functions of \iid sub-Gaussian random vectors) that will be frequently used in the following proofs. This result is related to finding a bound for the spectral norm of the residual between a covariance matrix and its sample estimate~\cite{Vershynin2012b} in the special case where a weighting affects the computation, and where both the
	corresponding weights and the vectors to which this residual applies are assumed to lie in known subspaces. 
	
	%%%%% New proposition ($\ell_\infty$-normalised weights) %%%%%%%%%%%%
	\begin{proposition}[Weighted Covariance Concentration in Subspaces]
		\label{prop:wg-cov-subsp} Consider a set of $q$ random vectors $\bs
		a_i \sim_{\iid} \bs a$, and two subspaces $\cl Z \subset \R^n$ and $\cl S
		\subset \R^q$ of dimensions $k \leq n$ and $s \leq q$, 
		respectively. Let $\bs Z \in \R^{n\times k}$ be an orthonormal basis of $\cl Z$, \ie $\bs Z^\top \bs Z = \I_k$.  
		Given $\delta \in (0,1)$, $t \geq 1$, provided $n \gtrsim t \log(q)$ and 
		\[
			\ts q \gtrsim \delta^{-2} (k + s)\log\big(\frac{n}{\delta}\big),
		\]
		with probability exceeding 
		\[
			1 - C \exp(- c \delta^2 q) - q^{-t}
		\]
		for some $C,c >0$ depending only on $\alpha$, we have for all $\bs w \in \cl S$
		\begin{equation}
		\label{eq:final-wg-concent}
			\big\|\tinv{q}\,\ts\sum_{i=1}^q w_i \bs Z^\top (\bs a_i \bs a^\top_i - \bs I_n) \bs Z\big\| \leq \delta\,\|{\bs w}\|_\infty.    
		\end{equation}
	\end{proposition}
	
	\begin{proof}[Proof of Prop.~\ref{prop:wg-cov-subsp}]
		
		Note that by the assumptions on the \rvs $a_{ij}$, $i\in [q]$, $j \in [n]$ we have
		that the random vectors $\a_i$, $i\in [q]$ are \iid centred ($\Ex
		\a_i = \vzer_n$) and isotropic ($\Ex \a_i \a_i^\top = \I_n$). By
		homogeneity of~\eqref{eq:final-wg-concent} it suffices to show the
		probability that the event
		\begin{align*}
		\cl E \coloneqq \big\{\big\|\tinv{q}\,\ts\sum_{i=1}^q w_i \bs Z^\top (\bs a_i \bs a^\top_i - \bs I_n) \bs Z\big\|\leq \delta\big\} 
		\end{align*}
		holds for all ${\bs w} \in \cl S \cap \mathbb{S}^{q-1}_\infty$. 
		% Maybe explain the meaning of homogeneity
		\medskip
		
		\noindent {\em Step 1: Concentration} \ Let us first observe that, since the rank-one matrices $\a_i \a^\top_i, i \in [q]$ are symmetric and $\bs Z^\top \bs Z = \bs I_k$,
		\begin{align*}
		\big\|\tinv{q}\,\ts\sum_{i=1}^q w_i \bs Z^\top (\a_i \a^\top_i - \bs I_n)
		\bs Z \big\|
		& = \sup_{\bs z \in {\bb S}_2^{k-1}} \tinv{q}\,\big|\ts\sum_{i=1}^q w_i                           
		(\bs z^\top \bs Z^\top \a_i \a^\top_i \bs Z \bs z - \|\bs Z \bs z\|^2)\big| \nonumber\\
		& = \sup_{\u \in \bs Z {\bb S}_2^{k-1}} \tinv{q}\,\big|\ts\sum_{i=1}^q w_i                        
		(\u^\top \a_i \a^\top_i \u - \|\u\|^2)\big| \nonumber\\
		& = \sup_{\u \in \bs Z\bb S_2^{n-1}} \tinv{q}\,\big|\ts\sum_{i=1}^q (V_i(\u) - w_i\|\u\|^2)\big|, 
		\end{align*}
		with $V_i(\u) \coloneqq w_i (\bs a^\top_i \u)^2, i \in [q]$ and $\bb E V_i(\u) = w_i \|\u\|^2$. 
		
		Let us fix $\u\in \bs Z\bb S^{k-1}_2$ and ${\bs w} \in
                \cl S^* \coloneqq \cl S \cap \bb S_\infty^{q-1}$ (which implies  $|w_i|\leq 1, i \in [q]$). Firstly, we observe that
                $V_i(\u)$ is a sub-exponential \rv since each $\a_i$ is formed by sub-Gaussian \rvs, so we can bound the sub-exponential norm 
		\begin{align*}
		\|V_i(\u)\|_{\psi_1} = |w_i|\|(\a_i^\top \u)^2\|_{\psi_1}       
		\leq 2 |w_i|\|\bs \a_i^\top \u \|^2_{\psi_2} \lesssim \alpha^2, 
		\end{align*}
		where we used the fact that $\|V\|^2_{\psi_2} \leq \|V^2\|_{\psi_1} \leq
		2\|V\|^2_{\psi_2}$ for any \rv $V$~\cite[Lemma~5.14]{Vershynin2012a} and that, given any $\bs{\sf v} \in \bb S_2^{n-1}$, 
		$\|\bs a^\top_i \bs {\sf v}\|^2_{\psi_2} \lesssim \|\bs {\sf v}\|^2\|X\|^2_{\psi_2} \lesssim \alpha^2$ by the rotational invariance of $\a_i$~\cite[Lemma~5.9]{Vershynin2012a}. 	
		Thus, the \rv $V\coloneqq\ts\sum_{i=1}^q V_i(\u)$ is in turn a sum of sub-exponential {variables} that concentrates
		around $\Ex V =\ts\sum_{i=1}^q w_i \|\u\|^2$, \ie for some
		$c>0$ and $\delta \in (0,1)$, applying~\cite[Cor.~5.17]{Vershynin2012a} yields
		\begin{align*}
		\ts \Pro[|V - \Ex V| > \delta q] < 2 \exp\big(-c q                                                                                                                            
		\min(\tfrac{\delta^2}{\alpha^4}, \tfrac{\delta}{\alpha^2}) \big) < 2 \exp\big(-c q \tfrac{\delta^2 \alpha^{-4}}{1+\alpha^{-2}} \big) < 2 \exp\big(- c_\alpha \delta^2 q\big), 
		\end{align*}
		for some $c_\alpha >0$ depending only on $\alpha$. Thus, with
		probability exceeding $1-2\exp(-c_\alpha \delta^2 q)$, the event
		\begin{equation}
		\label{eq:wg-concent1}
		\tinv{q} \big|\ts\sum_{i=1}^q (V_i(\u) - w_i \|\u\|^2)\big|
		\leq \delta 
		\end{equation}
		holds for some fixed $\u \in \bb S^{n-1}_2, {\bs w} \in \cl S^*$. 
		\medskip
		
		\noindent {\em Step 2: Covering} \ To obtain a {\em uniform}
		result we resort to a covering argument. Given two radii
		$\epsilon, \epsilon' > 0$ to be fixed later, let us take an
		$\epsilon$-net $\cl N_\epsilon$ of $\bs Z\,\bb B^{k}_2$ (\ie with respect to $\u$) and
		another $\epsilon'$-net $\cl W_{\epsilon'}$ of $\cl S \cap \bb
		B^q_\infty$. Noting that $\bs
		Z\,\bb S^{k-1}_2 \subset \bs Z\,\bb B^{k}_2$ and $\cl S^* \subset \cl S \cap \bb
		B^q_\infty$, $\cl N_\epsilon$ and $\cl W_{\epsilon'}$ are also the
		corresponding nets of $\bs
		Z\,\bb S^{k-1}_2$ and $\cl S^*$, respectively. Since $\bs Z\, \bb B^k_2$ is
		isomorphic to $\bb B^k_2$, a standard geometric argument
		provides that $|\cl N_\epsilon| \leq (1 + \tfrac{2}{\epsilon})^{k}
		\leq (\tfrac{3}{\epsilon})^k$~\cite{Pisier1999}. Moreover, given
		an orthonormal basis $\bs S \in \bb R^{q \times s}$ of $\cl S
		\subset \bb R^q$, since
		$\|\cdot\|_{\bs S,\infty} \coloneqq \|\bs S \cdot\|_{\infty}$ is a norm
		such that $\cl S \cap \bb B^q_\infty$ matches the unit ball
		$\bb B^s_{\bs S, \infty} \coloneqq
		\{\bs s \in \bb R^s: \|\cdot\|_{\bs S,\infty} \leq 1\}$, % in this norm,
		we also have $|\cl W_{\epsilon'}| \leq (1 + \tfrac{2}{\epsilon'})^{s}
		\leq (\tfrac{3}{\epsilon'})^s$~\cite[Prop.~4.10]{Pisier1999}.
		
		Since $\cl N_\epsilon \times \cl W_{\epsilon'} \subset \bs Z\, \bb
		B^k_2 \times (\cl S \cap \bb
		B^q_\infty)$ has no more than $|\cl
		N_\epsilon||\cl W_{\epsilon'}|$ points, by union bound on this set we find that
		\eqref{eq:wg-concent1} holds with probability exceeding 
		\begin{equation}
		\label{eq:wg-prob-mid}
		\ts 1 - 2\exp\big(k \log\big(\frac{3}{\epsilon}\big) + s \log\big(\frac{2}{\epsilon'}\big)
		-c_\alpha \delta^2 q\big)
		\end{equation}
		for all $\u \in \cl N_\epsilon, {\bs w} \in \cl W_{\epsilon'}$. 
		
		\medskip
		\noindent {\em Step 3: Continuity} \ To obtain a final, uniform
		result we require a continuity argument for points that lie in the
		vicinity of those in the net. For all $\u \in \bs Z\, \bb S_2^{k-1}$ and
		${\bs w} \in \cl S^*$ we can take the closest points in their
		respective nets $\u' \in \cl N_\epsilon,{\bs w}' \in \cl
		W_{\epsilon'}$. By using H\"older's inequality (\ie $|\bs v^\top
		\bs  v'| \leq \|\bs v\|_1 \|\bs v'\|_\infty$) and Cauchy-Schwarz's
		inequality, we see that 
		\begin{align}
		\tinv{q}\big|\ts\sum_{i=1}^q (V_i(\u) - w_i\|\u\|^2)\big| & \leq \tinv{q}\big|\ts\sum_{i=1}^q (w_i - w'_i) \u^\top \a_i \a^\top_i \u\big| \nonumber                                                                   \\ 
		& \quad + \tinv{q}\big|\ts\sum_{i=1}^q (w'_i \u^\top \a_i \a^\top_i \u - w'_i\|\u\|^2)\big| + \tinv{q}\big|\ts\sum_{i=1}^q (w'_i - w_i) \|\u\|^2\big|\nonumber \\
		& \leq \epsilon' (A + 1) + \tinv{q}\big|\ts\sum_{i=1}^q (w'_i \u^\top \a_i \a^\top_i \u - w'_i \|\u\|^2)\big|,                                              
		\label{eq:wg-last}
		\end{align}
		where we substituted ${\bs w} = {\bs w}'+({\bs w}-{\bs w}')$ in the first line, noting
		that $\a_i\a^\top_i \succeq 0, i \in [q]$ and defining $A
		\coloneqq \max_{i \in [q]} \|\a_i \a^\top_i\| = \max_{i \in [q]} \|\a_i\|^2$. We can then
		bound the last term in~\eqref{eq:wg-last} by substituting $\u =
		\u'+ (\u-\u')$ to obtain
		\begin{align}
		\tinv{q}\big|\ts\sum_{i=1}^q (w'_i \u^\top \a_i \a^\top_i \u - w'_i \|\u\|^2)\big| & \leq \tinv{q}\big|\ts\sum_{i=1}^q w'_i ( (\a^\top_i\u')^2 - (\bs a^\top_i\u)^2)\big| \nonumber                     \\ 
		&                                                                                                                 
		\quad +\tinv{q}\big|\ts\sum_{i=1}^q (w'_i (\u')^\top \a_i \a^\top_i \u'-w'_i\|\u'\|^2)\big|\nonumber\\
		& \quad+\tinv{q}\big|\ts\sum_{i=1}^q w'_i (\|\u'\|^2- \|\u\|^2)\big|\nonumber                                      \\
		& \leq {2\epsilon}(A+1) +\tinv{q}\big|\ts\sum_{i=1}^q w'_i \big((\u')^\top \a_i \a^\top_i \u' - \|\u'\|^2\big)\big|, 
		\label{eq:wg-secondlast}
		\end{align}
		which follows again by H\"older's inequality after noting that $\big| \|\u'\|^2- \|\bs
		u\|^2\big| = \big|(\u' - \u)^\top(\u' + \u)\big| \leq 2\epsilon$  and
		\begin{align}
		\big|(\a^\top_i\u')^2 - (\a^\top_i\u)^2\big| & \leq \big|(\u'-\u)^\top \a_i \a_i^\top (\u'+\u)\big| \leq \|\u'-\u\|\|\u'+\u\|\,\|\a_i \a^\top_i\|\leq 2 A \epsilon, \nonumber 
		\end{align}
		since both $\u,\u' \in \bb S^{n-1}_2$. 
		In conclusion, by bounding the last term in
		\eqref{eq:wg-secondlast} by $\delta$ using~\eqref{eq:wg-concent1} we have, 
		with the same probability~\eqref{eq:wg-prob-mid},
		\begin{align}
		\tinv{q}\big|\ts\sum_{i=1}^q (V_i(\u) - w_i \|\u\|^2)\big|
		& \leq (\epsilon'+2\epsilon) (A+1) + \delta. 
		\label{eq:wg-probref}
		\end{align}
		
Let us now bound the value of $A$. By Prop.~\ref{prop:truncation1}, we
know that there exists a $\vartheta >0$ depending on the distribution
of $X$ such that $\bb P[A \leq \vartheta n] \geq 1 - 2 q^{-t}$
provided $n \gtrsim t \log (mp)$ for some $t\geq 1$. Hence, by union bound with~\eqref{eq:wg-prob-mid}, we have that
this last event and the event
\eqref{eq:wg-probref} hold with probability exceeding $1 - 2 q^{-t}
- 
2\exp\big(k \log\big(\frac{3}{\epsilon}\big) + s
\log\big(\frac{2}{\epsilon'}\big) -c_\alpha \delta^2 q\big)$, in which
case 
\begin{align*}
  \tinv{q}\big|\ts\sum_{i=1}^q (V_i(\u) - w_i \|\u\|^2)\big|
  & \leq (\epsilon'+2\epsilon) (\vartheta n+1) + \delta. 
\end{align*}

% Let us now discuss the value of $A$, whose role is that of bounding the norm of sub-Gaussian random vectors $\a_i$ by ``truncating'' the probability it exceeds an upper bound. In fact, we have already discussed this as event $\cl E'$ in Prop.~\ref{prop:truncation1}; hence, 
% 		%    For any $i\in[q]$ and some $T\geq 1$, applying~\cite[Prop.
% 		%    5.17]{Vershynin2012a} provides a concentration inequality on the norm of the sub-Gaussian random vectors $\a_i$, \ie 
% 		%    \begin{align*}
% 		%      \Pro \big[\|\a_i\|^2 > n(1+T)\big] < \Pro \big[| \|\a_i\|^2- n | > n T\big] < 2 e^{-c_\alpha T^2 n}
% 		%    \end{align*}
% 		%    since $\|\a_i\|^2, i \in [q]$ are sub-exponential variables. 
% 		%    Thus, 
% 		by setting $\vartheta \coloneqq 2$, $q \coloneqq mp$, provided $n \gtrsim t \log(q)$ for $t \geq 1$ we have 
% 		\[
% 		\Pro[\cl E'^{\rm c}] = \Pro[A > 2n] < {5 {q}^{-t}}
% 		\] 
% 		when Prop.~\ref{prop:truncation1} holds.
		
Thus, setting $\epsilon' (\vartheta n+1) = \epsilon (\vartheta n+1) =
{\delta}$ and rescaling $\delta$ to $\tfrac{\delta}{4}$, we obtain % obtain a
% bound on
% \begin{align}
%   \tinv{q}\big|\ts\sum_{i=1}^q (V_i(\u) - w_i \|\u\|^2)\big|
%   & \leq \delta
% \end{align}
% \begin{align*}
%   \Pro[\cl E | A \leq 2 n] \nonumber & =\Pro\big[ \tinv{q}\big|\ts\sum_{i=1}^q (V_i(\u) - w_i \|\u\|^2)\big| \leq \delta \big\vert A \leq 2 n\big] \nonumber   \\ 
% 		& \geq \ts 1 - 2\exp\big(k \log(\tfrac{c n}{\delta}) + s \log(\tfrac{c' n}{\delta})-c_\alpha \delta^2 q\big), \label{eq:3} 
% 		\end{align*}
% 		for some universal constants $c,c'>0$ depending on
%                 $\vartheta$ and $\alpha$, and thus on the distribution of $X$. 
% 		% 
% 		We can finally obtain a probability bound on $\cl E$ by noting that
% 		$\Pro[\cl E] \geq \Pro[{\cl E} | A \leq 2n] - \Pro[A > 2n]$, \ie by collecting all the above results as
		\begin{align*}
		\Pro\big[ \tinv{q}|\ts\sum_{i=1}^q (V_i(\u) - w_i \|\u\|^2)| \leq \delta \big] \geq \ts 1 - 2\exp\big(k \log(\tfrac{c n}{\delta}) + s \log(\tfrac{c' n}{\delta})	-c_\alpha \delta^2 q\big) - 2q^{-t} 
		\end{align*}
		for $c,c',c_\alpha>0$ and $t \geq 1$. Therefore,
                by summarising the previous requirements, this last
                probability exceeds $1 - C[\exp(-c\delta^2 q) +
                q^{-t}]$ provided
		\begin{align*}
		\ts q \gtrsim \delta^{-2}\big(k \log(\frac{n}{\delta}) + s \log(\frac{n}{\delta})\big), \ n \gtrsim t \log(q) 
		\end{align*}
		which concludes the proof.
	\end{proof}
	
	We now adapt Prop.~\ref{prop:wg-cov-subsp} to the sensing model in Def.~\ref{def:modelsub} by the following corollary.
	
	\begin{corollary}[Application of Prop.~\ref{prop:wg-cov-subsp} to Blind Calibration with Subspace Priors]
		\label{coro:appl-wcovsub}
		Consider two subspaces $\cl B \subset \bb R^m$ and $\cl Z \subset \bb
		R^n$ with $\dim \cl B = h \leq m$ and $\dim \cl Z = k \leq n$, with 
		$\bs Z \in \bb R^{n \times k}$ an orthonormal basis of $\cl
		Z$, \ie $\bs Z^\top \bs Z = \I_k$. Let us define $mp$ random
		vectors $\bs a_{i,l} \sim_{\iid} \bs a$, $i\in [m]$, $l \in p$. Given $\delta \in (0,1)$, $t \geq 1$, provided $n \gtrsim t \log(mp)$ and 
		\[
		\ts mp \gtrsim \delta^{-2}(k+h)\log(\frac{n}{\delta}),
		\]
		with probability exceeding
		\begin{equation}
		\label{eq:pcoro2}
		1 - C \exp(- c \delta^2 m p) - (mp)^{-t}
		\end{equation} 
		for some $C,c >0$ depending only on $\alpha$, we have $\forall \bs \theta \in \cl B$,
		\begin{align}
		\|\tinv{mp}\,\ts\sum_{i=1}^m\ts\sum_{l=1}^p \theta_i \bs Z^\top
		(\ail \ail^\top - \I_n) \bs Z \|& \leq 
		\delta\,\|\bs \theta\|_\infty.
		\label{eq:coro-ineq-wcov}
		%      \\
		%      \|\tinv{mp}\,\ts\sum_{i=1}^m\ts\sum_{l=1}^p \bs Z^\top
		%      (\ail \ail^\top - \I_n) \bs Z \|&\leq
		%      \delta.
		%      \label{eq:coro-ineq-wcov-unwg}
		\end{align}
	\end{corollary}
	The proof is straightforward and given below.
	\begin{proof}[Proof of Cor.~\ref{coro:appl-wcovsub}]
		To prove~\eqref{eq:coro-ineq-wcov} only under the setup of this corollary, we  just have to set $q = mp$, $s=h$ and 
		${\bs w} = \vone_p \otimes \bs\theta \in \vone_p \otimes \cl B$ in
		Prop.~\ref{prop:wg-cov-subsp}, with $\cl S = \vone_p \otimes \cl B$ a subspace of $\R^{mp}$ of dimension~$h$. 
		%
		%    The proof of~\eqref{eq:coro-ineq-wcov-unwg} is an easy modification of the proof of 
		%	Prop.~\ref{prop:wg-cov-subsp} that we omit here. It corresponds to the case where the weight domain is
		%	summarised to the singleton $\{\bs 1_q = \bs 1_p \times
		%	\bs 1_m\}$, with $q = mp$. In such a context,
		%	the covering of the weight domain is summarised to this single point,
		%	with no impact on the corresponding union bound argument made on the
		%	covering of $\cl Z \cap \bb B^n$, and
		%	there is also no need of a continuity argument, \ie it is as if $s=0$
		%	in Prop.~\ref{prop:wg-cov-subsp}. Therefore,~\eqref{eq:coro-ineq-wcov-unwg}
		%	holds with probability exceeding $1 - C \exp(- c \delta^2 m p) -
		%	(mp)^{-t}$ provided $mp \gtrsim \delta^{-2} k \log \big(\tfrac{n}{\delta}\big)$. 
		%
		%	Finally, by union bound and up to a rescaling of the values $C,c >0$,
		%	we show easily that 
		%	both~\eqref{eq:coro-ineq-wcov} and
		%	\eqref{eq:coro-ineq-wcov-unwg} hold jointly with probability exceeding $1 - C e^{-c \delta^2 m p} -
		%	(mp)^{-t}$ provided $mp \gtrsim \delta^{-2} (k + h) \log (\tfrac{n}{\delta})$.
		%
		Hence, by straightforward substitution this statement holds with probability exceeding $1 - C e^{-c \delta^2 m p} - (mp)^{-t}$ provided $mp \gtrsim \delta^{-2} (k + h) \log (\tfrac{n}{\delta})$.
	\end{proof}
	
	A remark allows us to cover the sensing model of Def.~\ref{def:model} (as presented in Sec.~\ref{sec:nonconv}).
	\begin{remark}
		Cor.~\ref{coro:appl-wcovsub} is straightforwardly extended in absence of known subspace priors, \ie by setting % $\cl B = \bb R^m$
		% $\cl B = \onep$
		$\cl B \coloneqq \bb R^m$, $\cl Z \coloneqq
		\bb R^n$, with which $h=m$ and $k=n$.  
	\end{remark}
	\medskip
	
We also recall a fundamental result on the concentration of sums of
\iid random matrices, \ie matrix Bernstein's inequality. This one is
developed in \cite{Tropp2015} for random matrices with bounded
spectral norms; we here adopt its variant which assumes that these norms
have sub-exponential tail bounds. This result is only a slight adaptation 
of matrix Bernstein's inequality used in \cite{AhmedRechtRomberg2014}.

          	\begin{proposition}[Matrix Bernstein's Inequality,
                  adapted from Prop.~3 in \cite{AhmedRechtRomberg2014}]
		\label{prop:mbi}
		Let $\bs J_i \in \R^{n \times m}, i \in [q]$ be a
                sequence of $q$ \iid random matrices. Assume that the \rvs
                $N_i \coloneqq \|\bs J_i - \bb E \bs J_i\|$ are
                sub-exponential with norm $\|N_i\|_{\psi_1} \leq T$.
Define the {\em matrix variance}
		\begin{equation}
			\label{eq:matrix-var-T}
			v \coloneqq \max \big\{\big\Vert\ts\sum_i \bb E (\bs J_i  - \bb E \bs J_i) (\bs J_i -  \bb E \bs J_i)^\top \big\Vert, \big\Vert\ts\sum_i \bb E (\bs J_i  - \bb E \bs J_i)^\top (\bs J_i -  \bb E \bs J_i) \big\Vert \big\}.
		\end{equation}
		Then for $\delta \geq 0$ and some constant $c > 0$, 
		\begin{equation}
                  \label{eq:mbi-prob}
\ts			\Pro\big[\big\Vert \ts\sum_i \bs J_i  - \bb E \bs J_i \big\Vert \geq \delta \big] \leq (n+m) \exp\left(-\frac{c\,\delta^2}{v\, +\, T\log(\frac{T^2q}{v})\delta }\right).
		\end{equation}
	\end{proposition}
        \begin{proof}[Proof of Prop.~\ref{prop:mbi}]
          Under the conditions of this proposition we observe that,
          from the sub-exponentiality of $N_i$, $\bb E\exp\big(c \tfrac{N_i}{T}\big)
          \leq e$ for some universal constant $c>0$
          \cite[Eq.~(5.16)]{Vershynin2012a}. Therefore, by Jensen's inequality we find
          $$\bb E \exp\big(c N_i \tfrac{\log(2)}{T}\big) \leq \big(\bb E\exp\big(c \tfrac{N_i}{T}\big)\big)^{\log
            2} \leq e^{\log 2} = 2,$$ and define
          \begin{equation}
          \label{eq:Tprime-A3}
          T'\coloneqq \inf_{u\geq 0}\left\{\bb E \exp(\|\bs J_i - \bb E \bs
          J_i\|)^{\tinv{u}} \leq 2\right\} \leq \tfrac{T}{c \log(2)} \lesssim T.
          \end{equation}
          We now proceed to use our estimation~\eqref{eq:Tprime-A3} of $T'$ in \cite[Prop.~3]{AhmedRechtRomberg2014}; taking 
          $$
	          \ts \delta' \coloneqq C \max\{\sqrt{v}\sqrt{t + \log(n+m)},\ T'\log(\frac{T'\sqrt{q}}{\sqrt{v}})(t+\log(n+m))\}
          $$
          for some constant $C >0$, we have by the aforementioned proposition
          $$
          \ts \bb P\big[\big\Vert \ts\sum_i \bs J_i  - \bb E \bs J_i
          \big\Vert > \delta' \big] \leq e^{-t}. 
          $$
          However, since for some $s\geq 0$, $\max(\sqrt{\alpha s}, \beta s) \leq \beta s
          + \sqrt{\beta^2s^2 + 4\alpha s}$ and since we can take $2\delta'' = \beta s
          + \sqrt{\beta^2s^2 + 4\alpha s}$ we have $s = \tfrac{\delta''^2}{\alpha +
          \beta \delta''}$, and we obtain $\delta' \leq 2C\delta''$ by setting $s = t + \log(n+m)$,
          $\alpha = v$ and $\beta = T' \log(T'\sqrt{\tfrac{q}{v}})$. Thus, replacing this bound on $\delta'$ yields 
          $$
          \ts \bb P\big[\big\Vert \ts\sum_i \bs J_i  - \bb E \bs J_i 
          \big\Vert > 2 C\delta'' \big] \leq  \bb P\big[\big\Vert \ts\sum_i \bs J_i  - \bb E \bs J_i
          \big\Vert > \delta' \big] \leq e^{-t} =
          (n+m) \exp(- \frac{\delta''^2}{v +T' \log(T'\sqrt \frac{q}{v}) \delta''}). 
          $$
          Finally, setting $\delta = 2C\delta''$ and recalling that
          $T' \lesssim T$ provides the result in~\eqref{eq:mbi-prob}.
       \end{proof}

          Matrix Bernstein's inequality allows us to prove the following
          concentration result, that is used several times in these appendices
          and which characterises the concentration of a special form
          of ``matrix-weighted'' covariance.
          \begin{proposition}[Matrix-Weighted Covariance Concentration]
            \label{prop:m-w-covconc}
            Given the setup defined in Def.~\ref{def:modelsub}, $\delta\in(0,1)$ and $t \geq 1$, provided
            \begin{equation}
              \label{eq:nice-prop-first-cond}
             mp \gtrsim t \delta^{-2} \upmu_{\max}^2 h \log(mp)\log(mn),
           \end{equation}
           we have  
            \begin{equation}
              \label{eq:nice-prop-first-res}
            \big\| \tinv{p} \ts\sum_{i,l}  ({\bs B}^\top  {\bs c}_i  {\bs c}^\top_i {\bs
              B})\ (\bs Z \hat{\bs z})^\top\, (\ail \ail^\top - \I_n)\, {\bs Z}
            \hat{\bs z}  \big\|\ \leq\ \delta, 
            \end{equation}
            with probability exceeding $1-(mp)^{-t}$. 
            Moreover, we have 
            \begin{equation}
              \label{eq:nice-prop-snd-res}
            \big\| \tinv{p} \ts\sum_{i,l}  ({\bs c}_i  {\bs c}^\top_i)\ \hat{\bs x}^\top\, (\ail \ail^\top - \I_n)\, 
            \hat{\bs x}  \big\|\ \leq\ \delta, 
          \end{equation}
          with probability exceeding $1-(mp)^{-t}$ provided $p \gtrsim t \delta^{-2} \log(mp)\log(mn)$.
          \end{proposition}
          \begin{proof}[Proof of Prop.~{\ref{prop:m-w-covconc}}]
   
            Let us define a set of symmetric matrices 
            $$
            \{\bs V_{i,l} \coloneqq (\bs B^\top \bs c_i \bs
            c_i^\top \bs B)\, X^2_{i,l} \in \bb R^{h\times h},\ i\in
            [m], l\in [p]\}, 
            $$
            with the \rvs $X_{i,l} \coloneqq \ail^\top
            \hat{\bs x}$; in fact, all $X_{i,l} \sim X \coloneqq \bs a^\top \hat{\bs x}$, where $\bb E X^2 = 1$. We are now going to prove this proposition by invoking matrix Bernstein's inequality (in the form of Prop.~\ref{prop:mbi}) to show that
            $\|\frac{1}{p}(\sum_{i,l} \bs V_{i,l} - \bb E \bs V_{i,l})\| \leq
            \delta$ with high probability. This requires a bound $T$ over the norms \new{$\|\|\bs V_{i,l}\|\|_{\psi_1}$} and another  bound on the matrix variance $v$. 

            To begin with, note that $\bb E \bs V_{i,l} = (\bs B^\top \bs c_i \bs
            c_i^\top \bs B)$ and $\|\bs V_{i,l}\| \leq X_{i,l}^2\,\|\bs B^\top \bs c_i\|^2 \leq X_{i,l}^2\,\upmu_{\max}^2 \frac{h}{m}$, 
            where the coherence $\upmu_{\max}$ is defined in~\eqref{eq:mumax}. Therefore, 
            \begin{align*}
            \ts\big\|\|\bs V_{i,l} - \bb E \bs
            V_{i,l}\|\big\|_{\psi_1}&\ts \leq\ \|\bb E \bs
            V_{i,l}\| + \big\| \|\bs V_{i,l} \|\big\|_{\psi_1}\\
&\ts =\ 
            (1+\|X_{i,l}^2\|_{\psi_1})\,\upmu_{\max}^2 \frac{h}{m}
            \lesssim\ \|X_{i,l}\|^2_{\psi_2}\,\upmu_{\max}^2 \frac{h}{m}
            \lesssim\ \alpha^2\,\upmu_{\max}^2 \frac{h}{m},
            \end{align*}
            where we used the fact that $\|X\|^2_{\psi_2} \leq \|X^2\|_{\psi_1} \leq
            2\|X\|^2_{\psi_2}$ for any \rv $X$~\cite[Lemma~5.14]{Vershynin2012a}, and the
			rotational invariance of the random vectors $\ail \sim_{\iid} \bs a$ 
			that involves $\|\bs a^\top\hat{\bs x}\|^2_{\psi_2} \lesssim \alpha^2$~\cite[Lemma~5.9]{Vershynin2012a}.  
			Thus, there exists a $T>0$ such $\big\|\|\bs V_{i,l} - \bb E \bs V_{i,l}\|\big\|_{\psi_1}\leq T$ with $T \lesssim \upmu_{\max}^2 \frac{h}{m}$. 

            Secondly, by the symmetry of $\bs V_{i,l}$, the matrix variance is developed as follows:
            \begin{align*}
             \ts v \coloneqq \|\sum_{i,l}\,\bb E(\bs V_{i,l} - \bb E \bs
            V_{i,l})(\bs V_{i,l} - \bb E \bs
            V_{i,l})^\top\|&=\ts \|\sum_{i,l}\,\bb E(\bs V_{i,l}\bs V_{i,l}^\top) - (\bb E \bs
            V_{i,l})(\bb E \bs
            V_{i,l})^\top\|\\
			&\ts \leq \|\sum_{i,l}\bb E(\bs V_{i,l}\bs V_{i,l}^\top)\| + \|\sum_{i,l} (\bs B^\top \bs c_i \bs
            c_i^\top \bs B) (\bs B^\top \bs c_i \bs
            c_i^\top \bs B)\|\\
			&\ts \leq \|\sum_{i,l}\bb E(\bs V_{i,l}\bs V_{i,l}^\top)\| +
			\upmu_{\max}^2\tfrac{hp}{m}\,\|\bs B^\top (\sum_{i} \bs c_i \bs
            c_i^\top) \bs B\|\\
			&\ts = \|\sum_{i,l}\bb E(\bs V_{i,l}\bs V_{i,l}^\top)\| +
			\upmu_{\max}^2\tfrac{hp}{m},
            \end{align*}
            where we used $\bs B^\top (\sum_{i} \bs c_i \bs
            c_i^\top) \bs B = \bs B^\top\bs B= \bs I_h$. Moreover, using the same developments as above, we find
            \begin{align*}
            \ts \|\sum_{i,l}\bb E(\bs V_{i,l}\bs V_{i,l}^\top)\|&\ts = \|\sum_{i,l}(\bb E X^4_{i,l}) (\bs B^\top \bs c_i \bs
            c_i^\top \bs B) (\bs B^\top \bs c_i \bs
            c_i^\top \bs B)\| \lesssim \alpha^4 \upmu_{\max}^2 \tfrac{hp}{m},           
            \end{align*}
            where, again by rotational invariance,
            $\|X_{ij}\|_{\psi_2} = \|X\|_{\psi_2} = \|\bs a^\top \hat{\bs x}\|_{\psi_2}
            \lesssim \alpha$ so that $\bb E X_{ij}^4 = \bb E X^4 \lesssim
            \alpha^4$. Consequently, we have $v \lesssim \upmu_{\max}^2
            \tfrac{hp}{m}$. In addition, a lower bound on $v$ is provided by
            \begin{align*}
\ts v &\ts \coloneqq \|\sum_{i,l}\,\bb E(\bs V_{i,l} - \bb E \bs
            V_{i,l})(\bs V_{i,l} - \bb E \bs
            V_{i,l})^\top\|\\
&\ts = \|\sum_{i,l}\bb E(X^2_{i,l} - 1)^2\,(\bs B^\top \bs c_i \bs
            c_i^\top \bs B) (\bs B^\top \bs c_i \bs
            c_i^\top \bs B)\|\\
&\ts = p\,\bb E((\bs a^\top \hat{\bs x})^4 - 1)\,\|\sum_{i}(\bs B^\top \bs c_i \bs
            c_i^\top \bs B) (\bs B^\top \bs c_i \bs
            c_i^\top \bs B)\|, 
            \end{align*}
			where $X_{i,l} \sim X$ and $\bb E X^2 =1$. Let us then note that\footnote{This is easily shown by expanding the entries of $(\hat{\bs x}^\top \bs a)^2 \bs a \bs a^\top$ as done, \eg in the proof of \cite[Lemma~3.5]{WhiteSanghaviWard2015}.} 
			\begin{equation}
			\label{eq:expect_axsq_aa}
				{\bb E}(\hat{\bs x}^\top \bs a)^2 \bs a \bs a^\top = 2\hat{\bs x}\hat{\bs x}^\top + \bs I_{n} + (\bb E X^4 - 3)\diag(\hat{\bs x})^2.                            
			\end{equation}
			By this fact and the Bernoulli Restriction Hypothesis in Sec.~\ref{def:model}, there exists a $c>0$ such that  
			\begin{equation}
			  \label{eq:low-bound-fourth-mom}
			\ts \bb E (\bs a^\top \hat{\bs x})^4 - 1
			= \hat{\bs x}^\top [\bb E (\bs a^\top \hat{\bs x})^2 \bs a \bs a^\top]\,\hat{\bs x} 
			=2(1-\|\hat{\bs x}\|^4_4) + (\bb E X^4
			- 1) \|\hat{\bs x}\|^4_4 > \tfrac{c}{n},
			\end{equation}
			where we also used the fact that $\|\hat{\bs x}\|^4_4 \geq
			\tfrac{1}{n}\|\hat{\bs x}\|^4 = \tfrac{1}{n}$ which shows that the
			bound is also respected, whatever the value of $\hat{\bs x}$, provided $\bb E X^4 > 1$.
			
			Consequently, since 
			$$
			\ts \|\sum_{i}(\bs B^\top \bs c_i \bs
			            c_i^\top \bs B) (\bs B^\top \bs c_i \bs c_i^\top \bs B)\| \geq 
			\max_{i} \|(\bs B^\top \bs c_i \bs c_i^\top \bs B) (\bs B^\top \bs c_i \bs
			             c_i^\top \bs B)\| = (\tfrac{h}{m})^2\mu^4_{\max}, 
			$$
			we can consider that $v \gtrsim \mu^4_{\max} \frac{p}{n}
			(\frac{h}{m})^2$ so that, with the previous bound on $T$, we find 
			$$
			\ts \frac{T^2 mp}{v} \lesssim mp 
			\frac{h^2}{m^2} \frac{n}{p}
			(\frac{m}{h})^2 \leq mn
			$$ 
			Inserting the bounds on $T$, $v$ and $\log(\frac{T^2 mp}{v}) \lesssim \log(mn)$ in Prop.~\ref{prop:mbi} provides
			\begin{align*}
			  \ts \bb P\big[\|\frac{1}{p}(\sum_{i,l} \bs V_{i,l} - \bb E \bs V_{ij})\| >
			  \delta\big]&\ts \lesssim \exp\left(\log(h) - c \frac{\delta^2
			               mp}{\upmu_{\max}^2 h \log(mn)}\right) ,
			\end{align*}
			for some $c>0$. Consequently, we observe as announced that $\bb P[\|\frac{1}{p}(\sum_{i,l} \bs V_{i,l} - \bb E \bs V_{ij})\| > \delta] \lesssim (mp)^{-t}$ if $mp \gtrsim t \delta^{-2} \upmu_{\max}^2 h \log(mp)\log(mn)$. The last result of the proposition is simply obtained by replacing $\bs B$ by $\bs I_m$, \ie letting $h=m$ and $\upmu_{\max} = 1$ in the developments above. 
		\end{proof}
        
        \medskip

        Following a framework defined in \cite{AhmedRechtRomberg2014,LingStrohmer2015a}, matrix Bernstein's inequality also allows us to study the concentration of random projections in the sensing model of Def.~\ref{def:modelsub} when these are seen as a linear random operator acting on $\bs z \bs b^\top \in \bb R^{k \times h}$. Indeed, introducing the linear random map 
        \begin{equation}
          \label{eq:linear-gen-map}
        \ts \cl A: \bs W \in \bb
        R^{k \times h} \to \big\{\cl A_{il}(\bs W) \coloneqq \tinv{\sqrt p}\big\langle \bs Z^\top \ail \bs c^\top_i
        \bs B, \bs W \big\rangle_F:\ i \in [m], l \in [p]\big\} \in \bb
        R^{mp},
        \end{equation}
        for which the adjoint $\cl A^*: \bb R^{mp} \to \bb R^{k \times
        h}$ is such that 
        \begin{equation}
          \label{eq:linear-gen-map-adjoint-def}
        \ts \cl A^* \cl A(\bs W) =  \tinv{p}\sum_{i,l}\,\big\langle \bs Z^\top \ail \bs c^\top_i
        \bs B,\, \bs W \big\rangle_F \ \bs Z^\top \ail \bs c^\top_i
        \bs B,
      \end{equation}
      it is easy to see that the sensing model~\eqref{eq:sub-sens-model} is then equivalent to the action of $\cl A$ on the rank-1 matrix $\bs z\bs b^\top$, \ie for $i \in [m],\, l \in [p]$,
        \begin{equation}
          \label{eq:sub-sens-model-lin-mtx-oper}
          \ts y_{i,l} = g_i (\ail^\top \bs x) = (\bs c_i^\top \bs B \bs b)\, \ail^\top \bs Z \bs z = \langle \bs Z^\top \ail \bs c^\top_i
        \bs B,\, \bs z \bs b^\top \rangle_F = \sqrt{p}\,\cl A_{i,l}(\bs z \bs b^\top).
        \end{equation}
        
        Our developments actually need to characterise the concentration of $\cl A^*\cl A$ around its mean when evaluated on any element of a particular $(k+h)$-dimensional subspace $\cl M$ of $\bb R^{k \times h}$ to which $\bs z \bs b^\top$ belongs. This subspace is defined in the following proposition along with its projector.     
	\begin{proposition}
		\label{prop:subspaceM}
		In the setup of Def.~\ref{def:modelsub}, define the
                $(k+h)$-dimensional linear subspace
		\begin{equation}
		\label{eq:bigM}
			\cl M \coloneqq \big\{\,\u\,\hat{{\bs b}}^\top
                        + \hat{\bs z}\,\v^\top :\ \u \in \R^k,\, \v \in \R^{h}\big\}\ \subset\ \bb R^{k\times h}, 
		\end{equation}
		associated to $\bs x = \bs Z
                \bs z$ and $\bs g = \bs B \bs b$. Then the orthogonal
                projector on $\cl M$ is the linear operator
		\begin{equation}
		\label{eq:projM}
			{\cl P}_{\cl M}(\bs Q) \coloneqq \hat{\bs z} \hat{\bs z}^\top {\bs Q} + {\bs Q} \hat{{\bs b}} \hat{{\bs b}}^\top -  \hat{\bs z} \hat{\bs z}^\top {\bs Q} \hat{{\bs b}} \hat{{\bs b}}^\top.
		\end{equation}
	\end{proposition}
        The verification that~\eqref{eq:projM} is the correct
        projector, \ie that both ${\cl P}_{\cl M}(\bs Q) \in \cl M$
        and $\langle {\cl P}_{\cl M} (\bs Q), \bs Q - {\cl P}_{\cl
          M}(\bs Q)\rangle_F = 0$ for some $\bs Q \in \bb R^{k \times
          h}$ is long but extremely simple and omitted for the sake of
        brevity. Its linearity is also obvious.	
\medskip

        The announced concentration is then characterised in the next proposition and is important in the proof of
        Prop.~\ref{prop:regusub} in Sec.~\ref{sec:appC}. In a
        nutshell, it characterises a certain isometry
        property of $\cl A$ when the latter is \emph{restricted}~to matrices belonging to $\cl M$. This is closely related in principle to
        \cite[Cor.~2]{AhmedRechtRomberg2014} (and also \cite[Lemma
        4.3]{LingStrohmer2015a}) where random projections based on
        complex Gaussian random vectors $\ail$ are considered.  Our proof differs in that
        we provide results for real signals and \iid sub-Gaussian
        random sensing vectors. 

	\begin{proposition}[A Restricted Isometry for sub-Gaussian $\cl A$ on $\cl M$]
		\label{prop:standalone-ahmed}
		With the linear random map $\cl A$ defined in
               ~\eqref{eq:linear-gen-map} and the subspace $\cl M$
                associated to the projector ${\cl P}_{\cl M}$ defined in
                Prop.~\ref{prop:subspaceM}, given $\delta \in (0,1)$,
                $t\geq1$, and provided 
                \begin{subequations}
                  \label{eq:scseparate}
                \begin{align}                  
  \ts n&\ts \gtrsim t \log(mp),\\
\ts mp&\ts \gtrsim \delta^{-2} \max(k,\upmu_{\max}^2 h)\, \max\big(t
        \log(mp) \log(\frac{mn}{1-\rho}), \log(\frac{n}{\delta})\big),            
                \end{align}
                \end{subequations}
		we have 
                \begin{talign}
                  \label{eq:expect-AstA}
                  \bb E {\cl P}_{\cl M} \cl A^* \cl A\cl
                  P_{\cl M} &= {\cl P}_{\cl M},\\
                  \label{eq:concent-AstA}
\ts                  \big\| {\cl P}_{\cl M} \cl A^* \cl A\cl
                  P_{\cl M} - {\cl P}_{\cl M} \big\|&\coloneqq \sup_{\bs W \in \bb S_F^{k\times h}} \big|\langle \bs W, ({\cl P}_{\cl M} \cl A^* \cl A\cl
                  P_{\cl M} - {\cl P}_{\cl M}) \bs W\rangle_F\big|\ \leq\ \delta
                \end{talign}
                with probability exceeding $1 - C (mp)^{-t}$ for some $C > 0$.
	\end{proposition}
	
	% \begin{proposition}[A Concentration Inequality for an Inner Product in $\cl M$]
	% 	\label{prop:standalone-ahmed}
	% 	Let us consider the setup of Def.~\ref{def:modelsub}; define the function
	% 	\begin{equation}
	% 		\label{eq:muW}
	% 		\mu(\bs W) \coloneqq \tinv{mp} \ts\sum_{i,l} \big\langle{\cl P}_{\cl M}(\bs Z^\top \ail \bs c^\top_i {\bs B}), {\bs W} \big\rangle^2_F - \bs c^\top_i {\bs B} {\bs W^\top \bs W}  {\bs B}^\top\bs c_i
	% 	\end{equation}
	% 	Given $\bs W \in \bb S^{n \times m}_F$, provided 
	% 	\begin{equation}
	% 		\label{eq:scseparate}
	% 		k \gtrsim t \log mp, \ n \gtrsim t \log mp, \ p \gtrsim \log m, \ mp \gtrsim (k+h) \log n
	% 	\end{equation}
	% 	for $t \geq 1$, with probability exceeding
	% 	\begin{equation}
	% 		\label{eq:probseparate}
	% 		1 - C [\exp(-c \delta^2 p) + \exp(-c \delta^2 m p)+ (mp)^{-t}]
	% 	\end{equation}
	% 	for some $C,c > 0$ we have that $|\mu(\bs W)|\leq 12\tfrac{\delta}{m}$.
	% \end{proposition}

	\begin{proof}[Proof of Prop.~\ref{prop:standalone-ahmed}]
		%%%%%%%%% COPIED PART %%%%%%%%%
          Given some $t\geq 1$ and $\delta \in (0,1)$, the proof of this proposition is realised conditionally to an event $\cl E$ given by the joint verification of the following two properties for the $\{\ail\}$ in Def.~\ref{def:modelsub}, \ie
          \begin{subequations}
          \begin{eqnarray}
			\label{eq:determ-relations-nice-concent}&
                                                      \big\| \tinv{p} \ts\sum_{i,l}  ({\bs B}^\top  {\bs c}_i  {\bs c}^\top_i {\bs
                                                      B})\ \hat{\bs z}^\top {\bs Z}^\top\, (\ail \ail^\top - \I_n)\, {\bs Z}
                                                      \hat{\bs z}  \big\|\ \leq\ \delta,\\ 
            \label{eq:determ-relations-weighted-cov}&\ts \|\tinv{mp}\,\sum_{i,l} \theta_i \bs Z^\top
		(\ail \ail^\top - \I_n) \bs Z \| \leq 
		\delta\,\|\bs \theta\|_\infty,\ \forall \bs \theta
                                                    \in \cl B.
        \end{eqnarray}\label{eq:determ-relations}
        \end{subequations}
        \!\!In fact, we shall see that the probability of Prop.~\ref{prop:standalone-ahmed} is fixed by that of $\cl E$, which we bound as follows. From Cor.~\ref{coro:appl-wcovsub} and Prop.~\ref{prop:m-w-covconc}, and given some $c >0$, we can easily bound the probability of the complementary event
         \begin{equation}
           \label{eq:bound-clE-c}
         \bb P[\cl E^{\rm c}]\ \lesssim\ \exp(-c\delta^2 mp) +
         (mp)^{-t} \lesssim (mp)^{-t},           
         \end{equation}
         provided 
         \begin{equation}
           \label{eq:glob-requirements}
           \ts n \gtrsim t \log(mp),\ 
mp\ \gtrsim\ t \delta^{-2} \upmu_{\max}^2\,h \log(mp) \log(mn),\ mp
           \gtrsim \delta^{-2}(k+h) \log(\frac{n}{\delta}).            
         \end{equation}
         Conditionally on $\cl E$, the rest of the proof can thus assume that the relations listed in~\eqref{eq:determ-relations} hold deterministically. This conditioning is thus implicit in the rest of the developments. 
		\medskip

Before proceeding, let us first notice that, from~\eqref{eq:linear-gen-map-adjoint-def} and since $\scp{\bs A \bs B}{\bs C}_F = \scp{\bs B}{\bs A^\top \bs C}_F=\scp{\bs A}{\bs C\bs B^\top}_F$ for any matrices $\bs A, \bs B, \bs C$ with compatible dimensions, 
\begin{talign*}
\big\langle \bs W, {\cl P}_{\cl M} \cl A^* \cl A {\cl P}_{\cl M} \bs W\big \rangle_F&= \big\langle {\cl P}_{\cl M} \bs W, (\cl A^* \cl A)\, {\cl P}_{\cl M} \bs W\big \rangle_F = \tinv{p} \sum_{i,l}\,\big\langle \bs Z^\top \ail \bs c^\top_i
\bs B, {\cl P}_{\cl M} \bs W \big\rangle^2_F\\
&= \tinv{p}\sum_{i,l}\,\big\langle {\cl P}_{\cl M}\,\bs Z^\top \ail \bs c^\top_i
\bs B, {\cl P}_{\cl M}\,\bs W \big\rangle^2_F,  
\end{talign*}
where we used the fact that ${\cl P}_{\cl M}$ is an orthogonal projector on the subspace $\cl M \subset \bb R^{k \times h}$.
Moreover, for $\bs W' = {\cl P}_{\cl M} \bs W \in {\cl M}$ and since $\bs Z^\top\bs Z = \bs I_k$,
\begin{talign*}
\bb E \big\langle \bs W, {\cl P}_{\cl M} \cl A^* \cl A {\cl P}_{\cl M} \bs W\big \rangle_F&=\big\langle \bs W, {\cl P}_{\cl M} (\bb E \cl A^* \cl A ){\cl P}_{\cl M} \bs W\big \rangle_F\\
&= \tinv{p}\sum_{i,l}\, \bb E \big\langle {\cl P}_{\cl M}\,\bs Z^\top \ail \bs c^\top_i
\bs B, \bs W' \big\rangle^2_F\\
&= \tinv{p}\sum_{i,l}\, \bs c^\top_i\bs B\bs W'^\top\bs Z^\top \bb E(\ail \ail^\top)\bs Z \bs W' \bs B^\top  \bs c_i\\  
&= \sum_{i}\, \bs c^\top_i\bs B\bs W'^\top\bs W' \bs B^\top  \bs c_i\\  
&= \tr(\bs W'^\top\bs W') = \scp{\bs W}{({\cl P}_{\cl M}) \bs W}_F,
\end{talign*}
so that $\bb E(P_{\cl M}\cl A^* \cl A P_{\cl M}) = {\cl P}_{\cl M}$ since the equalities above are true for all $\bs W \in \bb S_F^{k\times h}$.

Therefore, to prove~\eqref{eq:concent-AstA} we must find an upper bound for $|M - \bb E M|$ with 
$$
M \coloneqq \tinv{p} \ts\sum_{i,l} \big\langle{\cl P}_{\cl
  M}(\bs Z^\top \ail \bs c^\top_i {\bs B}), {\bs W}
\big\rangle^2_F,
$$
for all $\bs W \in \cl M \cap \bb S^{k \times h}_{F}$. According to the definition of $\cl
                M$ and $\cl
                P_{\cl M}$ in Prop.~\ref{prop:subspaceM}, we proceed
                by first expanding 
		\[
			{\cl P}_{\cl M}  (\bs Z^\top \ail \bs c^\top_i {\bs B}) =  \hat{\bs z} \hat{\bs z}^\top \bs Z^\top \ail \bs c^\top_i {\bs B}(\I_{h} - \hat{{\bs b}}\hat{{\bs b}}^\top) + ({\bs B}_{i,\cdot} \hat{{\bs b}}) \bs Z^\top \ail \hat{{\bs b}}^\top,
		\]
		where we have collected in the last summand ${\bs
                  B}_{i,\cdot} \hat{{\bs b}} = \bs c_i^\top  {\bs B}
                \hat{{\bs b}} = \tfrac{g_i}{\|{\bs g}\|} = \hat{g}_i $. For a
                lighter notation, we define the vectors
		\[
			\bs u_{i,l}\coloneqq \hat{g}_i \bs Z^\top \ail, \quad \bs v_{i,l} \coloneqq (\I_{h} - \hat{{\bs b}}\hat{{\bs b}}^\top) {\bs B}^\top {\bs c}_i \ail^\top \bs Z \hat{\bs z},
		\]
		which let us rewrite ${\cl P}_{\cl M} (\bs Z^\top \ail
                \bs c^\top_i {\bs B})  = {\bs u}_{i,l} \hat{{\bs
                    b}}^\top + \hat{\bs z} {\bs v}^\top_{i,l}$; this
                form will be more convenient for some
                computations. We can then develop 
		\begin{align*}
			M & = 
			\tinv{p} \ts\sum_{i,l} \big\langle  {\bs u}_{i,l} \hat{{\bs b}}^\top + \hat{\bs z}{\bs v}^\top_{i,l}, {\bs W} \big\rangle^2_F \\
			& = {\tinv{p} \ts\sum_{i,l} \big\langle {\bs u}_{i,l} \hat{{\bs b}}^\top, \bs W\big\rangle^2_F} + {\tinv{p} \ts\sum_{i,l} \big\langle \hat{\bs z} {\bs v}^\top_{i,l},\bs W\big\rangle^2_F}
			+{\tfrac{2}{p} \ts\sum_{i,l}  \big\langle  {\bs u}_{i,l} \hat{{\bs b}}^\top,\bs W\big\rangle_F\big\langle \hat{\bs z} {\bs v}^\top_{i,l},\bs W\big\rangle_F} \\
			& = \underbrace{\tinv{p} \ts\sum_{i,l} \big\langle {\bs u}_{i,l}, \bs W \hat{{\bs b}}\big\rangle^2}_{M'} + \underbrace{\tinv{p} \ts\sum_{i,l} \big\langle {\bs v}_{i,l},\bs W^\top \hat{\bs z} \big\rangle^2}_{M''} 
			+\underbrace{\tfrac{2}{p} \ts\sum_{i,l}  \big\langle  {\bs u}_{i,l} ,\bs W \hat{{\bs b}} \big\rangle\big\langle {\bs v}_{i,l},\bs W^\top \hat{\bs z}\big\rangle}_{M'''},
		\end{align*}
		By linearity, we also have that $\bb E M = \bb E M' + \bb E M'' + \bb E M'''$. Hence,   
		$|M - \bb E M| \leq |M'-\bb E M'| + |M''-\bb E M''|
                +|M'''-\bb E M'''|$. Each of these three terms will
                now be carefully bounded. 

\medskip
\ \\
\noindent\emph{(i) Bound on $|M'-\bb E M'|$:} Let us recall that $g_i = 1 + \bs
                          c^\top_i {\bs B}^\perp {\bs b}^\perp = \bs
                          c^\top_i {\bs B} {\bs b} \leq 1+\rho$, with
                          $\rho \in [0,1)$. Then by~\eqref{eq:determ-relations-weighted-cov} we have
			\begin{align*}
			|M'-\bb E M'| & = \big\vert \tinv{p} \ts\sum_{i,l} \tfrac{g^2_i}{\|{\bs g}\|^2} \hat{{\bs b}}^\top \bs W^\top \bs Z^\top (\ail \ail^\top - \I_n) \bs Z \bs W \hat{{\bs b}} \big\vert \\
			& \leq \delta (1+\rho)^2 \|\bs W \hat{{\bs b}}\|^2 \leq \delta (1+\rho)^2 \|  \hat{{\bs b}} \hat{{\bs b}}^\top\|_F \|\bs W^\top \bs W\|_F \leq 4\delta.
			\end{align*}

\medskip
\ \\
\noindent\emph{(ii) Bound on $|M''-\bb E M''|$:} From
                         ~\eqref{eq:determ-relations-nice-concent}
                          the term $M''$ is so that 
			\begin{align*}
			|M''-\bb E M''| & = \big\vert \tinv{p} \ts\sum_{i,l}  (\hat{\bs z}^\top \bs W (\I_{h}-\hat{{\bs b}} \hat{{\bs b}}^\top) {\bs B}^\top  {\bs c}_i)^2\,\hat{\bs z}^\top {\bs Z}^\top (\ail \ail^\top - \I_n) {\bs Z} \hat{\bs z}  \big\vert  \\
			& = \big\vert \hat{\bs z}^\top \bs W
                          (\I_{h}-\hat{\bs b}\hat{\bs b}^\top)\ \big[\tinv{p}
                          \ts\sum_{i,l}  (\bs B^\top  {\bs c}_i {\bs
                          c}^\top_i \bs B)\,\hat{\bs z}^\top {\bs
                          Z}^\top (\ail \ail^\top - \I_n) {\bs Z}
                          \hat{\bs z}\big]\ (\I_{h}-\hat{\bs b}\hat{\bs b}^\top) \bs
                          W^\top \hat{\bs z} \big\vert  \\
&\leq \|(\I_{h}-\hat{\bs b}\hat{\bs b}^\top) \bs
                          W^\top \hat{\bs z}\|^2\,\big\|\tinv{p}
                          \ts\sum_{i,l}  (\bs B^\top  {\bs c}_i {\bs
                          c}^\top_i \bs B)\,\hat{\bs z}^\top {\bs
                          Z}^\top (\ail \ail^\top - \I_n) {\bs Z}
                          \hat{\bs z}\big\| \\
&\leq \|\bs
                          W^\top \hat{\bs z}\|^2\,\delta\ \leq\
  \delta.
		\end{align*}
% \mtodo{Here use
%                             rather Prop.~\ref{prop:m-w-covconc}}
% 			\begin{align*}
% 			|M''-\bb E M''| & = \big\vert \tinv{p} \ts\sum_{i,l}  (\hat{\bs z}^\top \bs W (\I_{h}-\hat{{\bs b}} \hat{{\bs b}}^\top) {\bs B}^\top  {\bs c}_i)^2\,\hat{\bs z}^\top {\bs Z}^\top (\ail \ail^\top - \I_n) {\bs Z} \hat{\bs z}  \big\vert  \\
% 			& = \big\vert \ts\sum_{i}  (\hat{\bs z}^\top
%                           \bs W (\I_{h}-\hat{{\bs b}} \hat{{\bs
%                           b}}^\top) {\bs B}^\top  {\bs c}_i)^2\,
%                           \tinv{p} \sum_l ((\hat{\bs z}^\top {\bs Z}^\top \ail)^2 - 1) \big\vert  \\
% 			& \leq \delta\big\vert \ts\sum_i \big({\bs c^\top_i} {\bs B} \big(\I_{h} - \hat{{\bs b}} \hat{{\bs b}}^\top\big)\bs W^\top \hat{\bs z}\big)^2\big\vert \\
% 			& = \delta\big\|{\bs B} (\I_{h} - \hat{{\bs
%                           b}}\hat{{\bs b}}^\top) \bs W^\top\big\|^2_F
%                           = \delta \big\|(\I_{h} - \hat{{\bs b}}\hat{{\bs b}}^\top) \bs W^\top\big\|^2_F
% 			\end{align*}
% 			since ${\bs B}^\top {\bs B} = \I_{h}$. Hence, we can bound 
% 			$
% 			\big\|\bs W \big(\I_{h} - \hat{{\bs b}} \hat{\bs b}^\top\big)\big\|_F \leq \|\bs W \|_F + \big\|\bs W \hat{{\bs b}} \hat{{\bs b}}^\top\big\|_F \leq 2
% 			$
% 			and $|M''-\bb E M''| \leq 2\delta$. 			
			
\medskip
\ \\
\noindent\emph{(iii) Bound on $|M'''-\bb E M'''|$:} As for the mixed term $M'''$,~\ie
			\[
			M''' = \tfrac{2}{p} \ts\sum_{i,l} \hat{{\bs b}}^\top \bs W^\top {\bs u}_{i,l} {\bs v}_{i,l}^\top \bs W^\top \hat{\bs z},
			\]
			this requires a more involved bounding technique. We begin by noting that 
			\begin{equation}
				\label{eq:E-uil-vil}
				\bb E \bs u_{i,l} \bs v^\top_{i,l} = \hat{g}_i  \hat{\bs z} {\bs c}_i^\top {\bs B}(\I_{h} - \hat{{\bs b}}\hat{{\bs b}}^\top),
			\end{equation}
			and that
			\[
\ts			\big\Vert\bb E \bs u_{i,l} \bs v^\top_{i,l}\big\Vert = \big\Vert \hat{g}_i \hat{\bs z} {\bs c}_i^\top {\bs B}(\I_{h} - \hat{{\bs b}}\hat{{\bs b}}^\top)\big\Vert \leq \tfrac{\|\bs g\|_\infty}{\|\bs g\|} \big\Vert {\bs c}_i^\top {\bs B}\big\Vert \ \leq  \frac{1+\rho}{\sqrt m} \upmu_{\max} \sqrt{\frac{h}{m}} \lesssim \upmu_{\max} \frac{\sqrt h}{m}.
			\]
			
			Given these facts about $\bb E\bs u_{i,l} \bs v^\top_{i,l}$, we can compute the expectation of $M'''$, \ie
			\begin{align*}
			\bb E M''' & = 2 \ts\sum_i \hat{g}_i \hat{{\bs b}}^\top \bs W^\top \hat{\bs z}(\hat{\bs z}^\top \bs W(\I_{h}- \hat{{\bs b}}\hat{{\bs b}}^\top) {\bs B}^\top {\bs c}_i) \\
			& = 2 \ts\sum_i {\bs c^\top_i {\bs B} \hat{{\bs b}}} \hat{{\bs b}}^\top \bs W^\top \hat{\bs z} \hat{\bs z}^\top \bs W(\I_{h}- \hat{{\bs b}}\hat{{\bs b}}^\top) {\bs B}^\top {\bs c}_i \\
			& = 2 \tr{\big( \hat{{\bs b}}^\top {\bs W}^\top \hat{\bs z} \hat{\bs z}^\top {\bs W} \underbrace{ (\I_{h} - \hat{{\bs b}} \hat{{\bs b}}^\top) \hat{{\bs b}} }_{\vzer_{h}} \big)} = 0
			\end{align*}
			Thus, we have $|M'''-\bb E M'''| = |M'''|$. Using the previous definitions, we also have
                        \begin{multline}
                          \label{eq:tmp-Mtierce}
			|M'''| = \big|\tfrac{2}{p} \hat{{\bs b}}^\top
                        \bs W^\top \big(\ts\sum_{i,l}\bs u_{i,l} \bs
                        v_{i,l}^\top\big)\bs W^\top \hat{\bs z}\big| % \\
                        \leq 2 \| \bs W \hat{{\bs b}}\| \| \bs
                        W^\top \hat{\bs z}\|\,\tfrac{1}{p}
                        \|\ts\sum_{i,l} {\bs u}_{i,l} {\bs
                          v}_{i,l}^\top \|\  \leq\ \tfrac{2}{p}
                        \|\ts\sum_{i,l} {\bs u}_{i,l} {\bs
                          v}_{i,l}^\top \|
                      \end{multline}
                    	Note how the last inequality highlights the norm of a sum of rectangular random matrices ${\bs
                          J}_{i,l} \coloneqq \bs u_{i,l} \bs
                        v^\top_{i,l}$, \ie $\bs J \coloneqq
                        \ts\sum_{i,l} {\bs J}_{i,l}$, whose
                        expectation $\bb E \bs J = \vzer_{n \times m}$, and whose summands are non-centred with $\bb E \bs J_{i,l}$ given in \eqref{eq:E-uil-vil}.

                        We can thus apply matrix Bernstein's inequality in the form of Prop.~\ref{prop:mbi} (with $q=mp$) by first computing the upper bound $T>\max_{i,l} \|\|\bs J_{i,l} - \bb E \bs J_{i,l}\|\|_{\psi_1}$ as well as the matrix variance $v$ given in~\eqref{eq:matrix-var-T}. This task is eased by our conditioning
                        over $\cl E$ for which all the relations in
                       ~\eqref{eq:determ-relations} hold. 

                        Firstly, the upper bound $T$ is obtained by noting that
                        $2\|\bs J_{i,l}\| = 2\|\bs u_{i,l}\| \|\bs
                        v_{i,l}\| \leq \|\bs u_{i,l}\|^2 + \|\bs
                        v_{i,l}\|^2$, so that $
                        \|\|\bs J_{i,l}\|\|_{\psi_1} \leq \|\|\bs u_{i,l}\|\|^2_{\psi_2} + \|\|\bs
                        v_{i,l}\|\|^2_{\psi_2}$, since $\|X\|^2_{\psi_2} \leq
                          \|X^2\|_{\psi_1} \leq 2 \|X\|^2_{\psi_2}$ for any \rv $X$ \cite[Lemma 5.14]{Vershynin2012a}. However 
                        \begin{align*}
                        \ts \|\|\bs u_{i,l}\|\|^2_{\psi_2} \leq
                        \|\|\bs u_{i,l}\|^2\|_{\psi_1}&\ts = \hat{g}_i^2
                        \|\|\bs Z^\top \ail\|^2\|_{\psi_1} \leq
                        \hat{g}_i^2 \sum_{j=1}^k \|(\bs Z_j^\top
                        \ail)^2\|_{\psi_1}\\
&\ts \leq
                        2 \hat{g}_i^2 \sum_{j=1}^k \|\bs Z_j^\top
                        \ail\|^2_{\psi_2} \lesssim  k
  \tfrac{(1+\rho)^2}{m} \lesssim \tfrac{k}{m},  
                        \end{align*}
                        where we used the rotational invariance of
                        $\ail \sim_\iid \bs a$. Moreover, since $\|\bs v_{i,l}\| = |{\hat{\bs z}^\top \bs
                          Z^\top \ail}| \|(\I_{h} - \hat{{\bs
                            b}}\hat{{\bs b}}^\top) {\bs B}^\top \bs
                        c_i\|$, we find $\|\|\bs v_{i,l}\|\|^2_{\psi_2} \ts \leq \||\hat{\bs
                                                              x}^\top \ail|\|^2_{\psi_2}\|{\bs B}^\top \bs
                        c_i\|^2 \lesssim \upmu^2_{\max} \tfrac{h}{m}$. 
                        Consequently, since we showed above that 
$\|\bb E \bs J_{i,l}\|_{\psi_1} = \|\bb E \bs u_{i,l}\bs v_{i,l}^\top\| \lesssim  \upmu_{\max} \frac{\sqrt h}{m}$, there exists a $T > 0$ such that $\|\|\bs J_{i,l} - \bb E\bs J_{i,l}\|\|_{\psi_1} \leq T$ for all $i\in[m]$ and $l\in [p]$, with $T \lesssim \tfrac{k + \upmu_{\max}^2 h}{m}$.

% 			\begin{equation}
% 			\label{eq:bound-normu}
% 			\ts \|\bs u_{i,l} \| \leq
%                         |\hat{g}_i| \|\bs
%                         Z^\top \ail\| \leq \sqrt{\tfrac{\vartheta t
%                             k \log (mpk)}{m}} (1+\rho) \leq 2\sqrt{\vartheta}\,\sqrt{\tfrac{t
%                             k \log (mp)}{m}}  
% 			\end{equation}			 
%                         from
%                        ~\eqref{eq:determ-relations-trunc-ail-subspace}
%                         and since $mp \geq k$ from
%                        ~\eqref{eq:glob-requirements}, and
% 			\begin{equation}
% 			\label{eq:bound-normv}
% 			\|\bs v_{i,l}\| = \big|{\hat{\bs z}^\top \bs
%                           Z^\top \ail}\big| \big\|(\I_{h} - \hat{{\bs
%                             b}}\hat{{\bs b}}^\top) {\bs B}^\top \bs
%                         c_i\big\| \leq \upmu_{\max} \sqrt{\tfrac{th \log(mp)}{m}}
% 			\end{equation}
%                         from
%                        ~\eqref{eq:determ-relations-trunc-ail-scp} and with $\upmu_{\max}$
%                           defined
%                           in~\eqref{eq:mumax}. Recalling
%                         what we showed above, \ie $\|\bb E \bs
%                         J_{i,l}\| \leq \tfrac{1+\rho}{\sqrt{m}} $,
%                         we find that we can take 
% 			\[
% \ts \|\bs u_{i,l}\|\|\bs v_{i,l}\| \leq 2\sqrt{\vartheta} \upmu_{\max}\,\frac{\sqrt{kh}}{m} t \log(mk)
%   \leq \sqrt{\vartheta}\, m^{-1}\,(k + \upmu_{\max}^2 h)\, t \log(mk) \eqqcolon T,
% 			\]
% \ie with $T \lesssim m^{-1}\,(k + \upmu_{\max}^2 h)\, t \log(mk)$.
			
                      Secondly, we have to bound the matrix variance $v$ defined in~\eqref{eq:matrix-var-T} associated to the matrix sequence $\{\bs J_{i,l}\}$. Let us first compute:
			\begin{align}
			\bb E(\bs J_{i,l} \bs J^\top_{i,l}) & = \bb E( \|\bs v_{i,l}\| \bs u_{i,l} \bs u^\top_{i,l}) \nonumber \\ 
			& = \hat{g}^2_i \bs c^\top_i {\bs B} (\I_{h} - \hat{{\bs b}}\hat{{\bs b}}^\top)  {\bs B}^\top \bs c_i \, {\bb E}[({\hat{\bs z}^\top \bs Z^\top \ail})^2 \bs Z^\top \ail \ail^\top \bs Z] \label{eq:mv1}\\
			\bb E(\bs J_{i,l}^\top \bs J_{i,l}) & =  \bb E(\|\bs u_{i,l}\| \bs v_{i,l} \bs v^\top_{i,l})\nonumber \\ 
                                                            & = \hat{g}^2_i (\I_{h} - \hat{{\bs b}}\hat{{\bs b}}^\top) {\bs B}^\top \bs c_i   \bs c^\top_i {\bs B} (\I_{h} - \hat{{\bs b}}\hat{{\bs b}}^\top) \, \bb E[ ({\hat{\bs z}^\top \bs Z^\top \ail})^2 \tr(\bs Z^\top \ail \ail^\top \bs Z)] \label{eq:mv2}\\
			(\bb E \bs J_{i,l}) (\bb E \bs J_{i,l})^\top & = \hat{g}^2_i \hat{\bs z} \bs c^\top_i {\bs B} \big(\I_{h}-\hat{{\bs b}}\hat{{\bs b}}^\top\big) {\bs B}^\top \bs c_i  \hat{\bs z}^\top \label{eq:mv3}\\
			(\bb E \bs J_{i,l})^\top (\bb E \bs J_{i,l}) & = \hat{g}^2_i (\I_{h} - \hat{{\bs b}}\hat{{\bs b}}^\top) {\bs B}^\top \bs c_i   \bs c^\top_i {\bs B} (\I_{h} - \hat{{\bs b}}\hat{{\bs b}}^\top), \label{eq:mv4} 
			\end{align}
                        which also shows that 
			\begin{align} \bb E(\bs J_{i,l} \bs
J^\top_{i,l}) & = \big(\hat{\bs z}^\top (\bb E \bs J_{i,l}) (\bb E \bs
J_{i,l})^\top \hat{\bs z}\big)\ {\bb E}[({\hat{\bs z}^\top \bs Z^\top
\ail})^2 \bs Z^\top \ail \ail^\top \bs Z] \label{eq:mv1-bis}\\
\bb E(\bs J_{i,l}^\top \bs J_{i,l}) & = \big((\bb E \bs J_{i,l})^\top (\bb
E \bs J_{i,l})\big)\ \bb E[ ({\hat{\bs z}^\top \bs Z^\top \ail})^2
\tr(\bs Z^\top \ail \ail^\top \bs Z)] \label{eq:mv2-bis}.
			\end{align}

			% To handle the fourth-order terms in~\eqref{eq:mv1} and~\eqref{eq:mv2} we will have to invoke again the previous assumptions that granted~\eqref{eq:bound-normu} and~\eqref{eq:bound-normv}. 
                        The variance $v$ is the maximum between $\big\Vert\ts\sum_{i,l}{\bb E}(\bs J_{i,l} - \bb E \bs J_{i,l})(\bs J_{i,l} - \bb E \bs J_{i,l})^\top\big\Vert$ and $\big\Vert\ts\sum_{i,l}{\bb E}(\bs J_{i,l} - \bb E \bs J_{i,l})^\top(\bs J_{i,l} - \bb E \bs J_{i,l})\big\Vert$. Concerning the first term, we have
			\begin{align*}
			\big\Vert\ts\sum_{i,l}{\bb E}(\bs J_{i,l} - \bb E \bs J_{i,l})(\bs J_{i,l} - \bb E \bs J_{i,l})^\top\big\Vert &  \leq \big\Vert\ts\sum_{i,l}{\bb E}(\bs J_{i,l}  \bs J_{i,l}^\top)\big\Vert + \big\Vert\ts\sum_{i,l}(\bb E \bs J_{i,l})(\bb E \bs J_{i,l})^\top\big\Vert. 
			\end{align*}
			A bound is obtained by noting that, from~\eqref{eq:mv3}, 
			\begin{align*}
			\big\Vert\ts\sum_{i,l}(\bb E \bs J_{i,l})(\bb E \bs J_{i,l})^\top\big\Vert & \leq \tfrac{p}{m} (1+\rho)^2 \big\|\tr\big({\bs B} \big(\I_{h}-\hat{{\bs b}}\hat{{\bs b}}^\top\big) {\bs B}^\top\big) \hat{\bs z} \hat{\bs z}^\top \big\| \leq  p h \tfrac{(1+\rho)^2}{m} \lesssim \tfrac{ph}{m}                                                                                                            
			\end{align*}
			since $0 \leq \tr\big({\bs B} \big(\I_{h}-\hat{{\bs b}}\hat{{\bs b}}^\top\big) {\bs B}^\top\big) =  \tr\big(\I_{h}-\hat{\bs b}\hat{\bs b}^\top\big) = h - 1 \leq h$, and, from~\eqref{eq:mv1}, 
                          \begin{align*}
                           \ts \big\Vert\ts\sum_{i,l}{\bb E}(\bs J_{i,l}  \bs
                        J_{i,l}^\top)\big\Vert&\ts = \big\Vert\ts\sum_{i,l} \hat{g}^2_i \bs c^\top_i {\bs B} (\I_{h} - \hat{{\bs b}}\hat{{\bs b}}^\top)  {\bs B}^\top \bs c_i \, \bs Z^\top{\bb E}[({\hat{\bs x}^\top \ail})^2 \ail \ail^\top]\,\bs Z\big\Vert\\
&\ts = \|\bs Z^\top{\bb E}[(\hat{\bs x}^\top \bs a)^2 \bs a \bs a^\top]\,\bs Z\|\ \big(\sum_{i,l} \hat{g}^2_i \bs c^\top_i {\bs B} (\I_{h} - \hat{{\bs b}}\hat{{\bs b}}^\top)  {\bs B}^\top \bs c_i\big). 
                          \end{align*}
                          
                          We have already recalled in~\eqref{eq:expect_axsq_aa} that ${\bb E}[(\hat{\bs x}^\top \bs a)^2 \bs a \bs a^\top] = 2\hat{\bs x}\hat{\bs x}^\top + \bs I_{n} + (\bb E X^4 - 3)\diag(\hat{\bs x})^2$. Thus, we have $\|\bs {\bb E}[(\hat{\bs x}^\top \bs a)^2 \bs a \bs a^\top]\| \leq 3 + |\bb E X^4 - 3| \lesssim \alpha^4$ and 
\begin{align*}
                           \ts \big\Vert\ts\sum_{i,l}{\bb E}(\bs J_{i,l}  \bs
                        J_{i,l}^\top)\big\Vert&\ts \lesssim \sum_{i,l} \tfrac{g^2_i}{\|\bs g\|^2} \bs c^\top_i {\bs B} (\I_{h} - \hat{{\bs b}}\hat{{\bs b}}^\top) \bs B^\top \bs c_i \lesssim \tfrac{ph}{m}, 
                          \end{align*}
where we got the last result from the same developments bounding $\|\ts\sum_{i,l}(\bb E \bs J_{i,l})(\bb E \bs J_{i,l})^\top\|$. Therefore, 
$\|\sum_{i,l}{\bb E}(\bs J_{i,l} - \bb E \bs J_{i,l})(\bs J_{i,l} - \bb E \bs J_{i,l})^\top\|\lesssim \tfrac{ph}{m}$.

			As for the second term in~\eqref{eq:matrix-var-T}, we proceed similarly and see that 
			\begin{align*}
			\big\Vert\ts\sum_{i,l}{\bb E}(\bs J_{i,l} - \bb E \bs J_{i,l})^\top(\bs J_{i,l} - \bb E \bs J_{i,l})\big\Vert \leq \big\Vert\ts\sum_{i,l}{\bb E}(\bs J_{i,l}^\top  \bs J_{i,l})\big\Vert + \big\Vert\ts\sum_{i,l}(\bb E \bs J_{i,l})^\top(\bb E \bs J_{i,l})\big\Vert. 
			\end{align*}
			Note then that, from~\eqref{eq:mv4} and $\sum_{i}\bs c_i   \bs c^\top_i = \bs I_m$, 
                        \begin{equation}
                          \label{eq:temp-res}
                        \big\Vert\ts\sum_{i,l}(\bb E \bs J_{i,l})^\top(\bb E \bs J_{i,l})\big\Vert \leq p \tfrac{(1+\rho)^2}{m} \| (\I_{h} - \hat{{\bs b}}\hat{{\bs b}}^\top) {\bs B}^\top {\bs B} (\I_{h} - \hat{{\bs b}}\hat{{\bs b}}^\top)\| \leq p \tfrac{(1+\rho)^2}{m} \big\Vert \I_{h} - \hat{{\bs b}}\hat{{\bs b}}^\top \big\Vert  \lesssim \tfrac{p}{m}.   
                        \end{equation}
                        Moreover, from~\eqref{eq:mv2} and the fact that $\bs 0 \preccurlyeq \bb E[ ({\hat{\bs x}^\top \bs a})^2 \bs a \bs a^\top ] \preccurlyeq C \alpha^4 \bs I_n$ some~$C > 0$ as established above from the bound on $\|\bb E[ ({\hat{\bs x}^\top \bs a})^2 \bs a \bs a^\top ]\|$, 
                        \begin{align*}
                        \ts\|\sum_{i,l}{\bb E}(\bs J_{i,l}^\top  \bs J_{i,l})\big\|&\ts = \|\sum_{i,l} \hat{g}^2_i (\I_{h} - \hat{{\bs b}}\hat{{\bs b}}^\top) {\bs B}^\top \bs c_i   \bs c^\top_i {\bs B} (\I_{h} - \hat{{\bs b}}\hat{{\bs b}}^\top) \, \bb E[ ({\hat{\bs z}^\top \bs Z^\top \ail})^2 \tr(\bs Z^\top \ail \ail^\top \bs Z)]\|\\
						&\ts = \|\sum_{i,l} \hat{g}^2_i (\I_{h} - \hat{{\bs b}}\hat{{\bs b}}^\top) {\bs B}^\top \bs c_i   \bs c^\top_i {\bs B} (\I_{h} - \hat{{\bs b}}\hat{{\bs b}}^\top)\| \, \bb E[ ({\hat{\bs x}^\top \bs a})^2 \tr(\bs Z^\top \bs a \bs a^\top \bs Z)]\\
						&\ts = \|\sum_{i,l} \hat{g}^2_i (\I_{h} - \hat{{\bs b}}\hat{{\bs b}}^\top) {\bs B}^\top \bs c_i   \bs c^\top_i {\bs B} (\I_{h} - \hat{{\bs b}}\hat{{\bs b}}^\top)\| \, \tr(\bs Z^\top \bb E[ ({\hat{\bs x}^\top \bs a})^2 \bs a \bs a^\top ]\bs Z)\\
						&\ts \leq p \tfrac{(1+\rho^2)}{m}\,\|(\I_{h} - \hat{{\bs b}}\hat{{\bs b}}^\top)\| \, \tr(\bs Z^\top \bb E[ ({\hat{\bs x}^\top \bs a})^2 \bs a \bs a^\top ]\bs Z)\\
						&\ts \lesssim \alpha^4 \tfrac{p}{m} \tr(\bs Z^\top \bs Z) \lesssim \tfrac{pk}{m},
                        \end{align*}
where we used the fact that $0\leq \tr \bs M \leq \tr \bs M'$ if $\bs 0 \preccurlyeq \bs M \preccurlyeq \bs M'$, for two matrices $\bs M$, $\bs M'$.

We can thus conclude that $\big\Vert\ts\sum_{i,l}{\bb E}[(\bs J_{i,l} - \bb E \bs J_{i,l})^\top(\bs J_{i,l} - \bb E \bs J_{i,l})]\big\Vert \lesssim \frac{pk}{m}$. Combining this with the results above yields $v \lesssim \frac{p \max(k,h)}{m} \leq \frac{p (k + h)}{m}$.   

As matrix Bernstein's inequality contains a term in $\log\big(\tfrac{T^2q}{\sqrt v}\big)$
(with $q=mp$ here) in the probability~\eqref{eq:mbi-prob} that must be
upper bounded, we also need to lower bound $v$. This is done by
realising that, by definition of~$v$, using~\eqref{eq:mv1},
\eqref{eq:mv1-bis} and~\eqref{eq:mv3}, there exists a constant $c>0$
such that
\begin{align*}
v&\geq\ts \big\Vert \sum_{i,l}{\bb E}[(\bs J_{i,l} - \bb E \bs J_{i,l})(\bs J_{i,l} - \bb E \bs J_{i,l})^\top]\big\Vert\\
&=\ts \big\Vert \sum_{i,l}{\bb E}[\bs J_{i,l}\bs J_{i,l}^\top] -  (\bb E \bs J_{i,l})(\bb E \bs J_{i,l})^\top]\big\Vert\\
&=\ts \big\Vert \sum_{i,l}\big(\hat{\bs z}^\top (\bb E \bs J_{i,l}) (\bb E \bs
J_{i,l})^\top \hat{\bs z}\big)\ {\bb E}[({\hat{\bs z}^\top \bs Z^\top
\ail})^2 \bs Z^\top \ail \ail^\top \bs Z] -  (\bb E \bs J_{i,l})(\bb E \bs J_{i,l})^\top]\big\Vert\\
&=\ts \big\Vert \sum_{i,l}\hat{g}^2_i \big(\bs c^\top_i {\bs B} \big(\I_{h}-\hat{{\bs b}}\hat{{\bs b}}^\top\big) {\bs B}^\top \bs c_i \big)\ \bs Z^\top \big({\bb E}[({\hat{\bs x}^\top \bs a})^2 \bs a \bs a^\top] -  \hat{\bs x} \hat{\bs x}^\top \big) \bs Z \big\Vert\\
&=\ts p\,\big(\sum_{i}\hat{g}^2_i \,\bs c^\top_i {\bs B} \big(\I_{h}-\hat{{\bs b}}\hat{{\bs b}}^\top\big) {\bs B}^\top \bs c_i\big)\ \big\Vert\bs Z^\top \big({\bb E}[({\hat{\bs x}^\top \bs a})^2 \bs a \bs a^\top] -  \hat{\bs x} \hat{\bs x}^\top \big) \bs Z \big\Vert\\
&\geq \ts\frac{(1-\rho)^2 p}{m} \,{\rm tr}\big({\bs B} \big(\I_{h}-\hat{{\bs b}}\hat{{\bs b}}^\top\big) {\bs B}^\top\big)\ \big\Vert\bs Z^\top \big({\bb E}[({\hat{\bs x}^\top \bs a})^2 \bs a \bs a^\top] -  \hat{\bs x} \hat{\bs x}^\top \big) \bs Z \big\Vert\\
&\geq \ts \frac{(1-\rho)^2p (h-1)}{m}\,\big\Vert\bs Z^\top \big({\bb E}[({\hat{\bs x}^\top \bs a})^2 \bs a \bs a^\top] -  \hat{\bs x} \hat{\bs x}^\top \big) \bs Z \big\Vert\\
&= \ts \frac{(1-\rho)^2p (h-1)}{m}\,\big\Vert\bs Z^\top \big({\bb E}[({\hat{\bs x}^\top \bs a})^2 \bs a \bs a^\top] -  \hat{\bs x} \hat{\bs x}^\top \big) \bs Z \big\Vert \|\hat{\bs z}\|^2\\
&\geq \ts \frac{(1-\rho)^2p (h-1) }{m}\,(\hat{\bs x}^\top \big({\bb E}[({\hat{\bs x}^\top \bs a})^2 \bs a \bs a^\top] -  \hat{\bs x} \hat{\bs x}^\top \big) \hat{\bs x}\\
&= \ts \frac{(1-\rho)^2p (h-1)}{m}\,({\bb E}({\hat{\bs x}^\top \bs a})^4 -
  1) > \frac{c(1-\rho)^2p h}{mn},
\end{align*}
% \hat{g}^2_i
where we used~\eqref{eq:low-bound-fourth-mom}, which holds under the Bernoulli Restriction Hypothesis in Sec.~\ref{def:model}.

Consequently, using this last result and the bound on $T$, we find the crude bound
% $$
% \ts \frac{T^2mp}{v} \lesssim mp \frac{mn}{c(1-\rho)^2ph} \upmu_{\max}^2
% \tfrac{kh}{m^2} \leq mp \frac{mn}{c(1-\rho)^2ph} \frac{m}{h}
% \tfrac{kh}{m^2} =\frac{mnk}{c(1-\rho)^2h} \leq \frac{(mn)^2}{c(1-\rho)^2}.
% $$
$$
\ts \frac{T^2mp}{v} \lesssim mp \frac{mn}{c(1-\rho)^2ph} \tfrac{k^2 + \upmu_{\max}^4 h^2}{m^2} \leq 
mp \frac{mn}{c(1-\rho)^2ph} \tfrac{n^2 + m^2}{m^2} \lesssim (\frac{mn}{1-\rho})^3.
$$

Using matrix Bernstein's inequality in Prop.~\ref{prop:mbi} with the bounds on $T$ and $v$ then yields
\begin{align*}
\ts \Pro \big[\tfrac{1}{p} \big\Vert\ts\sum_{i,l} {\bs J}_{i,l}\big\Vert > \delta]&\ts \leq (k+h) \exp\big(-\frac{c\,\delta^2p^2}{v\, +\, T\log(\frac{T^2mp}{v})\delta }\big)\\
&\ts \leq (k+h) \exp\big(-\frac{c\,\delta^2 mp}{(k +h)\, +\, (k + \upmu_{\max}^2 h)\log(\frac{mn}{c(1-\rho)})\delta }\big)\\
&\ts \leq (k+h) \exp\big(-\frac{c\,\delta^2 mp}{\max(k,\upmu_{\max}^2 h)\log(\frac{mn}{c(1-\rho)})}\big),%\\
%\frac{(mn)^2}{c(1-\rho)^2}\\
% &\ts \leq (k+h) \exp\bigg(-\frac{c'\,\delta^2mp}{\max(k,h)\, +\, \max(k, \upmu_{\max}^2 h)\log(\frac{\max(k, \upmu_{\max}^2 h)^2 m^2}{(1-\rho)^2})\delta }\bigg)\\
% &\ts \leq (k+h) \exp\bigg(-\frac{c''\,\delta^2mp} {\max(k, \upmu_{\max}^2 h) \log(\frac{\max(k, \upmu_{\max}^2 h) m}{(1-\rho)})}\bigg)
\end{align*}
for some $c,c' > 0$, where we used $\delta \in (0,1)$. Consequently, 
$$\Pro[|M''| > 2\delta ] \leq \Pro \big[\tfrac{1}{p} \big\Vert\ts\sum_{i,l} {\bs J}_{i,l}\big\Vert > \delta] \leq (mp)^{-t},
$$ provided 
\begin{equation}
  \label{eq:new-req-on-p}
\ts mp \gtrsim \delta^{-2} t \max(k, \upmu_{\max}^2 h) \log(mp)\log(\frac{mn}{1-\rho}),  
\end{equation}
for some $t \geq 1$.

  % From this, since $mp \gtrsim k + h$ from~\eqref{eq:glob-requirements}, we finally obtain that if
  %                       \begin{equation}
  %                         \label{eq:new-req-on-p}
  %                         \ts mp \gtrsim \delta^{-2} (k + \upmu_{\max}^2 h)\, t^2 \log(mp)^2,                        
  %                       \end{equation}
  %       		then, from~\eqref{eq:tmp-Mtierce}, 
  %                       \begin{equation}
  %                         \label{eq:bound-Mttt}
  %                       \ts \Pro \big[|M'''| > 2\delta\,\big|\,\cl E]\ \leq\ \Pro \big[\tfrac{1}{p} \big\Vert\ts\sum_{i,l}
  %                       {\bs J}_{i,l}\big\Vert > \delta\,\big|\,\cl
  %                       E \big]\ \leq\ (mp)^{-t}.  
  %                       \end{equation}                                            
        	
				\bigskip

                Finally, gathering all the previous bounds, we find by union bound\footnote{Since the probability that~\eqref{eq:determ-relations} do not hold jointly, \ie of the event $\cl E^{\rm c}$ is also bounded by $C' (mp)^{-t}$ for some $C' > 0$, the probability bound of the whole Prop.~\ref{prop:standalone-ahmed} without conditioning is unaltered. In detail, if its occurrence corresponds to an event ${\cl E}'$, then ${\bb P}[{\cl E}'] \geq 1-{\bb P}[({\cl E}')^{\rm c}] \geq 1 - {\bb P}[({\cl E}')^{\rm c}|{\cl E}] - {\bb P}[({\cl E})^{\rm c}] \geq 1 - C (mp)^{-t}$ for another $C > 0$, where ${\bb P}[({\cl E}')^{\rm c}|{\cl E}]$ follows from the requirement \eqref{eq:new-req-on-p}. } that, with probability exceeding $1-C(mp)^{-t}$ for some $C>0$,
                $$  
                |M - \bb E M| \leq |M'-\bb E M'| + |M''-\bb E M''| +|M'''-\bb E M'''| \leq 7\delta 
                $$
				provided all the requirements stated at the beginning of the proof in~\eqref{eq:glob-requirements} and the additional we found in~\eqref{eq:new-req-on-p} hold jointly, \ie if 
				\begin{align*}
					%  \label{eq:glob-requirements-bis}
					\ts n&\ts \gtrsim t \log(mp),\\
					mp&\ts \gtrsim\ t \delta^{-2} \upmu_{\max}^2\,h \log(mp) \log(mn),\\
					\ts mp&\ts \gtrsim \delta^{-2}(k+h) \log(\frac{n}{\delta}),\\
					\ts mp&\ts \gtrsim \delta^{-2} t \max(k, \upmu_{\max}^2 h) \log(mp)\log(\frac{mn}{1-\rho}),  
				\end{align*}
				or more simply if
				\begin{align*}
					%  \label{eq:glob-requirements-bis}
					\ts n&\ts \gtrsim t \log(mp),\\
					\ts mp&\ts \gtrsim \delta^{-2} \max(k,\upmu_{\max}^2 h)\, \max\big(t \log(mp) \log(\frac{mn}{1-\rho}), \log(\frac{n}{\delta})\big).            
				\end{align*}
% 
% However, the last condition is also compatible with
% the fact that then $\exp(-c \delta^2 p) \leq (mp)^{-t}$ since 
% $\ts p \gtrsim \delta^{-2} (1 + \frac{k}{m})\, \max(t^2 \log(mp)^2,
% \log(\frac{n}{\delta})) \geq \delta^{-2} t \log(mp)$ for $t \geq 1$
% and since we certainly have $\log(mp) > 1$. 
% Therefore, we can
% also simplify the probability above as $\bb P[|M - \bb E M| > 12
% \delta] \lesssim (mp)^{-t}$. We have therefore proved the statement of this proposition up to a
% simple rescaling of~$\delta$.
	\end{proof}

    We conclude this section by reporting a basic fact of convex analysis that will be used later on.
	\begin{proposition}[Projections are contractions (adapted from {\cite[Thm.~1.2.1]{Schneider2013}})]
		\label{lemma-ccs}
		Let $\cl C$ be a non-empty, closed, and convex set in $\R^q$. Let
		$$
                {\cl P}_{\cl
                  C} \bs u \coloneqq \arg\min_{\bs u' \in \cl C} \|\bs u - \bs u'\|
                $$ be the orthogonal projection on $\cl C$. Then
		for all $\u,\x \in \R^q$,
		\[
		\|{\cl P}_{\cl C} \u -{\cl P}_{\cl C} \x\|\leq \|\u-\x\|,
		\]
		\ie it is a {\em contraction}.
	\end{proposition}
	
	\section{Proofs on the Geometry of Non-Convex Blind Calibration}
	\label{sec:appB}
	This section provides the proofs for some simple facts appearing in Sec.~\ref{sec:geom-blind}. They are here presented in absence of subspace priors for the sake of simplicity. Similar arguments hold effortlessly with known subspace priors. 
	
	Let us first provide a proof of the bounds we have given on the relation between the naturally-induced pre-metric $\Delta_F$ and the one we actually used in the rest of the paper, $\Delta$.
	
	\begin{proof}[Proof of~\eqref{eq:boundsondelta}]	
		Firstly, note that $\|\bxi{\bs \gamma}^\top - \x{\bs g}^\top\|^2_F = \|(\bxi-\x){\bs \gamma}^\top + \x({\bs \gamma}-{\bs g})^\top\|^2_F = \|{\bs \gamma}\|^2 \|\bxi-\x\|^2 + \|\x\|^2\|{\bs \gamma}-{\bs g}\|^2 - 2(\bxi-\x)^\top \x ({\bs \gamma}-{\bs g})^\top {\bs \gamma}$. If we let $(\bxi,{\bs \gamma}), (\x,{\bs g}) \in {{\cl D}}_{\kappa, \rho}$ we have
		\begin{align*}
		\Delta_F(\bxi,{\bs \gamma}) = \tinv{m}\|\bxi{\bs \gamma}^\top - \x{\bs g}^\top\|^2_F & \leq (1+\rho^2) \|\bxi-\x\|^2 + \tfrac{\|\x\|^2}{m}\|{\bs \gamma}-{\bs g}\|^2 - \tfrac{2}{m}(\bxi-\x)^\top \x ({\bs \gamma}-{\bs g})^\top \eps      \\
		& \leq  (1+\rho^2) \|\bxi-\x\|^2 + \tfrac{\|\x\|^2}{m}\|{\bs \gamma}-{\bs g}\|^2 + \tfrac{2}{\sqrt{m}} \rho \|\bxi-\x\|\|\x\|\|{\bs \gamma}-{\bs g}\|\\
		& \leq (1+\rho+\rho^2) \|\bxi-\x\|^2 + (1+\rho) \tfrac{\|\x\|^2}{m}\|{\bs \gamma}-{\bs g}\|^2 \leq (1+2\rho) \Delta(\bxi,{\bs \gamma}),                
		\end{align*}
		where we used ${\bs \gamma}-{\bs g} = \eps-{\bs e}$, Cauchy-Schwarz and other simple norm inequalities, \ie $\rho \in [0, 1)$, $\|{\bs \gamma}\|\leq \sqrt{m(1+\rho^2)}$, $\|\eps\|\leq\sqrt{m}\rho$, together with the fact that 
		\[
			\tfrac{2}{\sqrt{m}} \rho \|\bxi-\x\|\|\x\|\|{\bs \gamma}-{\bs g}\|\leq \rho \|\bxi-\x\|^2 + \rho \tfrac{\|\x\|^2}{m}\|{\bs \gamma}-{\bs g}\|^2
		\]
		since for any $a,b \in \R, 2 ab \leq a^2 + b^2$.  
		Similarly, we have that
		\begin{align*}
		\Delta_F(\bxi,{\bs \gamma}) & \geq \tfrac{1}{m} m \|\bxi-\x\|^2 + \tfrac{\|\x\|^2}{m}\|{\bs \gamma}-{\bs g}\|^2 - \tfrac{2}{\sqrt{m}} \rho \|\bxi-\x\|\|\x\|\|{\bs \gamma}-{\bs g}\|\\
		& \geq (1-\rho) \|\bxi-\x\|^2 + (1-\rho) \tfrac{\|\x\|^2}{m}\|{\bs \gamma}-{\bs g}\|^2 = (1-\rho) \Delta(\bxi,{\bs \gamma}).                              
		\end{align*}%
	\end{proof}
	%%%%%% 
	
	Then, we start from the definitions in Table~\ref{tab:quants} to show some simple results on the asymptotic geometry of~\eqref{eq:bcp} in absence of priors. In fact, the use of known subspace priors does not change these results, that are asymptotic and hence do not take advantage of the sample complexity reduction provided by assuming such structures. Let us also recall the definition of $\cl V \coloneqq \R^n \times \onep$.
	
	Hereafter, we use a few times the \emph{Hermitian dilation} of matrices
	that is defined by ${\bs H}(\bs A) \coloneqq \big(
	\begin{smallmatrix}
	\bs 0&\bs A\\\bs A^\top&\bs 0  
	\end{smallmatrix}\big)$ for any matrix $\bs A$, with the important fact
	that $\|{\bs H}(\bs A)\|=\|\bs A\|$~\cite[Sec. 2.1.17]{Tropp2015}.
	\begin{proof}[Proof of Prop.~\ref{prop:expstatpoint}]
		Assuming $\x \neq \vzer_n$ and $\bs g \neq \vzer_m$, by setting
		$\Ex\bs\nabla^\perp f(\bxi,{\bs \gamma}) = \vzer_{m+n}$, \ie from the
		expressions found Table~\ref{tab:quants},
		$$\|\bs
		\gamma\|^2 \bs \xi - (\bs \gamma^\top \bs g)\bs x\ =\ \bs 0,\quad
		\text{and}\quad \|\bs \xi\|^2 \bs \varepsilon - (\bs \xi^\top\bs x) \bs
		e\ =\ \bs 0,
		$$
		we easily see that the stationary points in expectation are  
		\begin{align}
		\big\{(\bxi,{\bs \gamma}) \in \R^n \times \Pi_+^m: {\bs \gamma} = \vone_m +\eps,\ \bxi = \tfrac{m + \|\eps\|^2}{m + \eps^\top {\bs e}} \x,\ \eps = \tfrac{\bxi^\top \x}{\|\bxi\|^2} {\bs e} \big\}. 
		\end{align}	
		By setting $\eps = \tau {\bs e}, \tau \coloneqq \frac{\bxi^\top
			\x}{\|\bxi\|^2}$ and replacing it in $\bxi$ we find that
		$\tau$ must respect the equation $\tau = \tfrac{m +
			\tau\|{\bs e}\|^2}{m + \tau^2\|{\bs e}\|^2}$.  However, $\tau = 1$
		is the only feasible solution, \ie $\ts(\bxi,{\bs \gamma}) \equiv (\x,{\bs g})$,
		since assuming $\tau \neq 1$ would lead to $m = -
		\tau (1+\tau)\|{\bs e}\|^2 < 0$. 		
		Evaluated at this point, the expected value of the Hessian reads 
		$$\ts 
		\Ex \H^\perp f(\x,{\bs g}) =  \tfrac{1}{m} \begin{bsmallmatrix} \|{\bs g}\|^2 \I_n & \x {\bs e}^\top \\ 
		{\bs e} \x^\top & \|\x\|^2 \Ponep \end{bsmallmatrix}.
		$$
		Thus, letting $\bLambda \coloneqq \begin{bsmallmatrix} \|{\bs g}\|^2 \I_n & \vzer_{n\times m} \\ \vzer_{m \times n} & \|\x\|^2 \I_m
		\end{bsmallmatrix}$, we observe that
		\begin{align}
		\ts  m\bLambda^{-\frac{1}{2}} \Ex\H^\perp f(\x,{\bs g})
		\bLambda^{-\frac{1}{2}} = \ts \begin{bsmallmatrix}\I_n                                                                                                      & \vzer_{n\times m} \\ \vzer_{m \times n} &  \Ponep
		\end{bsmallmatrix} + \tfrac{1}{\|\x\|\|{\bs g}\|} {\bs H}(\x {\bs e}^\top) \succeq \big(1-\ts\frac{\rho}{\sqrt{1+\rho^2}}\big) \begin{bsmallmatrix}\I_n & \vzer_{n\times m} \\ \vzer_{m \times n} &  \Ponep
		\end{bsmallmatrix},
		\end{align}
		since $\|{\bs H}(\bs x\bs e^\top)\| = \|\bs x\bs e^\top\| =
		\|\bs x\|\|\bs e\| = \tfrac{\|\bs e\|}{\|\bs g\|} \|\bs x\|
		\|\bs g\|$, which can bounded by $\tfrac{\rho}{\sqrt{1+\rho^2}} \|\bs x\|\|\bs
		d\|$ since $\|{\bs g}\|= \sqrt{m+\|{\bs e}\|^2}$, $\|{\bs e}\| <
		\sqrt{m}{\rho}$ and $\tfrac{\|{\bs e}\|}{\sqrt{m+\|{\bs e}\|^2}}$ is an increasing
		function of $\|{\bs e}\|$. Thus, we have shown that for all $\rho \in [0,1), \Ex
		\H^\perp f(\x,{\bs g})\succ 0$ on~$\cl V = \bb R^n \times 1^\perp$, \ie $\Ex
		\H^\perp f(\x,{\bs g})\succ_{\cl V} 0$.
	\end{proof}
	
	%%%%%%%%%%%%%%%%%%%%%%% 
	% New proof %
	\begin{proof}[Proof of Prop.~\ref{prop:expconvexity}]
		We now define $\bLambda \coloneqq \begin{bsmallmatrix}\|{\bs \gamma}\|^2
		\I_n & \vzer_{n\times m} \\ \vzer_{m \times n} & \|\bxi\|^2
		\I_m \end{bsmallmatrix}$ and we observe that, by 
		the expression of the Hessian (see Table~\ref{tab:quants}),
		\begin{align}
		\label{eq:proof-expconvex-tmp1}
		\ts      m\bLambda^{-\frac{1}{2}} \Ex\big[\H^\perp f(\bxi,{\bs \gamma})\big]
		\bLambda^{-\frac{1}{2}} = 
		\begin{bsmallmatrix}\I_n
		& \vzer_{n \times m} \\ \vzer_{m \times n}&  \Ponep 
		\end{bsmallmatrix} + \tfrac{1}{\|\bxi\|\|{\bs \gamma}\|} {\bs H}(\bs \xi
		\bs \varepsilon^\top)
		+
		\tfrac{1}{\|\bxi\|\|{\bs \gamma}\|}
		{\bs H}(\bxi \eps^\top - \x {\bs e}^\top).
		\end{align}
		Similarly to the proof of Prop.~\ref{prop:expstatpoint}, it is easy to
		show that 
		$$\ts
		\begin{bsmallmatrix}\I_n
		& \vzer_{n \times m}\\ \vzer_{m \times n}&  \Ponep 
		\end{bsmallmatrix} + \tfrac{1}{\|\bxi\|\|{\bs \gamma}\|} {\bs H}(\bs \xi
		\bs \varepsilon^\top) \succeq \big(1-\ts\tfrac{\rho}{\sqrt{1+\rho^2}}\big) \begin{bsmallmatrix}\I_n & \vzer_{n\times m} \\ \vzer_{m \times n} &  \Ponep
		\end{bsmallmatrix},
		$$
		since $\|\bs \varepsilon\| \leq \sqrt m \rho$, $\|{\bs \gamma}\|=\sqrt{m + \|\bs
			\varepsilon\|^2}$ and $\|{\bs H}(\bs \xi \bs \varepsilon^\top)\| = \|\bs \xi \bs \varepsilon^\top\| \leq \|\bs \xi\| \|\bs \varepsilon\|
		\leq \tfrac{\rho}{\sqrt{1+ \rho^2}} \|\bs \xi\| \|{\bs \gamma}\| $.
		
		Moreover, by restricting the application of ${\bs H}(\bxi \eps^\top -
		\x {\bs e}^\top)$ in the \rhs of
		\eqref{eq:proof-expconvex-tmp1} to the space $\cl V = \bb R^n \times
		\bs 1^\perp_m$, we have for all $\bs u \in \cl V$
		$$\ts
		\bs u^\top {\bs H}(\bxi \eps^\top - \x {\bs e}^\top) \bs u = 
		\bs u^\top {\bs H}(\bxi \bs \gamma^\top - \x \bs g^\top) \bs u
		\leq \|\bs u\|^2 \|\bxi \bs \gamma^\top - \x \bs g^\top\| \leq \|\bs u\|^2 \|\bxi \bs \gamma^\top - \x \bs g^\top\|_F.
		$$ 
		Since $(\bs \xi,\bs \gamma) \in {\cl D}_{\kappa,\rho}$,
		\eqref{eq:boundsondelta} provides $\|\bxi \bs \gamma^\top - \x \bs
		g^\top\|_F = \Delta_F(\bs \xi,\bs \gamma)^{\ts\frac{1}{2}} \leq \sqrt{1+ 2\rho}\,
		\kappa \|\bs x\|$, so that 
		$$\ts
		\tfrac{1}{\|\bxi\|\|{\bs \gamma}\|} \bs u^\top {\bs H}(\bxi \eps^\top - \x {\bs e}^\top) \bs u
		\leq \tfrac{1}{\|\bxi\|\|{\bs \gamma}\|} \|\bs u\|^2 \sqrt{1+ 2\rho}\,
		\kappa \|\bs x\|
		\leq \tfrac{\kappa \sqrt{1+ 2\rho}}{(1-\kappa) \sqrt m} \|\bs u\|^2, 
		$$
		since by the norm bounds on $(\bxi,{\bs \gamma}) \in {\cl D}_{\kappa,\rho}$ we
		have $\vert\|\bxi\|-\|\x\|\vert\leq \|\bxi-\x\| \leq \kappa \|\x\|$,
		\ie $(1-\kappa) \|\x\| < \|\bxi\|< (1+\kappa)\|\x\|$, and
		$\|{\bs \gamma}\| \geq \sqrt m$. 
		Gathering this bound with the previous developments, we get  
		\[
		m\bLambda^{-\frac{1}{2}} \Ex\big[\H^\perp f(\bxi,{\bs \gamma})\big]
		\bLambda^{-\frac{1}{2}} \succeq_{\cl V} 
		\big(1-\ts\tfrac{\rho}{\sqrt{1+\rho^2}} - \tfrac{\kappa \sqrt{1+ 2\rho}}{(1-\kappa) \sqrt m}\big) \begin{bsmallmatrix}\I_n & \vzer_{n\times m} \\ \vzer_{m \times n} &  \Ponep \end{bsmallmatrix}. 
		\]
		This proves that the expected value of the Hessian is positive
		definite provided $\rho < 1 - \tfrac{\sqrt 3\,\kappa}{(1-\kappa)\sqrt m}$.
	\end{proof}

	\section{Proofs on the Convergence Guarantees of the Descent Algorithms}
	\label{sec:appC}
	
	As Alg.~\ref{alg1} is a special case of  Alg.~\ref{alg2}, this section focuses only on proving in order 
        Prop.~\ref{prop:initsub}, Prop.~\ref{prop:regusub}, and Thm.~\ref{theorem:convergence-subs}. Thus, we begin by proving how close the initialisation point $(\bs\zeta_0, \bs	\beta_0)$ used in Alg.~\ref{alg2} is with respect to the global minimiser $(\bs z,\bs b)$ that also corresponds to $(\bs x = \bs Z \bs z, \bs g =  \bs B \bs b)$. 
	\begin{proof}[Proof of Prop.~\ref{prop:initsub}]
		Recall that $\bs x = \bs Z \bs z \in \cl Z$ and ${\bs\zeta}_0 \coloneqq \tfrac{1}{m p} \ts\sum_{l}  g_i \bs Z^\top \ail \ail^\top \bs Z \bs z$, where $g_i = \bs c_i^\top \bs B \bs b = 1 + \bs c_i^\top \bs B^\perp \bs b^\perp$. Since $\bs 1_m^\top \, \bs g = m$ by construction,
		% ${\bs B}_{i,\cdot}$ being the $i$-th row of $\bs B$ 
		\begin{align*}
			\ts \|\bs\zeta_0 - \bs z\| = \big\|\tfrac{1}{mp}\ts\sum_{i,l} g_i 
			\bs Z^\top (\ail \ail^\top - \I_n) \bs Z \bs z \big \| \leq           
			\big\|\tfrac{1}{mp}\ts\sum_{i,l} g_i \bs Z^\top(\ail \ail^\top-         
			\I_n)\bs Z \big\| \| \bs z \|.                                   
		\end{align*}
		We then notice that
		\[
			\big\|\tfrac{1}{mp}\ts\sum_{i,l} g_i \bs Z^\top (\ail \ail^\top -
			\I_n)\bs Z \big\|
			\leq
			\big\|\tfrac{1}{mp}\ts\sum_{i,l} (g_i-1) \bs Z^\top (\ail \ail^\top -
			\I_n)\bs Z \big\|
			+ 
			\big\|\tfrac{1}{mp}\ts\sum_{i,l} \bs Z^\top (\ail \ail^\top -
			\I_n)\bs Z \big\|,
		\]
		where $\bs g \in {\cl B}_\rho$ by Def.~\ref{def:blind-calibr-with-sub}, and consequently $|g_i - 1|\leq \rho$. Thus,
		fixing $\delta' >0$ and $t \geq 1$, 
		provided $n \gtrsim t \log(mp)$ and 
		\[
			\ts mp \gtrsim
			(\delta')^{-2}(k+h)\log(\frac{n}{\delta'}),
		\]
		Cor.~\ref{coro:appl-wcovsub} gives, with probability exceeding $1 - C \exp(- c (\delta')^2 m p) - (mp)^{-t}$
		and some values $C,c >0$ depending only on  $\alpha$,  
		\[
			\big\|\tfrac{1}{mp}\ts\sum_{i,l} (g_i-1) \bs Z^\top (\ail \ail^\top -
			\I_n)\bs Z \big\|
			+ 
			\big\|\tfrac{1}{mp}\ts\sum_{i,l} \bs Z^\top (\ail \ail^\top -
			\I_n)\bs Z \big\|
			\leq (\rho + 1)\delta'.
		\]
		Setting $\delta'\coloneqq \ts\tfrac{\delta}{2}$ so that for $\rho <1$, $(\rho + 1)\delta' \leq \delta$ yields
		the statement of this proposition. 
		As for our choice of $\bs\beta_0\coloneqq\begin{bsmallmatrix}
		\sqrt m \\\vzer_{h-1}
		\end{bsmallmatrix}$ (that is identical to ${\bs \gamma}_0 \coloneqq \vone_m$), by definition $\|{\bs\beta}_0 - {\bs b}\| = \|\bs g-\vone_m\|=\|\bs e\| \leq \sqrt{m} \rho$. Thus, we have that $\Delta(\bs\zeta_0,{\bs\beta}_0) \leq (\delta^2+\rho^2)\|\bs z\|^2$, so $(\bs\zeta_0,{\bs\beta}_0) \in {\cl D}^{\rm s}_{\kappa,\rho}$ for $\kappa = \sqrt{\delta^2+\rho^2}$.
	\end{proof}
	Secondly, we introduce a definition that studies the projected gradient $\bs\nabla^\perp f^{\rm s} (\bs\zeta,{\bs\beta})$ on ${\cl D}^{\rm s}_{\kappa,\rho}$.
	\begin{definition}[Regularity condition in a $(\kappa,\rho)$-neighbourhood, adapted from {\cite[Condition 7.9]{CandesLiSoltanolkotabi2015}}]
		\label{def:regularity}
		We say that \new{$\bs\nabla^\perp f^{\rm s}(\bs\zeta, \bs\beta) = (\bs\nabla_{\bs \zeta} f^{\rm s}(\bs\zeta, \bs\beta)^\top, \bs\nabla^\perp_{\bs \beta} f^{\rm s}(\bs\zeta, \bs\beta)^\top)^\top\!$} verifies
		the {\em regularity condition} on ${\cl D}^{\rm s}_{\kappa,\rho}$ with
		constants $\beta_C,\beta_{L,\bs \zeta},\beta_{L,\bs \beta} > 0$ if, for all $(\bs\zeta, \bs\beta) \in {\cl D}^{\rm s}_{\kappa,\rho}$,
		\begin{align}
		\label{eq:req1}	\big\langle \bs\nabla^\perp f^{\rm s} (\bs\zeta, \bs\beta), \big[\begin{smallmatrix} \bs\zeta-\bs z \\ 
		{\bs\beta}-{\bs b} \end{smallmatrix}\big] \big\rangle & \geq \beta_C {\Delta}(\bs\zeta, \bs\beta),
                \end{align}\new{\vspace{-7mm}
		\begin{align}
		\label{eq:req2-zta} 
                  \|\bs\nabla_{\bs \zeta} f^{\rm s}(\bs\zeta, \bs\beta)\|^2 & \leq \beta_{L,\bs \zeta} {\Delta}(\bs\zeta, \bs\beta),\\
                  \|\bs\nabla^\perp_{\bs \beta} f^{\rm s}(\bs\zeta, \bs\beta)\|^2 & \leq \beta_{L,\bs \beta} {\Delta}(\bs\zeta, \bs\beta). 
		\end{align}\vspace{-7mm}}
	\end{definition}
	The role of Def.~\ref{def:regularity} is immediately understood by taking the gradient descent update from any point $(\bs\zeta,{\bs\beta}) \in {\cl D}^{\rm s}_{\kappa,\rho}$ to some $(\bs\zeta_+,\check{{\bs\beta}}_+)$ through the formula $\bs\zeta_+ \coloneqq \bs\zeta - \mu_{\bs\zeta} \bs\nabla_{\bs\zeta} f^{\rm s}(\bs\zeta,{\bs\beta})$ and $\check{{\bs\beta}}_+ \coloneqq {\bs\beta} - \mu_{\bs\beta} \bs\nabla^\perp_{\bs\beta} f^{\rm s}(\bs\zeta,{\bs\beta})$. It is then simply seen that the distance with respect to the global minimiser,
	\begin{align}
	{\Delta} (\bs\zeta_+,\check{{\bs\beta}}_+) & = \|\bs\zeta-\bs z\|^2 - 2 \mu_{\bs\zeta}\langle \bs\nabla_{\bs\zeta} f^{\rm s}(\bs\zeta,{\bs\beta}), \bs\zeta-\bs z \rangle + \mu^2_{\bs\zeta}\|\bs\nabla_{\bs\zeta} f^{\rm s}(\bs\zeta,{\bs\beta})\|^2 \nonumber                                                                                                               \\ 
	& \quad +\tfrac{\|\bs z\|^2}{m}\|\bs\beta-{\bs b}\|^2 - 2 \tfrac{\|\bs z\|^2}{m} \mu_{\bs\beta} \langle \bs\nabla^\perp_{\bs\beta} f^{\rm s}(\bs\zeta,{\bs\beta}), \bs\beta-{\bs b} \rangle + \tfrac{\|\bs z\|^2}{m} \mu^2_{\bs\beta} \| \bs\nabla^\perp_{\bs\beta} f^{\rm s}(\bs\zeta,{\bs\beta})\|^2 \nonumber \\
	& = \Delta(\bs\zeta,{\bs\beta}) -2 \mu \big(\langle \bs\nabla_{\bs\zeta}f^{\rm s}(\bs\zeta,{\bs\beta}), \bs\zeta-\bs z\rangle + \langle \bs\nabla^\perp_{\bs\beta} f^{\rm s}(\bs\zeta,\bs\beta), {\bs\beta}-{\bs g}\rangle\big) \nonumber                                                                   \\
	& \quad + {\mu^2} \big(\|\bs\nabla_{\bs\zeta} f^{\rm s}(\bs\zeta,\bs\beta)\|^2 + \tfrac{m}{\|\bs z\|^2} \|\bs\nabla^\perp_{\bs\beta} f^{\rm s}(\bs\zeta,\bs\beta)\|^2\big) \nonumber  \\
\label{eq:updatedist}&\new{\leq \ts \big(1-2 \mu \beta_C +  \mu^2(\beta_{L,\bs \zeta} + \frac{m \beta_{L,\bs \beta }}{\|\bs z\|^2})\big)\Delta(\bs\zeta,{\bs\beta})}
	% & \leq \ts \left(1-2 \mu \beta_C + \tfrac{\mu^2}{\min\big(1, \frac{\|\bs z\|^2}{m}\big)} \beta_L\right)\Delta(\bs\zeta,{\bs\beta})\label{eq:updatedist}
	\end{align}		
	where we have let $\mu_{\bs\zeta} \coloneqq \mu$, $\mu_{\bs\beta} \coloneqq \mu \tfrac{m}{\|\bs z\|^2}$ for some $\mu > 0$, and where the last line holds only if \new{the constants $\beta_C$, $\beta_{L,\bs \zeta}$, $\beta_{L,\bs \beta} >0$} are found for which Def.~\ref{def:regularity} applies. To highlight the values of these constants in~\eqref{eq:updatedist} we introduced Prop.~\ref{prop:regusub} which simply turns Def.~\ref{def:regularity} in a probabilistic statement, so the event that Def.~\ref{def:regularity} holds is a property of the neighbourhood ${\cl D}^{\rm s}_{\kappa,\rho}$, \ie verified uniformly on it with very high probability. Below, we provide a proof of this regularity condition.
	
	\begin{proof}[Proof of Prop.~\ref{prop:regusub}]
		 This proof develops the regularity condition in Def.~\ref{def:regularity} to highlight the values of \new{$\beta_C$, $\beta_{L,\bs \zeta}, \beta_{L,\bs \beta} > 0$ (\ie $\eta$, $L^2_{\bs \zeta}$ and $L^2_{\bs \beta}$ in Prop.~\ref{prop:regusub}, respectively)}; finding the bounds will require the application of the tools developed in Sec.~\ref{sec:appA}, as well as some simple considerations on ${\cl D}^{\rm s}_{\kappa,\rho}$ with $\kappa > 0$ and $\rho \in [0, 1)$. 

The following developments are built on the assumption that the events given by Corollary~\ref{coro:appl-wcovsub},
                Prop.~\ref{prop:m-w-covconc}
                and Prop.~\ref{prop:standalone-ahmed} are all occurring jointly. 
By union bound and combining all the related requirements, this happens with probability larger than $1 - C 
                (mp)^{-t}$ for some $C>0$ and $t
                \geq 1$, provided 
                \begin{align*}
                  \ts n&\ts \gtrsim t \log(mp),\\
                  \ts mp&\ts \gtrsim \delta^{-2} \max(k,\upmu_{\max}^2 h)\, \max\big(t
                          \log(mp) \log(\frac{mn}{1-\rho}), \log(\frac{n}{\delta})\big),            
                \end{align*}
                or simply, since we additionally and implicitly assume $\rho < \tfrac{1}{31}$ in the statement of Prop.~\ref{prop:regusub},
                \begin{align*}
                  \ts n&\ts \gtrsim t \log(mp),\\
                  \ts mp&\ts \gtrsim t\delta^{-2} (k + \upmu_{\max}^2 h)\, 
                          \log(m(p+n))^2 \log(\frac{1}{\delta}).
                \end{align*}

		\noindent {\em Bounded Curvature:} \ 
		Let us recall the definitions $\bs g = {\bs B} {\bs b}$ and $\bs \gamma = {\bs B} {\bs\beta}$, ${\bs\beta} \coloneqq \begin{bsmallmatrix} \sqrt{m} \\ \bs\beta^\perp \end{bsmallmatrix} \in \R^{h}$, which shall be convenient to keep a more compact notation. Note also that $\|{\bs\beta}\| = \|{\bs\gamma}\|$ and $\|{\bs b}\| = \|{\bs g}\|$. Then, we start this proof by developing the \lhs of~\eqref{eq:req1} using the gradients obtained\footnote{Note that the components $g_i = \bs B_{i,\cdot} \bs b$, $ \gamma_i = \bs B_{i,\cdot} \bs\beta$ are simply the inner products between the $i$-th row of $\bs B$ and the respective $h$-dimensional vector in the subspace model.} in Sec.~\ref{sec:blind-calibration-with-subspace-priors} as follows. Since $\bs c^\top_i \bs B = \bs B_{i,\cdot}$ (the $i$-th row vector of the tight frame~$\bs B$), 
		\begin{align}
			P & \coloneqq \langle \bs\nabla_{\bs\zeta} f^{\rm s}(\bs\zeta, {\bs\beta}), \bs\zeta - \bs z\rangle + \langle \bs\nabla^\perp_{{\bs\beta}} f^{\rm s}(\bs\zeta, \bs\beta), \bs\beta - {\bs b}\rangle \nonumber \\                      
				& = \tfrac{1}{mp} \ts\sum_{i,l} (\bs B_{i,\cdot} \bs\beta)^2 (\bs Z \bs\zeta)^\top \ail \ail^\top (\bs Z (\bs\zeta - \bs z)) - (\bs B_{i,\cdot} \bs\beta) (\bs B_{i,\cdot} \bs b)  (\bs Z \bs z)^\top \ail \ail^\top (\bs Z (\bs\zeta - \bs z)) \nonumber \\
				& \quad + \big[(\bs B_{i,\cdot} \bs\beta)  (\bs Z \bs\zeta)^\top \ail \ail^\top (\bs Z \bs\zeta) - (\bs B_{i,\cdot} \bs b) (\bs Z \bs\zeta)^\top \ail \ail^\top (\bs Z \bs z)\big] \bs B_{i,\cdot} (\bs\beta- \bs b) \nonumber \\
				& = \tfrac{1}{mp} \ts\sum_{i,l} (\bs B_{i,\cdot} \bs\beta)^2  (\bs\zeta-\bs z)^\top \bs Z^\top \ail \ail^\top \bs Z (\bs\zeta - \bs z) + (\bs B_{i,\cdot} \bs\beta) \bs B_{i,\cdot} (\bs\beta- \bs b) \bs z^\top \bs Z^\top \ail \ail^\top \bs Z (\bs\zeta-\bs z) \nonumber \\
				& \quad + \tfrac{1}{mp} \ts\sum_{i,l} \big[ (\bs B_{i,\cdot} \bs\beta) (\bs\zeta-\bs z)^\top \bs Z^\top \ail \ail^\top \bs Z (\bs\zeta - \bs z) + \bs B_{i,\cdot} (\bs\beta - \bs b) \bs z^\top \bs Z^\top \ail \ail^\top \bs Z \bs z\big] \bs B_{i,\cdot} (\bs\beta- \bs b) \nonumber \\
				& \quad + \tfrac{1}{mp} \ts\sum_{i,l}  \big[(2  \bs B_{i,\cdot} \bs\beta  - \bs B_{i,\cdot} \bs b) (\bs\zeta-\bs z)^\top \bs Z^\top \ail \ail^\top \bs Z  \bs z\big] \bs B_{i,\cdot} (\bs\beta- \bs b) \nonumber \\
				& = \tfrac{1}{mp} \ts\sum_{i,l} \big[(\bs B_{i,\cdot} \bs\beta)^2 + (\bs B_{i,\cdot} \bs\beta) \bs B_{i,\cdot} (\bs\beta - \bs b)\big]  (\bs\zeta-\bs z)^\top \bs Z^\top \ail \ail^\top \bs Z (\bs\zeta - \bs z) \nonumber \\
				& \quad + \tfrac{1}{mp} \ts\sum_{i,l}  (2  \bs B_{i,\cdot} \bs\beta + \bs B_{i,\cdot} (\bs\beta - \bs b)) \bs B_{i,\cdot} (\bs\beta - \bs b)  (\bs\zeta-\bs z)^\top \bs Z^\top \ail \ail^\top \bs Z  \bs z \nonumber \\
				& \quad + \tfrac{1}{mp} \ts\sum_{i,l}  (\bs B_{i,\cdot}(\bs\beta -\bs b))^2 \bs z^\top \bs Z^\top \ail \ail^\top \bs Z \bs z. 
		\end{align}
		Let us then add and subtract the term 
		\[
			P' \coloneqq \tfrac{1}{mp} \ts\sum_{i,l} (\bs B_{i,\cdot} \bs b (\bs\zeta-\bs z) + \bs B_{i,\cdot}(\bs\beta - \bs b) \bs z)^\top \bs Z^\top \ail \ail^\top \bs Z (\bs B_{i,\cdot} \bs b (\bs\zeta-\bs z) + \bs B_{i,\cdot}(\bs\beta - \bs b) \bs z),
		\]
		so that $P = P' + P''$, where we have collected
		\begin{align*}
		P'' \coloneqq P-P' & = \underbrace{\tfrac{1}{mp}  \ts\sum_{i,l}  \bs B_{i,\cdot}(\bs\beta - \bs b) [\bs B_{i,\cdot} \bs\beta  +\bs B_{i,\cdot} (\bs\beta + \bs b)]  (\bs\zeta-\bs z)^\top \bs Z^\top \ail \ail^\top \bs Z (\bs\zeta-\bs z)}_{Q'} \\ 
		& \quad + \underbrace{\tfrac{1}{mp}  \ts\sum_{i,l}  3(\bs B_{i,\cdot}(\bs\beta - \bs b))^2   (\bs\zeta-\bs z)^\top \bs Z^\top \ail \ail^\top \bs Z \bs z}_{Q''}.
		\end{align*}
		
		Let us first focus on finding a bound for the simpler $P''$; since this term may take negative values, we need an upper bound for $|P''|$ so that $P \geq P' - |P''|$. Having highlighted $Q',Q''$ so that $P'' = Q' + Q''$, clearly  $|P''| \leq |Q'| + |Q''|$. An upper bound for $|Q'|$ is easily found as follows; noting that the first component of $\bs\beta - \bs b$ is always $0$, it is observed that $\max_{i\in [m]} | \bs B_{i,\cdot}(\bs\beta - \bs b)| < 2\rho$ and that $\max_{i\in [m]} |  2 \bs B_{i,\cdot} \bs\beta +\bs B_{i,\cdot} \bs b| < 3(1+\rho)$. Thus, since $\rho \in [0,1)$,
		\[
		\max_{i\in [m]} | \bs B_{i,\cdot}(\bs\beta - \bs b) [\bs B_{i,\cdot} \bs\beta +\bs B_{i,\cdot} (\bs\beta + \bs b)]|\leq 2\rho\cdot 3(1+\rho) \leq 12 \rho,
		\]
		so we obtain by~Cor.~\ref{coro:appl-wcovsub} (\ie by~\eqref{eq:coro-ineq-wcov})
		\[
			|Q'| \leq  \big\| \tfrac{1}{mp} \ts\sum_{i,l}  \bs B_{i,\cdot}(\bs\beta - \bs b) [\bs B_{i,\cdot} \bs\beta+\bs B_{i,\cdot} (\bs\beta + \bs b)] \bs Z^\top \ail \ail^\top \bs Z\big\| \|\bs\zeta-\bs z\|^2 \leq 12 \rho (1+\delta) \|\bs\zeta-\bs z\|^2,
		\]
		which is simply $|Q'| \leq 24 \rho \|\bs\zeta-\bs
                z\|^2$ when $\delta \in (0,1)$. As for $Q''$, by a
                straightforward application of the Cauchy-Schwarz
                inequality we find that 
  \begin{align*}
    |Q''|&\leq 3 \big\vert \tinv{mp} \ts\sum_{i,l} (\bs B_{i,\cdot}
           (\bs\beta - \bs b))^4 (\bs z^\top \bs Z^\top \ail)^2
           \big\vert^{\frac{1}{2}}\big\vert \tinv{mp} \ts\sum_{i,l}
           ((\bs\zeta-\bs z)^\top \bs Z^\top \ail)^2
           \big\vert^{\frac{1}{2}}\\
&\leq 6\rho\,\big\vert \tinv{mp} \ts\sum_{i,l} (\bs B_{i,\cdot}
           (\bs\beta - \bs b))^2 (\bs z^\top \bs Z^\top \ail \ail^\top
  \bs Z \bs z)
           \big\vert^{\frac{1}{2}}\big\vert \tinv{mp} \ts\sum_{i,l}
           ((\bs\zeta-\bs z)^\top \bs Z^\top \ail)^2
           \big\vert^{\frac{1}{2}}.
  \end{align*}
  where we used the fact that $\max_i (\bs B_{i,\cdot} (\bs\beta-\bs b))^2 < 4\rho^2$. Moreover, from Prop.~\ref{prop:m-w-covconc},
\begin{align*}
&\big\vert \tinv{p} \ts\sum_{i,l} (\bs B_{i,\cdot}
           (\bs\beta - \bs b))^2 (\bs z^\top \bs Z^\top \ail \ail^\top
  \bs Z \bs z)
           \big\vert\\
&\leq  \|\bs\beta - \bs b\|^2  \|\bs z\|^2 \big\|\tinv{p} \ts\sum_{i,l} (\bs B^\top
  \bs c_i \bs c_i^\top \bs B) (\hat{\bs z}^\top \bs Z^\top \ail \ail^\top
  \bs Z \hat{\bs z})\big\|\\
&\leq (1+\delta) \|\bs\beta - \bs b\|^2 \|\bs z\|^2,
\end{align*}
and Cor.~\ref{coro:appl-wcovsub} provides
$$
\big\vert \tinv{mp} \ts\sum_{i,l}
           ((\bs\zeta-\bs z)^\top \bs Z^\top \ail)^2
           \big\vert \leq \|\bs\zeta-\bs z\|^2 \big\|\tinv{mp}\sum_{i,l}
           \bs Z^\top \ail \ail^\top \bs Z\big\| \leq (1+\delta) \|\bs\zeta-\bs z\|^2.
$$
Therefore, 
$$
|Q''|\ \leq\ 12\rho\,\tfrac{\|\bs\beta - \bs b\|}{\sqrt m} \|\bs z\|\|\bs\zeta-\bs z\|.
$$

		% \[
		% 	|Q''| \leq 3 \big\vert \tinv{mp} \ts\sum_{i,l} (\bs B_{i,\cdot} (\bs\beta - \bs b))^4 (\bs z^\top \bs Z^\top \ail)^2 \big\vert^{\frac{1}{2}}\big\vert \tinv{mp} \ts\sum_{i,l}  ((\bs\zeta-\bs z)^\top \bs Z^\top \ail)^2 \big\vert^{\frac{1}{2}}.
		% \]
		% Then we recall Prop.~\ref{prop:conc-ineq-sump}
                % \mtodo{Change this by using
                %   Prop.~\ref{prop:m-w-covconc} above.} with $\bs x \coloneqq \bs Z \bs z$ and Cor.~\ref{coro:appl-wcovsub} to upper bound the first and second factor respectively, yielding 		
		% %		\[
		% %			|Q''| \leq 3 \sqrt{2} \big\vert \tfrac{(2\rho)^2}{m} \ts\sum_{i}(\bs B_{i,\cdot}(\bs\beta - \bs b))^2 \|\bs z\|^2 \big\vert^{\frac{1}{2}}\sqrt{1+\delta} \|\bs\zeta -\bs z \| \leq 12 \rho \tfrac{\|\bs\beta-\bs b\|}{\sqrt m} \|\bs z\| \|\bs\zeta -\bs z \|.
		% %		\]
		% \[
		% |Q''| \leq 3 \sqrt{2} \big\vert \tinv{m} \ts\sum_{i}(\bs B_{i,\cdot}(\bs\beta - \bs b))^4 \|\bs z\|^2 \big\vert^{\frac{1}{2}}\sqrt{1+\delta} \|\bs\zeta -\bs z \| \leq 12 \rho \tfrac{\|\bs\beta-\bs b\|}{\sqrt m} \|\bs z\| \|\bs\zeta -\bs z \|,
		% \]
		% where the last inequality is obtained since \mtodo{improve} $$\big\vert
                % \ts\sum_{i}(\bs B_{i,\cdot}(\bs\beta - \bs b))^4
                % (\bs z^\top \bs Z^\top \ail)^2 \big\vert \leq \big\vert {(2\rho)^2} \ts\sum_{i}(\bs B_{i,\cdot}(\bs\beta - \bs b))^2 (\bs z^\top \bs Z^\top \ail)^2\big\vert.$$ 
		Hence, 
		%		\begin{align*}
		%			|P''| & \leq 12\rho (1+\delta)\|\bs\zeta-\bs z \|^2 +  6 \rho \tfrac{ \|\bs z \|}{m}\|\bs\beta-\bs b\|^2 + 6\rho \|\bs\zeta-\bs z\|^2 \\
		%			& = [12(1+\delta) + 6]\rho\|\bs\zeta-\bs z\|^2 +  6 \rho \tfrac{ \|\bs z\|}{m}\|\bs\beta-\bs b\|^2 \\
		%			& \leq 18\rho(1+\delta)\Delta(\bs\zeta,\bs\beta).                                                             
		%		\end{align*}
		\begin{align*}
			|P''| & \leq 24 \rho \|\bs\zeta-\bs z \|^2 +  6 \rho \tfrac{ \|\bs z \|}{m}\|\bs\beta-\bs b\|^2 + 6\rho \|\bs\zeta-\bs z\|^2 \\
			& = 30 \rho\|\bs\zeta-\bs z\|^2 +  6 \rho \tfrac{ \|\bs z\|}{m}\|\bs\beta-\bs b\|^2 \\
			& \leq 30 \rho \Delta(\bs\zeta,\bs\beta).                                                             
		\end{align*}
		
		We now move our focus to the slightly more cumbersome
                task of finding a lower bound for $P'$. Note that we
                can rewrite 
                \[ 
                  P' =  \tfrac{1}{mp} \ts\sum_{i,l} \tr\big(\bs c^\top_i {\bs B}\big[{\bs b} (\bs\zeta-\bs z)^\top + ({\bs\beta} - {\bs b}) \bs z^\top\big] \bs Z^\top \ail\big)^2
                  = \tfrac{1}{mp} \ts\sum_{i,l} \langle \bs Z^\top \ail \bs c^\top_i {\bs B}, \bs M\rangle^2_F,
		\]
                with the matrix $\bs M \coloneqq (\bs\zeta-\bs
                z)\|{\bs g}\|\hat{{\bs b}}^\top + \hat{\bs z} \|{\bs
                  z}\| ({\bs\beta}-{\bs b})^\top$. 

                Interestingly, $\bs M$ has a very particular structure as it
                belongs to the matrix subspace $\cl M$ of $\bb
                R^{k\times h}$ introduced in
                Prop.~\ref{prop:subspaceM}. Considering the orthogonal
                projector ${\cl P}_{\cl M}$ defined
                in~\eqref{eq:projM}, and since ${\cl P}_{\cl M} (\bs M)
                = \bs M \in \cl M$, we have
		\begin{equation}
			P' = \tinv{mp} \ts\sum_{i,l} \langle \bs Z^\top \ail \bs c^\top_i {\bs B}, {\cl P}_{\cl M} (\bs M) \rangle^2_F =  \tinv{mp} \ts\sum_{i,l} \langle{\cl P}_{\cl M}(\bs Z^\top \ail \bs c^\top_i {\bs B}), \bs M \rangle^2_F. 
			\label{eq:pprime}
		\end{equation}

		% Let us now expand the computation of 
		% \[
		% 	{\cl P}_{\cl M}  (\bs Z^\top \ail \bs c^\top_i {\bs B}) =  \hat{\bs z} \hat{\bs z}^\top \bs Z^\top \ail \bs c^\top_i {\bs B}(\I_{h} - \hat{{\bs b}}\hat{{\bs b}}^\top) + ({\bs B}_{i,\cdot} \hat{{\bs b}}) \bs Z^\top \ail \hat{{\bs b}}^\top
		% \]
		% where we have collected in the last summand ${\bs B}_{i,\cdot} \hat{{\bs b}} = \bs c_i^\top  {\bs B} \hat{{\bs b}} = \tfrac{g_i}{\|{\bs g}\|}$. For a lighter notation, we define
		% \[
		% 	\bs u_{i,l}\coloneqq \tfrac{g_i}{\|{\bs g}\|} \bs Z^\top \ail, \quad \bs v_{i,l} \coloneqq (\I_{h} - \hat{{\bs b}}\hat{{\bs b}}^\top) {\bs B}^\top {\bs c}_i \ail^\top \bs Z \hat{\bs z},
		% \]
		% which let us rewrite ${\cl P}_{\cl M} (\bs Z^\top \ail \bs c^\top_i {\bs B})  = {\bs u}_{i,l} \hat{{\bs b}}^\top + \hat{\bs z} {\bs v}^\top_{i,l}$; this form will be more convenient for some computations.
		
		We are now in the position of bounding $P' \geq {\bb E} P' - |P'-{\bb E} P'|$; hence, we proceed to find a lower bound for ${\bb E} P'$ and an upper bound for $|P'- {\bb E} P'|$. Firstly, note how the expectation 
		\begin{equation}
			\bb E P'  = \tinv{mp} \ts\sum_{i,l} {\bb E} (\bs c^\top_i {\bs B} \bs M^\top) \bs Z^\top \ail \ail^\top\bs Z (\bs M {\bs B}^\top \bs c_i )
						 = \tinv{m}{\rm tr}({\bs B} \bs M^\top \bs M{\bs B}^\top) = \tinv{m}\|\bs M\|^2_F
			\label{eq:expB}
		\end{equation}
		since ${\bs B}^\top {\bs B} = \I_{h}$. Moreover, recalling $\|{\bs g}\|^2 = m + \|\bs b^\perp\|^2 \geq m$ we can expand
		\begin{align}
			\|\bs M\|^2_F & = \|\bs\zeta-\bs z\|^2 \|{\bs g}\|^2 + \|\bs z\|^2 \|{\bs\beta} - {\bs b}\|^2 + 2 (\bs\zeta - \bs z)^\top \bs z {\bs b}^\top ({\bs\beta} - {\bs b}) \nonumber \\
			& \geq m(1-\rho) \|\bs\zeta-\bs z\|^2  + (1-\rho)\|\bs z\|^2  \|{\bs\beta} - {\bs b}\|^2                                                                          
			\label{eq:lbM}
		\end{align}
		where we used the fact that $\|\bs b^\perp \| = \|\bs B^\perp \bs b^\perp \| \leq \sqrt m \rho$, that the first component of ${\bs\beta} - {\bs b}$ is $0$ by construction (and therefore ${\bs b}^\top ({\bs\beta} - {\bs b})={(\bs b^\perp)}^\top ({\bs\beta}^\perp - {\bs b}^\perp)$), and that
		\[
			2 (\bs\zeta - \bs z)^\top \bs z {\bs b}^\top ({\bs\beta} - {\bs b}) \leq 2 \|\bs z\| \|{\bs b}^\perp\|\|\bs\zeta - \bs z\| \|{\bs\beta}-{\bs b}\| \leq m \rho \|\bs\zeta- \bs z\|^2 + \rho \|\bs z\|^2 \|{\bs\beta}- {\bs b}\|^2. 
		\]
		Replacing the lower bound~\eqref{eq:lbM} in~\eqref{eq:expB} yields $\bb E P' \geq (1-\rho) \Delta(\bs\zeta,\bs\beta)$. 
		\medskip

		We now have to find an upper bound for $|P'-{\bb E} P'|$. Notice that, using the linear random operator $\cl A$
                and its adjoint $\cl A^*$ introduced in~\eqref{eq:linear-gen-map} and~\eqref{eq:linear-gen-map-adjoint-def}, respectively,
                we have  
		\begin{talign*}
			&|P'-\bb E P'|\\
&\ts = \|\bs M\|^2_F \big\vert
                  \tinv{mp} \ts\sum_{i,l} \big\langle {\cl P}_{\cl M}
                  (\bs Z^\top \ail \bs c^\top_i {\bs B}), \widehat{\bs
                  M} \big\rangle^2_F - \bb E \big\langle {\cl P}_{\cl M}
                  (\bs Z^\top \ail \bs c^\top_i {\bs B}), \widehat{\bs
                  M} \big\rangle^2_F  \big\vert\\
% &\leq \tinv{m} \|\bs M\|^2_F \,\sup_{\bs W \in \cl M \cap \bb S^{n \times m}_F} \big\vert
%                   \tinv{p} \ts\sum_{i,l} \big\langle {\cl P}_{\cl M}
%                   (\bs Z^\top \ail \bs c^\top_i {\bs B}), \bs W \big\rangle^2_F - \bs c^\top_i {\bs
%                   B} \bs W^\top \bs W {\bs B}^\top
%                   \bs c_i \big\vert\\
&\leq \tinv{m} \|\bs M\|^2_F \,\sup_{\bs W \in \bb S^{n \times m}_F} \big\vert
                  \tinv{p} \ts\sum_{i,l} \big\langle {\cl P}_{\cl M}
                  (\bs Z^\top \ail \bs c^\top_i {\bs B}), \bs W \big\rangle^2_F - \bb E\big\langle {\cl P}_{\cl M}
                  (\bs Z^\top \ail \bs c^\top_i {\bs B}), \bs W \big\rangle^2_F \big\vert\\
                                     &= \tinv{m} \|\bs M\|^2_F \,\| {\cl P}_{\cl M}\cl A^*\cl A {\cl P}_{\cl M} - {\cl P}_{\cl M}\|,
                \end{talign*}
                with $\|{\cl P}_{\cl M}\cl A^*\cl A {\cl P}_{\cl
                  M} - {\cl P}_{\cl M}\|$ defined in
               ~\eqref{eq:concent-AstA} and 
                 where we used~\eqref{eq:linear-gen-map-adjoint-def} and the
                linearity of the orthogonal projector ${\cl P}_{\cl
                  M}$ on $\cl M$ in order to realise that
                \begin{talign*}
                \tinv{p}\sum_{i,l} \langle {\cl P}_{\cl M}
                (\bs Z^\top \ail \bs c^\top_i {\bs B}),\, \bs W
                \rangle^2_F&= \big\langle \big[\tinv{p}\sum_{i,l} \langle {\cl P}_{\cl M}
                  (\bs Z^\top \ail \bs c^\top_i {\bs B}),\, \bs W
                  \rangle_F {\cl P}_{\cl M}
                  (\bs Z^\top \ail \bs c^\top_i {\bs B})\big],\, \bs W
                  \big\rangle_F\\
                  &= \big\langle \big[ \tinv{p}\sum_{i,l} \langle 
                  \bs Z^\top \ail \bs c^\top_i {\bs B},\, {\cl P}_{\cl
                    M} \bs W
                  \rangle_F 
                  (\bs Z^\top \ail \bs c^\top_i {\bs B}) \big],\, {\cl P}_{\cl
                    M} \bs W
                  \big\rangle_F\\
                  &= \big\langle \cl A^* \cl A({\cl P}_{\cl
                    M} \bs W),\, {\cl P}_{\cl
                    M} \bs W \big\rangle_F = \big\langle {\cl P}_{\cl
                    M} \cl A^* \cl A {\cl P}_{\cl
                    M} \bs W,\, \bs W
                  \big\rangle_F,
                \end{talign*}
                and also $p^{-1}{\bb E} \sum_{i,l} \langle {\cl P}_{\cl M}
                (\bs Z^\top \ail \bs c^\top_i {\bs B}),\, \bs W
                \rangle^2_F = \bb E \big\langle {\cl P}_{\cl
                    M} \cl A^* \cl A {\cl P}_{\cl
                    M} \bs W,\, \bs W
                  \big\rangle_F  = \big\langle {\cl P}_{\cl
                    M} \bs W,\, \bs W
                  \big\rangle_F$ from~\eqref{eq:expect-AstA}.
                
                Therefore, since the event given by Prop.~\ref{prop:standalone-ahmed} holds by assumption, $\| {\cl P}_{\cl M}\cl A^*\cl A {\cl P}_{\cl M} - {\cl P}_{\cl M}\| \leq \delta$ and $|P'-\bb E P'| \leq \|\bs M\|^2_F \tfrac{\delta}{m}$. Thus, noting that 
		\[
			\tfrac{\|\bs M\|^2_F }{m} \leq (1+\rho)^2 \|\bs\zeta-\bs z\|^2 + \tfrac{\|\bs z\|^2}{m}\|\bs\beta -\bs b\|^2 + \tfrac{2}{\sqrt m}  \rho \|\bs z\| \|\bs\zeta-\bs z\| \|\bs\beta- \bs b\|^2\leq (1+\rho)^2\Delta(\bs\zeta, \bs\beta), 
		\]
		we conclude that $P' \geq  [(1-\rho)- \delta (1+\rho)^2] \Delta(\bs\zeta,\bs\beta)$. Plugging this result in $P$ yields
		\[
			P \geq  [(1-\rho)- \delta (1+\rho)^2 - 30 \rho] \Delta(\bs\zeta,\bs\beta) \geq (1-31\rho-4\delta)\Delta(\bs\zeta,\bs\beta),
		\]
		\ie $\langle \bs\nabla_{\bs\zeta} f^{\rm s}(\bs\zeta, {\bs\beta}), \bs\zeta - \bs z\rangle + \langle \bs\nabla^\perp_{{\bs\beta}} f^{\rm s}(\bs\zeta, \bs\beta), \bs\beta - {\bs b}\rangle \geq \eta\,\Delta(\bxi,\bs \gamma)$ with $\eta \coloneqq 1-31\rho-4\delta > 0$ if $\rho < \tfrac{1-4\delta}{31}$. 
		\medskip
		
		\noindent {\em  Lipschitz Gradient:} %%%%%%%%%%%%%%%%%%%%%%%%%%%%%%%%%%%%%%%%%%% NEW STUFF HERE (3/11/2017) 
\new{We proceed by bounding the following two quantities:
		\begin{align}
\label{eq:norm-nabla-zeta-sup}
\ts \| \bs\nabla_{\bs\zeta} f^{\rm s}(\bs\zeta, {\bs\beta})\|&\ts = \sup_{\u \in \bb S_2^{k-1}} \langle \bs\nabla_{\bs\zeta} f^{\rm s}(\bs\zeta, {\bs\beta}), \u \rangle,\\
\label{eq:norm-nabla-beta-sup}
\ts \|\bs\nabla^\perp_{{\bs\beta}} f^{\rm s}(\bs\zeta, \bs\beta)\| & \leq \ts \|\bs\nabla_{{\bs\beta}} f^{\rm s}(\bs\zeta, \bs\beta)\| = \sup_{\bs v \in \bb S_2^{h-1}} \langle \bs\nabla_{{\bs\beta}} f^{\rm s}(\bs\zeta, \bs\beta), \bs v \rangle.
		\end{align}
For the \rhs of \eqref{eq:norm-nabla-zeta-sup}, fixing $\bs u \in \bb
S_2^{k-1}$, using the developments of gradients obtained in
Sec.~\ref{sec:blind-calibration-with-subspace-priors} and \eqref{eq:gradw-bis}, and recalling (to maintain a lighter notation) that $\gamma_i = \bs c^\top_i\bs B \bs\beta$ and $g_i = \bs c^\top_i \bs B\bs b$, we observe that 
		\begin{align*}
			\langle\bs\nabla_{\bs\zeta} f^{\rm
                          s}(\bs\zeta, {\bs\beta}), \u \rangle&\ts = \tfrac{1}{mp} \ts\sum_{i,l}  \left[\gamma_i  \bs\zeta^\top\bs Z^\top \ail  - g_i \bs z^\top \bs Z^\top \ail\right]\gamma_i \ail^\top {\bs Z} \bs u  \\
&= \tinv{mp}\ts\sum_{i,l} \gamma^2_{i} (\bs\zeta - \bs z)^\top \bs Z^\top \ail \ail^\top \bs Z \bs u + \tinv{mp}\ts\sum_{i,l}\gamma_i  (\gamma_i - g_i) \bs z^\top \bs Z^\top \ail \ail^\top \bs Z \u  \nonumber                                                                             
                \end{align*}

                By rearranging its terms, the last expression can be further developed on ${\cl D}^{\rm s}_{\kappa,\rho}$ as  
		\begin{align*}
			& \ts \tinv{mp}\sum_{i,l} \gamma^2_{i} (\bs\zeta-\bs z)^\top \bs Z^\top \ail \ail^\top {\bs Z} {\bs u} + \tinv{mp}\ts\sum_{i,l}\gamma_i (\gamma_i - g_i)\,\bs z^\top  \bs Z^\top \ail \ail^\top \bs Z \u\\
&\ts \leq\ (1+\rho)^2 \|\bs\zeta-\bs z\| \,\big\|\frac{1}{mp} \sum_{i,l} \bs Z^\top \ail\ail^\top \bs Z \big\|\\
&\ts\quad + \big[\tinv{mp}\ts\sum_{i,l} (\gamma_i - g_i)^2\,\bs z^\top  \bs Z^\top \ail \ail^\top \bs Z \bs z\big]^{\frac{1}{2}}\big[\tinv{mp}\ts\sum_{i,l}\gamma^2_i \bs u^\top  \bs Z^\top \ail \ail^\top \bs Z \u\big]^{\frac{1}{2}},
		\end{align*}
where the second term has been bounded by the Cauchy-Schwarz inequality.

As we assumed that the events given by Cor.~\ref{coro:appl-wcovsub} and Prop.~\ref{prop:m-w-covconc} jointly hold, we have
\begin{equation}
  \label{eq:cor-prop-hold}
\ts \|\frac{1}{mp} \sum_{i,l} \bs Z^\top \ail\ail^\top \bs Z\| \leq 1+\delta \quad\text{and}\quad\|\frac{1}{p} \sum_{i,l} (\bs B^\top \bs c_i \bs c_i^\top\bs B)\,\bs z^\top \bs Z^\top \ail\ail^\top \bs Z \bs z\| \leq (1+\delta) \|\bs z\|^2,  
\end{equation}
so that  
\begin{eqnarray*}
&\ts  \tinv{mp}\ts\sum_{i,l}\gamma^2_i \bs u^\top  \bs Z^\top \ail \ail^\top \bs Z \u \ts \leq\ (1+\rho)^2 \|\frac{1}{mp} \sum_{i,l} \bs Z^\top \ail\ail^\top \bs Z\| \leq (1+\delta)(1+\rho)^2,\\
&\ts \tinv{mp}\ts\sum_{i,l} (\gamma_i - g_i)^2\,\bs z^\top  \bs Z^\top \ail \ail^\top \bs Z \bs z \ts \leq \tfrac{\|\bs \beta - \bs b\|^2}{m} \|\frac{1}{p} \sum_{i,l} (\bs B^\top \bs c_i \bs c_i^\top\bs B)\,\bs z^\top \bs Z^\top \ail\ail^\top \bs Z \bs z\| \leq \tfrac{(1+\delta) \|\bs \beta - \bs b\|^2 \|\bs z\|^2}{m}.
\end{eqnarray*}
Therefore, using $\max(\delta,\rho)<1$,
		\begin{align*}
                  \ts \langle\bs\nabla_{\bs\zeta} f^{\rm s}(\bs\zeta, {\bs\beta}), \u \rangle&\ts \leq (1+\rho)^2 (1+ \delta)\|\bs\zeta-\bs z\| + (1+\delta) (1+\rho)  \|\bs z\| \frac{\|\bs \beta - \bs b\|}{\sqrt m}\\
&\leq\ts 8 (\|\bs\zeta-\bs z\| + \|\bs z\| \frac{\|\bs \beta - \bs b\|}{\sqrt m}),
		\end{align*}
so that, from \eqref{eq:norm-nabla-zeta-sup}, 
\begin{equation}
  \label{eq:norm-nabla-zeta}
  \ts \|\bs\nabla_{\bs\zeta} f^{\rm s}(\bs\zeta, {\bs\beta})\|^2\ \leq\ 128\, (\|\bs\zeta-\bs z\|^2 + \|\bs z\|^2 \frac{\|\bs \beta - \bs b\|^2}{m})\ =\ 128\,\Delta(\bs \xi, \bs \gamma),
\end{equation}
\ie $L_{\bs \zeta} = 8\sqrt{2}$ in Prop.~\ref{prop:regusub}.

Concerning the \rhs of \eqref{eq:norm-nabla-beta-sup} 
we first note that
\begin{align}
\label{eq:nabla-beta-perp-tmp} \ts \|\bs\nabla_{{\bs\beta}} f^{\rm s}(\bs\zeta, \bs\beta)\| = \sup_{\bs v \in \bb S_2^{h-1}} \langle \bs\nabla_{{\bs \gamma}} f(\bxi,{\bs \gamma})\big\vert_{\bxi = {\bs Z} {\bs \zeta}, {\bs \gamma} = \bs B\bs \beta}, \bs B \bs v \rangle,
\end{align}
since $\bs\nabla_{\bs\beta}f^{\rm s}(\bs\zeta,{\bs\beta}) = \bs B^\top \bs\nabla_{{\bs \gamma}} f(\bxi,{\bs \gamma})\big\vert_{\bxi = {\bs Z} {\bs \zeta}, {\bs \gamma} = \bs B\bs \beta}$  from \eqref{eq:gradz}.

Therefore, fixing $\bs v \in \bb S_2^{h-1}$ and restricting our domains to ${\cl D}^{\rm s}_{\kappa,\rho}$, from the formulation of $\bs\nabla_{{\bs \gamma}} f(\bxi,{\bs \gamma})$ provided in Table~\ref{tab:quants}, we have
\begin{align*}
  \langle \bs\nabla^\perp_{{\bs\beta}} f^{\rm s}(\bs\zeta, \bs\beta), \bs v \rangle &\ts = \tfrac{1}{mp}\ts \sum_{i,l} (\bs \zeta^\top \bs Z^\top \ail) \big( \gamma_i \ail^\top \bs Z \bs \zeta - g_i \ail^\top \bs Z \bs z \big) {\bs c}_i^\top \bs B\bs v\\
&\ts = \tfrac{1}{mp}\ts \sum_{i,l} \gamma_i ({\bs c}_i^\top \bs B\bs v)\, \bs \zeta^\top \bs Z^\top \ail  \ail^\top \bs Z \bs \zeta\ -\ \tfrac{1}{mp}\ts \sum_{i,l} g_i ({\bs c}_i^\top \bs B\bs v)\, \bs \zeta^\top \bs Z^\top \ail \ail^\top \bs Z \bs z \\
&\ts = \tfrac{1}{mp}\ts \sum_{i,l} \gamma_i ({\bs c}_i^\top \bs B\bs v)\, (\bs \zeta - \bs z)^\top \bs Z^\top \ail  \ail^\top \bs Z \bs \zeta + \tfrac{1}{mp}\ts \sum_{i,l} (\gamma_i - g_i) ({\bs c}_i^\top \bs B\bs v)\, \bs \zeta^\top \bs Z^\top \ail \ail^\top \bs Z \bs z \\
&\ts \leq \big[\tfrac{1}{mp}\ts \sum_{i,l} \gamma^2_i (\bs \zeta - \bs z)^\top \bs Z^\top \ail  \ail^\top \bs Z (\bs \zeta - \bs z) \big]^{\frac{1}{2}}\big[\tfrac{1}{mp}\ts \sum_{i,l} ({\bs c}_i^\top \bs B\bs v)^2\, \bs \zeta^\top \bs Z^\top \ail  \ail^\top \bs Z \bs \zeta\big]^{\frac{1}{2}}\\
& + \big[\tfrac{1}{mp}\ts \sum_{i,l} (\gamma_i - g_i)^2 \, \bs z^\top \bs Z^\top \ail \ail^\top \bs Z \bs z\big]^{\frac{1}{2}}\big[\tfrac{1}{mp}\ts \sum_{i,l} ({\bs c}_i^\top \bs B\bs v)^2\, \bs \zeta^\top \bs Z^\top \ail \ail^\top \bs Z \bs \zeta\big]^{\frac{1}{2}} \\
&\ts \leq (1+\rho) \|\bs \zeta - \bs z\| \|\bs \zeta\| \big\|\tfrac{1}{mp}\ts \sum_{i,l} \bs Z^\top \ail  \ail^\top \bs Z \big\|\\
& + \big[\tfrac{1}{mp}\ts \sum_{i,l} (\bs \beta -\bs b)^\top \bs B^\top \bs c_i \bs c_i^\top \bs B (\bs \beta - \bs b) \, \bs z^\top \bs Z^\top \ail \ail^\top \bs Z \bs z\big]^{\frac{1}{2}}\big[\tfrac{1}{mp}\ts \sum_{i,l} \bs \zeta^\top \bs Z^\top \ail \ail^\top \bs Z \bs \zeta\big]^{\frac{1}{2}} \\
&\ts \leq (1+\rho) \|\bs \zeta - \bs z\| \|\bs \zeta\| \big\|\tfrac{1}{mp}\ts \sum_{i,l} \bs Z^\top \ail  \ail^\top \bs Z \big\|\\
& + \|\bs \beta -\bs b\| \|\bs \zeta\| \big\|\tfrac{1}{mp}\ts \sum_{i,l} (\bs B^\top \bs c_i \bs c_i^\top \bs B) \, \bs z^\top \bs Z^\top \ail \ail^\top \bs Z \bs z\big\|^{\frac{1}{2}}\big\|\tfrac{1}{mp}\ts \sum_{i,l} \bs Z^\top \ail \ail^\top \bs Z \big\|^{\frac{1}{2}},
\end{align*}        
where we have used two times the Cauchy-Schwarz inequality in the
first inequality, and later the fact that $|\bs c_i^\top \bs B \bs v| \leq \|\bs B \bs v\|_\infty \leq \|\bs B \bs v\| = \|\bs v\| = 1$.

Again, when the events given by Cor.~\ref{coro:appl-wcovsub} and Prop.~\ref{prop:m-w-covconc} jointly hold, from \eqref{eq:cor-prop-hold}, $\|\bs \zeta\| \leq \|\bs\zeta-\bs z\| + \|\bs z\| \leq (1 + \kappa) \|\bs z\|$ for $(\bs \zeta, \bs \beta) \in \cl D^{\rm s}_{\kappa, \rho}$ and $\max(\delta,\rho)<1$, we obtain the crude bounds
		\begin{align*}
                  \ts \langle \bs\nabla^\perp_{{\bs\beta}} f^{\rm s}(\bs\zeta, \bs\beta), \bs v \rangle &\ts \leq (1+\rho)(1+\delta)(1+\kappa) \|\bs \zeta - \bs z\| \|\bs z\| + (1+\kappa) (1+\delta) \|\bs \beta -\bs b\| \frac{\|\bs z\|^2}{\sqrt m} \\
&\ts \leq 4(1+\kappa)\|\bs z\| ( \|\bs\zeta-\bs z\| + \frac{\|\bs z\|}{\sqrt m}\,\|\bs \beta -\bs b\|).
		\end{align*}
Consequently, using \eqref{eq:nabla-beta-perp-tmp}, we find
\begin{equation}
  \label{eq:norm-nabla-beta}
  \ts \|\bs\nabla^\perp_{{\bs\beta}} f^{\rm s}(\bs\zeta, \bs\beta)\|^2\ \leq\ 32\,(1+\kappa)^2 \|\bs z\|^2 (\|\bs\zeta-\bs z\|^2 + \|\bs z\|^2 \frac{\|\bs \beta - \bs b\|^2}{m})\ =\ 32\,(1+\kappa)^2 \|\bs z\|^2\,\Delta(\bs \xi, \bs \gamma),
\end{equation}
\ie $L_{\bs \beta} = 4\sqrt 2(1+\kappa) \|\bs z\|$ in Prop.~\ref{prop:regusub}.
% Combining \eqref{eq:norm-nabla-zeta} with \eqref{eq:norm-nabla-beta}, we thus find
% \begin{equation}
%   \label{eq:norm-nabla-zeta-and-beta-weighted}
%   \ts \|\bs\nabla_{\bs\zeta} f^{\rm s}(\bs\zeta, {\bs\beta})\|^2 + \|\bs\nabla^\perp_{{\bs\beta}} f^{\rm s}(\bs\zeta, \bs\beta)\|^2\ \leq L^2 \Delta(\bs \xi, \bs \gamma),
% \end{equation}
% with $L = 8\sqrt 2 (1 + (1+\kappa) \|\bs z\|)$, and 
% \begin{equation}
%   \label{eq:norm-nabla-zeta-and-beta-weighted}
%   \ts \|\bs\nabla_{\bs\zeta} f^{\rm s}(\bs\zeta, {\bs\beta})\|^2 + \frac{1}{\|\bs z\|^2} \|\bs\nabla^\perp_{{\bs\beta}} f^{\rm s}(\bs\zeta, \bs\beta)\|^2\ \leq L'^2 \Delta(\bs \xi, \bs \gamma),
% \end{equation}
% with $L = 8\sqrt 2 (2 + \kappa)$.
}
\end{proof}

	\begin{proof}[Proof of Thm.~\ref{theorem:convergence-subs}]		
          This theorem needs that Prop.~\ref{prop:initsub} and
          Prop.~\ref{prop:regusub} hold jointly, where in this last proposition we set $\kappa = \kappa_0 =\sqrt{\delta^2 + \rho^2} < \sqrt 2$ so that
                $\eta = 1-31\rho-4\delta$, 
\new{
$$
L_{\bs \zeta} = 8\sqrt{2},\quad\text{and}\quad L_{\bs \beta} = 4\sqrt{2}\, (1+ \kappa_0) \|\bs z\| \leq 10 \sqrt{2} \|\bs z\|.
$$ 
}
By union bound and considering the requirements of both propositions, this happens with probability exceeding $1-C (mp)^{-t}$, for some $C>0$ and $t\geq1$, provided the strongest of the requirements of Prop.~\ref{prop:regusub} in $mp, n$ are met.  
Conditionally to this event, to prove the convergence of Alg.~\ref{alg2} with fixed step sizes 
$\mu_\bxi = \mu$ and $\mu_{\bs \gamma} = \mu \tfrac{m}{\|\x\|^2}$ for some $\mu > 0$ we need to show that, starting from $(\bs\zeta_0,{\bs\beta}_0)$ and at each iteration, the distance with respect to the global minimiser decreases. 

For $j = 0$ and by Prop.~\ref{prop:initsub} we start from $(\bs\zeta_0,{\bs\beta}_0) \in {\cl D}^{\rm s}_{\kappa_0,\rho}$. Moreover, for $j \in \bb N$, by recalling~\eqref{eq:updatedist} with the values $\beta_C = \eta$, \new{$\beta_{L,\bs \zeta} = L_{\bs \zeta}^2$ and $\beta_{L,\bs \beta} = L_{\bs \beta}^2$} set above 
\new{
		\begin{align*}
&{\Delta}\big(\bs\zeta_{j+1},{{\bs\beta}}_{j+1}\big)\\
&\ts \leq {\Delta}\big(\bs\zeta_{j+1},\check{{\bs\beta}}_{j+1}\big)\\
&\ts = {\Delta}(\bs\zeta_{j},{\bs\beta}_{j})\ -\ 2 \mu \big\langle \bs\nabla^\perp f^{\rm s}(\bs\zeta_{j},{\bs\beta}_{j}), \big[\begin{smallmatrix} \bs\zeta_{j}-\bs z \\ {\bs\beta}_{j}-{\bs b} \end{smallmatrix}\big]\big\rangle\ +\ \mu^2 (\|\bs\nabla_{\bs\zeta} f^{\rm s}(\bs\zeta_{j},{\bs\beta}_{j})\|^2 + \frac{m}{\|\bs z\|^2}\|\bs\nabla_{\bs\beta}^\perp f^{\rm s}(\bs\zeta_{j},{\bs\beta}_{j})\|^2) \nonumber \\
		& \leq \big(1- 2\mu\eta + \mu^2 (128 + 200 m)\big) {\Delta}(\bs\zeta_{j},{\bs\beta}_{j})\ \leq\ \big(1- 2\mu\eta + 400 \mu^2 m\big) {\Delta}(\bs\zeta_{j},{\bs\beta}_{j}),                                                                                                                            
		\end{align*}}
where in the first line we used the projection step ${\bs\beta}_{j+1} \coloneqq {\cl P}_{{\cl B}_\rho} \check{{\bs\beta}}_{j+1}$ as defined in Alg.~\ref{alg2}, and the fact that, ${\cl B}_\rho$ being a non-empty closed convex set with ${\bs b} = {\cl P}_{\cl B_\rho} {\bs b}$ by assumption, Prop.~\ref{lemma-ccs} ensures that $\|{\bs\beta}_{j+1}-{\bs b}\|\leq \|\check{{\bs\beta}}_{j+1}-{\bs b}\|$ and thus ${\Delta}\big(\bs\zeta_{j+1},{{\bs\beta}}_{j+1}\big) \leq {\Delta}\big(\bs\zeta_{j+1},\check{{\bs\beta}}_{j+1}\big)$.

From this we observe that if \new{$D \coloneqq 1- 2\mu\eta + 400 \mu^2 m \in (0,1)$} then $(\bs\zeta_{j}, \bs\beta_{j}) \in {\cl D}^{\rm
  s}_{\kappa_0,\rho}$ for all $j \in \bb N$. This occurs if $0 < \mu < \mu_0$ with
\new{
$$
\ts \mu_0 = \frac{\eta}{400\,m} \lesssim \frac{\eta}{m},
$$
}
for which
$$
D \coloneqq 1- 2\mu\eta + 400 \mu^2 m \leq 1 - \mu\eta.
$$
Thus, we find by recursion 
\begin{align*}
  {\Delta}\big(\bs\zeta_{j+1}, \bs\beta_{j+1}\big) & \leq (1 - \mu\eta)^{j+1} {\Delta}(\bs\zeta_{j},{\bs\beta}_{j}),                                                                                                                                              
\end{align*}
for all $j \in \bb N$, which yields~\eqref{eq:convres}. 
\end{proof}

	\begin{remark}
		\label{remark:subequalsnosub}
		Since it is legitimate to consider $\cl Z \coloneqq
                \bb R^n$ and $\cl B \coloneqq \bb R^m$, \ie according to the setup in Def.~\ref{def:model}, with
                all due substitutions the proofs of
                Prop.~\ref{prop:initsub} and~\ref{prop:regusub} are
                identical to the proofs of Prop.~\ref{prop:init}
                and~\ref{prop:regu}, respectively. Hence, the proof of
                Thm.~\ref{theorem:convergence} is identical to that
                of Thm.~\ref{theorem:convergence-subs}. 
	\end{remark}
	
	\begin{remark}
		\label{remark:onperp}
		The use of $\bs\nabla^\perp f(\bs\xi,\bs\gamma)$ instead of $\bs\nabla f(\bs\xi,\bs\gamma)$ in absence of subspace priors amounts to fixing the first basis vector $\bs B_{\cdot,1} = \tfrac{\vone_m}{\sqrt m}$ and the corresponding coefficient $b_1 = \beta_1 = \sqrt m$. Our proofs hold under this hypothesis, since this allows us to consider bounded variations around $\vone_m$ throughout the paper.
	\end{remark}

	\section{Proofs on the Stability of the Descent Algorithms}
	\label{sec:appD-sub}

	We now provide the proofs required to show that Alg.~\ref{alg2} is stable with respect to additive noise in the sensing model. These proofs are also valid in absence of subspace priors, \ie for Alg.~\ref{alg1}.
	
	\begin{proof}[Proof of Prop.~\ref{prop:initnoisy}]
	
		The proof is a modification of the one of Prop.~\ref{prop:initsub}. As for this one, we assume that Cor.~\ref{coro:appl-wcovsub}
		holds, which occurs with probability exceeding $1 - Ce^{-c\delta^2 mp} - (mp)^{-t}$ given $\delta \in(0,1)$ and for
		some values $C,c > 0$ depending only on $\alpha$, provided 
		$n \gtrsim t \log mp$ and $mp \gtrsim \delta^{-2}
		(k + h) \log\big(\tfrac{n}{\delta}\big)$. We start by recalling the value of $\tilde{\bs \zeta}_0 \coloneqq \tinv{mp}
		\ts\sum_l \bs Z^\top (\bs A_l)^\top \bs y_l$ in Table~\ref{tab:quantsnoise-sub}, where we see by the triangle inequality that 
		\begin{align} 
		\|\tilde{\bs \zeta}_0 - \bs z \| & \leq \big\| \tfrac{1}{mp}\ts\sum_{i,l} g_i \bs Z^\top (\ail \ail^\top-\I_n)\bs Z \bs z \big\| + \sup_{\u\in{\bb S}^{k-1}_2} \big\vert \tfrac{1}{mp}\ts\sum_{i,l} \nu_{i,l} \ail^\top \bs Z \u \big\vert.\nonumber 
		\end{align}
		
		In this case, with $\rho \in [0, 1)$, the first term above is bounded as 
		\[
		\big\| \tfrac{1}{mp}\ts\sum_{i,l} g_i \bs Z^\top (\ail \ail^\top-\I_n) \bs Z \bs z \big\| \leq
		\big\| \tfrac{1}{mp}\ts\sum_{i,l} g_i \bs Z^\top (\ail \ail^\top-\I_n) \bs Z \big\| \| \bs z \|
		\leq (1+\rho)\delta \|\bs z\| <2\delta\|\bs z\|.
		\]
		Moreover, by the Cauchy-Schwarz inequality the second term can
		be bounded as
		\[
		\ts \sup_{\u\in{\bb S}^{k-1}_2} \big\vert \tfrac{1}{mp}\ts\sum_{i,l} \nu_{i,l} \ail^\top \bs Z \u \big\vert \leq 
		\big(\tfrac{1}{mp}\ts\sum_{i,l} \nu^2_{i,l}
		\big)^{\frac{1}{2}} \sup_{\u\in{\bb
				S}^{k-1}_2}\big(\tfrac{1}{mp}\ts\sum_{i,l} (\ail^\top \bs Z \u)^2
		\big)^{\frac{1}{2}} \leq\sigma \sqrt{1+\delta} \leq \sqrt{2}\sigma.
		\] 
		where in the last inequality we have applied again Cor.~\ref{coro:appl-wcovsub} and the fact that $\bs u \in {\bb S}^{n-1}_2$. Consequently, 
		\begin{equation}
		\|\tilde{\bs \zeta}_0 - \bs z \| \leq 2\delta \|\bs z\| + \sqrt{2}\,\sigma, \label{eq:interm-noisy}
		\end{equation}
		and
                $$
                \Delta(\tilde{\bs \zeta}_0, \bs \beta_0) \leq (2\delta \|\bs z\| + \sqrt{2}\sigma)^2 + \|\bs z\|^2 \rho^2 
                \leq ( 8\delta^2 +\rho^2 + \tfrac{4\sigma^2}{\|\bs z\|^2}) \|\bs z\|^2
                $$
                The statement of this proposition is therefore proved since, from~\eqref{eq:dist_subs} and by a rescaling of $\delta$ to $\tfrac{\delta}{8}$ (which changes only the constants in the sample complexity requirements above), $(\tilde{\bs \zeta}_0, \tilde{\bs \beta}_0) \in \cl D^{\rm s}_{\tilde \kappa_0, \rho}$ with 
                $\tilde \kappa_0 \coloneqq \delta^2 + \rho^2 + \frac{4\sigma^2}{\|\bs z\|^2}$.
	\end{proof}
	
        \begin{proof}[Proof of Prop.~\ref{prop:regunoisy}]
          % \mtodo{Calling the noise-dependent terms $\bs \eta_{\bs \zeta}$ and $\bs \eta_{\bs \beta}$ is a bit confusing with the constant $\eta$ used before. I decided to use $\bs {\sf v}_.$, \ie recalling the noise $\nu_{i,l}$.}
          This proof assumes that the events given by Cor.~\ref{coro:appl-wcovsub}, Prop.~\ref{prop:m-w-covconc} and Prop.~\ref{prop:regusub} jointly hold. We can easily check that, by union bound, this happens with probability exceeding $1-C(mp)^{-t}$ for some $C>0$ and $t>1$ provided the conditions of Prop.~\ref{prop:regunoisy} on $mp$, $n$ and $\rho$ are met. To prove this proposition we start by defining the noise-dependent terms in the expressions of $\bs\nabla_{\bs \zeta} \tilde{f}^{\rm s}(\bs \zeta,\bs \beta)$ and $\bs\nabla_{\bs \beta} \tilde{f}^{\rm s}(\bs \zeta,\bs \beta)$ (see Table~\ref{tab:quantsnoise-sub}) as
		\begin{align*}
\ts		{\bs{\sf v}}_{\bs \zeta} \coloneqq \tfrac{1}{mp}\ts\sum_{i,l} \nu_{i,l} \gamma_i\, \bs Z^\top \ail,\quad {\bs{\sf v}}_{\bs \beta} \coloneqq  \tfrac{1}{mp}\ts\sum_{i,l} \nu_{i,l} (\ail^\top \bs Z \bs \zeta)\,\bs B^\top {\bs c}_i,\quad \text{and}\quad \bs {\sf v} \coloneqq \begin{pmatrix}
                     \ts {\bs{\sf v}}_{\bs \zeta}\\
                     \ts {\bs{\sf v}}_{\bs \beta}
                  \end{pmatrix}, 
		\end{align*}
		so that, with this notation,
		\begin{align}
                  \label{eq:rel-noisy-grad-grad}
                  \bs\nabla^\perp \tilde{f}^{\rm s}(\bs \zeta, \bs \beta) \coloneqq 
                  \ts \begin{pmatrix}
                    \ts \bs\nabla_{\bs \zeta} \tilde{f}^{\rm s}(\bs \zeta, \bs \beta)\\
                    \ts \bs\nabla_{\bs \beta}^\perp \tilde{f}^{\rm s}(\bs \zeta, \bs \beta)
                  \end{pmatrix}
		 = 
                  \begin{pmatrix}
                    \ts \bs\nabla_{\bs \zeta} f^{\rm s}(\bs \zeta, \bs \beta),\\
                    \ts \bs\nabla_{\bs \beta}^\perp f^{\rm s}(\bs \zeta, \bs \beta)
                  \end{pmatrix}
                  - 
                  \begin{pmatrix}
                     \ts {\bs{\sf v}}_{\bs \zeta}\\
                     \ts {\bs{\sf v}}_{\bs \beta}
                  \end{pmatrix} 
                  = \bs\nabla^\perp f^{\rm s}(\bs \zeta, \bs \beta) - \bs {\sf v},
                  \end{align}
                with $\bs\nabla_{\bs \beta}^\perp f^{\rm s}(\bs \zeta, \bs \beta) = \bs B^\top (\bs I_m - \frac{\bs 1 \bs 1^\top}{m})\bs B \bs\nabla_{\bs \beta} f^{\rm s}(\bs \zeta, \bs \beta)$.

Consequently, from~\eqref{eq:rel-noisy-grad-grad} and since the events in~\eqref{eq:reg-event-prop-reg-subspace} given by Prop.~\ref{prop:regusub} hold, we have 
\begin{subequations}
  \label{eq:tilde-grad-f-bounds}
  \begin{align}
  \big\langle \bs\nabla^\perp \tilde{f}^{\rm s}({\bs \zeta}, {\bs \beta}), \big[\begin{smallmatrix} {\bs \zeta}-{\bs z} \\ {\bs \beta}- {\bs b} \end{smallmatrix}\big] \big\rangle&\geq \eta\, {\Delta}({\bs \zeta},{\bs \beta}) - \big|\big\langle \bs{\sf v}, \big[\begin{smallmatrix} {\bs \zeta}-{\bs z} \\ {\bs \beta}- {\bs b} \end{smallmatrix}\big] \big\rangle\big|
  \end{align}\new{\vspace{-7mm}
  \begin{align}
    \|\bs\nabla_{\bs \zeta} \tilde{f}^{\rm s}({\bs \zeta}, {\bs \beta}) \|^2&\leq
  2L_{\bs \zeta}^2\ {\Delta}({\bs \zeta},{\bs \beta}) + 2\|\bs {\sf v}_{\bs \zeta}\|^2,\\
 \|\bs\nabla^\perp_{\bs \beta} \tilde{f}^{\rm s}({\bs \zeta}, {\bs
    \beta}) \|^2& \leq  \|\bs\nabla_{\bs \beta} \tilde{f}^{\rm s}({\bs
                  \zeta}, {\bs \beta}) \|^2 \leq
  2L_{\bs \beta}^2\ {\Delta}({\bs \zeta},{\bs \beta}) + 2\|\bs {\sf v}_{\bs \beta}\|^2,
\end{align}
}
\end{subequations}
\new{with $\eta = 1-31\rho-4\delta$, $L_{\bs \zeta} = 8\sqrt 2$, and $L_{\bs \beta} = 4\sqrt 2 (1+\kappa)\,\|\bs z\|$.}

Note that for any vector $\bs u \in \bb R^k$ and $\bs v \in \bb R^h$ we have, from Cor.~\ref{coro:appl-wcovsub}, Prop.~\ref{prop:m-w-covconc} and using the Cauchy-Schwarz inequality several times, 
		\begin{align}
                  \big\vert\langle {\bs{\sf v}}_{\bs \zeta},\bs u \rangle\big\vert = \big\vert\tfrac{1}{mp}\ts\sum_{i,l} \nu_{i,l} \gamma_i \ail^\top \bs Z \bs u\big\vert&\ts\leq \big(\tfrac{1}{mp}\ts\sum_{i,l} \gamma^2_i  \bs u^\top \bs Z^\top \ail \ail^\top \bs Z \bs u\big)^{\frac{1}{2}} 
		\big(\tfrac{1}{mp}\ts\sum_{i,l} \nu^2_{i,l}\big)^{\frac{1}{2}} \nonumber \\
		& \leq(1+\rho) \sigma \sqrt{1+\delta} \|\bs u\|  
		\label{eq:eta-zeta-up}
		\end{align}
                and 
		\begin{align}
		\big\vert\langle {\bs{\sf v}}_{\bs \beta}, \bs v \rangle \big\vert =  \big|\tfrac{1}{mp}\ts\sum_{i,l} \nu_{i,l} (\ail^\top \bs Z \bs \zeta)\,\bs v^\top\bs B^\top {\bs c}_i\big|&\ts = \big\vert\tfrac{1}{mp}\ts\sum_{i,l} \nu_{i,l} (\bs v^\top\bs B^\top {\bs c}_i) \ail^\top \bs Z \bs \zeta \big\vert\nonumber\\
&\ts \leq \tfrac{1}{\sqrt{mp}} \|\bs \nu\| \big[\tfrac{1}{mp}\sum_{i,l} (\bs v^\top\bs B^\top {\bs c}_i)^2 \bs \zeta^\top \bs Z^\top \ail \ail^\top \bs Z \bs \zeta\big]^{\frac{1}{2}} \nonumber\\
&\ts \leq \sigma \sqrt{1+\delta} \|\bs B \bs v\|_{\infty }\|\bs \zeta\|. \label{eq:eta-beta-up}
		\end{align}
Therefore, assigning $\bs u = {\bs \zeta}-{\bs z}$ and $\bs v = \bs \beta - \bs b$ with $(\bs \zeta,\bs \beta) \in \cl D^{\rm s}_{\kappa,\rho}$, we find that
\begin{align}
\ts \big|\big\langle \bs{\sf v}, \big[\begin{smallmatrix} {\bs \zeta}-{\bs z} \\ {\bs \beta}- {\bs b} \end{smallmatrix}\big] \big\rangle\big|
&\ts \leq (1+\rho) \sigma \sqrt{1+\delta} \|{\bs \zeta}-{\bs z}\| + \sigma \sqrt{1+\delta} \|\bs B ({\bs \beta}- {\bs b})\|_{\infty }\|\bs \zeta\|\nonumber\\
&\ts \leq \sqrt{2}\sigma(2 \kappa \|{\bs z}\| + \|\bs \zeta\|)\nonumber\\
&\ts \leq \sqrt{2}\sigma(1+ 3\kappa) \|{\bs z}\|\label{eq:bound-scp-nu},
\end{align}
where we used $\|\bs \zeta\| \leq \|\bs \zeta -\bs z\| + \|\bs z\|\leq (1+\kappa)\|\bs z\|$ and $\max(\delta,\rho)<1$.

Concerning the \new{norms of $\bs{\sf v}_{\bs \zeta}$ and $\bs{\sf
    v}_{\bs \beta}$}, they can be found by
bounding~\eqref{eq:eta-zeta-up} and~\eqref{eq:eta-beta-up} for any
$\bs u \in \bb S^{k-1}$ and $\bs v \in \bb S^{h-1}$, that is
\new{$\|\bs {\sf v}_{\bs \zeta}\| \leq (1+\rho) \sigma \sqrt{1+\delta}
  \leq 2\sqrt 2\sigma$ and $\|\bs {\sf v}_{\bs\beta}\| \leq \sigma \sqrt{1+\delta} \|\bs B \bs v\|_{\infty }\|\bs \zeta\| \leq \sigma \sqrt{1+\delta} \|\bs \zeta\| \leq \sqrt 2 \sigma (1+ \kappa) \|\bs z\|$.} Therefore, from~\eqref{eq:tilde-grad-f-bounds} and \eqref{eq:bound-scp-nu}, we find 
  \begin{align*}
  \big\langle \bs\nabla^\perp \tilde{f}^{\rm s}({\bs \zeta}, {\bs \beta}), \big[\begin{smallmatrix} {\bs \zeta}-{\bs z} \\ {\bs \beta}- {\bs b} \end{smallmatrix}\big] \big\rangle&\geq \eta\, {\Delta}({\bs \zeta},{\bs \beta}) - o_C\,\sigma
  \end{align*}\new{\vspace{-7mm}
\begin{align*}
  \|\bs\nabla_{\bs \zeta} \tilde{f}^{\rm s}({\bs \zeta}, {\bs \beta}) \|^2&\leq
  2 L_{\bs \zeta}^2\ {\Delta}({\bs \zeta},{\bs \beta}) + o_{L,\bs \zeta}\,\sigma^2,\\
  \|\bs\nabla^\perp_{\bs \beta} \tilde{f}^{\rm s}({\bs \zeta}, {\bs \beta}) \|^2&\leq
  2 L_{\bs \beta}^2\ {\Delta}({\bs \zeta},{\bs \beta}) + o_{L,\bs \beta}\,\sigma^2,
\end{align*}}
with $o_C \coloneqq \sqrt{2}(1+ 3\kappa) \|{\bs z}\|$, \new{$o_{L,\bs \zeta} := 16$, and $o_{L,\bs \beta} \coloneqq 4\,(1+\kappa)^2\,\|\bs z\|^2$.} 
\end{proof}

\begin{proof}[Proof of Thm.~\ref{theorem:stability}]
Given $\delta >0$, $\rho \in (0,1)$, $\kappa > 0$ and $t \geq 1$, we assume in this
proof that the events given by Prop.~\ref{prop:initnoisy} and
Prop.~\ref{prop:regunoisy} jointly hold for the values $\tilde{\kappa}_0^2 \coloneqq \delta^2 + \rho^2 + \frac{4\sigma^2}{\|\bs z\|^2}$, $\eta := 1-31\rho-4\delta$, \new{$L_{\bs \zeta} := 8\sqrt 2$, $L_{\bs \beta} := 4\sqrt 2 (1+\kappa)\,\|\bs z\|$, $o_C \coloneqq \sqrt{2}(1+ 3\kappa) \|{\bs z}\|$, $o_{L,\bs \zeta} := 16$, and $o_{L,\bs \beta} \coloneqq 4\,(1+\kappa)^2\,\|\bs z\|^2$} that were set in these propositions. Note that the precise value of $\kappa$ will be set later. 

Let the step sizes in Alg.~\ref{alg2} be fixed to 
$\mu_{\bs \zeta} = \mu$ and $\mu_{\bs \beta} = \mu \tfrac{m}{\|\bs
  z\|^2}$, for some $\mu > 0$.
Starting from $(\tilde{\bs \zeta}_0,\tilde{\bs \beta}_0)$ in Table~\ref{tab:quantsnoise-sub} for which $\Delta(\tilde{\bs \zeta}_0,\tilde{\bs \beta}_0) \leq \tilde \kappa_0^2 \|\bs x\|^2$ by assumption, recalling~\eqref{eq:updatedist} with $\beta_C = \eta$, \new{$\beta_{L,\bs \zeta} = 2L_{\bs \zeta}^2$ and $\beta_{L,\bs \beta} = 2L_{\bs \zeta}^2$}, using 
\eqref{eq:bounded-curvature-noise}, \new{\eqref{eq:bounded-lipschitz-noise-zeta} and \eqref{eq:bounded-lipschitz-noise-beta}}, we have that at iteration $J \in \bb N$ and assuming for now that $\kappa$ is large enough to have $(\bs\zeta_{j}, \bs\beta_{j}) \in \cl D^{\rm s}_{\kappa,\rho}$ for~$j\in \{0,\cdots,J-1\}$,
\new{		
\begin{talign*}
{\Delta} (\bs\zeta_{J}, \bs\beta_{J}) &\leq {\Delta} (\bs\zeta_{J},\check{{\bs\beta}}_{J})\\
	& \leq \Delta(\bs\zeta_{J-1},{\bs\beta}_{J-1})\\
&\quad - 2 \mu \big(\langle \bs\nabla_{\bs\zeta} \tilde{f}^{\rm s}(\bs\zeta_{J-1},{\bs\beta}_{J-1}), \bs\zeta_{J-1}-\bs z\rangle + \langle \bs\nabla^\perp_{\bs\beta} \tilde{f}^{\rm s}(\bs\zeta_{J-1},\bs\beta_{J-1}), {\bs\beta_{J-1}}-{\bs g}\rangle\big) \nonumber                                                                   \\
	& \quad + \mu^2 \big(\|\bs\nabla_{\bs\zeta} \tilde{f}^{\rm s}(\bs\zeta_{J-1},\bs\beta_{J-1})\|^2 + \frac{m}{\|\bs z\|^2} \|\bs\nabla^\perp_{\bs\beta} \tilde{f}^{\rm s}(\bs\zeta_{J-1},\bs\beta_{J-1})\|^2\big) \nonumber\\
	&\ts \leq \Delta(\bs\zeta_{J-1},{\bs\beta}_{J-1}) -\ 2\mu \eta\Delta(\bs \zeta_{J-1},{\bs \beta}_{J-1})\ +\
           2\mu o_C\sigma\\ 
&\ts \qquad + 2 \mu^2 (L_{\bs \zeta}^2 + L_{\bs \beta}^2 \frac{m}{\|\bs z\|^2}) \Delta(\bs \zeta_{J-1},{\bs \beta}_{J-1}) + \mu^2 (o_{L,\bs \zeta} + \frac{m}{\|\bs z\|^2} o_{L,\bs \beta}) \sigma^2\\
	&\leq D \Delta(\bs\zeta_{J-1},{\bs\beta}_{J-1}) + R\\ 
	&\leq D^{J} \tilde \kappa_0^2 \|\bs z\|^2 + (\sum_{j=0}^{J-1} D^j) R\\
	&\leq D^{J} \tilde \kappa_0^2 \|\bs z\|^2 + \tfrac{1}{1-D} R,
	\end{talign*}where $R:=2\sqrt{2}\mu (1+3\kappa)\|\bs z\|\,\sigma + \mu^2 (16 + 4 m (1+ \kappa)^2) \sigma^2$, $D \coloneqq 1-2\eta \mu + 800 \mu^2m$} and where the first line uses the fact that the projection on $\cl B_{\rho}$ defining $\bs \beta_{j+1} = {\cl P}_{\cl B_\rho} \check{\bs \beta}_{j+1}$ in Alg.~\ref{alg2} is a contraction (see Prop.~\ref{lemma-ccs}).

Thus, we see that a sufficient condition for ${\Delta} (\bs\zeta_{J}, \bs\beta_{J})$ to be bounded for any value of $J$ is to impose $D \in (0,1)$. \new{This occurs if
$0 < \mu < \mu_0$ with
$$
\ts \mu_0\ :=\ \frac{1}{800\, m\, \max(1,\frac{\sigma^2}{\|\bs z\|^2})}\eta\ \lesssim\ \frac{\eta}{m}\, \min(1,\frac{\|\bs z\|^2}{\sigma^2}),
$$}%
in which case $D < 1 -\eta\mu$ and, from the bound on $\mu$, 
\new{
\begin{talign*}
{\Delta} (\bs\zeta_{J}, \bs\beta_{J})&
\ts \leq (1 -\eta\mu)^{J} \tilde \kappa_0^2 \|\bs z\|^2 + \tfrac{1}{\eta\mu} (2\sqrt{2}\mu (1+3\kappa)\|\bs z\|\,\sigma + \mu^2 (32 + 4 m (1+ \kappa)^2) \sigma^2)\\
&\ts \leq (1 -\eta\mu)^{J} \tilde \kappa_0^2 \|\bs z\|^2 +\frac{2\sqrt{2}}{\eta} (1+3\kappa)\|\bs z\|\,\sigma + \frac{1}{800 m}\, (32 + 4 m (1+ \kappa)^2) \sigma^2 \min(1,\frac{\|\bs z\|^2}{\sigma^2})\\
&\ts \leq (1 -\eta\mu)^{J} \tilde \kappa_0^2 \|\bs z\|^2 + \frac{2\sqrt{2}}{\eta} (1+3\kappa)\|\bs z\|\,\sigma + (\frac{1}{25} + \frac{1}{200} (1+ \kappa)^2) \sigma^2 \min(1,\frac{\|\bs z\|^2}{\sigma^2})\\
&\ts \leq (1 -\eta\mu)^{J} \tilde \kappa_0^2 \|\bs z\|^2 + \frac{2\sqrt{2}}{\eta} (1+3\kappa)\|\bs z\|\,\sigma + \frac{1}{5} (1+ \kappa^2) \sigma^2 \min(1,\frac{\|\bs z\|^2}{\sigma^2})\\
&\ts \leq (1 -\eta\mu)^{J} \tilde \kappa_0^2 \|\bs z\|^2 +
\frac{2\sqrt{2}}{\eta} \|\bs z\|\,\sigma + \frac{1}{5} \sigma^2 + 
\frac{6\sqrt{2}}{\eta} \|\bs z\|\,\sigma\, \kappa + \frac{1}{5} \sigma^2\, \kappa^2 \min(1,\frac{\|\bs z\|^2}{\sigma^2}).
\end{talign*}}		
% However, from the bound on $\mu$ and the values of $o_C$ and $o_L$,
% \begin{talign*}
% \frac{2m\mu}{\min(1, \|\bs z\|^2)}
% o_L \sigma^2 + 2o_C\sigma&
% \leq \frac{\min(1, \|\bs z\|^2)}{512 (1+ 2(1+\kappa^2)\|\bs z\|^2) \min(1, \|\bs z\|^2)}\,
% 8(2  + (1+\kappa^2)\|\bs z\|^2) \eta \sigma^2 + 2\sqrt{2}(1+ 3\kappa) \|{\bs z}\|\sigma\\
% &\leq \frac{1}{32}\,
% \eta \sigma^2 + 2\sqrt{2}(1+ 3\kappa) \|{\bs z}\|\sigma,
% \end{talign*}
% so that
% \begin{talign*}
% {\Delta} (\bs\zeta_{J}, \bs\beta_{J}) &\leq (1 - \eta\mu)^{J} \tilde \kappa_0^2 \|\bs z\|^2 + \tfrac{2\sqrt{2}}{\eta}(1+ 3\kappa) \|{\bs z}\|\sigma  + \frac{1}{32}\, \sigma^2.
% \end{talign*}		
Coming back to the initial assumption that $\kappa$ is so that $\forall j \in [J]$, $(\bs\zeta_{J}, \bs\beta_{J}) \in \cl D^{\rm s}_{\kappa,\rho}$, if this must hold for all $J$, then the last bound must be smaller than $\kappa^2 \|\bs z\|^2$, \ie we must solve with respect to $\kappa$ the quadratic inequality
\new{$$
\ts 2  +
\frac{2\sqrt{2}}{\eta} \frac{\sigma}{\|\bs z\|} + \frac{21\sigma^2}{5\|\bs
  z\|^2} + 
\frac{6\sqrt{2}}{\eta} \frac{\sigma}{\|\bs z\|} \kappa + \frac{1}{5} \kappa^2 \leq \kappa^2.
$$
where we have used $\tilde{\kappa}_0^2 \leq 2 + \frac{4\sigma^2}{\|\bs
  z\|^2}$. 
} Since for $\alpha,\beta > 0$ we have $x^2 \geq \alpha + \beta x$ if $x \geq \alpha + \beta$, a few computations show that this is satisfied by taking
\new{$$
\ts \kappa \coloneqq  3  +
\frac{10\sqrt{2}}{\eta} \frac{\sigma}{\|\bs z\|} + \frac{6\sigma^2}{\|\bs
  z\|^2}.
$$
This shows that if $\sigma \lesssim \|\bs z\|$, then $\kappa \lesssim 1$.

}
Consequently, \new{under this condition}, coming back to the equation above and replacing $\tilde
\kappa_0$ by its value we find 
\new{
\begin{align*}
{\Delta} (\bs\zeta_{J}, \bs\beta_{J}) &\ts \leq (1 -\eta\mu)^{J} \big(\delta^2 +\rho^2 + \tfrac{4\sigma^2}{\|\bs z\|^2}\big) \|\bs z\|^2 +
\frac{2\sqrt{2}}{\eta} \|\bs z\|\,\sigma + \frac{1}{5} \sigma^2 + 
\frac{6\sqrt{2}}{\eta} \|\bs z\|\,\sigma\, \kappa + \frac{1}{5} \sigma^2\, \kappa^2 \min(1,\frac{\|\bs z\|^2}{\sigma^2})\\
&\ts \lesssim (1 -\eta\mu)^{J} \|\bs z\|^2 + \frac{1}{\eta}\,\sigma \|\bs z\|,
% &----\\
% &\ts \leq (1 -\eta\mu)^{J} \tilde \kappa_0^2 \|\bs z\|^2 + \frac{4}{5} \kappa^2 \|\bs z\|^2 + \frac{1}{5} \sigma^2\, \kappa^2 \min(1,\frac{\|\bs z\|^2}{\sigma^2})\\
% &\ts \leq (1 -\eta\mu)^{J} \tilde \kappa_0^2 \|\bs z\|^2 + \kappa^2 \|\bs z\|^2\\
% &\ts \leq (1 -\eta\mu)^{J} \tilde \kappa_0^2 \|\bs z\|^2 + (3  +
% \frac{10\sqrt{2}}{\eta} \frac{\sigma}{\|\bs z\|} + \frac{6\sigma^2}{\|\bs
%   z\|^2})^2\,\|\bs z\|^2\\
% &----\\
% &\ts \leq (1 -\eta\mu)^{J}
% \tilde \kappa_0^2 \|\bs z\|^2 + \tfrac{2\sqrt{2}}{\eta}\big(7 +
% \frac{24\sqrt{2}\sigma}{\eta\|\bs z\|} + \frac{15\sigma^2}{\|\bs
%   z\|^2}\big) \|{\bs z}\|\sigma  + \frac{1}{32}\, \sigma^2\\
% &\ts \leq (1 -\eta\mu)^{J} \big(\delta^2 +\rho^2 + \tfrac{4\sigma^2}{\|\bs z\|^2}\big) \|\bs z\|^2 +
% \big(\tfrac{25}{\eta} \tfrac{\sigma}{\|\bs z\|}+
% \big(\frac{100}{\eta^2} + \frac{1}{32}\big) \frac{\sigma^2}{\|\bs z\|^2} + \tfrac{43}{\eta}\tfrac{\sigma^3}{\|\bs z\|^3}\big) \|\bs z\|^2.
\end{align*}
which finally proves that, under the previous conditions on $\sigma$ and $\mu$, 
$$
\lim_{J \to + \infty} {\Delta} (\bs\zeta_{J}, \bs\beta_{J})\ \lesssim\ \tinv{\eta} \sigma \|\bs z\|.
$$
}
% \noindent Therefore, if $\sigma \lesssim \|\bs z\|$, we have 
% $\kappa \lesssim 1$ which involves 
% $$
% \ts \mu_0 \lesssim \tfrac{1}{m}\,\frac{\min(1, \|\bs z\|^2)}{(1+ \|\bs z\|^2)}\,\eta.
% $$
% and
% $$
% \lim_{J \to + \infty} {\Delta} (\bs\zeta_{J}, \bs\beta_{J})\ \lesssim\ \tinv{\eta^2} \sigma \|\bs z\|,
% $$
% which concludes the proof.
\end{proof}

%%% Local Variables:
%%% mode: latex
%%% TeX-master: "iai_blindcalibration_2016"
%%% End:

\footnotesize
\bibliographystyle{IEEEtran}
\bibliography{iaibiblio}

\end{document}

%%% Local Variables:
%%% TeX-command-extra-options: "-shell-escape"
%%% mode: latex
%%% TeX-master: t
%%% End: